\documentclass[12pt]{article}
\pdfoutput=1

\newif\ifTITsubmission
\makeatletter
\@ifclassloaded{IEEEtran}{\TITsubmissiontrue}{\TITsubmissionfalse}
\makeatother

\ifTITsubmission\else
\usepackage{fullpage}
\fi
\usepackage{amssymb}
\usepackage{amsmath}
\usepackage{bm}
\usepackage{amsthm}
\usepackage{dsfont}
\ifTITsubmission\else
\usepackage{libertine}
\fi
\usepackage{thmtools}
\makeatletter
\@ifundefined{newcounteralias}{}{%
  \renewcommand\thmt@autorefsetup{%
    \@xa\def\csname\thmt@envname autorefname\@xa\endcsname\@xa{\thmt@thmname}}%
}
\makeatother
\usepackage{mathtools}
\usepackage{mathrsfs}
\usepackage{booktabs}
\usepackage{array}
\usepackage{tikz}
\usepackage{float}
\usepackage{hyperref}
\usepackage[capitalise]{cleveref}
\usetikzlibrary{positioning,arrows.meta,calc,shapes.geometric}

\AtBeginDocument{%
  \setlength{\emergencystretch}{.5em}
  \setlength{\abovedisplayskip}{9pt plus 2pt minus 4pt}%
  \setlength{\belowdisplayskip}{9pt plus 2pt minus 4pt}%
  \setlength{\abovedisplayshortskip}{0pt plus 2pt}%
  \setlength{\belowdisplayshortskip}{5pt plus 2pt minus 2pt}%
}

\ifTITsubmission\else
\makeatletter
\renewcommand{\section}{\@startsection {section}{1}{\z@}
             {-3.5ex \@plus -1ex \@minus -.2ex}
             {2.3ex \@plus .2ex}
             {\normalfont\Large\scshape\bfseries}}
\renewcommand{\subsection}{\@startsection{subsection}{2}{\z@}
             {-3.25ex\@plus -1ex \@minus -.2ex}
             {1.5ex \@plus .2ex}
             {\normalfont\large\scshape\bfseries}}
\renewcommand{\subsubsection}{\@startsection{subsubsection}{3}{\z@}
             {-3.25ex\@plus -1ex \@minus -.2ex}
             {1.5ex \@plus .2ex}
             {\normalfont\normalsize\scshape\bfseries}}
\renewcommand{\paragraph}{\@startsection{paragraph}{4}{\z@}
             {-2.5ex\@plus -.5ex \@minus -.2ex}
             {.75ex \@plus .15ex}
             {\normalfont\large\sffamily\bfseries\color{black!45!black}}}
\makeatother
\fi

\newtheorem{theorem}{Theorem}[section]

\newtheorem{lemma}[theorem]{Lemma}
\newtheorem{prop}[theorem]{Proposition}
\newtheorem{cor}[theorem]{Corollary}

\theoremstyle{definition}
\newtheorem{definition}[theorem]{Definition}

\newtheorem{remark}[theorem]{Remark}
\newtheorem{fact}[theorem]{Fact}
\newtheorem{example}[theorem]{Example}

\crefname{prop}{Proposition}{Propositions}
\Crefname{prop}{Proposition}{Propositions}
\crefname{cor}{Corollary}{Corollaries}
\Crefname{cor}{Corollary}{Corollaries}
\crefname{fact}{Fact}{Facts}
\Crefname{fact}{Fact}{Facts}



\def\K{\mathcal{K}}

\def\Q{\mathcal{Q}}

\def\U{\mathcal{U}}
\def\V{\mathcal{V}}
\def\W{\mathcal{W}}

\newcommand {\set} [1] {\ensuremath{ \left\lbrace #1 \right\rbrace }}

\newcommand{\normthree}[1]{{\left\vert\kern-0.25ex\left\vert\kern-0.25ex\left\vert #1 \right\vert\kern-0.25ex\right\vert\kern-0.25ex\right\vert}}
\newcommand {\br} [1] {\ensuremath{ \left( #1 \right) }}
\newcommand {\Br} [1] {\ensuremath{ \left[ #1 \right] }}

\newcommand {\norm} [1] {\ensuremath{ \left\| #1 \right\| }}

\newcommand {\abs} [1] {\ensuremath{ \left| #1 \right| }}
\newcommand{\grayhighlight}[1]{%
  #1
}

\newcommand {\bra} [1] {\ensuremath{ \left\langle #1 \right| }}
\newcommand {\ket} [1] {\ensuremath{ \left| #1 \right\rangle }}
\newcommand {\ketbratwo} [2] {\ensuremath{ \left| #1 \middle\rangle \middle\langle #2 \right| }}
\newcommand {\ketbra} [1] {\ketbratwo{#1}{#1}}

\DeclareMathOperator{\RelEnt}{D}

\newcommand {\Tr} {\ensuremath{ \mathrm{Tr} }}

\newcommand {\id} {\ensuremath{\mathds{1}}}

\newcommand{\iu}{\mathrm{i}}
\newcommand{\pauliX}{\mathsf{X}}
\newcommand{\pauliY}{\mathsf{Y}}
\newcommand{\pauliZ}{\mathsf{Z}}

\newcommand{\Matrix} {\ensuremath{\mathcal{M}}}

\newcommand {\Ent} {\ensuremath{\mathrm{Ent}}}

\newcommand{\MIP}{\mathsf{MIP}}
\newcommand{\NEXP}{\mathsf{NEXP}}
\newcommand{\RE}{\mathsf{RE}}
\newcommand{\poly}{\mathrm{poly}}

\newcommand{\mobius}{Möbius}

\newcommand{\sfG}{\mathsf{G}}
\DeclareMathOperator{\Corr}{Corr}
\DeclareMathOperator{\Val}{Val}

\DeclareMathOperator{\supp}{supp}

\newcommand{\cL}{\mathcal L}

\newcommand{\cE}{\mathcal E}
\newcommand{\cG}{\mathcal G}
\newcommand{\bll}{\boldsymbol\lambda}
\newcommand{\dd}{\,\mathrm d}

\newcommand{\one}{\mathbf 1}
\newcommand{\bfS}{\mathbf{S}}
\newcommand{\bfr}{\mathbf{r}}

\newcommand{\wh}{\widehat}
\newcommand{\wt}{\widetilde}
\newcommand{\tA}{t_A}
\newcommand{\tB}{t_B}

\newcommand{\namednote}[3]
{\ThereShouldBeZeroNotes\textcolor{#2}{\textbf{#1:} #3}}

\newcommand{\znote}[1]
{\namednote{Haigang Zhou}{red}{#1}}

\newcommand{\pnote}[1] 
{\textcolor{blue}{\textbf{Ynote:} #1}}

\newcommand{\jnote}[1]
{\textcolor{brown}{\textbf{Jnote:}#1}}

\usepackage[english]{babel}
\usepackage{graphicx}

\title{A Dichotomy for $\mathsf{MIP}^*$ in the Presence of Unital Noise}

\author{
Yangjing Dong\textsuperscript{$*,\dagger$}
\and Zhenyu Jiang\textsuperscript{$*$}
\and Minglong Qin\textsuperscript{$\ddagger$}
\and Penghui Yao\textsuperscript{$*,\mathsection$}
\and Haigang Zhou\textsuperscript{$*$}
}

\date{}

\begin{document}
\maketitle

\begingroup
\renewcommand{\thefootnote}{\fnsymbol{footnote}}

\footnotetext[1]{\scriptsize
State Key Laboratory of Novel Software Technology,
Nanjing University, Nanjing 210023, China. Emails: Yangjing Dong (dongmassimo@gmail.com);
Zhenyu Jiang (zhenyujiang@smail.nju.edu.cn);
Penghui Yao (phyao1985@gmail.com);
Haigang Zhou (hgzhou2003@outlook.com).
}

\footnotetext[2]{\scriptsize
Aalto University, Finland.
}

\footnotetext[3]{\scriptsize
Centre for Quantum Technologies,
National University of Singapore,
Singapore 117543, Singapore.
Email: mlqin6@gmail.com.
}

\footnotetext[4]{\scriptsize
Hefei National Laboratory,
Hefei 230088, China.
}

\endgroup

\maketitle

\begin{abstract}
We study two-prover one-round quantum interactive proofs in which the provers share
arbitrarily many copies of a noisy EPR pair~\cite{QY21,QY23}. Prior work~\cite{DongEtAl24} showed
that if the shared state $\psi$ has quantum maximal correlation $\rho_{\max}(\psi) < 1$, and the answer size is constant, then the computational advantage of the system vanishes, namely
then $\mathrm{MIP}^{\psi}[\mathrm{poly}, O(1)] = \mathrm{MIP} = \mathrm{NEXP}$, while
noiseless EPR pairs give $\mathrm{MIP}^{\mathrm{EPR}}[\mathrm{poly}, O(1)] = \mathrm{RE}$. That result leaves two cases open: answers of polynomial length, and noise that keeps
$\rho_{\max}(\psi) = 1$, which unital noise can do. We settle both for unital noise by giving a dichotomy for $\mathsf{MIP}^*$ in the presence of unital qubit noise.
Equivalently, let $\psi$ be any two-qubit state with maximally mixed marginals. We show that

\[\mathsf{MIP}^{\psi}[\poly,\poly]
=\mathsf{MIP}
=\mathsf{NEXP},~\text{if $\psi$ is separable or $\rho_{\max}(\psi)<1$};\]
\[\mathsf{MIP}^{\psi}[\poly,\poly]
=\mathsf{RE},
\text{if $\psi$ is entangled and $\rho_{\max}(\psi)=1$.}\]

Beyond this complexity-theoretic dichotomy, our proof introduces two techniques of potential independent interest.  The first technique approximates quantum measurements by
positive low-degree operators. We prove that the square roots of the measurement elements, after applying the smoothing operation, can be approximated by low-degree operators, with error controlled in the Schatten 4-norm. This further gives a positive low-degree approximation with small normalization error. Second, we introduce \emph{Pauli folding}, a randomized dimension-reduction method for low-degree quantum operators. For fixed accuracy, the number of qubits after
folding is polynomial in the Pauli degree and the number of outcomes and does not depend
on $n$. Moreover, Pauli folding preserves Hilbert–Schmidt inner products in expectation and is approximately multiplicative. This is an exponential improvement over the dimension reductions via invariance principles in all previous works~\cite{QY21,QY23,DongEtAl24}.

\end{abstract}

\setcounter{tocdepth}{2}
\tableofcontents


\section{Introduction}
\label{sec:introduction}

Quantum entanglement changes the power of multiprover interactive proofs dramatically. After a long line of work~\cite{cleve2004consequences,kobayashi2003quantum,ito2012multi,natarajan2019neexp},  Ji, Natarajan, Vidick, Wright,
and Yuen proved the landmark result $\MIP^*=\RE$~\cite{JNVWY21,JNVWY22}, where $\MIP^*$ is the class of languages decided by quantum multiprover interactive proof systems and $\RE$ is the class of recursively enumerable
languages, while the classical
analogue $\MIP$ equals $\NEXP$, the class of languages decidable in nondeterministic
exponential time~\cite{BFL91,FL92}.
The equality $\MIP^*=\RE$ already holds for
the systems with two provers and one round of interaction, where the provers are restricted to share
an unbounded number of {\em perfect} EPR pairs. 

The perfectness of EPR pairs seems brittle. It is worth asking  what happens to the computational advantages of 
$\MIP^*$ when the entanglement is noisy. Qin and
Yao~\cite{QY21} initiated the study of quantum multiprover interactive proof systems, where  the provers share arbitrarily many copies
of a fixed noisy bipartite state $\psi$, independent of the input.
They use {\em quantum maximal correlation}, a measure of correlation introduced by Beigi~\cite{Beigi13}, generalizing
classical maximal correlation~\cite{Renyi1959}, and focused on \emph{noisy maximally entangled states}, namely
states with maximally mixed marginals and quantum maximal
correlation $\rho_{\max}(\psi)<1$.
For such states they gave a computable bound on the number of
copies of $\psi$ needed to approximate the optimal acceptance probability to any desired
accuracy. It follows that these proof systems recognize only decidable
languages~\cite{QY21}.
They subsequently extended this decidability result to the setting in which the verifier sends quantum questions and receives quantum
answers~\cite{QY23}. These results establish decidability but leave open whether the
computational power still exceeds the classical class $\NEXP$.

A sharper complexity-theoretic characterization was obtained by
Dong, Fu, Natarajan, Qin, Xu, and Yao~\cite{DongEtAl24}. They studied
two-prover one-round $\MIP$ in which the provers share
arbitrarily many copies of a fixed noisy maximally entangled state
$\psi$ and the answers have constant length. They proved
\[
    \MIP^{\psi}[\poly,O(1)]
       = \NEXP \quad\mbox{and}\quad \MIP^{\mathrm{EPR}}[\poly,O(1)]=\RE,
\]
where $\MIP^\psi[q,a]$ denotes the two-prover one-round proof class
using arbitrarily many copies of $\psi$,
with question and answer lengths bounded by $q$ and $a$ bits,
respectively.
Thus, in the constant-answer regime, the quantum advantages vanish in the presence of noise.

Two questions remain open. First, {\em would this vanishing also happen when the answers have polynomial length?} The $\NEXP$
characterization of~\cite{DongEtAl24} assumes that the provers return
only a constant number of bits. 
Second, {\em will all noise kill the computational advantages of~$\MIP^*$?}  It is worth noticing that all of the preceding noisy-state results assume
that the quantum maximal correlation satisfies
$\rho_{\max}(\psi)<1$. This
assumption excludes some natural noisy EPR states: unital qubit noise such as dephasing
can produce mixed states with maximally mixed marginals and $\rho_{\max}(\psi)=1$. For
these states, maximal correlation one does not tell us whether the shared state is
entangled, since even a shared classical random bit has maximal correlation
one~\cite{Beigi13}.


In this work, we resolve both questions for two-qubit resources with
maximally mixed marginals. Equivalently, we obtain a complete
classification $\MIP^{\psi}[\poly,\poly]$, where $\psi$ is an EPR pair subjected to unital qubit noise. Our
first result extends the $\NEXP$ characterization in~\cite{DongEtAl24} from constant-length answers to polynomial-length
answers. More generally, we prove that if $\psi$ is a two-qubit state
with maximally mixed marginals and either $\psi$ is separable or
$\rho_{\max}(\psi)<1,$
then
\[
    \MIP^{\psi}[\poly,\poly]
       = \NEXP.
\]
Thus neither polynomially long answers nor arbitrary classical
correlations increase the computational power beyond $\NEXP$ in
these cases.

The complementary regime behaves completely differently. We prove
that if $\psi$ is entangled and
$\rho_{\max}(\psi)=1,$
then the players recover the full computational power of perfect
entanglement:
\[
    \MIP^{\psi}[\poly,\poly]=\MIP^{\psi}[\poly,O(1)]
       = \RE.
\]
In particular, perfect EPR pairs are not the only resources capable
of realizing the full power of $\MIP^*$: certain mixed noisy states
do so as well. Combining the two directions gives a complete
dichotomy for two-qubit states with maximally mixed marginals.
Within this family, the computational power of the shared resource
is determined jointly by its maximal correlation and whether it is
entangled. We state the precise statements next. 




\subsection{Main Result}




Throughout, let $\psi$ be a fixed two-qubit state with maximally
mixed marginals, and assume that its entries are computable in polynomial time: given $b\in\mathbb N$, a rational approximation
to each entry with error at most $2^{-b}$ can be computed in time
$\poly(b)$.
The two provers may share
arbitrarily many copies of $\psi$ and perform arbitrary local
measurements on them, with no access to any other shared
quantum resource. We use $\MIP^\psi[q,a]$ to denote two-prover one-round interactive
proof systems in this setting, where $q$ and $a$ bound the question
and answer lengths in bits, respectively, as functions of the
input length. Our main result gives a
complete dichotomy for the computational power of these
proof systems.
\begin{theorem}
\label{thm:main}
Let $\psi\in \Matrix_2\otimes \Matrix_2$ be a fixed two-qubit state with maximally
mixed marginals. 

\begin{enumerate}
\item If $\psi$ is separable or $\rho_{\max}(\psi)<1$, then
\[
\MIP^{\psi}[\poly,\poly]
=
\NEXP
=
\MIP.
\]

\item If $\psi$ is entangled and $\rho_{\max}(\psi)=1$, then
\[
   \MIP^{\psi}[\poly,\poly]
  =
  \RE.
\]

\end{enumerate}
\end{theorem}




\Cref{thm:main} shows that unital noise may eliminate the computational
advantage of shared EPR pairs in two-prover one-round $\MIP$.  If the noise reduces quantum maximal correlation below
one or renders the state separable, the computational power
collapses to the classical counterpart $\MIP=\NEXP$. Otherwise, the
full power of $\MIP^*=\RE$ survives. Surprisingly, the quantum advantage of $\MIP^*$ is robust to certain noise, such as small dephasing noise, although the provers cannot exchange messages and thus cannot implement any quantum error correction scheme.


\subsection{Proof Overview}

We formulate the proof in terms of two-player entangled
games. For each fixed input, a two-prover one-round protocol
defines such a game, and the verifier's acceptance
probability equals the winning probability. 
 





\subsubsection{The case $\rho_{\max}(\psi)<1$}

Our proof follows the framework that Ghazi, Kamath, and Sudan~\cite{GKS16}, introduced to
show that non-interactive simulation of joint distributions is decidable, which was
developed further in~\cite{DMN18,GKR18}. In that problem, two parties receive
arbitrarily many i.i.d.\ samples from a fixed correlated distribution and, without
communicating, try to jointly sample a target distribution. The key step is to show that
any simulation using arbitrarily many samples can be approximated by one using a bounded
number of samples. Later, Qin and Yao~\cite{QY21} initiated the study
of the quantum counterpart of this problem, also known as quantum local state transformation~\cite{Beigi13}, in which shared samples are replaced by copies of a bipartite
quantum state and local functions by quantum measurements.
They proved decidability for nonlocal games with noisy
maximally entangled states and subsequently extended
the framework to fully quantum games~\cite{QY23}.
~\cite{DongEtAl24} refined this approach
to obtain a sharp $\NEXP$ upper bound for constant-length answers.

Similar to the classical setting, the upper bounds in~\cite{QY21,QY23,DongEtAl24}  are all obtained by proving an upper bound $D$ on the number of shared noisy states such that the provers need only share $D$ copies to simulate any strategy when they share arbitrarily many copies of the states.  To be more specific, let $s$ and $t$ denote the maximum numbers of questions
and answers, respectively, and fix the resource state and approximation error. The best previous bound on the number of copies is  $D=\poly(s)\exp(\poly(t))$
in~\cite{DongEtAl24}, which is exponential in the number of asnwers. To avoid writing out the full measurement matrices, their algorithm uses low-degree Pauli coefficients, whose number per operator scales as $D^{O(k)}$ for degree $k$. 

Two steps in~\cite{DongEtAl24} make the dependence on $t$ expensive. First, in previous arguments, measurement operators are smoothed by applying depolarizing channels, and their high-degree Pauli terms are truncated to obtain low-degree approximated operators. Nonetheless, truncation does not preserve positivity, so the errors from approximating each outcome separately add up over all $t$ answers. Thus, the degrees of the approximated operators have to grow polynomially in $t$. Second, the dimension of the measurement operators is reduced via an invariance principle, which converts low-degree operators to Gaussian polynomials, and further employs a dimension-reduction theorem for Gaussian polynomials~\cite{GKR18}. But this process inevitably loses a factor exponential in the degree, as the invariance principle relies on a quantum hypercontractive inequality that causes a factor exponential in the degree. As a result, these bounds do not give an \textsf{NEXP} upper bound for polynomial-length answers.

We keep the overall structure: regularize an arbitrary strategy, reduce it to bounded
dimension, and round it back to a valid strategy.
What changes is how the middle steps
are carried out and analyzed.
First, we smooth the entire POVM and bound the total change by a relative entropy of at most $\log t$ using the chain rule for relative entropy and Gao--Rouz\'e's inequality~\cite[Theorem~1.2]{GR22}, an exponential improvement over previous work. Second, instead of truncating the high-degree part of the smoothed measurement
operators, we approximate the
block operator formed by their square roots with a low-degree block operator, which automatically gives a positive low-degree approximation. The main tool is a matrix-valued extension of Mendel and
Naor~\cite[Theorem~5.1]{MN14}, which gives degree $k=\poly(\log t)$. The resulting operators $Q_a^\dagger Q_a$ are positive by construction, and the block approximation controls their total normalization error. Third, we reduce the dimension directly on the matrix-valued Pauli expansions with a new
randomized map, which we call \emph{Pauli folding}. It is approximately multiplicative on
low-degree operators, so we can control normalization and the game value without an
invariance principle or a detour through Gaussian space. This removes the exponential
loss in~\cite{GKR18,QY21,QY23,DongEtAl24}. Section~\ref{subsection:technicalcontribution} describes the second and third steps in more detail.

Together, these ingredients yield upper bounds on the copies of shared states $D=\poly(s,t)$ and the degrees of operators $k=\poly(\log t)$. The low-degree coefficients form a certificate of size
$\mathrm{poly}(s,t)\,D^{O(k)}$, which is exponential in the input length when questions
and answers have polynomial length. Its normalization and value can be checked in
exponential time, and rounding guarantees that every accepted certificate comes from a
genuine strategy with almost the same value. This gives the NEXP upper bound; the
matching lower bound is~\cite[Corollary~2]{DongEtAl24}.

\subsubsection{The case with $\rho_{\max}(\psi)=1$}

~\cite{DongEtAl24} gives constant-answer protocols for
$\mathrm{RE}$ using perfect EPR pairs. At $\rho_{\max}(\psi)=1$,
we run these protocols using copies of $\psi$ instead.
The main
step is a local conversion: for any finite number of EPR pairs, finitely many copies of
$\psi$ can be turned into a state with an exact EPR component whose weight is bounded
below independently of the target dimension.

We first observe that any two qubit state $\psi$ are equivalent to a canonical form $\Psi_\lambda$:

\begin{equation*}
  \Psi_\lambda
  =
  \frac14\left(
    \id\otimes\id
    +\lambda\pauliX\otimes\pauliX^{\mathsf T}
    +\lambda\pauliY\otimes\pauliY^{\mathsf T}
    +\pauliZ\otimes\pauliZ^{\mathsf T}
  \right).
\end{equation*} 
up to local unitaries, where $\lambda\in[0,1]$ and $\psi$ is separable exactly when
$\lambda=0$.
For $\lambda>0$, we use a perfect radius-one Hamming partition to
split the computational basis into code registers and logical registers. After a local
relabeling of the basis and discarding the code registers, the representatives of the
logical basis states are at pairwise Hamming distance at most two, so their coherence
factors are at least $\lambda^2$. Then we prove that, when the provers are provided with sufficiently many copies of $\Psi_{\lambda}$ with $\lambda>0$, they can use bilateral Pauli twirling to obtain
\[F\Phi_d+(1-F)\xi,\]
for any integer $d$, where $\Phi_d$ is a $d$-dimensional maximally entangled states and $F\geq\lambda^2$ independently of $m$. At the same time, we apply Yuen's parallel-repetition theorem~\cite{Yuen16} to boost the completeness--soundness gap such that the soundness is sufficiently smaller than $\lambda^2$. We further apply the original measurements for perfect EPR states, which keep at least a
$\lambda^2$ fraction of their acceptance probability. We then amplify the
remaining positive constant completeness by OR repetition on
independent resource blocks.
A union bound controls soundness even against arbitrary joint
entangled strategies.
Since $\lambda>0$ is fixed, the number of repetitions does not
depend on the input length, so the answers stay of constant length.

\subsection{Technical Contributions}\label{subsection:technicalcontribution}

Two tools from our proof are interesting on their own: a positive low-degree approximation
for quantum measurements and dimension reduction for positive low-degree operators.

\paragraph{Smoothing and positive low-degree approximation}
Smoothing followed by high-degree truncation is a standard regularization method in classical Fourier analysis~\cite{ODonnell2014,FeldmanKV2017}. Pauli truncation, which uses the fact that noise
suppresses high-degree terms, plays a similar role in classical simulation of noisy
quantum circuits~\cite{AharonovGLLV2023,SchusterYGY2025,AngrisaniMRCH2026}. In our setting, however, low degree alone is not enough: the approximating operators must keep positive, not just close to positive. Nonetheless, truncation breaks positivity, and bounding the error one operator at a time is
expensive when the family is large. Our scheme handles positivity, low degree, and the
total error together.

Fix a question $x$ to Alice, let $\mathscr{A}=\set{A_a}_a$ be the $t$-outcome POVMs performed by Alice. The treatment for Bob's strategies is the same.  Unlike previous work, we treat the POVM $\set{A_a}_a$ as a single block-diagonal operator $\sum_a\ketbra{a}\otimes A_a$ rather than $t$ individual operators. We apply a depolarizing channel to the POVM operators and get $\sum_a\ketbra{a}\otimes \mathcal{T}_u(A_a)$. When $\rho_{\max}(\rho)<1$, the complete strong data-processing inequality \cite[Theorem~1.2]{GR22} bounds the
contribution of each qubit by the decrease in relative entropy it causes. These
decreases add up to at most the initial relative entropy, which is at most $\log t$. Hence smoothing changes the game value by an amount controlled by $\log t$, not by a
polynomial in $t$.

The more challenging part is to approximate $\mathcal{T}_u(A_a)$ by positive low-degree operators. The low-degree truncation in Pauli degree gives operators that are only approximately positive, in the Schatten 2-norm.  To get a positive low-degree approximation, we further apply a depolarizing channel to the square-root of the POVM element and obtain $\mathcal{T}_{\theta}\sqrt{\mathcal{T}_u(A_a)}$ and show that 
\[\sqrt{\mathcal{T}_u(A_a)}\approx\mathcal{T}_{\theta}\sqrt{\mathcal{T}_u(A_a)}\]
in the normalized Schatten $4$-norm, using the quantum Stroock--Varopoulos
inequality~\cite{BDR20}. 

We then prove that if a Hermitian operator is close in the normalized Schatten $4$-norm to a slightly
depolarized version of itself, then it is also close to a low-degree operator, with
bounds that do not depend on the dimension. The proof transfers the high-degree decay
estimate of Mendel and Naor~\cite{MN14} to the
Pauli setting and uses Schatten duality.
The theorem holds for general Hermitian operators,
not only measurement operators. This implies that $\sqrt{\mathcal{T}_u(A_a)}$ also has a low-degree approximation, denoted by  $Q_a$.
Thus, the operator
$Q_a^\dagger Q_a$ is positive semidefinite and has
at most twice the degree of $Q_a$, which is close to $\mathcal{T}_u(A_a)$ . We also
ensure that $\sum_a Q_a^\dagger Q_a$ is close to the
identity, so these operators nearly form a valid measurement.

\paragraph{Dimension reduction for positive low-degree operators}
Dimension reduction is a standard tool in classical analysis and algorithms, but it is
much more difficult for quantum operators. Measurements must satisfy positivity and normalization, while their correlations depend on the noncommutative structure of the underlying operators. Without
further assumptions, dimension reduction can fail: Harrow, Montanaro, and Short~\cite{HarrowMontanaroShort2011} show that no distribution over quantum channels can substantially reduce dimension while preserving Hilbert–Schmidt distances between arbitrary quantum states with high probability. Stark and Harrow~\cite{StarkHarrow2016} studied what can and
cannot be achieved when compressing positive semidefinite factorizations and quantum
models. Some structure is therefore needed, and here it is low degree.

Given a family of positive operators
$\set{M_a=Q_a^\dagger Q_a}_{a=1}^t$ on $n$ qubits, where each $Q_a$
has degree at most $k$ and $\sum_a M_a$ is close
to the identity, our goal is to construct positive
operators $\set{\widetilde M_a}_a$ on $D$ qubits
that retain nearly the original game value while
keeping the sums $\sum_a\widetilde M_a$ close to
the identity.
For fixed approximation error, the bound on $D$
is polynomial in $k$ and $t$.

To compress the operators, we introduce a novel dimension reduction map {\em Pauli folding}: Let $h\colon[n]\to[D]$ uniformly from
all functions and $r\in \mathbb{F}_2^n$ uniformly, independently. Set
$H\colon \mathbb{F}_2^n\to \mathbb{F}_2^D$ to be:
\[
(H\alpha)_j=\sum_{i:h(i)=j}\alpha_i,
\qquad \alpha\in G^n,\quad j\in[D].
\]
The Pauli folding
$\Phi_{h,r}\colon \Matrix_2^{\otimes n}\to \Matrix_2^{\otimes D}$ is defined on Pauli basis and extended to $\Matrix_2^{\otimes n}$ by linearity. 
\begin{equation*}
\Phi_{h,r}(P_\alpha)=\chi_r(\alpha)P_{H\alpha},
\end{equation*}
where $\chi_r(\alpha)=\prod_i\chi_{r_i}(\alpha_i)$ is the character of $\mathbb{F}_2^n$.

Its key property is approximate multiplicativity: for operators $A$ and
$B$ of Pauli degree at most $k$,

\begin{equation}\label{eqn:mulintro}
    \mathbb E_{h,r}
    \norm{
      \Phi_{h,r}(AB)-\Phi_{h,r}(A)\Phi_{h,r}(B)
    }_2^2
    \leq
    \frac{\poly(k)}{D}\norm{A}_4^2\norm{B}_4^2.
\end{equation}  
The bound is independent of $n$
and $A$ and $B$ need not be positive or Hermitian. To keep positivity, we set $\widetilde M_a=\Phi(Q_a)^\dagger\Phi(Q_a),$
and approximate multiplicativity bounds the difference between $\widetilde{M}_a$ and
$\Phi(Q_a^\dagger Q_a)$. 
We also show that applying $\Phi$
to both players' operators approximately preserves
their correlations in expectation. Together with the bound on the
normalization error, this lets us fix a single map $\Phi$ that achieves the required
reduction.

To see why \cref{eqn:mulintro} holds, consider its Boolean analogue: take $h:[n]\rightarrow[D]$ and signs $s_1,\ldots,s_n$, define$(T_{h,s}f)(y_1,\ldots, y_D)=f(s_1y_{h(1)},\ldots,s_ny_{h(n)}).$
This reduces the number of variables from $n$ to $D$. It is easy to verify that $T_{h,s}(fg)=(T_{h,s}f)(T_{h,s}g),$
for every $h,s$ without any degree assumption.  

However, the argument fails when extending to the Pauli basis due to the non-commutativity of matrices. To see how multiplicativity fails, let $h(i)=h(i')=j$ for $i\neq i'$, and take
$A=X_iX_{i'}$ and $B=Z_iZ_{i'}$. Since $XZ=-iY$, we have
$AB=(-i)^2\,Y_iY_{i'}=-Y_iY_{i'}$. After folding, the labels in the fiber cancel:
$\Phi(A)=\chi_r(\alpha)I$, $\Phi(B)=\chi_r(\beta)I$, and
$\Phi(AB)=-\chi_r(\alpha+\beta)I$, so $\Phi(AB)=-\Phi(A)\Phi(B)$.
The operators $X_i$ and $Z_{i'}$ commute because they act on different qubits,
but after hashing both act on qubit $j$, where $X$ and $Z$ anticommute.

To understand why the approximate multiplicativity holds, first consider the case $D=n$
with $h$ chosen to be a permutation. There are then no collisions: folding
simply relabels the coordinates, while the random signs correspond to local
Pauli conjugations. Both operations preserve products, so
$\Phi(AB)=\Phi(A)\Phi(B)$ exactly. When $D<n$, different input coordinates
can be mapped to the same output coordinate as described above, and thus this equality can fail. Low degree suggests that such failures should be rare. Any two Pauli
strings of degree at most $k$ together involve at most $2k$ coordinates.
For a random hash into $D$ coordinates, the probability of a collision
among them is $O(k^2/D)$; without a collision, their product is preserved
exactly. Thus, a sufficiently large polynomial choice of $D$ makes the
error unlikely for each pair of Pauli strings. The challenge is to turn
this observation into a bound for entire operators, which may contain
many such strings even when their degree is small.

Turning this intuition into a rigorous mathematical proof is mathematically involved. The difficulty is that the same input Pauli label may
arise from many different products of terms. These contributions receive
the same random sign, so averaging does not separate them. If we bound them individually, the error would be either exponential in $k$ or depend on $n$. Our proof therefore controls these contributions together, using the structure of the operator products rather than estimating each coefficient separately.

We bound the error by collision patterns using Möbius inversion on set partitions, a form of inclusion–exclusion. A pattern requiring \(q\) independent coordinate identifications carries a factor \(D^{-q}\) and involves at most \(2q\) colliding coordinates. The main task is to bound the combined contribution of all placements of such a pattern without introducing a dependence on \(n\). We encode these placement sums as derivatives of auxiliary polynomials. On the unit cube, these polynomials are averages of trace expressions involving Pauli conjugations of the input operators. We further investigate the analytical properties of these auxiliary polynomials, adopt the Markov brothers’ inequality to control their derivatives at a cost polynomial in \(k\) per differentiation, independent of $n$. A second Möbius inversion enforces distinct coordinate choices. After also counting the possible collision patterns, the total cost remains polynomial in \(k\) per identification, while each required coordinate identification contributes a factor \(1/D\). Choosing \(D\) sufficiently large polynomially in \(k\) therefore makes the contributions geometrically summable.

We are not aware of any dimension reduction for matrices that is also approximately multiplicative. We expect approximate multiplicativity, and the dimension reduction it yields, to be
useful beyond this paper. It is natural to ask which other structured families of
quantum measurements or correlations admit dimension reduction in this way, and whether
approximate multiplicativity can serve as a general tool for other applications.


\subsection*{Acknowledgment}

Y.D., Z.J., P.Y. and H.Z were supported by the National Natural Science Foundation of China (Grant Nos. 62332009 and 12347104), the Quantum Science and Technology—National Science and Technology Major Project (Grant No. 2021ZD0302901), the NSFC/RGC Joint Research Scheme (Grant No. 12461160276), the Natural Science Foundation of Jiangsu Province (No. BK20243060), the Fundamental and Interdisciplinary Disciplines Breakthrough Plan of the Ministry of Education of China (No. JYB2025XDXM118), the "111 Center"(No. B26023), and the Fundamental Research Funds for the Central Universities (Grant no. 2026300376).
Most of the work was done when Y.D. was a Ph.D. student at Nanjing University.
Y.D. was also funded by the European Union (ERC, DistQuant, 101266668). Views and opinions expressed are however those of the author(s) only and do not necessarily reflect those of the European Union or the European Research Council. Neither the European Union nor the granting authority can be held responsible for them.
M.Q. was supported by the National Research Foundation, Singapore through the National Quantum Office, hosted in A*STAR, under its Centre for Quantum Technologies Funding Initiative (S24Q2d0009).

\section{Preliminaries}

\subsection{Notation and Conventions}
\label{subsec:notation-and-conventions}
For every $d\geq1$, let $\Matrix_d:=\Matrix_d(\mathbb C)$ denote the set of
$d\times d$ complex matrices, and write $\id$ for the identity whenever its
dimension is clear.  For every integer $n\geq1$, write
$[n]:=\{1,\ldots,n\}$, and set $\iu:=\sqrt{-1}$.  We use
\begin{equation}
\label{eq:unitary-and-rotation-groups}
  \begin{aligned}
    \U(d)&:=\{U\in \Matrix_d:U^\dagger U=\id\},\\
    \mathcal{SO}(d)&:=\{R\in \Matrix_d(\mathbb R):R^{\mathsf T}R=\id,
      \ \det R=1\}.
  \end{aligned}
\end{equation}
In particular,
$\U(1)=\{z\in\mathbb C:|z|=1\}$.

For $t\geq1$, write $\ell_2^t$ for $\mathbb C^t$ equipped with its
standard inner product and Euclidean norm:
\begin{equation}
\label{eq:euclidean-inner-product-and-norm}
  \langle z,w\rangle_{\ell_2^t}
  :=\sum_{a=1}^t\overline{z_a}w_a,
  \qquad
  \norm z_{\ell_2^t}
  :=\left(\sum_{a=1}^t|z_a|^2\right)^{1/2}.
\end{equation}
In particular, this norm is not normalized by the dimension~$t$.

For $X\in \Matrix_d$, write
\begin{equation}
\label{eq:normalized-trace}
  \tau_d(X):=\frac1d\Tr(X)
\end{equation}
for the \emph{normalized trace} on $\Matrix_d$.  Thus the normalized trace on
$\Matrix_2^{\otimes n}\cong \Matrix_{2^n}$ is $\tau_{2^n}$.

For states $\rho$ and $\sigma$ with
$\supp(\rho)\subseteq\supp(\sigma)$, their \emph{quantum relative entropy} is
\begin{equation}
\label{eq:quantum-relative-entropy}
  \RelEnt(\rho\|\sigma)
  :=\Tr\!\left(\rho(\log\rho-\log\sigma)\right).
\end{equation}
We set $\RelEnt(\rho\|\sigma)=+\infty$ if the support condition fails.
All logarithms in entropy expressions have base $e$.

For positive $R,S\in \Matrix_d$ with $\tau_d(R)=\tau_d(S)=1$, define
\begin{equation}\label{eqn:relent}
  \RelEnt_{\tau_d}(R\|S)
  :=\tau_d\!\left(R(\log R-\log S)\right),
\end{equation}
with the same support convention.  
The scalar $\log d$ terms cancel after rescaling, so
\begin{equation*}
  \RelEnt(R/d\|S/d)=\RelEnt_{\tau_d}(R\|S).
\end{equation*}

Applying the same quantum channel to two states cannot increase their
relative entropy.  This is the \emph{data-processing inequality}.

\begin{prop}[{\cite[Theorem~5.35]{tqi}}]
  \label{prop:relative-entropy-data-processing}
  Let $\Phi\colon \Matrix_d\to \Matrix_e$ be a completely positive trace-preserving
  linear map.  For all states $\rho,\sigma\in \Matrix_d$,
  \begin{equation*}
    \RelEnt(\Phi(\rho)\|\Phi(\sigma))\leq \RelEnt(\rho\|\sigma).
  \end{equation*}
  If $e=d$, then every pair of positive operators $R,S\in \Matrix_d$ with
  $\tau_d(R)=\tau_d(S)=1$ also satisfies
  \begin{equation*}
    \RelEnt_{\tau_d}(\Phi(R)\|\Phi(S))\leq \RelEnt_{\tau_d}(R\|S).
  \end{equation*}
\end{prop}

The normalized-trace assertion follows by applying the first inequality
to $R/d$ and $S/d$ and using the preceding rescaling identity.

For $X\in \Matrix_d$, set $|X|:=(X^\dagger X)^{1/2}$ and for $1\leq p<\infty$,  the \emph{normalized Schatten $p$-norm} is defined as
\begin{equation}
\label{eq:normalized-schatten-norm}
  \norm X_p:=\tau_d(|X|^p)^{1/p},
  \qquad
  \norm X_\infty:=\norm X_{\mathrm{op}}.
\end{equation}
We write $L_2(\tau_d)$ for the Hilbert space $\Matrix_d$ equipped with the
normalized Hilbert--Schmidt inner product
\begin{equation}
\label{eq:normalized-hilbert-schmidt-inner-product}
  \langle X,Y\rangle_2:=\tau_d(X^\dagger Y).
\end{equation}
Its induced norm is the normalized Schatten $2$-norm in
\eqref{eq:normalized-schatten-norm}.  We
omit the trace from the notation $L_2(\tau_d)$ when it is clear from context.

If $\mathcal H_0$ is a finite-dimensional Hilbert space and
$W\colon\mathbb C^d\to\mathcal H_0\otimes\mathbb C^d$ is a linear operator, its normalized
Schatten norm is defined using the input trace:
\begin{equation}
\label{eq:block-schatten-norm}
  \norm W_p
  :=\tau_d\!\left((W^\dagger W)^{p/2}\right)^{1/p}
  \quad(1\leq p<\infty),
  \qquad
  \norm W_\infty:=\norm W_{\mathrm{op}}.
\end{equation}
In particular,
$\norm W_p^2=\norm{W^\dagger W}_{p/2}$ for $p\geq2$.

We use the quantum Pinsker inequality in the following normalized-trace form.

\begin{prop}[Pinsker inequality, {\cite[Theorem~5.38]{tqi}}]
  \label{prop:quantum-pinsker}
  Let $d\geq1$ and $\tau=\tau_d$.  For all positive $R,S\in\Matrix_d$ with
  $\tau(R)=\tau(S)=1$,
  \begin{equation*}
    \norm{R-S}_1^2\leq2\RelEnt_\tau(R\|S).
  \end{equation*}
\end{prop}

This follows by converting the cited inequality to logarithms with base $e$ and
applying it to $R/d$ and $S/d$, since
$\Tr|(R-S)/d|=\norm{R-S}_1$ and
$\RelEnt(R/d\|S/d)=\RelEnt_\tau(R\|S)$.

\subsection{Pauli Analysis}
\label{subsec:pauli-analysis}

For standard accounts of Pauli
expansions, degree decompositions, and their relation to operator Fourier
analysis, see \cite[Section~5]{MO10} and
\cite[Section~2.3]{QY23}.

\begin{definition}[Pauli labels]
  \label{def:pauli-labels}
  Let $\mathbb F_2=\{0,1\}$ be the field whose arithmetic is modulo~$2$.
  The one-qubit Pauli label group is the four-element abelian group
  \begin{equation*}
    G:=\mathbb F_2^2
    =\{(0,0),(1,0),(0,1),(1,1)\},
  \end{equation*}
  with componentwise addition.  We write $0=(0,0)$ for its identity.
  Since addition is modulo~$2$, one has $s+s=0$ and hence $-s=s$ for every
  $s\in G$.
\end{definition}

\begin{definition}[Pauli matrices]
  \label{def:one-qubit-pauli-matrices}
  Set the Pauli matrices
  \[
    \pauliX:=\begin{pmatrix}0&1\\1&0\end{pmatrix},
    \qquad
    \pauliY:=\begin{pmatrix}0&-\iu\\ \iu&0\end{pmatrix},
    \qquad
    \pauliZ:=\begin{pmatrix}1&0\\0&-1\end{pmatrix}.
  \]
  For $(a,b)\in G$, define $ P_{(a,b)}:=\iu^{ab}\pauliX^a\pauliZ^b$, thus
  \begin{equation*}
    P_{(0,0)}=\id,\quad
    P_{(1,0)}=\pauliX,\quad
    P_{(0,1)}=\pauliZ,\quad
    P_{(1,1)}=\pauliY.
  \end{equation*}
\end{definition}

\begin{definition}[Pauli strings and weight]
  \label{def:pauli-strings-and-weight}
  For $n\geq1$, the elements of $G^n$ are added coordinatewise, and
  $0\in G^n$ also denotes the all-zero label.
  For $\alpha=(\alpha_1,\ldots,\alpha_n)\in G^n$, its associated
  \emph{Pauli string}, coordinate support, and Pauli weight are
  \begin{equation*}
    P_\alpha:=P_{\alpha_1}\otimes\cdots\otimes P_{\alpha_n},
    \qquad
    \supp(\alpha):=\{i\in[n]:\alpha_i\neq0\},
    \qquad
    |\alpha|:=|\supp(\alpha)|.
  \end{equation*}
\end{definition}


\begin{fact}
  \label{fact:pauli-orthonormal-basis}
  Let $n\geq1$.
  For all $\alpha,\beta\in G^n$,
  \begin{equation*}
    P_\alpha^\dagger=P_\alpha,
    \qquad
    P_\alpha^2=\id,
    \qquad
    \tau_{2^n}(P_\alpha)=\delta_{\alpha,0},
    \qquad
    \tau_{2^n}(P_\alpha^\dagger P_\beta)=\delta_{\alpha,\beta},
  \end{equation*}
  where $\delta_{\alpha,\beta}=\one_{\{\alpha=\beta\}}$ is the Kronecker delta.  Hence
  $(P_\alpha)_{\alpha\in G^n}$ is an orthonormal basis of
  $\Matrix_2^{\otimes n}$.
\end{fact}


\begin{definition}[$\V$-valued Pauli expansion]
  \label{def:coefficient-valued-pauli-expansion}
  Let $n\geq1$, let $\V$ be a complex vector space, and let
  $F\in \V \otimes \Matrix_2^{\otimes n}$.  For $\alpha\in G^n$, define
  \begin{equation*}
    \widehat F(\alpha)
    :=
    (\operatorname{id}_{\V}\otimes\varphi_\alpha)(F)
    \in \V,
    \qquad
    \varphi_\alpha(X):=\langle P_\alpha,X\rangle_2.
  \end{equation*}
  Then $F$ has the unique expansion
  \begin{equation*}
    F=\sum_{\alpha\in G^n}\widehat F(\alpha)\otimes P_\alpha.
  \end{equation*}
  We call $\widehat F(\alpha)$ the Pauli coefficient of $F$ at
  $\alpha$.
\end{definition}

\begin{definition}[Pauli support, degree, and truncation]
  \label{def:pauli-support-degree-and-truncation}
  Let $n\geq1$, let $\V$ be a complex vector space, and let
  $F=\sum_{\alpha\in G^n}\widehat F(\alpha)\otimes P_\alpha$ be as in
  \cref{def:coefficient-valued-pauli-expansion}.  Its \emph{Pauli support}
  and \emph{Pauli degree} are
  \begin{equation*}
    \operatorname{supp}_{\mathrm P}(F)
      :=\{\alpha\in G^n:\widehat F(\alpha)\neq0\},
      \qquad
    \deg(F):=
      \max_{\alpha\in\operatorname{supp}_{\mathrm P}(F)}|\alpha|,
  \end{equation*}
  with $\deg(0):=-\infty$.  For an integer $k\geq0$, the corresponding \emph{degree-$k$ truncation} is
  \begin{equation*}
    F^{\leq k}:=\sum_{|\alpha|\leq k}
      \widehat F(\alpha)\otimes P_\alpha.
  \end{equation*}
\end{definition}

For a complex coefficient space $\V$ and a linear map
$\Phi\colon \Matrix_2^{\otimes n}\to \Matrix_2^{\otimes n}$, we also write $\Phi$ for
its extension $\operatorname{id}_\V \otimes\Phi$ to
$\V\otimes \Matrix_2^{\otimes n}$.  Thus $\Phi$ acts only on the matrix factor.
For $\V=\mathbb C^t$, this convention reads
\begin{equation}
\label{eq:coefficientwise-map-action}
  \Phi\!\left(\sum_{a=1}^t|a\rangle\otimes W_a\right)
  =\sum_{a=1}^t|a\rangle\otimes\Phi(W_a),
  \qquad W_a\in \Matrix_2^{\otimes n}.
\end{equation}
We use this convention for the depolarizing maps that smooth the
measurements.  Their action on Pauli expansions is particularly simple.
\begin{definition}
  \label{def:depolarizing-maps}
  Let $\Pi\colon\Matrix_2\to \Matrix_2$ be the map
  $\Pi(X):=\tau_2(X)\id$.  For $u\geq0$ and $n\geq1$, define the one-qubit
  \emph{depolarizing map} and its $n$-qubit tensor product by
  \begin{equation*}
    \mathcal D_u:=\Pi+e^{-u}(\operatorname{id}_{\Matrix_2}-\Pi),
    \qquad
    \mathcal T_u:=\mathcal D_u^{\otimes n}.
  \end{equation*}
\end{definition}

\begin{prop}
  \label{prop:depolarizing-pauli-action}
  Let $n\geq1$, $u\geq0$, and $\alpha\in G^n$.  Then
  \begin{equation*}
    \mathcal T_u(P_\alpha)=e^{-u|\alpha|}P_\alpha.
  \end{equation*}
  Consequently, for any complex vector space $\V$ and
  $F\in \V\otimes \Matrix_2^{\otimes n}$,
  \begin{equation*}
    \mathcal T_u(F)
    =\sum_{\alpha\in G^n}e^{-u|\alpha|}
      \widehat F(\alpha)\otimes P_\alpha.
  \end{equation*}
\end{prop}



\begin{prop}
  \label{prop:pauli-parseval}
  Let $n\geq1$ and let $\V$ be a finite-dimensional complex inner-product
  space.  Equip
  $\V\otimes \Matrix_2^{\otimes n}$ with the Hilbert tensor-product inner product
  determined by
  \begin{equation}
\label{eq:tensor-product-inner-product}
    \langle e\otimes A,f\otimes B\rangle_{\V \otimes L_2}
    :=\langle e,f\rangle_\V\,\tau_{2^n}(A^\dagger B).
  \end{equation}
  For $F,H\in \V \otimes \Matrix_2^{\otimes n}$, the following hold.
  \begin{enumerate}
    \item \emph{(Parseval)}
      \begin{equation*}
        \norm{F}_{\V\otimes L_2}^2
        =\sum_{\alpha\in G^n}
          \norm{\widehat F(\alpha)}_\V^2.
      \end{equation*}
    \item \emph{(Inner product)}
      \begin{equation*}
        \langle F,H\rangle_{\V\otimes L_2}
        =\sum_{\alpha\in G^n}
          \left\langle\widehat F(\alpha),
          \widehat H(\alpha)\right\rangle_\V.
      \end{equation*}
    \item \emph{(Constant coefficient)}
      For every $e\in \V$,
      \begin{equation*}
        (\operatorname{id}_\V \otimes\tau_{2^n})(F)=\widehat F(0),
        \qquad
        \langle e\otimes\id,F\rangle_{\V\otimes L_2}
        =\langle e,\widehat F(0)\rangle_\V.
      \end{equation*}
  \end{enumerate}
\end{prop}

\begin{proof}
  Expanding both arguments in the Pauli basis gives
  \begin{equation*}
    \langle F,H\rangle_{\V \otimes L_2}
    =\sum_{\alpha,\beta\in G^n}
      \left\langle\widehat F(\alpha),
      \widehat H(\beta)\right\rangle_\V
      \tau_{2^n}(P_\alpha^\dagger P_\beta).
  \end{equation*}
  Pauli orthogonality removes all terms with $\alpha\neq\beta$, proving the
  inner-product identity; Parseval follows by taking $H=F$.  Finally,
  $P_0=\id$ and
  $\varphi_0=\tau_{2^n}$ give the two constant-coefficient identities.
\end{proof}

The choices $\V=\mathbb C$, $\V=\mathbb C^t$, and $\V=\Matrix_{d_0}$ give the three
forms used below.

Taking $\V=\mathbb C$ identifies
$\mathbb C\otimes \Matrix_2^{\otimes n}$ with $\Matrix_2^{\otimes n}$.  Thus every
$A\in \Matrix_2^{\otimes n}$ has the expansion
\begin{equation}
\label{eq:scalar-pauli-expansion}
  A=\sum_{\alpha\in G^n}\widehat A(\alpha)P_\alpha,
  \qquad
  \widehat A(\alpha)
  =\tau_{2^n}(P_\alpha^\dagger A)
  =\tau_{2^n}(P_\alpha A).
\end{equation}
\begin{prop}[{\cite[Fact~2.11]{QY21}}]
  \label{prop:scalar-pauli-identities}
  Let $n\geq1$ and $A,B\in \Matrix_2^{\otimes n}$.  Then
  \begin{enumerate}
    \item \emph{(Parseval)}
      \(
        \norm A_2^2
        =\sum_{\alpha\in G^n}|\widehat A(\alpha)|^2.
      \)
    \item \emph{(Inner product)}
      \(
        \tau_{2^n}(A^\dagger B)
        =\sum_{\alpha\in G^n}
          \overline{\widehat A(\alpha)}\widehat B(\alpha).
      \)
    \item \emph{(Constant coefficient)}
      \(
        \langle\id,A\rangle_2
        =\tau_{2^n}(A)
        =\widehat A(0).
      \)
  \end{enumerate}
\end{prop}

Let
$W\colon(\mathbb C^2)^{\otimes n}\to
\mathbb C^t\otimes(\mathbb C^2)^{\otimes n}$ be a linear operator.
In the standard bases, its $t2^n\times2^n$ matrix has the block form
\begin{equation}
\label{eq:block-operator-form}
  W=\begin{pmatrix}W_1\\ \vdots\\ W_t\end{pmatrix}
   =\sum_{a=1}^t|a\rangle\otimes W_a,
  \qquad W_a\in \Matrix_2^{\otimes n}.
\end{equation}
Thus $Wv=\sum_{a=1}^t|a\rangle\otimes W_av$ for every
$v\in(\mathbb C^2)^{\otimes n}$.
We call such an operator a \emph{block operator}.  This is the
$\V=\mathbb C^t$ case of \cref{def:coefficient-valued-pauli-expansion}:
\begin{equation}
\label{eq:block-pauli-expansion}
  W=\sum_{\alpha\in G^n}\widehat W(\alpha)\otimes P_\alpha,
  \qquad
  \widehat W(\alpha)
  =\sum_{a=1}^t|a\rangle\widehat{W_a}(\alpha)
  \in\mathbb C^t.
\end{equation}
Recall the normalized Schatten $p$-norm defined in~\eqref{eq:block-schatten-norm}. The following proposition follows directly from~\cref{prop:pauli-parseval}.

\begin{prop}
  \label{prop:block-operator-pauli-identities}
  Let $n,t\geq1$, and let $W,Z\colon(\mathbb C^2)^{\otimes n}\to
  \mathbb C^t\otimes(\mathbb C^2)^{\otimes n}$ be block operators.  Then
  \begin{enumerate}
    \item \emph{(Parseval)}
      \(
        \norm W_2^2
        =\sum_{\alpha\in G^n}
          \norm{\widehat W(\alpha)}_{\ell_2^t}^2.
      \)
    \item \emph{(Inner product)
    }
      \(
        \tau_{2^n}(W^\dagger Z)
        =\sum_{\alpha\in G^n}
          \left\langle\widehat W(\alpha),
          \widehat Z(\alpha)\right\rangle_{\ell_2^t}.
      \)
  \end{enumerate}
\end{prop}

Taking $\V=\Matrix_{d_0}$ with normalized Hilbert--Schmidt inner product
\(
  \langle C,D\rangle_2:=\tau_{d_0}(C^\dagger D)
\), every $A\in \Matrix_{d_0}\otimes \Matrix_2^{\otimes n}$ has the expansion
\begin{equation*}
  A=\sum_{\alpha\in G^n}\widehat A(\alpha)\otimes P_\alpha,
  \qquad
  \widehat A(\alpha)\in \Matrix_{d_0}.
\end{equation*}
Here the Pauli support records the nonzero matrix coefficients, and degree
and truncation are determined by their labels exactly as in the general
definition.

\begin{prop}
  \label{prop:operator-valued-pauli-identities}
  Let $n,d_0\geq1$ and $A,B\in \Matrix_{d_0}\otimes \Matrix_2^{\otimes n}$.  Then
  \begin{enumerate}
    \item \emph{(Parseval)}
      \(
        \norm A_2^2
        =\sum_{\alpha\in G^n}\norm{\widehat A(\alpha)}_2^2.
      \)
    \item \emph{(Inner product)}
      \(
        (\tau_{d_0}\otimes\tau_{2^n})(A^\dagger B)
        =\sum_{\alpha\in G^n}
          \tau_{d_0}\!\left(\widehat A(\alpha)^\dagger
          \widehat B(\alpha)\right).
      \)
  \end{enumerate}
\end{prop}

\subsection{Nonlocal Games, Multiprover Interactive Proof Systems and Parallel Repetition}


An nonlocal game is specified by
\(
  \sfG=(\mathcal X,\mathcal Y,\mathcal A,\mathcal B,\mathsf p_{\sfG},V),
\)
where the question sets $\mathcal X$, $\mathcal Y$ and the answer sets
$\mathcal A$, $\mathcal B$ are finite, $\mathsf p_{\sfG}$ is a probability
distribution on $\mathcal X\times\mathcal Y$, and
$V(a,b\mid x,y)\in\{0,1\}$ is called \emph{decision predicate}.  The referee samples
$(x,y)$ according to $\mathsf p_{\sfG}$, sends $x$ to Alice and $y$ to Bob,
and receives answers $a\in\mathcal A$ and $b\in\mathcal B$.  The referee
accepts if and only if $V(a,b\mid x,y)=1$.  We write
\begin{equation}
\label{eq:game-question-marginals}
  \mathsf p_X(x):=\sum_y\mathsf p_{\sfG}(x,y),
  \qquad
  \mathsf p_Y(y):=\sum_x\mathsf p_{\sfG}(x,y)
\end{equation}
for the question marginals.

A \textit{strategy} on a bipartite state $\psi$ is specified by the POVM
families\footnote{By a slight abuse of terminology, we also refer to families
of operators approximating POVMs as strategies when no confusion can arise.}
\begin{equation*}
  \mathscr A=(\mathscr A^x)_{x\in\mathcal X},
  \qquad
  \mathscr A^x=(A_a^x)_{a\in\mathcal A},
  \qquad
  \mathscr B=(\mathscr B^y)_{y\in\mathcal Y},
  \qquad
  \mathscr B^y=(B_b^y)_{b\in\mathcal B}.
\end{equation*}
Alice uses the measurement operators $A_a^x$, and Bob uses
$(B_b^y)^{\mathsf T}$.\footnote{This convention simplifies the
correlation-channel and Pauli calculations below. All transposes are taken
in the computational basis. Since transposition maps POVMs bijectively to
POVMs, it does not infect the strategies.}
The \textit{value} of this strategy is its acceptance probability:
\begin{equation}
\label{eq:strategy-value}
  \Val_{\sfG}^{\psi}(\mathscr A,\mathscr B)
  :=\sum_{x,y,a,b}\mathsf p_{\sfG}(x,y)V(a,b\mid x,y)
    \Tr\!\left[(A_a^x\otimes(B_b^y)^{\mathsf T})\psi\right].
\end{equation}
Here the resource state is fixed independently of the game: for a
fixed two-qubit state $\psi$, the players may use $\psi^{\otimes n}$ for any
finite $n$, but may not replace it by another bipartite state. Define
\begin{equation}
\label{eq:fixed-resource-game-values}
  \omega_n^*(\sfG,\psi)
  :=\sup_{\mathscr A,\mathscr B}
    \Val_{\sfG}^{\psi^{\otimes n}}(\mathscr A,\mathscr B),
  \qquad
  \omega_\infty^*(\sfG,\psi)
  :=\sup_{n\geq1}\omega_n^*(\sfG,\psi),
\end{equation}
where the first supremum is over all local POVM strategies on $n$ copies.
When the state is clear, we abbreviate
$\Val_{\sfG}^{\psi^{\otimes n}}(\mathscr A,\mathscr B)$ as
$\Val_{\sfG}(\mathscr A,\mathscr B)$.


\begin{definition}[Fixed-resource $\MIP^\psi$]
  \label{def:polynomial-answer-noisy-mip}
  Fix a bipartite state $\psi$, and let $\mathsf q, \mathsf a$ be two fixed functions.
  A language belongs to
  \(
    \MIP^{\psi}[\mathsf q,\mathsf a]
  \)
  if it admits a uniform family of one-round protocols in which a classical
  randomized verifier interacts with two noncommunicating, computationally
  unbounded quantum provers. The verifier samples question pairs and
  evaluates decision predicates in time polynomial in the input bit length.
  The lengths of the questions and answers are upper bounded by
  $\mathsf q$ and $\mathsf a$, respectively, as functions of the input bit length, and the protocol has a constant completeness--soundness gap. The provers share arbitrarily finite copies of $\psi$,  but no other shared quantum resource.
  Allowing shared classical randomness does not change
  the value, since it only forms convex combinations of strategies.
\end{definition}


For a finite game $\sfG$, write $\omega^*(\sfG)$ for the supremum of its
acceptance probabilities over all finite-dimensional bipartite shared
states and all local POVM strategies.  We now define the repeated games.

\begin{definition}[AND and OR parallel repetition]
\label{def:and-or-parallel-repetition}
Fix a finite game
$\sfG=(\mathcal X,\mathcal Y,\mathcal A,\mathcal B,\mathsf p_{\sfG},V)$
and an integer $k\geq1$.  In a $k$-fold repetition of $\sfG$, the referee draws
\[
  (\mathbf x,\mathbf y)
  =\bigl((x_1,\ldots,x_k),(y_1,\ldots,y_k)\bigr)
  \sim\mathsf p_{\sfG}^{\otimes k}
\]
and receives answer strings
$\mathbf a=(a_1,\ldots,a_k)\in\mathcal A^k$ and
$\mathbf b=(b_1,\ldots,b_k)\in\mathcal B^k$.  Thus
\begin{equation*}
  \mathsf p_{\sfG}^{\otimes k}(\mathbf x,\mathbf y)
  :=\prod_{i=1}^k\mathsf p_{\sfG}(x_i,y_i).
\end{equation*}
The \emph{AND repetition} $\sfG^{\otimes k}$ accepts if every coordinate
accepts, whereas the \emph{OR repetition} $\operatorname{OR}_k(\sfG)$
accepts if at least one coordinate accepts.  Set
$V_i:=V(a_i,b_i\mid x_i,y_i)$.  Their respective decision predicates are
\begin{equation*}
  V_{\wedge}^{(k)}(\mathbf a,\mathbf b\mid\mathbf x,\mathbf y)
  :=\prod_{i=1}^k V_i,
  \qquad
  V_{\vee}^{(k)}(\mathbf a,\mathbf b\mid\mathbf x,\mathbf y)
  :=1-\prod_{i=1}^k(1-V_i).
\end{equation*}
A strategy for either repeated game is a general POVM strategy
\[
  \mathscr A^{\mathbf x}
    =(A_{\mathbf a}^{\mathbf x})_{\mathbf a\in\mathcal A^k},
  \qquad
  \mathscr B^{\mathbf y}
    =(B_{\mathbf b}^{\mathbf y})_{\mathbf b\in\mathcal B^k}.
\]
The measurements may depend on the entire local question string and may be
joint across all $k$ coordinates.  In particular, no product structure is
imposed on the strategy.  The corresponding entangled values are
\begin{align*}
  \omega^*(\sfG^{\otimes k})
  &:=\sup_{\rho,\mathscr A,\mathscr B}
    \Val_{\sfG^{\otimes k}}^\rho(\mathscr A,\mathscr B),\\
  \omega^*\!\left(\operatorname{OR}_k(\sfG)\right)
  &:=\sup_{\rho,\mathscr A,\mathscr B}
    \Val_{\operatorname{OR}_k(\sfG)}^\rho(\mathscr A,\mathscr B),
\end{align*}
where both suprema range over all finite-dimensional bipartite shared
states $\rho$ and the POVM families described above.
The repetition number $k$ counts instances of the
game and is unrelated to the number $n$ of copies of the resource state.
\end{definition}

\begin{lemma}
  \label{lem:or-repetition-bounds}
  For every finite game
  $\sfG=(\mathcal X,\mathcal Y,\mathcal A,\mathcal B,\mathsf p_{\sfG},V)$
  and every $k\geq1$,
  \begin{equation*}
    1-\bigl(1-\omega^*(\sfG)\bigr)^k
    \leq \omega^*\!\left(\operatorname{OR}_k(\sfG)\right)
    \leq k\omega^*(\sfG).
  \end{equation*}
  Moreover, for every fixed bipartite resource state $\psi$,
  \begin{equation*}
    1-\bigl(1-\omega_\infty^*(\sfG,\psi)\bigr)^k
    \leq
    \omega_\infty^*\!\left(\operatorname{OR}_k(\sfG),\psi\right)
    \leq
    k\omega_\infty^*(\sfG,\psi).
  \end{equation*}
  The finite-copy upper bound
  \begin{equation*}
    \omega_n^*\!\left(\operatorname{OR}_k(\sfG),\psi\right)
    \leq k\omega_n^*(\sfG,\psi)
  \end{equation*}
  also holds for every $n\geq1$.
\end{lemma}

\begin{proof}
  Let $(\rho,\mathscr A,\mathscr B)$ be a single-game strategy with value
  $v=\Val_{\sfG}^{\rho}(\mathscr A,\mathscr B)$.  On $\rho^{\otimes k}$,
  use the product POVMs
  \[
    \widetilde A_{\mathbf a}^{\mathbf x}
      :=\bigotimes_{i=1}^k A_{a_i}^{x_i},
    \qquad
    \widetilde B_{\mathbf b}^{\mathbf y}
      :=\bigotimes_{i=1}^k B_{b_i}^{y_i}.
  \]
  The coordinates are independent under this strategy, so
  \[
    \Val_{\operatorname{OR}_k(\sfG)}^{\rho^{\otimes k}}
      (\widetilde{\mathscr A},\widetilde{\mathscr B})
    =1-\prod_{i=1}^k(1-v)=1-(1-v)^k.
  \]
  Taking the supremum over single-game strategies gives the first lower
  bound.  If $\rho=\psi^{\otimes n}$, this construction uses $nk$ copies of
  $\psi$, and hence
  \[
    \omega_{nk}^*\!\left(\operatorname{OR}_k(\sfG),\psi\right)
    \geq 1-\bigl(1-\omega_n^*(\sfG,\psi)\bigr)^k.
  \]
  Taking the supremum over $n$ gives the lower bound for
  $\omega_\infty^*$.

  For the upper bounds, fix any repeated-game strategy on a state $\rho$,
  and let $E_i$ be the event $V(a_i,b_i\mid x_i,y_i)=1$.  For fixed values
  of all question pairs except $(x_i,y_i)$, the families
  \[
    \left(\sum_{\mathbf a:\,a_i=a}A_{\mathbf a}^{\mathbf x}\right)_
      {a\in\mathcal A},
    \qquad
    \left(\sum_{\mathbf b:\,b_i=b}B_{\mathbf b}^{\mathbf y}\right)_
      {b\in\mathcal B}
  \]
  are POVMs for a single play of $\sfG$ on the same state $\rho$.
  Since the question pairs are independent, the remaining pair
  $(x_i,y_i)$ still has distribution $\mathsf p_{\sfG}$.
  Averaging over the fixed questions therefore gives
  $\Pr(E_i)\leq\omega^*(\sfG)$, and
  $\Pr(E_i)\leq\omega_n^*(\sfG,\psi)$ when $\rho=\psi^{\otimes n}$.
  Thus
  \[
    \Val_{\operatorname{OR}_k(\sfG)}^\rho(\mathscr A,\mathscr B)
    =\Pr\!\left(\bigcup_{i=1}^k E_i\right)
    \leq\sum_{i=1}^k\Pr(E_i).
  \]
  Taking suprema gives the upper bounds for $\omega^*$ and $\omega_n^*$. Taking the
  supremum over $n$ gives the upper bound for $\omega_\infty^*$.
\end{proof}

\begin{lemma}
  \label{lem:and-repetition-lower-bounds}
  For every finite game
  $\sfG=(\mathcal X,\mathcal Y,\mathcal A,\mathcal B,\mathsf p_{\sfG},V)$
  and every $k\geq1$,
  \[
    \omega^*(\sfG^{\otimes k})\geq\omega^*(\sfG)^k.
  \]
  Moreover, for every fixed bipartite resource state $\psi$,
  \[
    \omega_\infty^*(\sfG^{\otimes k},\psi)
    \geq\omega_\infty^*(\sfG,\psi)^k.
  \]
\end{lemma}

\begin{proof}
  Let $(\rho,\mathscr A,\mathscr B)$ be a single-game strategy with value
  $v=\Val_{\sfG}^{\rho}(\mathscr A,\mathscr B)$.  Use $\rho^{\otimes k}$
  and the product POVMs constructed above.  Independence across
  coordinates gives
  \[
    \Val_{\sfG^{\otimes k}}^{\rho^{\otimes k}}
      (\widetilde{\mathscr A},\widetilde{\mathscr B})
    =\prod_{i=1}^k v=v^k.
  \]
  Taking the supremum over single-game strategies proves the first
  inequality.  If $\rho=\psi^{\otimes n}$, the construction uses $nk$
  copies of $\psi$, so
  \[
    \omega_{nk}^*(\sfG^{\otimes k},\psi)
    \geq\omega_n^*(\sfG,\psi)^k.
  \]
  Taking the supremum over $n$ proves the second inequality.
\end{proof}

For an upper bound on the AND repetition, we use the following parallel repetition theorem for general nonlocal games.

\begin{theorem}[{\cite[Theorem~1]{Yuen16}}]
  \label{thm:entangled-parallel-repetition}
  There is an absolute constant $C>0$ with the following property.
  Let $\sfG=(\mathcal X,\mathcal Y,\mathcal A,\mathcal B,\mathsf p_{\sfG},V)$
  be a nonlocal game, let $0<\epsilon<1$, and set
  $s:=\log|\mathcal A\times\mathcal B|$.
  If $\omega^*(\sfG)\leq1-\epsilon$, then for every integer $r\geq2$,
  \begin{equation*}
    \omega^*(\sfG^{\otimes r})
    \leq\frac{C s\log r}{\epsilon^{17}r^{1/4}}.
  \end{equation*}
\end{theorem}

This theorem bounds arbitrary joint entangled strategies for the AND
repetition.  Its dependence on the game is through the soundness gap and
the answer-alphabet sizes, not the question-alphabet sizes.  When the gap
and answer lengths are fixed, a constant number of repetitions therefore
reduces the soundness to any prescribed positive constant. 
\footnote{Very recently, OpenAI has released a proof of a parallel repetition theorem for all nonlocal games with exponential decay\cite{OpenAI2026QuantumParallelRepetition}. However, polynomial decay suffices for our purpose.}

\subsection{Maximal Correlation and Correlation Channels}
\label{subsec:maximal-correlation-and-correlation-channels}

We use Beigi's quantum maximal correlation \cite{Beigi13}, restricted here
to states with maximally mixed marginals.

\begin{definition}[Quantum maximal correlation]
  Let $\psi_{AB}\in \Matrix_{d_A}\otimes \Matrix_{d_B}$ be a bipartite state with
  maximally mixed marginals,
  \(
    \psi_A=\frac{\id}{d_A},
    \psi_B=\frac{\id}{d_B}.
  \)
  Equivalently, the marginal densities relative to the normalized traces are
  both $\id$.  The \emph{quantum maximal correlation} of $\psi_{AB}$ is
  \begin{equation*}
    \rho_{\max}(\psi_{AB})
    :=\max\left\{
      \abs{\Tr\!\left[\psi_{AB}(X\otimes Y^\dagger)\right]}
      \;:\;
      \substack{
        X\in \Matrix_{d_A},\ Y\in \Matrix_{d_B},\\
        \tau_{d_A}(X)=\tau_{d_B}(Y)=0,\quad
        \norm X_2=\norm Y_2=1
      }
    \right\}.
  \end{equation*}
\end{definition}

\begin{fact}[{\cite[Theorems~1 and~2(a)]{Beigi13}}]
  Let $\psi_{AB}\in \Matrix_{d_A}\otimes \Matrix_{d_B}$ be a bipartite state with
  maximally mixed marginals.  Its maximal correlation satisfies
$0\leq\rho_{\max}(\psi_{AB})\leq1$ and tensorizes:
\(
  \rho_{\max}(\psi_{AB}^{\otimes n})
  =\rho_{\max}(\psi_{AB}).
\)
\end{fact}

\begin{fact}[Depolarized EPR pair]~\cite[Lemma 3.9]{QY21}
Let $  \ket{\Phi_2}:=\frac{1}{\sqrt2}(\ket{00}+\ket{11}),$
\begin{equation*}
  \psi_\varepsilon^{\mathrm{depol}}
  :=(1-\varepsilon)\ketbra{\Phi_2}
    +\varepsilon\frac{\id}{2}\otimes\frac{\id}{2},
  \qquad 0<\varepsilon<1.
\end{equation*}
Then $\rho_{\max}(\psi_\varepsilon^{\mathrm{depol}})=1-\varepsilon$.
\end{fact}
It is worth noticing that not all noise decreases the maximal correlation.
\begin{example}[Dephased EPR pair]
For $0\leq\varepsilon\leq1$, let
\begin{equation*}
  \psi_\varepsilon^{\mathrm{deph}}
  :=(1-\varepsilon)\ketbra{\Phi_2}
    +\frac{\varepsilon}{2}
      \bigl(\ketbra{00}+\ketbra{11}\bigr).
\end{equation*}
Take the Pauli observable
$\pauliZ=\ketbra{0}-\ketbra{1}$.  Since
$\tau_2(\pauliZ)=0$, $\norm{\pauliZ}_2=1$, and
$
  \Tr\!\left[
    \psi_\varepsilon^{\mathrm{deph}}(\pauliZ\otimes\pauliZ)
  \right]=1,
$
the definition of maximal correlation and the upper bound
$\rho_{\max}\leq1$ force
$
  \rho_{\max}(\psi_\varepsilon^{\mathrm{deph}})=1.
$
\end{example}

In this work, we associate with every two-qubit state $\psi$ having maximally
mixed marginals a \emph{correlation channel} $\mathcal K_\psi$, which recasts
its maximal correlation in operator-theoretic terms.  Under the normalized
Choi isomorphism~\cite[Section~2.2.2]{tqi}, the shared state $\psi$ is
precisely the Choi state of $\mathcal K_\psi$.  This formulation is closely
related to Beigi's
operator-Schmidt and super-operator description of quantum maximal
correlation~\cite[Theorem~1 and Section~III]{Beigi13}.

\begin{definition}[Correlation channel]\label{def:correlation channel}
  \label{def:correlation-channel}
  Let $\psi\in \Matrix_2\otimes \Matrix_2$ be a two-qubit state with maximally mixed
  marginals.  The \emph{correlation channel} associated with $\psi$ is the
  unique linear map $\mathcal K_\psi\colon \Matrix_2\to \Matrix_2$ satisfying
  \begin{equation*}
    \frac12\Tr\Br{X\mathcal K_\psi(Y)}
    =\Tr\!\left[(X\otimes Y^{\mathsf T})\psi\right],
    \qquad X,Y\in \Matrix_2.
  \end{equation*}
  For $n$ copies, we write
  \(
    \mathcal K_{\psi^{\otimes n}}=\mathcal K_\psi^{\otimes n}.
  \)
\end{definition}

\begin{prop}
  \label{prop:correlation-channel-choi-properties}
  Let $\psi$ and $\mathcal K_\psi$ be as in
  \cref{def:correlation-channel}.  Then:
  \begin{enumerate}
    \item[(i)] $\psi$ is the normalized Choi state of $\mathcal K_\psi$, i.e.
    \(
      \psi
      =
      (\mathcal K_\psi\otimes\operatorname{Id})(\ketbra{\Phi_2}).
    \)

    \item[(ii)] $\mathcal K_\psi$ is completely positive.

    \item[(iii)] $\mathcal K_\psi$ is unital and trace-preserving.
  \end{enumerate}
\end{prop}

\begin{proof}
  Set
  $J=(\mathcal K_\psi\otimes\operatorname{Id})(\ketbra{\Phi_2})$.
  Expanding $\ketbra{\Phi_2}$ in matrix units and using
  \cref{def:correlation-channel}, we obtain, for all
  $X,Y\in \Matrix_2$,
  \[
    \Tr\!\left[(X\otimes Y^{\mathsf T})J\right]
    =\tau_2\!\left(X\mathcal K_\psi(Y)\right)
    =\Tr\!\left[(X\otimes Y^{\mathsf T})\psi\right].
  \]
  Such tensor products span $\Matrix_2\otimes \Matrix_2$, so $J=\psi$, proving (i).

  Since $J=\psi\geq0$, the Choi representation theorem
  \cite[Theorem~2.22]{tqi} implies that $\mathcal K_\psi$ is completely
  positive, which proves (ii).

  Finally, using $\psi_A=\psi_B=\id/2$ in the defining pairing yields
  \[
    \tau_2\!\left(X\mathcal K_\psi(\id)\right)
    =\Tr(X\psi_A)=\tau_2(X),
    \qquad
    \tau_2\!\left(\mathcal K_\psi(Y)\right)
    =\Tr(Y^{\mathsf T}\psi_B)=\tau_2(Y).
  \]
  The normalized trace pairing is nondegenerate, so the first identity,
  valid for every $X$, forces $\mathcal K_\psi(\id)=\id$.  The second
  identity is exactly trace preservation.  This proves (iii).
\end{proof}

Let
\begin{equation}
\label{eq:traceless-qubit-space}
  L_2^0(\tau_2):=\{X\in \Matrix_2:\tau_2(X)=0\}
\end{equation}
be the traceless subspace of $L_2(\tau_2)$, with the inherited inner product.

\begin{prop}
  \label{prop:maximal-correlation-channel-norm}
  If $\psi\in \Matrix_2\otimes \Matrix_2$ has maximally mixed marginals, then
  \begin{equation*}
    \rho_{\max}(\psi)
    =
    \norm{\mathcal K_\psi|_{L_2^0(\tau_2)}}_{2\to2}
    =\sup_{\substack{X\in L_2^0(\tau_2)\\ \norm{X}_{2}=1}}\norm{\mathcal K_\psi(X)}_{2}.
  \end{equation*}
\end{prop}

\begin{proof}
  For every admissible pair $X,Y$ in the definition of maximal correlation,
  the Choi pairing gives
  \[
    \Tr\!\left[\psi(X\otimes Y^\dagger)\right]
    =\tau_2\!\left(X\mathcal K_\psi(\overline Y)\right)
    =\left\langle X^\dagger,\mathcal K_\psi(\overline Y)\right\rangle_2.
  \]
  Adjoint and entrywise conjugation each map the traceless unit sphere
  onto itself.  Since $\mathcal K_\psi$ preserves trace, maximizing the
  absolute value of this inner product gives its operator norm on
  $L_2^0(\tau_2)$.
\end{proof}

Now we introduce a canonical form of two-qubit quantum states.
\begin{definition}
  \label{def:canonical-form}
  A state $\psi_{\boldsymbol\lambda}\in \Matrix_2\otimes \Matrix_2$ is in
  \emph{canonical form}, or is a \emph{canonical state}, if
  \begin{equation*}
    \psi_{\boldsymbol\lambda}
    =\frac14\sum_{s\in G}\lambda_sP_s\otimes P_s^{\mathsf T},
    \qquad \lambda_s\in\mathbb R,\quad \lambda_0=1.
  \end{equation*}
  Its \emph{canonical correlation channel} is the correlation channel
  $\mathcal K:=\mathcal K_{\psi_{\boldsymbol\lambda}}\colon \Matrix_2\to \Matrix_2$;
  equivalently, $\mathcal K(P_s)=\lambda_sP_s$ for every $s\in G$.
  A canonical form of a state $\psi$ is a canonical state locally
  unitarily equivalent to $\psi$; the canonical form of a state need not be unique.
\end{definition}


For the Pauli basis $(P_s)_{s\in\mathbb F_2^2}$, we call a state
$\psi_{\boldsymbol\lambda}=\frac14\sum_{s\in\mathbb F_2^2}
\lambda_sP_s\otimes P_s^{\mathsf T}$, with real $\lambda_s$ and
$\lambda_0=1$, a \emph{canonical state}; this expression is its
\emph{canonical form}.  Its \emph{canonical correlation channel} is
$\mathcal K_{\psi_{\boldsymbol\lambda}}$.
The following proposition shows that local unitaries transform any
$\psi\in \Matrix_2\otimes \Matrix_2$ with maximally mixed marginals into this form
without changing game values.
This technique has been used extensively in the literature~\cite{904522,BETHRUSKAI2002159,PhysRevA.83.052108,king2014hypercontractivity}.
\begin{restatable}[Local-unitary canonical form]{prop}{localunitarycanonicalform}
  \label{prop:pauli-canonicalization}
  Let $\psi\in \Matrix_2\otimes \Matrix_2$ satisfy
  $\psi_A=\psi_B=\id/2$, and let
  $(P_s)_{s\in\mathbb F_2^2}$ be the Pauli basis.  There
  are single-qubit local unitaries $U_A,U_B$, and real numbers
  $\boldsymbol{\lambda}=(\lambda_s)$, $s\in\mathbb F_2^2$, with
  $\lambda_0=1$, such that
  \begin{enumerate}
    \item The state $\psi$ is locally equivalent to a canonical state
    \[
    \psi_{\boldsymbol\lambda}
    :=(U_A\otimes U_B)\psi(U_A\otimes U_B)^\dagger
    =\frac14\sum_{s\in\mathbb F_2^2}\lambda_s P_s\otimes P_s^{\mathsf T}.
    \]
  \item For every $n\geq1$ and every
  $\alpha,\beta\in(\mathbb F_2^2)^n$, we have
  \(
    \Tr\!\left[
      (P_\alpha\otimes P_\beta^{\mathsf T})
      \psi_{\boldsymbol\lambda}^{\otimes n}
    \right]
    =\mathbf 1_{\{\alpha=\beta\}}
      \prod_{i=1}^n\lambda_{\alpha_i}.
  \)
  \item The canonical correlation channel
  $\mathcal K:=\mathcal K_{\psi_{\boldsymbol\lambda}}$ satisfies
  \[
    \mathcal K(P_s)=\lambda_sP_s,
    \qquad
    \rho_{\max}(\psi)=\max_{s\neq0}|\lambda_s|.
  \]
  \item The game value is invariant under the corresponding local-unitary
  change of basis.  More precisely, for every game
  $\sfG=(\mathcal X,\mathcal Y,\mathcal A,\mathcal B,\mathsf p_{\sfG},V)$,
  every $n\geq1$, and every POVM strategy $(\mathscr A,\mathscr B)$ for
  $\sfG$ on $\psi^{\otimes n}$, define
  $C_a^x
    :=U_A^{\otimes n}A_a^x(U_A^{\otimes n})^\dagger,$
  $D_b^y
    :=\overline{U_B}^{\otimes n}B_b^y
      (\overline{U_B}^{\otimes n})^\dagger.$
  Write $\mathscr C=(C_a^x)_{x,a}$ and
  $\mathscr D=(D_b^y)_{y,b}$.  Then $(\mathscr C,\mathscr D)$ is a POVM
  strategy on $\psi_{\boldsymbol\lambda}^{\otimes n}$, and
  \(
    \Val_{\sfG}^{\psi_{\boldsymbol\lambda}^{\otimes n}}
      (\mathscr C,\mathscr D)
    =
    \Val_{\sfG}^{\psi^{\otimes n}}(\mathscr A,\mathscr B).
  \)
  \end{enumerate}

\end{restatable}

\begin{proof}
    By \cite[Equation 15]{PhysRevA.83.052108},
    using local unitaries $U_A$ and $U_B$,
    we can recast the density matrix for an arbitrary two–qubit state
    in the Bloch normal form
    \begin{equation*}
    \psi_{\boldsymbol\lambda}
    =(U_A\otimes U_B)\psi(U_A\otimes U_B)^\dagger
    =\frac14\br{\id_{4} 
      + \sum_{s\neq 0}a_sP_s\otimes\id_2
      + \sum_{s\neq 0}b_s\id_2\otimes P_s
      + \sum_{s\neq 0}c_sP_s\otimes P_s}.
    \end{equation*}
    Since $\psi$ has maximally mixed marginals, so does $\psi_{\boldsymbol\lambda}$.
    This implies that $a_s = b_s = 0$ for all $s\in\mathbb{F}_2^2$ and $s\neq 0$.
    Since $X=X^T, Y=-Y^T$, and $Z=Z^T$,
    this concludes item 1.
    We can then verify
    item 2, item 3, and item 4 easily by their definition.
    
    In particular,
    for item 2, we have
    for $s,t\in\mathbb F_2^2$, the diagonal form reads
$\Tr[(P_s\otimes P_t^{\mathsf T})\psi_{\boldsymbol\lambda}]
=\mathbf1_{\{s=t\}}\lambda_s$.  Taking the product of this identity over
the $n$ tensor factors proves the second assertion.

For item 3, for any $s, t\in\mathbb{F}^2_2$.
By the definition of correlation channel \cref{def:correlation channel},
\begin{equation*}
    \frac{1}{2}\Tr\Br{P_t\mathcal{K}(P_s)} = \Tr\Br{\br{P_t\otimes P_s^T}\psi_{\boldsymbol{\lambda}}} = \mathbf1_{\{s=t\}}\lambda_s,
\end{equation*}
which implies $\mathcal{K}(P_s) = \lambda_sP_s$.
The maximal correlation is invariant under local transformations.
So, by
\cref{prop:maximal-correlation-channel-norm},
\begin{equation*}
  \rho_{\max}(\psi) = \rho_{\max}(\psi_{\boldsymbol{\lambda}}) =\sup_{\substack{X: \tau\Br{X}=0\\ \norm{X}_{2}=1}}\norm{\mathcal K(X)}_{2} = \sup_{\substack{X: \tau\Br{X}=0\\ \norm{X}_{2}=1}}\sqrt{\sum_{s\neq 0}\abs{\lambda_s}^2\abs{\widehat{X}(s)}^2} \le \max_{s\neq 0}\abs{\lambda_s}.
\end{equation*}

For item 4,
fix $n\geq1$ and a POVM strategy $(\mathscr A,\mathscr B)$ on
$\psi^{\otimes n}$, and define $C_a^x,D_b^y$ as in the statement.  Unitary
conjugation preserves positivity, and
$\sum_aC_a^x=U_A^{\otimes n}\id(U_A^{\otimes n})^\dagger=\id$.
Similarly, $\sum_bD_b^y=\id$.
Moreover,
$(D_b^y)^{\mathsf T}
=U_B^{\otimes n}(B_b^y)^{\mathsf T}(U_B^{\otimes n})^\dagger$.
Thus
$(\mathscr C,\mathscr D)$ is valid a POVM strategy.  For every $x,y,a,b$,
  the definitions of $\psi_{\boldsymbol\lambda}$ and cyclicity of the trace
  imply
\[
  \begin{aligned}
    \Tr\!\left[
      (C_a^x\otimes(D_b^y)^{\mathsf T})\psi_{\boldsymbol\lambda}^{\otimes n}
    \right]
    &=\Tr\!\left[
      (U_A^{\otimes n}\otimes U_B^{\otimes n})
      (A_a^x\otimes(B_b^y)^{\mathsf T})\psi^{\otimes n}
      (U_A^{\otimes n}\otimes U_B^{\otimes n})^\dagger
    \right] \\
    &=\Tr\!\left[
      (A_a^x\otimes(B_b^y)^{\mathsf T})\psi^{\otimes n}
    \right].
  \end{aligned}
\]
over $x,y,a,b$ to obtain the fourth assertion.
\end{proof}



\subsection{Correlation Function}
\label{subsec:bob-coordinate-convention}

Fix the canonical state $\psi_{\boldsymbol\lambda}$ from
\cref{prop:pauli-canonicalization}, put
$\mathcal K:=\mathcal K_{\psi_{\boldsymbol\lambda}}$, and write
\begin{equation}
\label{eq:canonical-correlation-weights}
  \mathcal K_n:=\mathcal K^{\otimes n},
  \qquad
  w_{\boldsymbol\lambda}(\alpha)
  :=\prod_{i=1}^n\lambda_{\alpha_i}.
\end{equation}
The canonical form and the diagonal action of the correlation channel are
\begin{equation*}
  \psi_{\boldsymbol\lambda}^{\otimes n}
  =\frac{1}{4^n}\sum_{\alpha\in G^n}
    w_{\boldsymbol\lambda}(\alpha)
    P_\alpha\otimes P_\alpha^{\mathsf T},
  \qquad
  \mathcal K_n(P_\alpha)
  =w_{\boldsymbol\lambda}(\alpha)P_\alpha.
\end{equation*}

\begin{definition}
\label{def:correlation-functional}
For $X,Y\in\Matrix_2^{\otimes n}$, define
the \emph{correlation function} 
\begin{align*}
  \Corr_{\psi,n}(X,Y)
  &:=\tau_{2^n}\!\left(X\mathcal K_n(Y)\right) \notag\\
  &=\Tr\!\left[(X\otimes Y^{\mathsf T})
    \psi_{\boldsymbol\lambda}^{\otimes n}\right] \notag\\
  &=\sum_{\alpha\in G^n}
    w_{\boldsymbol\lambda}(\alpha)
    \widehat X(\alpha)\widehat Y(\alpha).
\end{align*}
\end{definition}
The second expression is the tensorized Choi representation, and the final
one follows from the canonical Pauli expansion above and Pauli
orthogonality.

\begin{prop}
  \label{prop:corr-contraction}
  Let $n\geq1$, and let $\Corr_{\psi,n}$ be the correlation function
  defined in \cref{def:correlation-functional} for $n$ copies of the
  canonical state $\psi_{\boldsymbol\lambda}$.
  Then, for
  every $X,Y\in \Matrix_2^{\otimes n}$,
\[
  \abs{\Corr_{\psi,n}(X,Y)}
  \leq\norm X_2\norm Y_2.
\]
\end{prop}

\begin{proof}
  Each canonical coefficient satisfies
  \[
    |\lambda_s|
    =\abs{\Tr\!\left[
      (P_s\otimes P_s^{\mathsf T})\psi_{\boldsymbol\lambda}
    \right]}
    \leq\norm{P_s\otimes P_s^{\mathsf T}}_\infty
    =1,
  \]
  because $\psi_{\boldsymbol\lambda}$ is a state and every Pauli matrix is
  unitary.  Consequently,
  $|w_{\boldsymbol\lambda}(\alpha)|
   =\prod_i|\lambda_{\alpha_i}|\leq1$.

  By the Pauli expansion of the correlation function,
  \[
    \Corr_{\psi,n}(X,Y)
    =\sum_{\alpha\in G^n}
      w_{\boldsymbol\lambda}(\alpha)
      \widehat X(\alpha)\widehat Y(\alpha).
  \]
  We estimate this sum in two steps:
  \begin{align*}
    \abs{\Corr_{\psi,n}(X,Y)}
    &\leq
      \left(\sum_{\alpha\in G^n}|\widehat X(\alpha)|^2\right)^{1/2}
      \left(\sum_{\alpha\in G^n}
        |w_{\boldsymbol\lambda}(\alpha)|^2
        |\widehat Y(\alpha)|^2\right)^{1/2}\\
    &\leq
      \left(\sum_{\alpha\in G^n}|\widehat X(\alpha)|^2\right)^{1/2}
      \left(\sum_{\alpha\in G^n}|\widehat Y(\alpha)|^2\right)^{1/2},
  \end{align*}
  The first inequality is Cauchy--Schwarz, and the second uses
  $|w_{\boldsymbol\lambda}(\alpha)|\leq1$.  Pauli Parseval identifies the
  two factors in the last line with $\norm X_2$ and $\norm Y_2$,
  respectively.
\end{proof}

By \eqref{eq:strategy-value}, we have
\begin{prop}
  \label{prop:game-value-correlation-functional}
  Fix a game
  $\sfG=(\mathcal X,\mathcal Y,\mathcal A,\mathcal B,\mathsf p_{\sfG},V)$.
  Let $n\geq1$, and let $\mathscr A=(A_a^x)_{x,a}$ and
  $\mathscr B=(B_b^y)_{y,b}$ be a POVM strategy for $\sfG$ on
  $\psi_{\boldsymbol\lambda}^{\otimes n}$.  Then
  \begin{equation*}
    \Val_{\sfG}^{\psi_{\boldsymbol\lambda}^{\otimes n}}
      (\mathscr A,\mathscr B)
    =\sum_{x,y,a,b}\mathsf p_{\sfG}(x,y)V(a,b\mid x,y)
      \Corr_{\psi,n}(A_a^x,B_b^y).
  \end{equation*}
\end{prop}

\subsection{POVM block operators}
\label{subsec:game-value-for-block-operators}


In this work we will be working heavily on POVM measurements
and their approximations.
A POVM measurement is a set of positive semidefinite operators.
To handle the operators in a POVM measurement as a single mathematical object,
and to enforce the positive semidefinite property,
in this work we will represent POVM measurements by \emph{block operators}:
For a block operator $Q=\sum_a|a\rangle\otimes Q_a$, the operators
$Q_a^\dagger Q_a$ are positive semidefinite and satisfy
\begin{equation*}
  \sum_aQ_a^\dagger Q_a=Q^\dagger Q.
\end{equation*}
Thus $(Q_a^\dagger Q_a)_a$ is a POVM measurement whenever $Q^\dagger Q=\id$.  We call the
block operators $Q$ arising in the approximation steps below \emph{block
approximants}.
Fix a game
$\sfG=(\mathcal X,\mathcal Y,\mathcal A,\mathcal B,\mathsf p_{\sfG},V)$.
For a POVM strategy $(\mathscr A,\mathscr B)$ for $\sfG$,
its \emph{block operator} representation is thus
\begin{equation}
\label{eq:povm-block-operators}
  W_{\mathscr A}^x
  :=\sum_{a\in\mathcal A}|a\rangle\otimes\sqrt{A_a^x},
  \qquad
  W_{\mathscr B}^y
  :=\sum_{b\in\mathcal B}|b\rangle\otimes\sqrt{B_b^y}.
\end{equation}
They satisfy
$(W_{\mathscr A}^x)^\dagger W_{\mathscr A}^x=\id$ and
$(W_{\mathscr B}^y)^\dagger W_{\mathscr B}^y=\id$.
For block operators
$Q^x=\sum_{a\in\mathcal A}|a\rangle\otimes Q_a^x$ and
$R^y=\sum_{b\in\mathcal B}|b\rangle\otimes R_b^y$, they define the POVM strategy
\begin{equation}
\label{eq:block-operator-families}
  \begin{aligned}
  \mathscr Q&=\bigl(\bigl((Q_a^x)^\dagger Q_a^x\bigr)_{a\in\mathcal A}\bigr)_{x\in\mathcal X},\\
  \mathscr R&=\bigl(\bigl((R_b^y)^\dagger R_b^y\bigr)_{b\in\mathcal B}\bigr)_{y\in\mathcal Y}.
  \end{aligned}
\end{equation}
The value formula in \eqref{eq:strategy-value} gives
\begin{equation}
\label{eq:block-game-value}
  \begin{aligned}
  \Val_{\sfG}^\psi(\mathscr Q,\mathscr R)
  =\sum_{x,y,a,b}&\mathsf p_{\sfG}(x,y)V(a,b\mid x,y)\\
  &\cdot\Tr\Br{\psi\bigl((Q_a^x)^\dagger Q_a^x\otimes
       ((R_b^y)^\dagger R_b^y)^{\mathsf T}\bigr)}.
  \end{aligned}
\end{equation}
If $(Q^x)^\dagger Q^x=\id$ and $(R^y)^\dagger R^y=\id$ questionwise, these
families are POVMs.  In particular, taking $Q^x=W_{\mathscr A}^x$ and
$R^y=W_{\mathscr B}^y$ gives
\begin{equation*}
  \mathscr Q=\mathscr A,\qquad \mathscr R=\mathscr B.
\end{equation*}
For approximately normalized families $\mathscr Q$ and $\mathscr R$,
the same value formula applies, although its value need not be a probability.


For fixed questions $(x,y)\in\mathcal X\times\mathcal Y$ and block operators
$Q_1,Q_2,R_1,R_2$ on $n$ copies, define the \emph{block correlation function} as
\begin{equation}
\label{eq:block-correlation-functional}
  \cG_{x,y}^{\psi^{\otimes n}}(Q_1,Q_2;R_1,R_2)
  :=\sum_{a\in\mathcal A,\,b\in\mathcal B}V(a,b\mid x,y)
  \Tr\!\left[
    \psi^{\otimes n}
    \bigl(Q_{1,a}^\dagger Q_{2,a}\otimes
          (R_{1,b}^\dagger R_{2,b})^{\mathsf T}\bigr)
  \right].
\end{equation}
Then
\begin{equation*}
  \Val_{\sfG}^{\psi^{\otimes n}}(\mathscr Q,\mathscr R)
  =\mathbb E_{(x,y)\sim\mathsf p_{\sfG}}\cG_{x,y}^{\psi^{\otimes n}}(Q^x,Q^x;R^y,R^y).
\end{equation*}
Recall that, for a block operator $Q$ on $n$ qubits,
\begin{equation*}
  \norm Q_4=\tau_{2^n}\!\left((Q^\dagger Q)^2\right)^{1/4}.
\end{equation*}
Similar to \cref{prop:corr-contraction}, we have the following bound for the
block correlation function:
\begin{lemma}
  \label{lem:game-expression-holder}
  Fix a game
  $\sfG=(\mathcal X,\mathcal Y,\mathcal A,\mathcal B,\mathsf p_{\sfG},V)$.
  Let $\psi\in\Matrix_2\otimes\Matrix_2$ be a state with
  $\psi_A=\psi_B=\id/2$.
  Let $n\geq1$ and $(x,y)\in\mathcal X\times\mathcal Y$.  For $j=1,2$, let
  $Q_j\colon(\mathbb C^2)^{\otimes n}\to
    \mathbb C^{\mathcal A}\otimes(\mathbb C^2)^{\otimes n}$
  and $R_j\colon(\mathbb C^2)^{\otimes n}\to
    \mathbb C^{\mathcal B}\otimes(\mathbb C^2)^{\otimes n}$
  be block operators with blocks $(Q_{j,a})_{a\in\mathcal A}$ and
  $(R_{j,b})_{b\in\mathcal B}$, respectively.
   For the function 
  $\cG_{x,y}^{\psi^{\otimes n}}$ in
  \eqref{eq:block-correlation-functional}, one has
  \begin{equation*}
    \abs{\cG_{x,y}^{\psi^{\otimes n}}(Q_1,Q_2;R_1,R_2)}
    \leq\prod_{j=1}^2\norm{Q_j}_4\norm{R_j}_4,
  \end{equation*}
\end{lemma}

\begin{proof}
  For each answer pair $(a,b)\in\mathcal A\times\mathcal B$,
  \begin{align*}
    Q_{1,a}^\dagger Q_{2,a}\otimes
      (R_{1,b}^\dagger R_{2,b})^{\mathsf T}
    &=Q_{1,a}^\dagger Q_{2,a}\otimes
      R_{2,b}^{\mathsf T}\overline{R_{1,b}}\\
    &=(Q_{1,a}\otimes\overline{R_{2,b}})^\dagger
      (Q_{2,a}\otimes\overline{R_{1,b}}).
  \end{align*}
  Thus Cauchy--Schwarz gives
  \begin{align*}
    &\left|\Tr\!\left[\psi^{\otimes n}
      \bigl(Q_{1,a}^\dagger Q_{2,a}\otimes
        (R_{1,b}^\dagger R_{2,b})^{\mathsf T}\bigr)\right]\right|\\
    &\qquad\leq
      \Tr\!\left[\psi^{\otimes n}
        \bigl(Q_{1,a}^\dagger Q_{1,a}\otimes
          (R_{2,b}^\dagger R_{2,b})^{\mathsf T}\bigr)\right]^{1/2}\\
    &\qquad\quad\times
      \Tr\!\left[\psi^{\otimes n}
        \bigl(Q_{2,a}^\dagger Q_{2,a}\otimes
          (R_{1,b}^\dagger R_{1,b})^{\mathsf T}\bigr)\right]^{1/2}.
  \end{align*}
  Multiply by $V(a,b\mid x,y)$ and sum over $(a,b)$.  The triangle
  inequality and another application of Cauchy--Schwarz yield
  \begin{align*}
    \abs{\cG_{x,y}^{\psi^{\otimes n}}(Q_1,Q_2;R_1,R_2)}
    &\leq
    \left(
      \sum_{a,b}V(a,b\mid x,y)
      \Tr\!\left[\psi^{\otimes n}
        \bigl(Q_{1,a}^\dagger Q_{1,a}\otimes
          (R_{2,b}^\dagger R_{2,b})^{\mathsf T}\bigr)\right]
    \right)^{1/2}\\
    &\quad\times
    \left(
      \sum_{a,b}V(a,b\mid x,y)
      \Tr\!\left[\psi^{\otimes n}
        \bigl(Q_{2,a}^\dagger Q_{2,a}\otimes
          (R_{1,b}^\dagger R_{1,b})^{\mathsf T}\bigr)\right]
    \right)^{1/2}\\
    &\leq
    \Tr\!\left[\psi^{\otimes n}
      \bigl(Q_1^\dagger Q_1\otimes
        (R_2^\dagger R_2)^{\mathsf T}\bigr)\right]^{1/2}
    \Tr\!\left[\psi^{\otimes n}
      \bigl(Q_2^\dagger Q_2\otimes
        (R_1^\dagger R_1)^{\mathsf T}\bigr)\right]^{1/2},
  \end{align*}
  where the second inequality uses $0\leq V(a,b\mid x,y)\leq1$.

  For any block operators $Q$ and $R$, Cauchy--Schwarz yields
  \begin{align*}
    &\Tr\!\left[\psi^{\otimes n}
      \bigl(Q^\dagger Q\otimes(R^\dagger R)^{\mathsf T}\bigr)\right]\\
    &\qquad\leq
      \Tr\!\left[\psi_A^{\otimes n}(Q^\dagger Q)^2\right]^{1/2}
      \Tr\!\left[\psi_B^{\otimes n}
        \bigl((R^\dagger R)^{\mathsf T}\bigr)^2\right]^{1/2}\\
    &\qquad=
      \norm Q_4^2\norm R_4^2,
  \end{align*}
  where the last equality uses
  $\psi_A=\psi_B=\id/2$ and invariance of the trace under transposition.
  Applying this estimate to the two factors above proves
  \cref{lem:game-expression-holder}.
\end{proof}

\section{Proof of the Main Theorem}
\label{sec:proof-main-theorem}

\subsection{Proof Outline}

The two alternatives in \cref{thm:main} require separate arguments.
When $\rho_{\max}(\psi)<1$, the measurements are smoothed by the
depolarizing semigroup, their square-root block operators are approximated by
low-degree operators, and these operators are compressed directly in the
Pauli basis before being rounded back to POVMs.  At
$\rho_{\max}(\psi)=1$, positivity forces the canonical form
$\Psi_\lambda$.  A perfect Hamming-code partition and bilateral Pauli
twirling then produce an exact, unflagged maximally entangled component of
weight at least $\lambda^2$.  We now state the principal results and
intermediate statements underlying these two branches.  Their proofs appear
in the subsequent sections; together, they imply \cref{thm:main}.

\begin{figure}[H]
  \centering
  \begin{tikzpicture}[
      x=1cm,
      y=1cm,
      process box/.style={draw=blue!55!black,fill=blue!3,
        rounded corners=2pt,align=center,font=\footnotesize,
        text width=46mm,minimum height=10.5mm,
        inner xsep=3.5pt,inner ysep=3pt},
      wide process box/.style={process box,text width=62mm},
      branch box/.style={diamond,aspect=2.25,
        draw=orange!75!black,fill=orange!12,very thick,
        align=center,font=\footnotesize\bfseries,
        minimum width=28mm,minimum height=9mm,
        inner xsep=1.5pt,inner ysep=1pt},
      route arrow/.style={-{Latex[length=1.6mm]},semithick},
      route label/.style={font=\scriptsize,fill=white,inner sep=0.8pt},
      nexp result/.style={draw=blue!65!black,fill=blue!8,
        rounded corners,very thick,font=\small\bfseries,
        minimum width=18mm,minimum height=9mm},
      re result/.style={draw=red!65!black,fill=red!8,
        rounded corners,very thick,font=\small\bfseries,
        minimum width=18mm,minimum height=9mm}
    ]
    \node[wide process box] (canonical) at (6.7,0)
      {\textbf{Local-Unitary Canonicalization}\\
       $\psi\mapsto\psi_{\boldsymbol\lambda}$ in Diagonal Pauli Form\\[-1pt]
       {\scriptsize \cref{prop:pauli-canonicalization}}};

    \node[branch box] (noisy) at (2.5,-1.55)
      {$\rho_{\max}(\psi)<1$};
    \node[process box] (smooth) at (2.5,-3.05)
      {\textbf{Smoothing}\\[-1pt]
       {\scriptsize \cref{thm:whole-smoothing}}};
    \node[process box] (degree) at (2.5,-4.45)
      {\textbf{Low-Degree Approximation}\\[-1pt]
       {\scriptsize \cref{thm:low-degree-approximation-new}}};
    \node[process box] (reduce) at (2.5,-5.85)
      {
       \textbf{Dimension Reduction}\\[-1pt]
       {\scriptsize \cref{thm:folded-block-output}}};
    \node[process box] (round) at (2.5,-7.35)
      {\textbf{Polar Rounding}\\[-1pt]
       {\scriptsize \cref{lem:polar-soundness}}};
    \node[process box] (verification) at (2.5,-8.95)
      {\textbf{Succinct Verification}\\[-1pt]
       {\scriptsize \cref{prop:succinct-verifier}}};
    \draw[route arrow] (canonical) -- (noisy);
    \draw[route arrow] (smooth) -- (degree);
    \draw[route arrow] (noisy) -- (smooth);
    \draw[route arrow] (degree) -- (reduce);
    \draw[route arrow] (reduce) -- (round);
    \draw[route arrow] (round) -- (verification);

    \node[branch box] (boundary) at (10.9,-1.55)
      {$\rho_{\max}(\psi)=1$};
    \node[process box] (boundary-form) at (10.9,-3.05)
      {\textbf{Boundary Canonical Form}\\
       Positivity Forces $\psi_{\boldsymbol\lambda}=\Psi_\lambda$\\[-1pt]
       {\scriptsize \cref{prop:boundary-canonical}}};
    \draw[route arrow] (canonical) -- (boundary);
    \draw[route arrow] (boundary) -- (boundary-form);

    \node[branch box] (separable) at (7.0,-4.5) {Separable};
    \draw[route arrow] (boundary-form) --
      node[route label,above] {$\lambda=0$} (separable);

    \node[branch box] (entangled) at (10.9,-4.5) {Entangled};
    \node[process box] (hamming) at (10.9,-6.05)
      {\textbf{Hamming Partition}\\[-1pt]
       {\scriptsize \cref{lem:boundary-hamming-partition,lem:boundary-local-extraction}}};
    \node[process box] (twirl) at (10.9,-7.65)
      {\textbf{Bilateral Pauli Twirling}\\
       {\scriptsize \cref{thm:boundary-epr-component}}};
    \node[process box] (simulation) at (10.9,-9.25)
      {\textbf{One-Sided Simulation}\\
       \textbf{and Amplification}\\[-1pt]
       {\scriptsize \cref{cor:boundary-epr-simulation,lem:or-repetition-bounds,thm:boundary-re-lower-bound}}};
    \draw[route arrow] (boundary-form) --
      node[route label,right] {$\lambda>0$} (entangled);
    \draw[route arrow] (entangled) -- (hamming);
    \draw[route arrow] (hamming) -- (twirl);
    \draw[route arrow] (twirl) -- (simulation);

    \node[nexp result] (nexp) at (2.5,-10.95) {$\NEXP$};
    \node[re result] (re) at (10.9,-10.95) {$\RE$};
    \draw[route arrow] (verification) -- (nexp);
    \draw[route arrow] (separable) -- (7.0,-10.95) -- (nexp.east);
    \draw[route arrow] (simulation) -- (re);
  \end{tikzpicture}
  \caption{Roadmap for the proof of \cref{thm:main}}
  \label{fig:technical-contribution-routes}
\end{figure}

\subsection{\texorpdfstring{$\rho_{\max}<1$ Implies $\NEXP$}
  {rho-max < 1 Implies NEXP}}

\begin{theorem}[Finite-copy approximation]
  \label{thm:finite-copy-approximation}
  There are absolute constants $C,\gamma>0$ such that the following
holds. Let $\psi\in \Matrix_2\otimes \Matrix_2$ be a fixed state with maximally
mixed marginals $\psi_A=\psi_B=\id/2$ and $\rho_{\max}(\psi)<1$.
There exists a constant $\alpha_\psi\in[0,1)$, depending only on
$\psi$, with the following property.

Let $\sfG=(\mathcal X,\mathcal Y,\mathcal A,\mathcal B,
\mathsf p_{\sfG},V)$ be a nonlocal game, set
$t:=\max\{|\mathcal A|,|\mathcal B|\}$, and let
$0<\varepsilon<1/10$. Define
\[
  k=\left\lceil
    \left(\frac{C(\log(2t))^2}
    {(1-\alpha_\psi)\varepsilon^{14}}\right)^{1/\gamma}
  \right\rceil,
  \qquad
  D=\left\lceil
    Ck^{16}\max\left\{
      \frac{t^2}{\varepsilon^2},
      \frac{t}{\varepsilon^6}
    \right\}
  \right\rceil.
\]
Then, for every $n\geq1$ and every local POVM strategy
$\mathscr A,\mathscr B$ on $\psi^{\otimes n}$, there exists
a local POVM strategy $\mathscr A^\sharp,\mathscr B^\sharp$
on $\psi^{\otimes D}$ satisfying
\[
  \Val_{\sfG}^{\psi^{\otimes D}}
    (\mathscr A^\sharp,\mathscr B^\sharp)
  \geq
  \Val_{\sfG}^{\psi^{\otimes n}}
    (\mathscr A,\mathscr B)-C\varepsilon.
\]
\end{theorem}

Below we state the necessary steps to prove \cref{thm:finite-copy-approximation}, and provide the proof of \cref{thm:finite-copy-approximation} at the end of this subsection.

\paragraph{Step 1: Smoothing.}


First we show that adding a small amount of noise to the original strategies, i.e., \emph{smoothing},
will not alter the game value too much.
\begin{restatable}{theorem}{wholesmoothing}
  \label{thm:whole-smoothing}
  Fix a game
  $\sfG=(\mathcal X,\mathcal Y,\mathcal A,\mathcal B,\mathsf p_{\sfG},V)$
  and an integer $n\geq1$, and use the
  resource state $\psi^{\otimes n}$, where $\psi\in \Matrix_2\otimes \Matrix_2$ satisfies
  $\psi_A=\psi_B=\id/2$ and $\rho_{\max}(\psi)<1$.
  Let $\mathscr A=(\mathscr A^x)_{x\in\mathcal X}$ and
  $\mathscr B=(\mathscr B^y)_{y\in\mathcal Y}$ be local POVM strategies on
  $(\mathbb C^2)^{\otimes n}$, and put
  $t=\max\{|\mathcal A|,|\mathcal B|\}$.  Let
  $0<w\leq1/4$.  Let $\mathcal T_u$ be the
  $n$-qubit depolarizing semigroup.  For arbitrary choices
  $u_x\in[w,2w]$ and $v_y\in[w,2w]$, define the smoothed POVMs
  \begin{align}\label{eq:smoothed-definition}
    A_a^{x,u_x}:=\mathcal T_{u_x}(A_a^x),
    &\qquad
    \mathscr A^{x,u_x}:=(A_a^{x,u_x})_{a\in\mathcal A},
    \notag \\
    B_b^{y,v_y}:=\mathcal T_{v_y}(B_b^y),
    &\qquad
    \mathscr B^{y,v_y}:=(B_b^{y,v_y})_{b\in\mathcal B}.
  \end{align}
  Write $\mathscr A^{(u)}=(\mathscr A^{x,(u)})_x$ and
  $\mathscr B^{(v)}=(\mathscr B^{y,(v)})_y$.  Then
  \begin{equation*}
    \abs{
      \Val_{\sfG}(\mathscr A,\mathscr B)
      -\Val_{\sfG}(\mathscr A^{(u)},\mathscr B^{(v)})
    }
    \leq
    2\sqrt{
      \frac{\alpha_\psi}{1-\alpha_\psi}\,
      w\log(2t)
    }.
  \end{equation*}
  Here $0\leq\alpha_\psi<1$ is the complete-SDPI constant from
  \cref{lem:complete-sdpi}.
\end{restatable}


Note that the map $\mathcal{T}_u$ is completely positive and trace preserving,
so after applying $\mathcal{T}_u$,
the output operators automatically form POVM measurements.
Also, \cref{thm:whole-smoothing} allows us to choose the noise parameters $u_x$ (resp. $v_y$) freely in the regime $[w, 2w]$ for each individual $x\in\mathcal{X}$ (resp. $y\in\mathcal{Y}$).
This facilitates further processing.

\paragraph{Step 2: Positive Low-Degree Approximation.}

The following theorem shows that our smoothed strategies can further be approximated by low-degree operators.
However, the low-degree truncation process does not automatically preserve positivity.
So we adapt the \emph{block operator} representation of POVM's defined in \cref{subsec:game-value-for-block-operators}.
In particular, we will construct low-degree operators $Q^x_a$ and $R^y_b$,
so that the approximate POVM's formed by the PSD operators $\br{Q^x_a}^\dagger Q^x_a$
and $\br{R^y_b}^\dagger R^y_b$ preserve the game value.

For an $x\in\mathcal{X}$,
we will represent the operators $\br{Q^x_a}_a$
succinctly by a block operator
$Q^x = \sum_{a}\ket{a}\otimes Q^x_a$
as in \cref{subsec:game-value-for-block-operators},
and keep using this representation in the future steps,
to preserve positivity of the POVM in the game strategies.

\begin{restatable}[Low-degree approximation]{theorem}{lowdegreeapproximationnew}
\label{thm:low-degree-approximation-new}
  There are absolute constants $C, \gamma$ with the following property.
  Let $n \ge 1, t\ge 2, w > 0, \eta\in(0,1)$.
  Fix a game $\sfG=(\mathcal X,\mathcal Y,\mathcal A,\mathcal B,\mathsf p_{\sfG},V)$ such that $t=\max\set{\abs{\mathcal A}, \abs{\mathcal B}}$, a two-qubit state $\psi$ satisfying $\psi_A = \psi_B = \id/2$.
  Let $\mathscr A=(\mathscr A^x)_{x\in\mathcal X}$ and
  $\mathscr B=(\mathscr B^y)_{y\in\mathcal Y}$ be local POVM strategies on
  $(\mathbb C^2)^{\otimes n}$.

  Let $k$ be an integer greater than $\br{\frac{C\log\br{2t}}{w\eta^4}}^{1/\gamma}$.
  There exist smoothing-time families $u=(u_x)_{x\in\mathcal X}$ and
  $v=(v_y)_{y\in\mathcal Y}$ with $u_x,v_y\in[w,2w]$ for all $x,y$,
  and families of block operators
  $Q^x=\sum_a\ket{a}\otimes Q_a^x$ and
  $R^y=\sum_b\ket{b}\otimes R_b^y$, such that, for the smoothed strategies
  $\mathscr A^{(u)}=(\mathscr A^{x,u_x})_x$ and
  $\mathscr B^{(v)}=(\mathscr B^{y,v_y})_y$ defined in
  \cref{eq:smoothed-definition}, the following hold.
  \begin{itemize}
      \item Each operator $Q^x_a$ and $R^y_b$ has Pauli degree at most $k$.
      \item
        For any $x\in\mathcal X$ and $y\in\mathcal Y$, we have
        \begin{equation*}
          \norm{\br{Q^x}^\dagger Q^x - \id}_2 \le \br{2+\eta}\eta
          \quad\text{and}\quad
          \norm{\br{R^y}^\dagger R^y - \id}_2 \le \br{2+\eta}\eta.
        \end{equation*}
      \item
        Define approximate strategies $\mathscr Q=\br{\mathscr{Q}^x}_x$ and $\mathscr R=\br{\mathscr{R}^y}_y$,
  where $\mathscr{Q}^x = \br{\br{Q^x_a}^\dagger Q^x_a}_a$
  and $\mathscr{R}^y = \br{\br{R^y_b}^\dagger R^y_b}_b$, then
        \begin{equation*}
        \abs{\Val_\sfG\br{\mathscr Q, \mathscr R} - \Val_\sfG\br{\mathscr{A}^{(u)}, \mathscr{B}^{(v)}}} \le 15\eta.
        \end{equation*}
  \end{itemize}

\end{restatable}

\paragraph{Step 3: Dimension Reduction.}

For a hash $h\colon[n]\to[D]$ and $\alpha\in(\mathbb F_2^2)^n$, set
\begin{equation}
\label{eq:hashed-pauli-label}
  (H\alpha)_j:=\sum_{i:h(i)=j}\alpha_i\in\mathbb F_2^2.
\end{equation}
Let $r_1,\ldots,r_n$ be independent uniform elements of $\mathbb F_2^2$.
For $r_i=(r_{i,1},r_{i,2})$ and $\alpha_i=(\alpha_{i,1},\alpha_{i,2})$, define
\begin{equation}
\label{eq:hashed-pauli-map}
  \chi_r(\alpha):=(-1)^{\sum_{i=1}^n(r_{i,1}\alpha_{i,1}+r_{i,2}\alpha_{i,2})},
  \qquad
  \Phi_{h,r}(P_\alpha):=\chi_r(\alpha)P_{H\alpha}.
\end{equation}

The following theorem shows that the map $\Phi_{h,r}$ can reduce the dimension of block approximants while approximately preserving the game value.

\begin{restatable}{theorem}{foldedblockoutput}
  \label{thm:folded-block-output}
  There is an absolute constant $C>0$ with the following property.  Fix a
  game $\sfG=(\mathcal X,\mathcal Y,\mathcal A,\mathcal B,\mathsf p_{\sfG},V)$,
  let $n,k,D\geq1$, and use $n$ copies of a canonical state
  $\psi_{\boldsymbol\lambda}$, whose marginals satisfy
  $(\psi_{\boldsymbol\lambda})_A=(\psi_{\boldsymbol\lambda})_B=\id/2$.  Let
  $0<\delta<1$, and put
  $t=\max\{|\mathcal A|,|\mathcal B|,2\}$.  Let
  $(Q^x)_{x\in\mathcal X}$ and
  $(R^y)_{y\in\mathcal Y}$ be collections of block approximants
  $Q^x\colon(\mathbb C^2)^{\otimes n}\to
  \mathbb C^{\mathcal A}\otimes(\mathbb C^2)^{\otimes n}$ and
  $R^y\colon(\mathbb C^2)^{\otimes n}\to
  \mathbb C^{\mathcal B}\otimes(\mathbb C^2)^{\otimes n}$.
  Write $\mathscr Q,\mathscr R$ for the families of measurement elements
  defined in \eqref{eq:block-operator-families}.
  Each block $Q_a^x$ and $R_b^y$ has Pauli degree at most $k$, and the
  common random map $\Phi_{h,r}$ is applied to these blocks. Suppose they satisfy
  \begin{equation*}
    \norm{(Q^x)^\dagger Q^x-\id}_2,
    \norm{(R^y)^\dagger R^y-\id}_2
    \leq\delta
  \end{equation*}
  for every question $x,y$.  If
  \(\,D\geq
    Ck^{16}\max\left\{
      \frac{t^2}{\delta^{2/3}},
      \frac{t}{\delta^2}
    \right\},
  \)
  then there is one deterministic seed $\omega_0=(h,r)$ for which the
  following holds.  Define
  \(
    \widetilde Q_a^x:=\Phi_{\omega_0}(Q_a^x),
    \widetilde R_b^y:=\Phi_{\omega_0}(R_b^y)
  \)
  and form the block operators $\widetilde Q^x=\sum_a|a\rangle\otimes\widetilde Q_a^x$
  and $\widetilde R^y=\sum_b|b\rangle\otimes\widetilde R_b^y$.
  Define the approximate strategies
  \begin{equation*}
    \widetilde{\mathscr Q}
    :=\bigl(\bigl((\widetilde Q_a^x)^\dagger\widetilde Q_a^x\bigr)_a\bigr)_x,
    \qquad
    \widetilde{\mathscr R}
    :=\bigl(\bigl((\widetilde R_b^y)^\dagger\widetilde R_b^y\bigr)_b\bigr)_y.
  \end{equation*}
  These operators satisfy
  \begin{equation*}
    \mathbb E_{x\sim\mathsf p_X}
      \norm{(\widetilde Q^x)^\dagger\widetilde Q^x-\id}_2^2,
      \quad
      \mathbb E_{y\sim\mathsf p_Y}
      \norm{(\widetilde R^y)^\dagger\widetilde R^y-\id}_2^2
    \leq C\delta^{4/3},
  \end{equation*}
  and
  \begin{equation*}
    \Val_{\sfG}(\widetilde{\mathscr Q},\widetilde{\mathscr R})
    \geq
    \Val_{\sfG}(\mathscr Q,\mathscr R)
    -C\delta^{1/3}.
  \end{equation*}
\end{restatable}


\paragraph{Step 4: Rounding.}

Finally, we round the dimension-reduced block approximants to exact POVMs.

\begin{restatable}{lemma}{polarsoundness}
  \label{lem:polar-soundness}
  Fix a game
  $\sfG=(\mathcal X,\mathcal Y,\mathcal A,\mathcal B,\mathsf p_{\sfG},V)$,
  an integer $n\geq1$, and the resource state
  $\psi^{\otimes n}$ satisfying $\psi_A=\psi_B=\id/2$.  Let
  $(Q^x)_{x\in\mathcal X}$ and
  $(R^y)_{y\in\mathcal Y}$ be collections of block approximants
  $Q^x\colon(\mathbb C^2)^{\otimes n}\to
    \mathbb C^{\mathcal A}\otimes(\mathbb C^2)^{\otimes n}$ and
  $R^y\colon(\mathbb C^2)^{\otimes n}\to
    \mathbb C^{\mathcal B}\otimes(\mathbb C^2)^{\otimes n}$.
  Write $\mathscr Q,\mathscr R$ for the families of measurement elements
  defined in \eqref{eq:block-operator-families}.
  Suppose, for some $0\leq\nu\leq1$, that
  \begin{equation*}
    \mathbb E_{x\sim\mathsf p_X}\norm{(Q^x)^\dagger Q^x-\id}_2^2\leq\nu^4,
    \qquad
    \mathbb E_{y\sim\mathsf p_Y}\norm{(R^y)^\dagger R^y-\id}_2^2\leq\nu^4.
  \end{equation*}
  Then there exist an Alice POVM strategy $\mathscr A^\sharp$ and a Bob
  POVM strategy $\mathscr B^\sharp$, both on $(\mathbb C^2)^{\otimes n}$, such that
  \begin{equation*}
    \abs{
      \Val_{\sfG}(\mathscr Q,\mathscr R)
      -\operatorname{Val}_{\sfG}(\mathscr A^\sharp,\mathscr B^\sharp)
    }
    \leq15\nu.
  \end{equation*}
\end{restatable}

With all the lemmas above, we are now ready to prove~\cref{thm:finite-copy-approximation}. The parameter choices and their dependencies are summarized
in \cref{fig:parameter-dependencies}.

\begin{figure}[tbp]
  \centering
  \begin{tikzpicture}[
    x=1cm,y=1cm, font=\small,
    stage/.style={draw,align=center,text width=3.15cm,
      minimum height=1cm,inner sep=5pt},
    result/.style={align=center,text width=5.25cm,inner sep=3pt},
    edge/.style={-{Stealth[length=2mm]},thin},
    tag/.style={fill=white,inner sep=2pt,font=\footnotesize,align=center}
  ]
    \node[stage,text width=2.15cm] (auxinput) at (0,-2.75) {Obtain $\alpha_\psi$\\
    (\cref{lem:complete-sdpi})};
    \node[stage,text width=2.15cm] (realinput) at (0,-4.75)
      {Given $t,\psi,\varepsilon$\\
       (\cref{thm:finite-copy-approximation})};

    \node[draw, dotted, minimum width=3cm, minimum height=3.7cm] (input) at (0, -3.75) {};

    \node[stage, text width=5cm] (smooth) at (8,0)
      {Set $\displaystyle w\gets c_0\frac{(1-\alpha_\psi)\varepsilon^2}{\log(2t)}$\\(\cref{thm:whole-smoothing}, Smoothing)};
    \node[stage, text width=7cm] (degree) at (8,-2.5)
      {Set $\displaystyle k\gets\left\lceil\left(
        \frac{C(\log(2t))^2}{(1-\alpha_\psi)\varepsilon^{14}}
        \right)^{\!1/\gamma}\right\rceil$\\
       (\cref{thm:low-degree-approximation-new}, Low-degree approximation)};
    \node[stage] (fold) at (5,-5.15)
      {Pauli folding\\(\cref{thm:folded-block-output})};
    \node[stage] (round) at (5,-7.1)
      {Rounding\\(\cref{lem:polar-soundness})};

    \draw[edge] (realinput.north) to node[tag, right] {$\psi$} (auxinput);

    \draw[edge] (input.north) to[out=90,in=180]
      node[tag,below right] {$t,\varepsilon, \alpha_\psi$} (smooth.west);
    \draw[edge] (input.40) to node[tag,above] {$\eta\gets c_1\varepsilon^3$} node[tag, below] {$t,\varepsilon$} (degree.west);
    \draw[edge] (input.317) to node[tag,above] {$\delta\gets\varepsilon^3$} node[tag,below] {$t,\varepsilon$} (fold.west);
    \draw[edge] (input.south) to[out=-90,in=180]
      node[tag,pos=.57,above right] {$\nu\gets C_1\varepsilon$} (round.west);

    \draw[edge] (smooth) -- node[tag,right]
      {$w$}
      (degree);
    \draw[edge] (degree) -- node[tag,below right]
      {$k$} (fold);
    \draw[edge] (fold) -- (round);

    \node[result,draw,inner sep=7pt] (dimension) at (10.2,-5.15)
      {$\displaystyle D\gets\left\lceil Ck^{16}\max\left\{
        \frac{t^2}{\varepsilon^2},\frac{t}{\varepsilon^6}
        \right\}\right\rceil$};
    \draw[edge] (fold.east) -- (dimension.west);
  \end{tikzpicture}
  \caption{Dependency of parameters in the proof of
    \cref{thm:finite-copy-approximation}.
    Here $c_0>0$ is sufficiently small, $0<c_1\leq1/3$,
    and $C,C_1>0$ are sufficiently large absolute constants;
    $\varepsilon$ is sufficiently small that $\nu\leq1$; $\gamma$ is the absolute constant in \cref{thm:low-degree-approximation-new}. Arrows represent the flow of parameters.}
  \label{fig:parameter-dependencies}
\end{figure}

\begin{proof}[Proof of \cref{thm:finite-copy-approximation}]
  Since game values lie in $[0,1]$, after enlarging the absolute constants it
  suffices to prove the theorem for $0<\varepsilon\leq\varepsilon_0$, where
  $\varepsilon_0>0$ is a sufficiently small absolute constant.  We choose
  $\varepsilon_0$ below so that the parameter in the final application of
  \cref{lem:polar-soundness} is at most one.

  First apply \cref{prop:pauli-canonicalization} to obtain a canonical
  state $\psi_{\boldsymbol\lambda}$ and a corresponding POVM strategy
  $(\mathscr A,\mathscr B)$ on $\psi_{\boldsymbol\lambda}^{\otimes n}$ with
  the same value as the original strategy on $\psi^{\otimes n}$.  Use the
  complete-SDPI constant $0\leq\alpha_\psi<1$ for the canonical correlation
  channel supplied by \cref{lem:complete-sdpi}.

  Choose
  \begin{equation*}
    w=c_0\frac{(1-\alpha_\psi)\varepsilon^2}{\log(2t)}
  \end{equation*}
  with $c_0>0$ sufficiently small.  By
  \cref{thm:whole-smoothing}, smoothing Alice's and Bob's POVMs changes
  the game value by at most $\varepsilon/10$.

  Set $\eta=c_1\varepsilon^3$ with $0<c_1\leq1/3$.  Apply
  \cref{thm:low-degree-approximation-new}
  to $(\mathscr A,\mathscr B)$ to choose the smoothing times and block approximants.
  With this choice of $w$ and $\eta=c_1\varepsilon^3$, the lower bound on
  $k$ in \cref{thm:low-degree-approximation-new} becomes
  \begin{equation*}
    k^\gamma
    \geq
    \frac{C(\log(2t))^2}{(1-\alpha_\psi)\varepsilon^{14}},
  \end{equation*}
  which is ensured by the stated choice of $k$.  The resulting blocks
  have normalization error at most $(2+\eta)\eta\leq3\eta\leq\varepsilon^3$.
  Moreover,
  the game value is changed by at most $15\eta=O(\varepsilon^3)$.

  Apply \cref{thm:folded-block-output} with $\delta=\varepsilon^3$.
  The stated choice of $D$ meets its dimension bound.
  Hence one deterministic seed produces
  $D$-qubit operators whose value is at least the value before dimension
  reduction minus
  $C\varepsilon$ and whose two average squared normalization errors are
  $O(\varepsilon^4)$.

  The implicit constant in the $O(\varepsilon^4)$ normalization bound fixes
  an absolute constant $C_1$ such that the normalization hypothesis of
  \cref{lem:polar-soundness} holds with $\nu=C_1\varepsilon$.  Choose
  $\varepsilon_0\leq C_1^{-1}$, so that
  $\nu\leq1$, and apply \cref{lem:polar-soundness}.  The
  dimension-reduced operators become POVM families for Alice and Bob, with
  a further value loss at most $15C_1\varepsilon$.  Adding the losses and increasing the absolute constant
  proves the desired estimate.  The compressed strategy acts on
  $D$ copies of $\psi_{\boldsymbol\lambda}$.  Conjugation by the inverse fixed
  local unitaries transports this strategy to $D$ copies of the original
  state $\psi$ without changing its value.
\end{proof}

\begin{restatable}{theorem}{succinctverifier}
  \label{prop:succinct-verifier}
  Let $\psi\in \Matrix_2\otimes \Matrix_2$ be a fixed state satisfying
  $\psi_A=\psi_B=\id/2$ and $\rho_{\max}(\psi)<1$.  Then
  \begin{equation*}
    \MIP^{\psi}[\poly,\poly]
    \subseteq\NEXP.
  \end{equation*}
\end{restatable}

The proof of \cref{prop:succinct-verifier} is given in
\cref{subsec:succinct-certificates}.

\begin{cor}[$\NEXP$ characterization]
  \label{cor:noisy-mip-equality}
  Fix a two-qubit state $\psi$ with maximally mixed marginals and
  $\rho_{\max}(\psi)<1$.  Then
  \begin{equation*}
    \grayhighlight{
      \MIP^{\psi}[\poly,O(1)]
      =\MIP^{\psi}[\poly,\poly]
      =\NEXP=\MIP.}
  \end{equation*}
\end{cor}

\begin{proof}
  By \cref{prop:succinct-verifier}, polynomial-answer protocols are contained
  in $\NEXP$.  Conversely, the lower inclusion in
  \cite[Corollary~2]{DongEtAl24} applies to $\psi$: its maximally mixed
  marginals and maximal correlation strictly below one are precisely the
  noisy-MES conditions of that result.  Hence
  \begin{equation*}
    \NEXP\subseteq\MIP^{\psi}[\poly,O(1)]
    \subseteq\MIP^{\psi}[\poly,\poly]
    \subseteq\NEXP.
  \end{equation*}
  Together with the classical identity $\MIP=\NEXP$ \cite{BFL91,FL92},
  these inclusions prove the claimed equalities.
\end{proof}

\subsection{\texorpdfstring{The Case $\rho_{\max}=1$}
  {The Case rho-max = 1}}
\label{sec:maxcorr-one}

At the boundary $\rho_{\max}=1$, the canonical form is governed by a
single parameter $\lambda\in[0,1]$.  The value $\lambda=0$ corresponds
to a separable, purely classical resource, whereas every $\lambda>0$
retains enough coherence to simulate finite-EPR strategies with a loss
that is independent of the number of EPR pairs.

For $0\leq\lambda\leq1$, define
\begin{equation}
\label{eq:canonical-boundary-family}
  \Psi_\lambda
  :=
  \frac12\left(
    \ketbra{00}+\ketbra{11}
    +\lambda\ket{00}\!\bra{11}
    +\lambda\ket{11}\!\bra{00}
  \right).
\end{equation}
Equivalently, in the Pauli basis,
\begin{equation*}
  \Psi_\lambda
  =
  \frac14\left(
    \id\otimes\id
    +\lambda\pauliX\otimes\pauliX^{\mathsf T}
    +\lambda\pauliY\otimes\pauliY^{\mathsf T}
    +\pauliZ\otimes\pauliZ^{\mathsf T}
  \right).
\end{equation*} 
The following proposition asserts that $\rho$ precisely captures the separability of $\psi$.
\begin{prop}[Boundary canonical form]
  \label{prop:boundary-canonical}
  Let $\psi\in \Matrix_2\otimes \Matrix_2$ have maximally mixed marginals and
  satisfy $\rho_{\max}(\psi)=1$.  Then $\psi$ is locally unitarily
  equivalent to $\Psi_\lambda$ for some $\lambda\in[0,1]$.
  Moreover,
  \[
    \psi \text{ is separable}
    \quad\Longleftrightarrow\quad
    \lambda=0.
  \]
\end{prop}

The full proof of~\cref{prop:boundary-canonical} appears in
\cref{sec:maximal-correlation-one}. As separable states can be locally generated by classical correlation, and the tensor product of separable states is still separable, the following proposition is obvious.


\begin{restatable}[Separable resources are classical]{prop}{separableresourceclassical}
  \label{prop:separable-resource-classical}
  Let $\psi$ be a fixed finite-dimensional separable bipartite state.  For
  every nonlocal game
  $\sfG=(\mathcal X,\mathcal Y,\mathcal A,\mathcal B,\mathsf p_{\sfG},V)$,
  optimization over arbitrary finite tensor powers of $\psi$ has exactly
  the classical value.  Consequently,
  \begin{equation*}
    \MIP^{\psi}[\poly,O(1)]
    =\MIP^{\psi}[\poly,\poly]
    =\NEXP.
  \end{equation*}
\end{restatable}
We will give a complete proof in \cref{sec:maximal-correlation-one}. 


\begin{restatable}[$\RE$ Lowerbound]{theorem}{boundaryrelowerbound}
  \label{thm:boundary-re-lower-bound}
  For every fixed $0<\lambda\leq1$, consider the boundary canonical form $\Psi_\lambda$, we have
  \begin{equation*}
    \RE\subseteq
    \MIP^{\Psi_\lambda}[\poly,O(1)].
  \end{equation*}
As a corollary, 
\[\RE\subseteq
    \MIP^{\psi}[\poly,O(1)]\] 
    for any entangled state $\psi\in \Matrix_2\otimes \Matrix_2$ with maximally mixed marginals and
  satisfying $\rho_{\max}(\psi)=1$. 

\end{restatable}

The proof of \cref{thm:boundary-re-lower-bound} is given in \cref{sec:maximal-correlation-one}. When $\lambda>0$, the main ingredient of this theorem,
\cref{thm:boundary-epr-component}, establishes a uniform one-sided simulation
of every finite-EPR strategy, retaining at least $\lambda^2$ times its
acceptance probability, independently of the number of EPR pairs.
Combining this simulation with the constant-answer EPR
characterization of $\RE$, Yuen's parallel-repetition theorem, and a
constant-size OR amplification yields
\[
  \RE
  \subseteq
  \MIP^{\psi}[\poly,O(1)].
\]

\begin{restatable}[$\RE$ upperbound]{theorem}{fixedresourcereupperbound}
  \label{lem:fixed-resource-re-upper-bound}
  For every fixed finite-dimensional bipartite state $\psi$,
  \begin{equation*}
    \MIP^{\psi}[\poly,\poly]\subseteq\RE.
  \end{equation*}
\end{restatable}

The proof of \cref{lem:fixed-resource-re-upper-bound} is given in \cref{sec:maximal-correlation-one}. The proof of this theorem follows from the general recursive-enumerability upper bound for fixed
resource states.

Using a combination of the above theorems and propositions, we can prove the following key result when $\rho=1$:



It's also sufficient to prove \cref{thm:main}.

\begin{proof}[Proof of \cref{thm:main}]
  If $\psi$ is separable, \cref{prop:separable-resource-classical} applies.
  If $\psi$ is entangled and $\rho_{\max}(\psi)<1$, the equality with
  $\NEXP=\MIP$ is \cref{cor:noisy-mip-equality}.  The remaining case is
  entangled with $\rho_{\max}(\psi)=1$, which follows from \cref{thm:boundary-re-lower-bound} and \cref{lem:fixed-resource-re-upper-bound}.
\end{proof}

\section{The Case \texorpdfstring{$\rho_{\max}<1$}
  {rho-max < 1 Implies NEXP}}
\label{sec:strict-noise-proof}

\subsection{Step 1: Smoothing}

In this subsection, we prove \cref{thm:whole-smoothing}, restated below.

\wholesmoothing*

This is an exponential improvement of~\cite[Lemma 6.1]{QY21}. The main ingredience  is that we treat a POVM $\mathscr A=(A_a)_{a=1}^t$ as a single block density operator$Z_{\mathscr A}
=
t\sum_{a=1}^t|a\rangle\langle a|\otimes A_a$, rather than several individual operators as in~\cite{QY21}. This enable us to bound the difference between the values of the games in terms of quantum relative entropy by Pinsker inequality and the chain rule of quantum relative entropy. The key component of the proof is the complete strong data processing inequality for quantum relative entropy~{GR22}, which establish the connection between the change of values and quantum maximal correlations. 

\subsubsection{Complete SDPI}

We state the complete strong data processing inequality for the correlation channel, which follows directly from \cite[Theorem~1.2]{GR22}.

\begin{restatable}[Complete SDPI]{lemma}{completesdpi}
  \label{lem:complete-sdpi}
Let $\psi\in \Matrix_2\otimes \Matrix_2$ satisfy
$\psi_A=\psi_B=\id/2$ and $\rho_{\max}(\psi)<1$.  Choose a canonical form
$\psi_{\boldsymbol\lambda}$ of $\psi$ as in
\cref{prop:pauli-canonicalization}, let
$\mathcal K_{\psi_{\boldsymbol\lambda}}$ be its canonical correlation
channel, and let
$\Pi(X)=\tau_2(X)\id$.  There is a constant $0\leq\alpha_\psi<1$, depending
only on $\psi$, such that for every $m\geq1$ and every density matrix
$\rho\in \Matrix_m\otimes \Matrix_2$ with $\Tr\rho=1$,
\begin{equation*}
\RelEnt\!\left(
(\operatorname{id}_{\Matrix_m}\otimes\K_{\psi_{\boldsymbol\lambda}})(\rho)
\,\middle\|\,
(\operatorname{id}_{\Matrix_m}\otimes\Pi)(\rho)
\right)
\le
\alpha_\psi
\RelEnt\!\left(
\rho
\,\middle\|\,
(\operatorname{id}_{\Matrix_m}\otimes\Pi)(\rho)
\right).
\end{equation*}
The same coefficient works for every ancilla dimension $m$.
\end{restatable}

The proof is given in \cref{app:complete-sdpi}.  The normalized-trace
form used below follows by rescaling.

\begin{cor}
\label{cor:complete-sdpi-normalized}
Let $\K,\Pi$, and $\alpha_\psi$ be as in \cref{lem:complete-sdpi}.
Every positive $Z\in \Matrix_m\otimes \Matrix_2$ with $\tau(Z)=1$ satisfies
\begin{equation*}
\RelEnt_{\tau}\!\left(
(\operatorname{id}_{\Matrix_m}\otimes\K)(Z)
\,\middle\|\,
(\operatorname{id}_{\Matrix_m}\otimes\Pi)(Z)
\right)
\le
\alpha_\psi
\RelEnt_{\tau}\!\left(
Z
\,\middle\|\,
(\operatorname{id}_{\Matrix_m}\otimes\Pi)(Z)
\right).
\end{equation*}
\end{cor}

\begin{proof}
Since $\tau=\tau_{2m}$, the operator $\rho=Z/(2m)$ is a density
matrix.  Apply \cref{lem:complete-sdpi} to $\rho$ and use linearity of
the channels together with
$\RelEnt(R/(2m)\|S/(2m))=\RelEnt_{\tau}(R\|S)$ for positive $R,S$ with
$\tau(R)=\tau(S)=1$.  Thus the same coefficient $\alpha_\psi$ applies.
\end{proof}

\subsubsection{Partial Depolarizing Maps}

\begin{definition}
  \label{def:partial-depolarizing-maps}
  For $n\geq1$ and $i\in[n]$, let $\Pi_i$
  apply $\Pi=\tau_2(\cdot) \id$ to the $i$th tensor factor of $\Matrix_2^{\otimes n}$ and
  act as the identity on all other factors.  For
  $S=\{i_1<\cdots<i_k\}\subseteq[n]$, define the \emph{partial depolarizing map}
  \begin{equation*}
    \Pi_S:=\Pi_{i_1}\circ\cdots\circ\Pi_{i_k},
    \qquad
    \Pi_\varnothing:=\operatorname{id}_{\Matrix_2^{\otimes n}}.
  \end{equation*}
  For $R\in\Matrix_2^{\otimes n}$, $\Pi_S(R)$ is obtained by applying $\Pi$
  to each qubit in $S$, leaving the other tensor factors unchanged.
\end{definition}

The map $\Pi_{[n]}=\Pi^{\otimes n}$ completely depolarizes all $n$ qubits.
Expanding
$\mathcal T_u=(e^{-u}\operatorname{id}_{\Matrix_2}+(1-e^{-u})\Pi)^{\otimes n}$
gives
\begin{equation*}
  \mathcal T_u
  =\sum_{S\subseteq[n]}
    e^{-u(n-|S|)}(1-e^{-u})^{|S|}\Pi_S.
\end{equation*}
Thus $\mathcal T_u$ is the average of $\Pi_S$ when each qubit is included
in $S$ independently with probability $1-e^{-u}$.

\begin{prop}
  \label{prop:correlation-channel-depolarizing-commutation}
  Let $\mathcal K=\mathcal K_{\psi_{\boldsymbol\lambda}}$ be the canonical
  correlation channel of a canonical state $\psi_{\boldsymbol\lambda}$. Then
  \begin{equation*}
    \mathcal K\circ\Pi=\Pi\circ\mathcal K=\Pi.
  \end{equation*}
  Moreover, let
  $\mathcal K_n:=\mathcal K^{\otimes n}$, for every $u\geq0$,
  \(
    \mathcal K_n\mathcal T_u=\mathcal T_u\mathcal K_n.
  \)
\end{prop}

\begin{proof}
  Since $\mathcal K$ is unital and trace preserving, for every $X\in \Matrix_2$,
  \[
    \mathcal K(\Pi(X))
    =\tau_2(X)\mathcal K(\id)
    =\tau_2(X)\id
    =\Pi(X)
  \]
  and
  \[
    \Pi(\mathcal K(X))
    =\tau_2(\mathcal K(X))\id
    =\tau_2(X)\id
    =\Pi(X).
  \]
  This proves the first assertion.
  Substituting $\mathcal K\circ\Pi=\Pi\circ\mathcal K=\Pi$ into
  $\mathcal D_u=\Pi+e^{-u}(\operatorname{id}-\Pi)$ shows directly that
  $\mathcal K \mathcal D_u=\mathcal D_u\mathcal K$.  Taking tensor powers proves the second
  assertion.
\end{proof}

We next prove the relative-entropy chain rule for partial depolarizing maps.

\begin{lemma}
\label{lem:partial-depolarizing-chain-rule}
Let $m,n\geq1$, and equip $\Matrix_m\otimes\Matrix_2^{\otimes n}$ with the
normalized trace $\tau:=\tau_m\otimes\tau_{2^n}$.  Use the maps $\Pi_S$ from
\cref{def:partial-depolarizing-maps}, extended by the identity on $\Matrix_m$.
For every positive $R\in\Matrix_m\otimes\Matrix_2^{\otimes n}$ with
$\tau(R)=1$ and every $S\subseteq T\subseteq[n]$,
\begin{equation*}
\RelEnt_\tau(R\|\Pi_T(R))
=
\RelEnt_\tau(R\|\Pi_S(R))
+
\RelEnt_\tau(\Pi_S(R)\|\Pi_T(R)).
\end{equation*}
\end{lemma}

\begin{proof}
Fix $S\subseteq[n]$ and a positive
$R\in\Matrix_m\otimes\Matrix_2^{\otimes n}$ with $\tau(R)=1$.
By definition, $\Pi_S$ is positive and trace preserving, and for all
$X,Y\in\Matrix_m\otimes\Matrix_2^{\otimes n}$,
\[
\tau\!\left(X\Pi_S(Y)\right)
=\tau\!\left(\Pi_S(X)Y\right).
\]
Without loss of generality, assume that the qubits in $S$ come last.  Then
$\Pi_S(R)=B\otimes\id_{2^{|S|}}$, where
$B\in\Matrix_{m2^{n-|S|}}$ is positive.  Let $P$ be the orthogonal
projector onto $\ker\Pi_S(R)$.  Since
$\ker\Pi_S(R)=(\ker B)\otimes\mathbb C^{2^{|S|}}$,
$P$ also has an identity factor on the qubits in $S$, so $\Pi_S(P)=P$.
Consequently,
\begin{align*}
\norm{\sqrt{R}P}_2^2
&=\tau(PRP)=\tau(PR)\\
&=\tau\!\left(\Pi_S(P)R\right)
=\tau\!\left(P\Pi_S(R)\right)=0.
\end{align*}
Thus $RP=0$, which gives $\ker\Pi_S(R)\subseteq\ker R$, or equivalently
$\supp(R)\subseteq\supp(\Pi_S(R))$.
Define each logarithm on the support of its argument and extend it by
zero on the kernel.  The same tensor form gives
\[
\log\Pi_S(R)=(\log B)\otimes\id_{2^{|S|}},
\qquad
\Pi_S\!\left(\log\Pi_S(R)\right)=\log\Pi_S(R).
\]
For $S\subseteq T$, the identity $\Pi_T\circ\Pi_S=\Pi_T$ and the
support inclusion proved above also give
$\supp(\Pi_S(R))\subseteq\supp(\Pi_T(R))$.
Moreover, the tensor form of $\log\Pi_T(R)$ gives
$\Pi_S(\log\Pi_T(R))=\log\Pi_T(R)$.
Using the trace identity above, we obtain
\begin{align*}
&\RelEnt_\tau(R\|\Pi_T(R))-\RelEnt_\tau(R\|\Pi_S(R))\\
&\quad=\tau\!\left(R\bigl(\log\Pi_S(R)-\log\Pi_T(R)\bigr)\right)\\
&\quad=\tau\!\left(\Pi_S(R)\bigl(\log\Pi_S(R)-\log\Pi_T(R)\bigr)\right)\\
&\quad=\RelEnt_\tau(\Pi_S(R)\|\Pi_T(R)),
\end{align*}
which proves the claim.
\end{proof}

\subsubsection{Relative Entropy Bounds}

\begin{lemma}
\label{lem:output-influence}
Let $m,n\geq1$, and equip $\Matrix_m\otimes \Matrix_2^{\otimes n}$ with the normalized
trace $\tau:=\tau_m\otimes\tau_{2^n}$.  Let $\K\colon \Matrix_2\to \Matrix_2$ be the
canonical correlation channel from \cref{lem:complete-sdpi}.  For $i\in[n]$,
let $\K_i$ act as $\K$ on the $i$th qubit and as the identity on $\Matrix_m$
and all other qubits.  Use the maps $\Pi_i$ and $\Pi_S$ from
\cref{def:partial-depolarizing-maps}, extended by the identity on $\Matrix_m$.  If
$Z\in \Matrix_m\otimes \Matrix_2^{\otimes n}$ is positive and $\tau(Z)=1$, define
\begin{equation*}
Y:=(\K_1\circ\cdots\circ\K_n)(Z)
=(\operatorname{id}_{\Matrix_m}\otimes\K^{\otimes n})(Z).
\end{equation*}
 Then
\begin{equation*}
\sum_{i=1}^n
\RelEnt_{\tau}(Y\|\Pi_i(Y))
\le
\frac{\alpha_\psi}{1-\alpha_\psi}
\RelEnt_{\tau}(Z\|\Pi_{[n]}(Z)),
\end{equation*}
where $\RelEnt_{\tau}$ is the normalized-trace relative entropy in~\eqref{eqn:relent}, and
$\alpha_\psi<1$ is the complete-SDPI constant given in~\cref{lem:complete-sdpi}.
\end{lemma}


\begin{proof}
If $\alpha_\psi=0$, \cref{cor:complete-sdpi-normalized} implies
$\K=\Pi$, so $Y=\Pi_{[n]}(Z)$ and the assertion is immediate.  Assume
$\alpha_\psi>0$ for the rest of the proof.
Define
\[
Z_0=Z,
\qquad
Z_i=\K_i(Z_{i-1})
\quad(1\le i\le n).
\]
The channels $\K_i$ commute with all $\Pi_j$, and
\(
\Pi_i\circ\K_i=\Pi_i.
\)
Hence
\[
\Pi_i(Z_i)=\Pi_i(Z_{i-1}),
\qquad
\Pi_{[n]}(Z_i)=\Pi_{[n]}(Z_{i-1}).
\]

Applying~\cref{cor:complete-sdpi-normalized} to $Z_{i-1}$ with every
other register treated as an ancilla, we have 
\begin{align*}
\RelEnt_\tau(Z_i\|\Pi_i(Z_i))
&=\RelEnt_\tau(\K_i(Z_{i-1})\|\Pi_i(Z_{i-1}))\\
&\leq\alpha_\psi\RelEnt_\tau(Z_{i-1}\|\Pi_i(Z_{i-1})).
\end{align*}
Apply \cref{lem:partial-depolarizing-chain-rule} with $S=\{i\}$ and
$T=[n]$ first to $Z_{i-1}$ and then to $Z_i$.  Subtracting the two
resulting equalities produces
\begin{align*}
&\RelEnt_{\tau}(Z_{i-1}\|\Pi_{[n]}(Z_{i-1}))
-\RelEnt_{\tau}(Z_i\|\Pi_{[n]}(Z_i))
\\
&\quad=
\Bigl[
\RelEnt_{\tau}(Z_{i-1}\|\Pi_i(Z_{i-1}))
+
\RelEnt_{\tau}(\Pi_i(Z_{i-1})\|\Pi_{[n]}(Z_{i-1}))
\Bigr]
\\
&\qquad-
\Bigl[
\RelEnt_{\tau}(Z_i\|\Pi_i(Z_i))
+
\RelEnt_{\tau}(\Pi_i(Z_i)\|\Pi_{[n]}(Z_i))
\Bigr]
\\
&\quad=
\RelEnt_{\tau}(Z_{i-1}\|\Pi_i(Z_{i-1}))
-\RelEnt_{\tau}(Z_i\|\Pi_i(Z_i))\\
&\quad\ge
\left(\frac1{\alpha_\psi}-1\right)
\RelEnt_{\tau}(Z_i\|\Pi_i(Z_i))\\
&\quad=
\frac{1-\alpha_\psi}{\alpha_\psi}
\RelEnt_{\tau}(Z_{i}\|\Pi_i(Z_{i})),
\end{align*}
where the second equality uses both
$\Pi_i(Z_i)=\Pi_i(Z_{i-1})$ and
$\Pi_{[n]}(Z_i)=\Pi_{[n]}(Z_{i-1})$.  The inequality uses the preceding
SDPI bound.

After the $i$th step, apply the remaining channels
$\K_{i+1},\ldots,\K_n$ to both arguments.
\Cref{prop:relative-entropy-data-processing} and the preceding estimate imply
\[
\RelEnt_{\tau}(Y\|\Pi_i(Y))
\le
\RelEnt_{\tau}(Z_i\|\Pi_i(Z_i))
\le
\frac{\alpha_\psi}{1-\alpha_\psi}
\left[
\RelEnt_{\tau}(Z_{i-1}\|\Pi_{[n]}(Z_{i-1}))
-\RelEnt_{\tau}(Z_i\|\Pi_{[n]}(Z_i))
\right].
\]
Summing over $i$ cancels the intermediate relative-entropy terms and gives
\begin{align*}
\sum_{i=1}^n\RelEnt_{\tau}(Y\|\Pi_i(Y))
&\le
\frac{\alpha_\psi}{1-\alpha_\psi}
\sum_{i=1}^n\left[
\RelEnt_{\tau}(Z_{i-1}\|\Pi_{[n]}(Z_{i-1}))
-\RelEnt_{\tau}(Z_i\|\Pi_{[n]}(Z_i))
\right]\\
&=
\frac{\alpha_\psi}{1-\alpha_\psi}
\left[
\RelEnt_{\tau}(Z\|\Pi_{[n]}(Z))
-\RelEnt_{\tau}(Y\|\Pi_{[n]}(Y))
\right]\\
&\le
\frac{\alpha_\psi}{1-\alpha_\psi}
\RelEnt_{\tau}(Z\|\Pi_{[n]}(Z)),
\end{align*}
where the last step uses nonnegativity of relative entropy.
\end{proof}


\begin{lemma}
\label{lem:entropy-cost-heat}
Let $m,n\geq1$,
and use the maps $\Pi_i$ and $\Pi_S$ from \cref{def:partial-depolarizing-maps},
extended by the identity on $\Matrix_m$.  Define
$\mathcal T_u:=\operatorname{id}_{\Matrix_m}\otimes\mathcal T_u^{(n)}$, where
$\mathcal T_u^{(n)}(P_\alpha)=e^{-u|\alpha|}P_\alpha$ is the $n$-qubit
depolarizing semigroup.
For every positive
$Y\in \Matrix_m\otimes \Matrix_2^{\otimes n}$ with $\tau(Y)=1$ and every $u\geq0$,
\begin{equation*}
\RelEnt_{\tau}(Y\|\mathcal T_u(Y))
\le
(1-e^{-u})
\sum_{i=1}^n\RelEnt_{\tau}(Y\|\Pi_i(Y)).
\end{equation*}
Consequently,
\begin{equation*}
\norm{Y-\mathcal T_u(Y)}_1
\le
\sqrt{
2(1-e^{-u})
\sum_{i=1}^n\RelEnt_{\tau}(Y\|\Pi_i(Y))
}.
\end{equation*}
\end{lemma}

\begin{proof}
For fixed $Y$, the function $Z\mapsto\RelEnt_\tau(Y\|Z)$ is convex.
Therefore,
\begin{align*}
\RelEnt_{\tau}(Y\|\mathcal T_u(Y))
&=\RelEnt_{\tau}\!\left(Y\,\middle\|\,\mathbb E_\bfS\Pi_{\bfS}(Y)\right)\\
&\le\mathbb E_\bfS \RelEnt_{\tau}(Y\|\Pi_{\bfS}(Y)).
\end{align*}
For a fixed set $S=\{i_1<\cdots<i_m\}$, iterating
\cref{lem:partial-depolarizing-chain-rule} decomposes this relative entropy as
\[
\RelEnt_{\tau}(Y\|\Pi_S(Y))
=
\sum_{j=1}^m
\RelEnt_{\tau}\!\left(
\Pi_{\{i_1,\ldots,i_{j-1}\}}(Y)
\,\middle\|\,
\Pi_{\{i_1,\ldots,i_j\}}(Y)
\right)
\leq
\sum_{j=1}^m \RelEnt_{\tau}(Y\|\Pi_{i_j}(Y)),
\]
where $\{i_1,\ldots,i_{j-1}\}=\varnothing$ when $j=1$.
The $j$th summand is at most
$\RelEnt_{\tau}(Y\|\Pi_{i_j}(Y))$ by
\cref{prop:relative-entropy-data-processing}, applied to the preceding
commuting averaging maps.  Taking the expectation
over $\bfS$ proves the first assertion of \cref{lem:entropy-cost-heat}.
The second assertion follows from Pinsker \cref{prop:quantum-pinsker} applied to
$Y$ and $\mathcal T_u(Y)$.
\end{proof}

For a POVM $\mathscr A=(A_a)_{a=1}^t$, define its block density
\begin{equation}
\label{eq:povm-block-density}
Z_{\mathscr A}
=
t\sum_{a=1}^t|a\rangle\langle a|\otimes A_a.
\end{equation}
Then
\(
\tau(Z_{\mathscr A})=1\),
and
\(0\leq Z_{\mathscr A}\leq t\id
\).
A crucial observation is that we can upper bound the relative entropy between $Z_{\mathscr A}$ and $\Pi_{[n]}(Z_{\mathscr A})$ by $\log t$, instead of $\poly(t)$.
\begin{lemma}\label{lem:entropy-upper-bound}
  \begin{equation}\label{eq:entopy-upper-bound}
    \RelEnt_{\tau}(Z_{\mathscr A}\|\Pi_{[n]}(Z_{\mathscr A})) \le \log t.
  \end{equation}
\end{lemma}
\begin{proof}
\begin{align*}
\RelEnt_{\tau}(Z_{\mathscr A}\|\Pi_{[n]}(Z_{\mathscr A}))
&=
\tau\Br{Z_{\mathscr A}\log Z_{\mathscr A}} - \tau\Br{Z_{\mathscr A}\log \Pi_{[n]}(Z_{\mathscr A})} \\
&\le
\tau\Br{Z_{\mathscr A}\log Z_{\mathscr A}} \\
&\le
\tau\Br{Z_{\mathscr A}}\cdot\log t \\
&=
\log t.
\end{align*}
The first inequality follows since
\begin{align*}
    \tau\Br{Z_{\mathscr A}\log\Pi_{[n]}(Z_{\mathscr A})} &= \sum_{a=1}^t\tau\Br{A_a}\cdot\log\br{t\cdot\tau\Br{A_a}} \\
    &= \log t - \sum_{a=1}^t\tau\Br{A_a}\log\br{1/\tau\Br{A_a}}
    \ge 0,
\end{align*}
where we can see the term $\sum_{a=1}^t\tau\Br{A_a}\log\br{1/\tau\Br{A_a}}$ as the entropy of the $t$-element distribution, in which the $a$ element occurs with probability $\tau\Br{A_a}$.
The second inequality follows since $0\leq Z_{\mathscr A} \leq t\id$.
\end{proof}


\subsubsection{\texorpdfstring{Proof of \cref{thm:whole-smoothing}}
  {Proof of the Smoothing Theorem}}
\begin{proof}
The assumption $\rho_{\max}(\psi)<1$ ensures the existence of the
complete-SDPI constant $\alpha_\psi<1$ in \cref{lem:complete-sdpi}.
By \cref{prop:pauli-canonicalization}, we may assume the resource state in
canonical form without changing the original game value.  The conjugations
by $U_A^{\otimes n}$ and $\overline{U_B}^{\otimes n}$ commute with smoothing by
\cref{prop:depolarizing-product-unitaries}, so this reduction also preserves
the smoothed game value.

First smooth Alice only.  For fixed $(x,y)$, put
\begin{equation}\label{eqn:Faxy}
F_a^{x,y}
=
\sum_bV(a,b\mid x,y)B_b^y,
\qquad
F^{x,y}
=
\sum_a|a\rangle\langle a|\otimes F_a^{x,y}.    
\end{equation}

Each $F_a^{x,y}$ is a sum of a subset of the elements of the POVM
family $(B_b^y)_b$.  Since $\sum_bB_b^y=\id$,
\(
0\le F^{x,y}\le\id.
\)

Let
\(
Y_x=(\operatorname{id}_{\mathbb C^t}\otimes\K_n)(Z_{\mathscr A^x}).
\)
For fixed $(x,y)$, define
\begin{equation}
\label{eq:alice-smoothing-change}
  \Delta_A(x,y):=\sum_{a,b}V(a,b\mid x,y)
  \Tr\!\left[
    \bigl((A_a^x-A_a^{x,(u)})\otimes(B_b^y)^{\mathsf T}\bigr)\psi^{\otimes n}
  \right].
\end{equation}
Using the self-adjointness of $\K_n$ and
\cref{prop:correlation-channel-depolarizing-commutation}, we obtain
\begin{align*}
\Delta_A(x,y)
&=
\sum_{a,b}V(a,b\mid x,y)
\tau_{2^n}\!\left(
(\id-\mathcal T_{u_x})A_a^x\,\K_n(B_b^y)
\right)\\
&=
\tau\!\left(
(Y_x-\mathcal T_{u_x}Y_x)F^{x,y}
\right),
\end{align*}
where $F^{x,y}$ is defined in Eq.~\eqref{eqn:Faxy}. The operator $Y_x-\mathcal T_{u_x}Y_x$ is Hermitian and has trace zero. By the fact that $0\leq F^{x,y}\leq \id$,
\[
|\Delta_A(x,y)|
\le
\frac12\norm{Y_x-\mathcal T_{u_x}Y_x}_1.
\]
Together, \cref{lem:output-influence,lem:entropy-cost-heat} and the entropy
bound \cref{eq:entopy-upper-bound} established immediately above imply
\begin{align*}
|\Delta_A(x,y)|
&\le
\sqrt{
\frac{\alpha_\psi}{2(1-\alpha_\psi)}
(1-e^{-u_x})\log t
}\le
\sqrt{
\frac{\alpha_\psi}{1-\alpha_\psi}
w\log(2t)
},
\end{align*}
because $1-e^{-u_x}\leq u_x\leq2w$ and $\log t\leq\log(2t)$.

To smooth Bob next, keep Alice's smoothed POVM
$\mathscr A^{x,(u)}=(A_a^{x,(u)})_a$ fixed.
For fixed $(x,y)$, put
\begin{equation*}
  E_b^{x,y}:=\sum_aV(a,b\mid x,y)A_a^{x,(u)}.
\end{equation*}
Since $A_a^{x,(u)}\geq0$ and $V(a,b\mid x,y)\in\{0,1\}$,
\begin{equation*}
  0\leq E_b^{x,y}\leq\sum_a A_a^{x,(u)}
  =\mathcal T_{u_x}\!\left(\sum_a A_a^x\right)=\id.
\end{equation*}

With Alice's strategy $\mathscr A^{(u)}=(A_a^{x,(u)})_{x,a}$, define
\begin{equation}
\label{eq:bob-smoothing-change}
  \Delta_B(x,y):=\sum_{a,b}V(a,b\mid x,y)
  \Tr\!\left[
    \bigl(A_a^{x,(u)}\otimes(B_b^y-B_b^{y,(v)})^{\mathsf T}\bigr)\psi^{\otimes n}
  \right].
\end{equation}
Using the definition of the correlation channel and
$\K_n(\id-\mathcal T_{v_y})=(\id-\mathcal T_{v_y})\K_n$, we obtain
\begin{align*}
  \Delta_B(x,y)
  &=\sum_{a,b}V(a,b\mid x,y)\tau_{2^n}\!\left(
    A_a^{x,(u)}\K_n\!\left(B_b^y-B_b^{y,(v)}\right)\right)\\
  &=\sum_b\tau_{2^n}\!\left(
    E_b^{x,y}(\id-\mathcal T_{v_y})\K_n(B_b^y)
    \right).
\end{align*}
Applying the same estimates to the block density of the POVM
$(B_b^y)_b$, with $0\leq E_b^{x,y}\leq\id$, gives
\begin{equation*}
  |\Delta_B(x,y)|
  \leq\sqrt{\frac{\alpha_\psi}{2(1-\alpha_\psi)}
    (1-e^{-v_y})\log t}
  \leq\sqrt{\frac{\alpha_\psi}{1-\alpha_\psi}w\log(2t)},
\end{equation*}
where $1-e^{-v_y}\leq v_y\leq2w$.
The quantities $\Delta_A(x,y)$ and $\Delta_B(x,y)$ are signed differences
and need not be nonnegative.  By \eqref{eq:strategy-value},
\eqref{eq:alice-smoothing-change}, and \eqref{eq:bob-smoothing-change},
\begin{align*}
  \Val_{\sfG}(\mathscr A,\mathscr B)
    -\Val_{\sfG}(\mathscr A^{(u)},\mathscr B)
  &=\mathbb E_{(x,y)\sim\mathsf p_{\sfG}}\Delta_A(x,y),\\
  \Val_{\sfG}(\mathscr A^{(u)},\mathscr B)
    -\Val_{\sfG}(\mathscr A^{(u)},\mathscr B^{(v)})
  &=\mathbb E_{(x,y)\sim\mathsf p_{\sfG}}\Delta_B(x,y).
\end{align*}
Adding these identities cancels the intermediate value
$\Val_{\sfG}(\mathscr A^{(u)},\mathscr B)$.  Taking the absolute value
therefore gives the error in game value:
\begin{align*}
  &\abs{\Val_{\sfG}(\mathscr A,\mathscr B)
    -\Val_{\sfG}(\mathscr A^{(u)},\mathscr B^{(v)})}\\
  &\quad=\abs{\mathbb E_{(x,y)\sim\mathsf p_{\sfG}}
    \bigl(\Delta_A(x,y)+\Delta_B(x,y)\bigr)}\\
  &\quad\leq\mathbb E_{(x,y)\sim\mathsf p_{\sfG}}
    \bigl(|\Delta_A(x,y)|+|\Delta_B(x,y)|\bigr)\\
  &\quad\leq2\sqrt{\frac{\alpha_\psi}{1-\alpha_\psi}w\log(2t)}.
\end{align*}
\end{proof}

\subsection{Step 2: Approximation by Positive Low-Degree Operators}
In this subsection, we prove \cref{thm:low-degree-approximation-new}, restated below.
\lowdegreeapproximationnew*
We break the proof of \cref{thm:low-degree-approximation-new} into two parts.
First, 
we show that any POVM operator can be approximated by {\em positive} low-degree operators.
\begin{restatable}[Low-degree block approximation]{lemma}{lowdegreeblockapproximationnew}
  \label{lem:low-degree-block-approximationnew}
  There are absolute constants $C,\gamma>0$ with the following property.
  Let $n\geq1$, $t\geq2$, $w>0$, and $\eta\in(0,1)$, and let
  $\mathscr A=(A_a)_{a=1}^t$ be a POVM on $(\mathbb C^2)^{\otimes n}$.
  Let $(\mathcal T_u)_{u\geq0}$ be the tensorized depolarizing semigroup.  For $u\geq0$,
  define the smoothed POVM and block operator
  \begin{equation*}
    M_a(u):=\mathcal T_u(A_a),
    \qquad
    W_u:=\sum_{a=1}^t|a\rangle\otimes\sqrt{M_a(u)}.
  \end{equation*}
  Let $k$ be an integer greater than
  \(
    \left(
      \frac{C\log(2t)}{w\eta^4}
    \right)^{1/\gamma}.
  \)
  There is a $u\in[w,2w]$ and operators
  $Q_1,\ldots,Q_t\in \Matrix_2^{\otimes n}$, each of Pauli degree at most $k$, such that
  \(
    \norm{W_u-Q}_4\leq\eta,
  \)
  where
  \(
    Q=\sum_{a=1}^t|a\rangle\otimes Q_a.
  \)
  Moreover,
  \(
    \norm{Q^\dagger Q-\id}_2\leq(2+\eta)\eta.
  \)
\end{restatable}
Second, we show that an approximation in the normalized Schatten $4$-norm results to a good approximation to the original game value.

\begin{restatable}{lemma}{blockvaluestabilitynew}
  \label{lem:block-value-stabilitynew}
  Fix a game
  $\sfG=(\mathcal X,\mathcal Y,\mathcal A,\mathcal B,\mathsf p_{\sfG},V)$
  and a two-qubit state $\psi$ satisfying $\psi_A=\psi_B=\id/2$.
  Let integer $n\geq1$.
  Let
  $(A^x)_{x\in\mathcal X}$,
  $(B^y)_{y\in\mathcal Y}$,
  $(Q^x)_{x\in\mathcal X}$, and
  $(R^y)_{y\in\mathcal Y}$ be collections of block operators with
  \begin{align*}
    &A^x = \sum_{a\in\mathcal{A}}\ket{a}\otimes A^x_a,\quad A^x_a\in\Matrix_{2^n}, \\
    &B^y = \sum_{b\in\mathcal{B}}\ket{b}\otimes B^y_b,\quad B^y_b\in\Matrix_{2^n}, \\
    &Q^x = \sum_{a\in\mathcal{A}}\ket{a}\otimes Q^x_a,\quad Q^x_a\in\Matrix_{2^n}, \\
    &R^y = \sum_{b\in\mathcal{B}}\ket{b}\otimes R^y_b,\quad R^y_b\in\Matrix_{2^n}, \\
  \end{align*}
  Assume that, for every $x\in\mathcal X$ and $y\in\mathcal Y$,
  \begin{equation*}
    (A^x)^\dagger A^x=\id,
    \qquad
    (B^y)^\dagger B^y=\id.
  \end{equation*}
  Suppose that for some $\eta\in[0,1]$,
  \begin{equation*}
    \norm{Q^x-A^x}_4\leq\eta,
    \qquad
    \norm{R^y-B^y}_4\leq\eta,\qquad
    \forall x,y.
  \end{equation*}
  Define the strategies $\mathscr A = \br{\mathscr{A}^x}_x$, $\mathscr B = \br{\mathscr{B}^y}_y$, 
  and
  $\mathscr Q = \br{\mathscr{Q}^x}_x$, $\mathscr R = \br{\mathscr{R}^y}_y$
  such that
  $\mathscr{A}^x = \br{\br{A^x_a}^\dagger A^x_a}_a$,
  $\mathscr{B}^y = \br{\br{B^y_b}^\dagger B^y_b}_b$,
  $\mathscr{Q}^x = \br{\br{Q^x_a}^\dagger Q^x_a}_a$,
  $\mathscr{R}^y = \br{\br{R^y_b}^\dagger R^y_b}_b$.
  Then
  \begin{equation*}
    \abs{
      \Val_{\sfG}\br{\mathscr Q,\mathscr R}
      -\Val_{\sfG}\br{\mathscr A, \mathscr B}
    }
    \leq15\eta.
  \end{equation*}
\end{restatable}
The proof of \cref{lem:block-value-stabilitynew} is deferred to \cref{sec:appendix:block-value-stability}.

\begin{proof}[Proof of \cref{thm:low-degree-approximation-new}]
We combine \cref{lem:low-degree-block-approximationnew} and \cref{lem:block-value-stabilitynew}.
For each $x\in\mathcal X$, apply \cref{lem:low-degree-block-approximationnew}
to the POVM $\mathscr A^x=(A_a^x)_a$.
This gives a time $u_x\in[w,2w]$ and operators $Q_1^x,\dots,Q_t^x$,
each of Pauli degree at most $k$, such that, writing
$A^{x,u_x}=\sum_{a=1}^t\ket{a}\otimes\sqrt{\mathcal T_{u_x}(A_a^x)}$
and $Q^x=\sum_{a=1}^t\ket{a}\otimes Q_a^x$,
\begin{equation*}
    \norm{A^{x,u_x} -  Q^x}_4 \le \eta
\quad\text{and}\quad\norm{\br{Q^x}^\dagger Q^x - \id}_2\le\br{2+\eta}\eta.
\end{equation*}
Similarly, applying the same lemma to each $\mathscr B^y$ gives a time
$v_y\in[w,2w]$ and a block operator $R^y$ of Pauli degree at most $k$.
With $B^{y,v_y}$ defined analogously, we have
\begin{equation*}
    \norm{B^{y, v_y} - R^y}_4 \le \eta
\quad\text{and}\quad\norm{\br{R^y}^\dagger R^y-\id}_2\le(2+\eta)\eta.
\end{equation*}
Since the smoothed measurements are POVMs,
\begin{equation*}
  (A^{x,u_x})^\dagger A^{x,u_x}
  =\sum_a\mathcal T_{u_x}(A_a^x)=\id,
  \qquad
  (B^{y,v_y})^\dagger B^{y,v_y}
  =\sum_b\mathcal T_{v_y}(B_b^y)=\id.
\end{equation*}
For the approximate strategies
$\mathscr{Q} = \br{\mathscr{Q}^x}_x$,
$\mathscr{R} = \br{\mathscr{R}^y}_y$,
by \cref{lem:block-value-stabilitynew},
we have
\begin{equation*}
    \abs{\Val_{\sfG}\br{\mathscr Q, \mathscr R} - \Val_{\sfG}\br{\mathscr{A}^{(u)}, \mathscr{B}^{(v)}}}\le 15\eta.
\end{equation*}

\end{proof}



The remainder of this subsection will be devoted to proving \cref{lem:low-degree-block-approximationnew}. Before giving a complete proof, here we give a short preview on the high-level ideas.

Fix $w>0$ and an integer $k\geq1$.  Starting from a POVM
$\mathscr A=(A_a)_{a=1}^t$, set
\[
M_a(u)=\mathcal T_u(A_a),
\qquad
W_u=\sum_{a=1}^t|a\rangle\otimes\sqrt{M_a(u)}.
\]
The goal is to choose $u\in[w,2w]$ and replace the rectangular block operator
$W_u$ by an operator $Q=\sum_a\ket{a}\otimes Q_a$ of Pauli degree at most $k$, while controlling both
$\norm{W_u-Q}_4$ and the normalization error $\norm{Q^\dagger Q-\id}_2$.
Then $Q_a^\dagger Q_a$ is a positive low-degree approximation of $M_a(u)$.

We first prove the following approximation bound for Hermitian matrices.
For every Hermitian matrix
$S\in\Matrix_{d_0}\otimes\Matrix_{2^n}$, setting $\theta=bk^{-\gamma}$ gives
\begin{equation*}
\operatorname{dist}_4(S,E_{\leq k})
\leq
2\norm{S-\mathcal T_\theta S}_4,
\end{equation*}
where $b,\gamma>0$ are absolute constants independent of $d_0$ and $n$,
and $E_{\leq k}$ is the subspace of operators of Pauli degree at most $k$.
Here
\(
\operatorname{dist}_4(S,E_{\le k})
:=\inf_{R\in E_{\le k}}\norm{S-R}_4.
\)

Our proof is inspired by Mendel and Naor~\cite[Theorem~5.1]{MN14}, who
proved an $L_p$-norm decay bound for the heat semigroup on the Boolean cube,
applied to functions whose Fourier support is contained in degrees at least $k$.

To obtain the corresponding bound for quantum depolarization, we map each
operator to a function on $G^n$ by Pauli unitary conjugation.
By \cref{lem:pauli-four-point}, this preserves
the norm and converts Pauli terms into Fourier terms of the same degree,
while quantum depolarization becomes classical depolarization on $G^n$.
In \cref{lem:four-point-boolean}, we compare this classical depolarizing
semigroup with the heat semigroup on the $2n$-bit cube in the
$L_{4/3}$-norm.  The Mendel--Naor bound then gives
the Schatten $4/3$ decay estimate in \cref{cor:S43-Pauli-tail}.

The distance formula in \cref{lem:pauli-distance-duality} yields the
approximation bound above.  Finally, placing
a rectangular block operator and its adjoint in the off-diagonal blocks
of a Hermitian matrix gives the approximation bound for rectangular block
operators in \cref{lem:rectangular-tail}.

It remains to choose $u\in[w,2w]$ such that
$\norm{W_u-\mathcal T_\theta W_u}_4$ is small.
Integrating the entropy derivative over $[w,2w]$ gives a time $u$ for which
the Stroock--Varopoulos inequality implies
\[
\sum_a\cE_{2,\cL}\!\left(\sqrt{M_a(u)}\right)
\leq\sum_a\cE_{1,\cL}(M_a(u))
\leq\frac{\log t}{4w},
\]
where $\cE_{2,\cL}$ and $\cE_{1,\cL}$ are defined in~\eqref{eq:tracial-dirichlet-forms}. 
The Pauli expansions of the square roots then yield
\[
\norm{W_u-\mathcal T_\theta W_u}_4^4
\leq
b\frac{k^{-\gamma}\log t}{w}.
\]

Applying the rectangular approximation lemma with this bound produces $Q$;
Schatten H\"older then
controls $Q^\dagger Q-\id$ and completes the proof of
\cref{lem:low-degree-block-approximationnew}.

\subsubsection{Approximation in the Schatten 4-Norm}

Recall that we use $\Matrix_d$ to denote the set of $d\times d$ complex matrices.
In this subsection we prove a ``low-degree'' approximation result for matrices in $\Matrix_{d_0}\otimes\Matrix_{2^n}$, with respect to the normalized Schatten 4-norm,
where the ``low-degree''
applies to only the $\Matrix_{2^n}$ part.

\begin{restatable}{theorem}{thmsquaretail}
\label{thm:square-tail}
There are absolute constants $b,\gamma>0$ with the following
property.
Let $d_0,n\ge1$,
and let $S$ be a Hermitian operator in
\(
\Matrix_{d_0}\otimes\Matrix_{2^n}.
\)
For every integer $k\ge1$,
there is a Hermitian operator $R$,
with the low-degree expansion
\[
R=\sum_{|\alpha|\le k}\widehat R(\alpha)\otimes P_\alpha,
\quad \widehat R(\alpha)\in \Matrix_{d_0},
\]
such that
\begin{equation*}
\norm{S-R}_4
\le
2\norm{S-\br{\id_{d_0}\otimes\mathcal T_{bk^{-\gamma}}^{\otimes n}}(S)}_4.
\end{equation*}
The constants $b$ and $\gamma$ are independent of $d_0$ and $n$.
\end{restatable}

The proof of~\cref{thm:square-tail} is inspired by the techniques introduced by Mendel and Naor~\cite{MN14} to bound the norm of high-order terms of a Banach-space-valued Boolean functions.  We first briefly introduced the framework therein.
Here we use the Banach-space-valued Walsh expansion on $\Omega_N = \mathbb{F}_2^N$,
and define $\chi_\xi(x)=(-1)^{\xi\cdot x}$.
For $u\geq0$, the
Boolean heat semigroup is the Fourier multiplier
\begin{equation}
\label{eq:boolean-heat-semigroup}
\mathsf B_u f
:=
\sum_{\xi\in\Omega_N}e^{-u|\xi|}\widehat f(\xi)\chi_\xi.
\end{equation}

For each $d_0,n\ge1$, write
\begin{equation}
\label{eq:schatten-coefficient-space}
X_{n,d_0}
=
\left(\Matrix_{d_0}\otimes \Matrix_2^{\otimes n},\norm{\cdot}_{4/3}\right).
\end{equation}
We use the following form of the Mendel--Naor estimate for these spaces.
\begin{restatable}{lemma}{schattenbooleandecay}
\label{lem:schatten-boolean-decay}
There are absolute constants $a_0>0$, $C_0\ge1$, and $B>2$ such that, for
all integers $d_0,n,N,k\ge1$, every $u>0$, and every
$f\in L_{4/3}^{>k}(\Omega_N;X_{n,d_0})$,
\begin{equation*}
\norm{\mathsf B_u f}_{L_{4/3}(\Omega_N;X_{n,d_0})}
\le
C_0e^{-a_0k\min\{u,u^B\}}
\norm f_{L_{4/3}(\Omega_N;X_{n,d_0})}.
\end{equation*}
\end{restatable}

The proof is given in \cref{app:schatten-boolean-decay}.

Now, we consider the depolarizing Semigroup on $\mathbb F_2^2$.
Let $G=\mathbb F_2^2$, with the uniform probability measure.
The classical depolarizing semigroup on $G$ is
defined, for scalar- or Banach-space-valued functions $f$ on $G$, by 
\begin{equation}
\label{eq:four-point-heat-semigroup}
\mathsf H_t^{(4)}f
=e^{-t}f+(1-e^{-t})\mathbb E_G f.
\end{equation}
Equivalently, its action on the characters of $G$ is
\[
\mathsf H_t^{(4)}\chi_0=\chi_0,
\qquad
\mathsf H_t^{(4)}\chi_\xi=e^{-t}\chi_\xi
\quad(\xi\ne0).
\]
Its $n$-fold product satisfies
\begin{equation}
\label{eq:four-point-product-semigroup}
\mathsf H_t^{(4,n)}\chi_\alpha
=e^{-t|\alpha|}\chi_\alpha,
\qquad
\alpha\in G^n,
\end{equation}
where $|\alpha|=|\{i\in[n]:\alpha_i\ne0\}|$.

\begin{lemma}[Pauli conjugation]
\label{lem:pauli-multiplier-implementation}
Let $G=\mathbb F_2^2$, and let $\vartheta(a,b)=(b,a)$, extended
coordinatewise to $G^n$. For $\rho\in G^n$, set
$U_\rho=P_{\vartheta(\rho)}$. Then, for every $\alpha\in G^n$,
\[
    U_\rho P_\alpha U_\rho^\dagger
    = \chi_\rho(\alpha)P_\alpha,
    \qquad
    \chi_\rho(\alpha)
    = (-1)^{\sum_{j=1}^n \rho_j\cdot\alpha_j}.
\]
\end{lemma}

\begin{proof}
For one qubit, the Pauli commutation relations give
\[
    P_{(b,a)}P_{(c,d)}P_{(b,a)}^\dagger
    = (-1)^{ac+bd}P_{(c,d)}.
\]
Taking the tensor product over all coordinates gives the claim.
\end{proof}

\begin{lemma}
\label{lem:pauli-four-point}
Let $d_0,n\ge1$ and let
\(
A\in \Matrix_{d_0}\otimes \Matrix_2^{\otimes n}.
\)
Using \cref{def:coefficient-valued-pauli-expansion} with $\V=\Matrix_{d_0}$,
write the operator-valued Pauli expansion
\(
A=\sum_{\alpha\in G^n}\widehat A(\alpha)\otimes P_\alpha,
\widehat A(\alpha)\in \Matrix_{d_0}.
\)
For every $t\geq0$, let the depolarizing semigroup act only on the $n$
qubits:
\begin{equation*}
\mathcal T_t(A)
:=\sum_{\alpha\in G^n}
e^{-t|\alpha|}\widehat A(\alpha)\otimes P_\alpha.
\end{equation*}
For $\rho\in G^n$, let $U_\rho=P_{\vartheta(\rho)}$ be the unitary matrix from
\cref{lem:pauli-multiplier-implementation}.  Define the operator-valued
function $J(A):G^n\to \Matrix_{d_0}\otimes \Matrix_2^{\otimes n}$ by
\begin{equation*}
J(A)(\rho)
:=
(\id\otimes U_\rho)A(\id\otimes U_\rho)^\dagger.
\end{equation*}
Then, for every $1\le p<\infty$,
\begin{equation*}
\left(\mathbb E_{\rho\in G^n}
  \norm{J(A)(\rho)}_p^p\right)^{1/p}
=\norm A_p.
\end{equation*}
Its Fourier degree satisfies $\deg_G(J(A))=\deg(A)$.
Moreover,
\begin{equation*}
J(\mathcal T_tA)=\mathsf H_t^{(4,n)}J(A).
\end{equation*}
\end{lemma}

\begin{proof}
The isometry follows immediately from unitary invariance of Schatten norms:
\[
\mathbb E_{\rho\in G^n}\norm{J(A)(\rho)}_p^p
=
\mathbb E_\rho
\norm{(\id\otimes U_\rho)A(\id\otimes U_\rho)^\dagger}_p^p
=
\norm A_p^p.
\]
The Pauli-multiplier implementation in
\cref{lem:pauli-multiplier-implementation} rewrites $J(A)$ as
\begin{equation*}
J(A)(\rho)
=
\sum_{\alpha\in G^n}
\chi_\rho(\alpha)\widehat A(\alpha)\otimes P_\alpha
=
\sum_{\alpha\in G^n}
\chi_\alpha(\rho)\widehat A(\alpha)\otimes P_\alpha,
\end{equation*}
where the second equality uses the symmetry of the standard pairing on
$G^n$.  In the convention of
\cref{def:coefficient-valued-fourier-expansion}, the Fourier
coefficient of $J(A)$ is therefore
\begin{equation*}
\widehat{J(A)}(\alpha)=\widehat A(\alpha)\otimes P_\alpha.
\end{equation*}
Here $\widehat A(\alpha)$ is the Pauli coefficient of $A$, whereas
$\widehat{J(A)}(\alpha)$ lies in $\Matrix_{d_0}\otimes \Matrix_2^{\otimes n}$.
Since $P_\alpha$ is nonzero, the two coefficients vanish for exactly the
same labels, so $\deg_G(J(A))=\deg(A)$.
Finally, both $\mathcal T_t$ and
$\mathsf H_t^{(4,n)}$ multiply
the $\alpha$-summand by $e^{-t|\alpha|}$, which proves the last identity
in \cref{lem:pauli-four-point}.
\end{proof}

The identity $J\mathcal T_t=\mathsf H_t^{(4,n)}J$ gives the following
commutative diagram:
\begin{center}
\begin{tikzpicture}[
  >=Latex,
  node distance=18mm and 45mm,
  every node/.style={font=\normalsize}
]
  \node (A) {$A$};
  \node[right=of A] (TA) {$\mathcal T_tA$};
  \node[below=of A] (JA) {$J(A)$};
  \node[below=of TA] (JTA) {$J(\mathcal T_tA)$};
  \draw[->] (A) -- node[above] {$\mathcal T_t$} (TA);
  \draw[->] (A) -- node[left] {$J$} (JA);
  \draw[->] (TA) -- node[right] {$J$} (JTA);
  \draw[->] (JA) -- node[below] {$\mathsf H_t^{(4,n)}$} (JTA);
\end{tikzpicture}
\end{center}
For every $A\in\Matrix_{d_0}\otimes\Matrix_2^{\otimes n}$,
the right-then-down path equals
$J(\mathcal T_tA)$, whereas the down-then-right path equals
$\mathsf H_t^{(4,n)}J(A)$.  Thus commutativity of the diagram is exactly the
last identity in \cref{lem:pauli-four-point}.  More explicitly, the left vertical
  arrow sends the Pauli summand $\widehat A(\alpha)\otimes P_\alpha$ to the
function
\[
\rho\longmapsto
\chi_\rho(\alpha)\widehat A(\alpha)\otimes P_\alpha.
\]
Thus $J$ converts Pauli degree into Fourier degree while preserving the
multiplier $e^{-t|\alpha|}$.

We next compare $\mathsf H_t^{(4,n)}$ with the heat semigroup on the Boolean cube after identifying
$G^n$ with the $2n$-bit cube $\Omega_{2n}$.  Both use the Walsh characters
of \cref{def:binary-characters}.  Their weights differ: the former counts
nonzero two-bit blocks, while the latter counts individual nonzero bits.

\begin{lemma}
\label{lem:four-point-boolean}
There are absolute constants $a_1>0$, $C_0\ge1$, and $B>2$ such that, for every
$d_0,n,k\ge1$, every $t>0$, and every
\[
f\in L_{4/3}(G^n;X_{n,d_0})
\]
satisfying $\widehat f(\alpha)=0$ whenever $|\alpha|\leq k$,
\begin{equation*}
\norm{\mathsf H_t^{(4,n)}f}_{L_{4/3}(G^n;X_{n,d_0})}
\le
C_0e^{-a_1k\min\{t,t^B\}}
\norm f_{L_{4/3}(G^n;X_{n,d_0})}.
\end{equation*}
\end{lemma}

\begin{proof}
Set $q=e^{-t/2}$.  On one two-bit block, the semigroup $\mathsf B_{t/2}$
from \eqref{eq:boolean-heat-semigroup} has, in the ordering
$00,10,01,11$, the multipliers
\begin{equation*}
(1,q,q,q^2).
\end{equation*}
Define the probability measure on $G$
\begin{equation}
\label{eq:four-point-smoothing-measure}
\nu_q
=
\frac{1+q}{2}\delta_{00}
+
\frac{1-q}{2}\delta_{11}
\end{equation}
where $\delta_z$ denotes the probability measure concentrated at $z$.
For a complex Banach space $X$ and $F:G\to X$, define $\mathsf C_q$ using
the convolution in \cref{def:finite-abelian-convolution}:
\begin{equation}
\label{eq:four-point-averaging-operator}
  (\mathsf C_qF)(x):=(\nu_q*F)(x)
  =\frac{1+q}{2}F(x)+\frac{1-q}{2}F(x+11).
\end{equation}
Here subtraction equals addition in $G$.  Its multiplier on $\chi_\xi$ is
\[
\sum_{z\in G}\nu_q(\{z\})\chi_\xi(z)
=
\frac{1+q}{2}
+
\frac{1-q}{2}(-1)^{\xi\cdot11},
\]
so its four multipliers are
\begin{equation*}
(1,q,q,1).
\end{equation*}
The entrywise product of the two multiplier lists above is
\[
(1,q^2,q^2,q^2)
=
(1,e^{-t},e^{-t},e^{-t}),
\]
the multipliers of $\mathsf H_t^{(4)}$.  Therefore, under the identification
$G^n\cong\mathbb F_2^{2n}$,
\begin{equation*}
\mathsf H_t^{(4,n)}
=
\mathsf C_q^{\otimes n}\mathsf B_{t/2}^{(2n)}.
\end{equation*}

For $1\leq p<\infty$, Minkowski's inequality and translation invariance imply
\[
\begin{aligned}
\norm{\mathsf C_qF}_{L_p(G;X)}
&=
\left(\mathbb E_x\left\|\mathbb E_{z\sim\nu_q}F(x+z)\right\|_X^p\right)^{1/p}\\
&\le
\mathbb E_{z\sim\nu_q}
\left(\mathbb E_x\norm{F(x+z)}_X^p\right)^{1/p}
=
\norm F_{L_p(G;X)}.
\end{aligned}
\]
The same holds for $\mathsf C_q^{\otimes n}$.

Each nonzero pair $\alpha_i\in G$ contains one or two nonzero bits.  Hence the
Boolean degree of a character is at least its $G$-degree.  Thus
$\widehat f(\alpha)=0$ for $|\alpha|\leq k$ implies that all Boolean
Fourier coefficients of degree at most $k$ vanish.
Apply \cref{lem:schatten-boolean-decay} at time $t/2$:
\begin{align*}
\norm{\mathsf H_t^{(4,n)}f}_{4/3}
&=\norm{\mathsf C_q^{\otimes n}\mathsf B_{t/2}^{(2n)}f}_{4/3}\\
&\le\norm{\mathsf B_{t/2}^{(2n)}f}_{4/3}\\
&\le
C_0
\exp\!\left\{-a_0k\min\left\{\frac t2,
\left(\frac t2\right)^B\right\}\right\}
\norm f_{4/3}.
\end{align*}
Since
\[
\min\left\{\frac t2,\left(\frac t2\right)^B\right\}
\ge
2^{-B}\min\{t,t^B\},
\]
the claimed estimate follows with $a_1=2^{-B}a_0$.
\end{proof}


Combining \cref{lem:pauli-four-point,lem:four-point-boolean} gives the following
bound for operators whose Pauli coefficients vanish in degrees at most $k$.

\begin{cor}
\label{cor:S43-Pauli-tail}
Let $d_0,n,k\ge1$, and let $A\in \Matrix_{d_0}\otimes \Matrix_2^{\otimes n}$ be supported
on Pauli degrees strictly greater than $k$.  Then, for every $t>0$,
\begin{equation*}
\norm{\mathcal T_tA}_{4/3}
\le
C_0e^{-a_1k\min\{t,t^B\}}\norm A_{4/3},
\end{equation*}
where the constants are independent of $n$ and $d_0$.
Here $\mathcal T_t$ and the normalized Schatten norms are those in
\cref{lem:pauli-four-point}, and $a_1,C_0,B$ are the constants from
\cref{lem:four-point-boolean}.
\end{cor}

\begin{proof}
By \cref{lem:pauli-four-point}, $\widehat{J(A)}(\alpha)=0$ whenever
$|\alpha|\leq k$, and
\[
J(\mathcal T_tA)=\mathsf H_t^{(4,n)}J(A).
\]
Apply \cref{lem:four-point-boolean} to $J(A)$ and then use the isometry in
\cref{lem:pauli-four-point}:
\[
\begin{aligned}
\norm{\mathcal T_tA}_{4/3}
&=
\norm{J(\mathcal T_tA)}_{L_{4/3}(G^n;X_{n,d_0})}\\
&=
\norm{\mathsf H_t^{(4,n)}J(A)}_{L_{4/3}(G^n;X_{n,d_0})}\\
&\le
C_0e^{-a_1k\min\{t,t^B\}}
\norm{J(A)}_{L_{4/3}(G^n;X_{n,d_0})}\\
&=
C_0e^{-a_1k\min\{t,t^B\}}\norm A_{4/3}.
\end{aligned}
\]
\end{proof}


\begin{definition}
\label{def:pauli-degree-subspaces}
For integers $d_0,n,k\geq1$, define the subspaces
$E_{\le k},E_{>k}\subseteq\Matrix_{d_0}\otimes\Matrix_2^{\otimes n}$ by
\begin{align*}
E_{\le k}&:=\{R:\widehat R(\alpha)=0
  \text{ whenever }|\alpha|>k\},\\
E_{>k}&:=\{B:\widehat B(\alpha)=0
  \text{ whenever }|\alpha|\leq k\}.
\end{align*}
\end{definition}

\begin{lemma}
\label{lem:pauli-distance-duality}
Let $d_0,n,k\geq1$ be integers, let
$\tau=\tau_{d_0}\otimes\tau_{2^n}$, and let $E_{\le k},E_{>k}$
be the subspaces in \cref{def:pauli-degree-subspaces}.
Then, for every $A\in\Matrix_{d_0}\otimes\Matrix_2^{\otimes n}$,
\begin{equation*}
\operatorname{dist}_4(A,E_{\le k})
:=\inf_{R\in E_{\le k}}\norm{A-R}_4
=\sup_{\substack{B\in E_{>k}\\\norm B_{4/3}\leq1}}
  \abs{\tau(B^\dagger A)}.
\end{equation*}
\end{lemma}

\begin{proof}
Pauli orthogonality gives
\[
B\in E_{>k}
\quad\Longleftrightarrow\quad
\tau(B^\dagger R)=0
\quad\text{for every }R\in E_{\le k}.
\]
Thus, for $B\in E_{>k}$ with $\norm B_{4/3}\leq1$ and $R\in E_{\le k}$,
Schatten H\"older gives
\[
\abs{\tau(B^\dagger A)}
=\abs{\tau(B^\dagger(A-R))}
\leq\norm B_{4/3}\norm{A-R}_4
\leq\norm{A-R}_4.
\]
Taking the supremum over $B$ and the infimum over $R$ proves one inequality.

Choose a minimizer $R_0\in E_{\le k}$, which exists because $E_{\le k}$
is a finite-dimensional subspace, and put $C=A-R_0$.
If $C=0$, both sides are zero.  Otherwise, for every $R\in E_{\le k}$,
minimality gives
\[
0=\left.\frac{\dd}{\dd t}\norm{C-tR}_4^4\right|_{t=0}
=-4\operatorname{Re}\tau\!\left((C^\dagger C)C^\dagger R\right).
\]
Applying the same identity to $iR\in E_{\le k}$ also makes the imaginary
part zero.  Therefore
\[
B:=\frac{C(C^\dagger C)}{\norm C_4^3}
\]
satisfies $\tau(B^\dagger R)=0$ for all $R\in E_{\le k}$, so $B\in E_{>k}$.
Moreover,
\begin{align*}
\norm B_{4/3}^{4/3}
&=\frac{\tau((C^\dagger C)^2)}{\norm C_4^4}=1,\\
\tau(B^\dagger A)
&=\tau(B^\dagger C)
=\frac{\tau((C^\dagger C)^2)}{\norm C_4^3}
=\norm C_4
=\operatorname{dist}_4(A,E_{\le k}).
\end{align*}
This $B$ attains the claimed supremum and proves the reverse inequality.
\end{proof}

We are now ready to prove \cref{thm:square-tail}.
\thmsquaretail*
\begin{proof}[Proof of \cref{thm:square-tail}]
Use the subspaces $E_{\le k}$ and $E_{>k}$ from
\cref{def:pauli-degree-subspaces}.
If $k\ge n$, then $E_{\le k}=\Matrix_{d_0}\otimes\Matrix_{2^n}$, and
$R=S$ satisfies the theorem.  We therefore assume $1\le k<n$.

By \cref{lem:pauli-distance-duality}, with
$\langle B,A\rangle=\tau\Br{B^\dagger A}$,
we have
\begin{equation*}
\operatorname{dist}_4(A,E_{\le k})
=
\sup_{\substack{B\in E_{>k}\\ \norm B_{4/3}\le1}}
\abs{\langle B,A\rangle}.
\end{equation*}

The map $\mathcal T_t$ preserves $E_{>k}$ and is self-adjoint for this
pairing.  Moreover, \cref{cor:S43-Pauli-tail} bounds its norm on $E_{>k}$ by
$C_0e^{-a_1k\min\{t,t^B\}}$.  Consequently, for every
$A\in\Matrix_{d_0}\otimes\Matrix_2^{\otimes n}$,
\begin{align*}
\operatorname{dist}_4(\mathcal T_tA,E_{\le k})
&=
\sup_{\substack{B\in E_{>k}\\ \norm B_{4/3}\le1}}
\abs{\langle B,\mathcal T_tA\rangle} \notag\\
&=
\sup_{\substack{B\in E_{>k}\\ \norm B_{4/3}\le1}}
\abs{\langle \mathcal T_tB,A\rangle} \notag\\
&\le
C_0e^{-a_1k\min\{t,t^B\}}
\operatorname{dist}_4(A,E_{\le k}).
\end{align*}
Here the last inequality uses the fact that $\mathcal T_tB$ remains in
$E_{>k}$ and has Schatten $4/3$-norm at most
$C_0e^{-a_1k\min\{t,t^B\}}$ whenever $\norm B_{4/3}\le1$.

Set
\(
\gamma=\frac1B.
\)
Choose an absolute $b\ge1$ large enough that
\(
C_0e^{-a_1\min\{b,b^B\}}\le\frac12,
\)
and put
\(
\theta=bk^{-\gamma},
\)
which implies that
\[
k\min\{\theta,\theta^B\}
=
\begin{cases}
k\theta^B=b^B, & \theta\le1,\\
k\theta=bk^{1-1/B}\ge b, & \theta\ge1.
\end{cases}
\]
With $t=\theta$, the preceding distance estimate becomes, for every
$A\in\Matrix_{d_0}\otimes\Matrix_2^{\otimes n}$,
\begin{equation*}
\operatorname{dist}_4(\mathcal T_\theta A,E_{\le k})
\le\frac12\operatorname{dist}_4(A,E_{\le k}).
\end{equation*}
Apply this contraction with $A=S$.  The triangle inequality then shows that
\begin{align*}
\operatorname{dist}_4(S,E_{\le k})
&\le
\norm{S-\mathcal T_\theta S}_4
+
\operatorname{dist}_4(\mathcal T_\theta S,E_{\le k})\\
&\le
\norm{S-\mathcal T_\theta S}_4
+
\frac12\operatorname{dist}_4(S,E_{\le k}).
\end{align*}
Therefore
\begin{equation*}
\operatorname{dist}_4(S,E_{\le k})
\le
2\norm{S-\mathcal T_\theta S}_4.
\end{equation*}
Since $E_{\le k}$ is finite dimensional, there exists $R_0\in E_{\le k}$
such that
\[
\norm{S-R_0}_4=\operatorname{dist}_4(S,E_{\le k}).
\]
Since $S=S^\dagger$ and $E_{\le k}$ is closed under adjoints, the
Hermitian part
\(
R=\displaystyle\frac{R_0+R_0^\dagger}{2}
\)
belongs to $E_{\le k}$ and satisfies
\[
\norm{S-R}_4
=
\norm{\operatorname{Re}(S-R_0)}_4
\le
\norm{S-R_0}_4.
\]
Combining this with the preceding distance estimate proves
\cref{thm:square-tail}, since $\theta=bk^{-\gamma}$.
\end{proof}

\begin{remark}
\label{rem:S4-distance-duality}
In this remark, write
$S_p:=(\Matrix_{d_0}\otimes\Matrix_2^{\otimes n},\norm{\cdot}_p)$
for $1<p<\infty$, where the Schatten norms are normalized by
$\tau_{d_0}\otimes\tau_{2^n}$.
Use the subspaces $E_{\le k}$ and $E_{>k}$ from
\cref{def:pauli-degree-subspaces}.

Consider the quotient space $S_4/E_{\le k}$.  For $A\in S_4$, denote its
coset by $[A]:=A+E_{\le k}$.  Its quotient norm is
\begin{equation}
\label{eq:low-degree-quotient-norm}
\norm{[A]}_{S_4/E_{\le k}}
=
\inf_{R\in E_{\le k}}\norm{A-R}_4
=
\operatorname{dist}_4(A,E_{\le k}).
\end{equation}
Because $\mathcal T_t(E_{\le k})\subseteq E_{\le k}$, the map $\mathcal T_t$
induces a well-defined map on the quotient by
\begin{equation}
\label{eq:induced-depolarizing-map}
\widetilde{\mathcal T}_t[A]:=[\mathcal T_tA].
\end{equation}
With respect to the normalized trace pairing
$\langle B,A\rangle=(\tau_{d_0}\otimes\tau_{2^n})(B^\dagger A)$, define the
annihilator of $E_{\le k}$ by
\begin{equation}
\label{eq:low-degree-annihilator}
E_{\le k}^{\perp}
:=
\left\{B\in S_{4/3}:\langle B,R\rangle=0
\text{ for every }R\in E_{\le k}\right\}.
\end{equation}
Pauli orthogonality identifies it with $E_{>k}$.
The dual of the quotient is identified isometrically as
\[
\left(S_4/E_{\le k}\right)^*
\cong
E_{\le k}^{\perp}
\subseteq S_{4/3},
\]
and, under this identification, the adjoint of
$\widetilde{\mathcal T}_t$ is the restriction of $\mathcal T_t^*$ to
$E_{\le k}^{\perp}$.  Since $\mathcal T_t^*=\mathcal T_t$, the high-degree
estimate in \cref{cor:S43-Pauli-tail} implies
\[
\norm{\widetilde{\mathcal T}_t}_{S_4/E_{\le k}\to S_4/E_{\le k}}
=
\sup_{\substack{B\in E_{\le k}^{\perp}\\\norm B_{4/3}\leq1}}
\norm{\mathcal T_tB}_{4/3}
\le
C_0e^{-a_1k\min\{t,t^B\}}.
\]
Applying this operator bound to $[A]$ and using the definition of the
quotient norm in \eqref{eq:low-degree-quotient-norm} recovers the
distance estimate proved above.

Once the factor-$1/2$ distance contraction is known, the last step of the proof
can equivalently be written entirely in the quotient:
\[
\begin{aligned}
\norm{[S]}_{S_4/E_{\le k}}
&\le
\norm{[S-\mathcal T_\theta S]}_{S_4/E_{\le k}}
+
\norm{\widetilde{\mathcal T}_\theta[S]}_{S_4/E_{\le k}}\\
&\le
\norm{S-\mathcal T_\theta S}_4
+
\frac12\norm{[S]}_{S_4/E_{\le k}}.
\end{aligned}
\]
This bounds $\operatorname{dist}_4(S,E_{\le k})$.  Since $E_{\le k}$ is
finite dimensional, the infimum is attained by some $R\in E_{\le k}$,
as shown at the end of the proof.
\end{remark}

\subsubsection{Rectangular Operators}

\begin{lemma}
\label{lem:rectangular-tail}
Let $n,t\geq1$ be integers and $b,\gamma>0$ be the absolute constants from
\cref{thm:square-tail}.
Let 
$W\colon(\mathbb C^2)^{\otimes n}\to\mathbb C^t\otimes(\mathbb C^2)^{\otimes n}$
be a block operator.
For every integer $k\ge1$, there exists a block operator
$Q\colon(\mathbb C^2)^{\otimes n}\to\mathbb C^t\otimes(\mathbb C^2)^{\otimes n}$
with Pauli degree at most $k$
\[
Q=\sum_{|\alpha|\le k}\widehat Q(\alpha)\otimes P_\alpha,
\qquad \widehat Q(\alpha)\in\mathbb C^t,
\]
such that
\begin{equation*}
\norm{W-Q}_4
\le2\norm{W-\mathcal T_{bk^{-\gamma}}W}_4.
\end{equation*}
The constants $b,\gamma$ are independent of $t$ and $n$.
\end{lemma}

\begin{proof}
To apply the square-matrix approximation theorem to a rectangular block
operator, we place it in the off-diagonal block of a Hermitian matrix.
For a block operator
$A:(\mathbb C^2)^{\otimes n}\to\mathbb C^t\otimes(\mathbb C^2)^{\otimes n}$
with the Pauli expansion
\[
A=\sum_{\alpha\in G^n}\widehat A(\alpha)\otimes P_\alpha,
\qquad
\widehat A(\alpha)\in\mathbb C^t,
\]
define its \emph{Hermitian dilation} by
\begin{equation}
\label{eq:hermitian-dilation}
\mathscr H(A)=
\begin{pmatrix}
0&A^\dagger\\
A&0
\end{pmatrix}.
\end{equation}
Equivalently,
\[
\mathscr H(A)
=
\sum_{\alpha\in G^n}
\begin{pmatrix}
0&\widehat A(\alpha)^\dagger\\
\widehat A(\alpha)&0
\end{pmatrix}
\otimes P_\alpha.
\]
Thus $A$ and $\mathscr H(A)$ have exactly the same Pauli support and the
same Pauli degree.  The depolarizing map $\mathcal T_\theta$ acts only on
the $(\mathbb C^2)^{\otimes n}$ factor, so
\begin{equation*}
\mathscr H(\mathcal T_\theta A)=\mathcal T_\theta\mathscr H(A).
\end{equation*}
For every rectangular $A$,
\[
\mathscr H(A)^2
=
\begin{pmatrix}
A^\dagger A&0\\
0&AA^\dagger
\end{pmatrix}.
\]
Moreover,
\[
(\Tr_{\mathbb C^t}\otimes\tau_{2^n})\!\left((AA^\dagger)^2\right)
=
\tau_{2^n}\!\left((A^\dagger A)^2\right),
\]
by cyclicity of the ordinary trace (equivalently, $AA^\dagger$ and
$A^\dagger A$ have the same nonzero singular values).  Hence
\begin{equation*}
\grayhighlight{
\norm{\mathscr H(A)}_4^4
=\frac{2}{t+1}\norm A_4^4.}
\end{equation*}

Apply \cref{thm:square-tail} to the Hermitian square matrix
$\mathscr H(W)$.  We obtain a Hermitian operator $R$ of Pauli degree at
most $k$ such that, with $\theta=bk^{-\gamma}$,
\begin{equation*}
\norm{\mathscr H(W)-R}_4
\le
2\norm{\mathscr H(W)-\mathcal T_\theta\mathscr H(W)}_4.
\end{equation*}
The approximant $R$ need not itself be off-diagonal.  Write its block
decomposition as
\[
R=
\begin{pmatrix}
R_{00}&R_{10}^\dagger\\
R_{10}&R_{11}
\end{pmatrix},
\qquad
R_{00}=R_{00}^\dagger,\qquad R_{11}=R_{11}^\dagger,
\]
and define the rectangular block operator
\begin{equation*}
Q:=R_{10}:(\mathbb C^2)^{\otimes n}\longrightarrow
\mathbb C^t\otimes(\mathbb C^2)^{\otimes n}.
\end{equation*}
Taking the $(1,0)$ block acts only on the coefficient space
$\mathbb C\oplus\mathbb C^t$ and therefore does not change any Pauli label.  Since
$R$ has degree at most $k$, so does $Q$.

It remains to check that discarding the diagonal blocks does not increase
the approximation error.  Let
\[
J_0=\begin{pmatrix}\id&0\\0&-\id\end{pmatrix},
\qquad
\mathcal P_{\mathrm{off}}(X)=\frac12(X-J_0XJ_0).
\]
Since unitary conjugation by $J_0$ preserves every Schatten $p$-norm,
\begin{equation*}
\norm{\mathcal P_{\mathrm{off}}(X)}_p
\le
\frac12\bigl(\norm X_p+\norm{J_0XJ_0}_p\bigr)
=\norm X_p.
\end{equation*}
The map $\mathcal P_{\mathrm{off}}$ fixes every Hermitian dilation,
commutes with $\mathcal T_\theta$, and does not increase Pauli degree.  Direct
multiplication of the $2\times2$ blocks shows that
\[
\mathcal P_{\mathrm{off}}(R)=
\frac12
\left[
\begin{pmatrix}R_{00}&R_{10}^\dagger\\ R_{10}&R_{11}\end{pmatrix}
-
\begin{pmatrix}R_{00}&-R_{10}^\dagger\\-R_{10}&R_{11}\end{pmatrix}
\right]
=
\begin{pmatrix}0&Q^\dagger\\Q&0\end{pmatrix}
=\mathscr H(Q).
\]
Thus the precise passage is $R\mapsto Q=R_{10}$ and
$\mathscr H(Q)=\mathcal P_{\mathrm{off}}(R)$; in general one does
\emph{not} have $R=\mathscr H(Q)$.  Using
$\mathcal P_{\mathrm{off}}(\mathscr H(W))=\mathscr H(W)$, we now have
\begin{equation*}
\mathscr H(W)-\mathscr H(Q)
=
\mathcal P_{\mathrm{off}}\bigl(\mathscr H(W)-R\bigr).
\end{equation*}
Combining these identities with the contractivity of
$\mathcal P_{\mathrm{off}}$, we obtain
\begin{align*}
\left(\frac{2}{t+1}\right)^{1/4}\norm{W-Q}_4
&=\norm{\mathscr H(W)-\mathscr H(Q)}_4\\
&\le\norm{\mathscr H(W)-R}_4\\
&\le2\norm{\mathscr H(W)-\mathcal T_\theta\mathscr H(W)}_4\\
&=2\left(\frac{2}{t+1}\right)^{1/4}
  \norm{W-\mathcal T_\theta W}_4.
\end{align*}
Canceling the common factor $(2/(t+1))^{1/4}$ proves
\cref{lem:rectangular-tail}.
\end{proof}

\subsubsection{Proof of \cref{lem:low-degree-block-approximationnew}: Approximation of the Block Operator}

\begin{proof}[Proof of \cref{lem:low-degree-block-approximationnew}]
Let $b,\gamma$ be the constants from \cref{lem:rectangular-tail}, and
choose the constant $C$ in the theorem so that $C\geq16b$.
We now return to a fixed POVM $\mathscr A=(A_a)_{a=1}^t$.  For $u>0$, set
\begin{equation}
\label{eq:smoothed-block-density}
M_a(u)=\mathcal T_u(A_a),
\qquad
Z_u=t\sum_{a=1}^t|a\rangle\langle a|\otimes M_a(u).
\end{equation}
Since $\tau(Z_u)=1$ and $0\leq Z_u\leq t\id$, the definition of
$\Ent$ in \eqref{eq:tracial-entropy}
specializes to
\begin{equation*}
0\leq\Ent(Z_u)=\tau(Z_u\log Z_u)\leq\log t.
\end{equation*}
By \cref{prop:tracial-entropy-derivative},
with the definition $\bar\cL=\id_{\Matrix_t}\otimes\cL$,
where $\cL$ is the generator of $\mathcal{T}_u$, we have
\begin{equation*}
-\frac{\dd}{\dd u}\Ent(Z_u)
=4\cE_{1,\bar\cL}(Z_u).
\end{equation*}
Integrating this identity gives
\[
4\int_w^{2w}\cE_{1,\bar\cL}(Z_u)\,\dd u
=\Ent(Z_w)-\Ent(Z_{2w})
\leq\log t.
\]
Hence there exists $u\in[w,2w]$ such that
\begin{equation*}
\cE_{1,\bar\cL}(Z_u)
\leq\frac{\log t}{4w}.
\end{equation*}
Fix such a time and define
\begin{equation*}
W_u=\sum_{a=1}^t|a\rangle\otimes\sqrt{M_a(u)}.
\end{equation*}
Apply \cref{cor:depolarizing-stroock-varopoulos} to each $M_a(u)$ and sum up
the resulting inequalities over $a$:
\begin{equation}\label{eqn:sumcE}
\grayhighlight{\sum_a\cE_{2,\cL}\!\left(\sqrt{M_a(u)}\right)
\leq
\sum_a\cE_{1,\cL}(M_a(u))
=
\cE_{1,\bar\cL}(Z_u)
\leq
\frac{\log t}{4w}.}
\end{equation}
For the equality, use \eqref{eq:tracial-dirichlet-forms}:
\begin{align*}
\cE_{1,\bar\cL}(Z_u)
&=\frac14\bar\tau\!\left((\bar\cL Z_u)\log Z_u\right)\\
&=\frac1{4t}\sum_a\tau_{2^n}\!\left(
  t\,\cL(M_a(u))\bigl((\log t)\id+\log M_a(u)\bigr)
  \right)\\
&=\sum_a\cE_{1,\cL}(M_a(u))
  +\frac{\log t}{4}\sum_a\tau_{2^n}\!\left(\cL(M_a(u))\right)\\
&=\sum_a\cE_{1,\cL}(M_a(u)),
\end{align*}
where zero blocks contribute $0$, and trace preservation gives
\[
\tau_{2^n}\!\left(\cL(M_a(u))\right)
=-\frac{\dd}{\dd u}\tau_{2^n}(M_a(u))
=-\frac{\dd}{\dd u}\tau_{2^n}(A_a)
=0.
\]

By \eqref{eq:block-pauli-expansion} and the $\ell_2^t$-norm in
\eqref{eq:euclidean-inner-product-and-norm}, the Pauli coefficients of $W_u$ satisfy
\begin{equation*}
\norm{\widehat W_u(\alpha)}_{\ell_2^t}^2
=
\sum_{a=1}^t\left|\widehat{\sqrt{M_a(u)}}(\alpha)\right|^2.
\end{equation*}
Let $\theta=bk^{-\gamma}$.  The vector-valued Pauli Parseval identity in
\cref{prop:block-operator-pauli-identities}, the multiplier formula for
$\mathcal T_\theta$, and $(1-e^{-x})^2\leq x$ for $x\geq0$ imply
\begin{align*}
\norm{W_u-\mathcal T_\theta W_u}_2^2
&=\sum_\alpha(1-e^{-\theta|\alpha|})^2
\norm{\widehat W_u(\alpha)}_{\ell_2^t}^2\\
&\le\theta\sum_\alpha|\alpha|
\norm{\widehat W_u(\alpha)}_{\ell_2^t}^2\\
&=\theta\sum_{a=1}^t\sum_{\alpha\in G^n}|\alpha|
\left|\widehat{\sqrt{M_a(u)}}(\alpha)\right|^2\\
&=\theta\sum_{a=1}^t
\cE_{2,\cL}\!\left(\sqrt{M_a(u)}\right).
\end{align*}
The last equality follows from
\cref{prop:depolarizing-dirichlet-pauli-expansion}.

Since $W_u^\dagger W_u=\sum_aM_a(u)=\id$,
\cref{prop:depolarizing-contractivity} implies
$\norm{\mathcal T_\theta(W_u)}_\infty\leq\norm{W_u}_\infty=1$.
The triangle inequality therefore yields
\[
\norm{W_u-\mathcal T_\theta W_u}_\infty
\leq\norm{W_u}_\infty+\norm{\mathcal T_\theta(W_u)}_\infty
\leq2.
\]
Therefore
\begin{align*}
\norm{W_u-\mathcal T_\theta W_u}_4^4
&\le\norm{W_u-\mathcal T_\theta W_u}_\infty^2
\norm{W_u-\mathcal T_\theta W_u}_2^2\\
&\le4\theta\sum_a\cE_{2,\cL}\!\left(\sqrt{M_a(u)}\right)
\le b\frac{k^{-\gamma}\log t}{w},
\end{align*}
where the last inequality is by~\eqref{eqn:sumcE}. By \cref{lem:rectangular-tail}, there is a block operator $Q$ of Pauli
degree at most $k$ such that
\begin{equation*}
\norm{W_u-Q}_4^4
\le16\norm{W_u-\mathcal T_\theta W_u}_4^4
\le16b\frac{k^{-\gamma}\log t}{w}
\le\eta^4,
\end{equation*}
where the last inequality uses the assumed lower bound on $k$ and
$C\geq16b$.
Since $W_u^\dagger W_u=\id$, Schatten H\"older gives
\[
\norm{Q^\dagger Q-\id}_2
\le\norm{Q-W_u}_4(\norm Q_4+\norm{W_u}_4)
\le\eta(2+\eta),
\]
because $\norm{W_u}_4=1$ and $\norm Q_4\le1+\eta$.
\end{proof}

\subsection{Step 3: Dimension Reduction}
\label{subsec:dimension-reduction}

In this subsection, we prove \cref{thm:folded-block-output}, restated below.

\foldedblockoutput*

\subsubsection{Pauli Folding}

Let $n,D\ge1$ be integers.  Choose $h\colon[n]\to[D]$ uniformly from
all functions and $r\in G^n$ uniformly, independently.  Recall
$H\colon G^n\to G^D$ from \eqref{eq:hashed-pauli-label}:
\[
(H\alpha)_j=\sum_{i:h(i)=j}\alpha_i,
\qquad \alpha\in G^n,\quad j\in[D].
\]
The linear map
$\Phi_{h,r}\colon \Matrix_2^{\otimes n}\to \Matrix_2^{\otimes D}$
from \eqref{eq:hashed-pauli-map} is given by
\begin{equation*}
\Phi_{h,r}(P_\alpha)=\chi_r(\alpha)P_{H\alpha},
\end{equation*}
where $\chi_r(\alpha)=\prod_i\chi_{r_i}(\alpha_i)$ is the character from
\cref{def:pauli-label-fourier-analysis}.

\begin{prop}
\label{prop:folded-norm-expectation}
Let $n,D\ge1$, fix a map $h\colon[n]\to[D]$, and choose $r\in G^n$
uniformly.  For every $X\in \Matrix_2^{\otimes n}$,
\begin{equation*}
\mathbb E_r\norm{\Phi_{h,r}(X)}_2^2
=\norm X_2^2,
\end{equation*}
where the norms on the left and right use the normalized traces
$\tau_{2^D}$ and $\tau_{2^n}$, respectively.
\end{prop}

\begin{proof}
Expand $X=\sum_\alpha\widehat X(\alpha)P_\alpha$.
Character orthogonality gives
$\mathbb E_r[\overline{\chi_r(\alpha)}\chi_r(\beta)]
=\mathbf1_{\{\alpha=\beta\}}$.
Thus, on expanding the squared norm, all cross terms vanish and
\begin{align*}
\mathbb E_r\norm{\Phi_{h,r}(X)}_2^2
&=\sum_\alpha|\widehat X(\alpha)|^2
\tau_{2^D}\!\left(P_{H\alpha}^\dagger P_{H\alpha}\right)\\
&=\sum_\alpha|\widehat X(\alpha)|^2
=\norm X_2^2.
\end{align*}
Here $P_{H\alpha}^\dagger P_{H\alpha}=\id$, and the last equality is
Pauli Parseval from \cref{prop:scalar-pauli-identities}.
\end{proof}

\subsubsection{Approximate Multiplicativity}

To prove \cref{thm:folded-block-output}, we first bound the failure of
$\Phi_{h,r}$ to preserve products, where the parameters are not optimized.

\begin{lemma}
  \label{lem:multiplication-error}
  There are absolute constants $C_0,C_{\mathrm{mul}}>0$ with the following property.  Let
  $n,k,D\geq1$.  Choose $h\colon[n]\to[D]$ uniformly from all functions
  and $r\in(\mathbb F_2^2)^n$ uniformly, independently, and let $\Phi_{h,r}$ be
  the randomized map in \eqref{eq:hashed-pauli-map}.  If $A,B\in \Matrix_2^{\otimes n}$ have Pauli
  degree at most $k$ and $D\geq C_0k^{16}$, then
  \begin{equation*}
    \mathbb E_{h,r}
    \norm{
      \Phi_{h,r}(AB)-\Phi_{h,r}(A)\Phi_{h,r}(B)
    }_2^2
    \leq
    \frac{C_{\mathrm{mul}}k^{16}}{D}\norm{A}_4^2\norm{B}_4^2.
  \end{equation*}
  Consequently, for every degree-$k$ operator $Q\in \Matrix_2^{\otimes n}$,
  \begin{equation*}
    \mathbb E_{h,r}
    \norm{
      \Phi_{h,r}(Q^\dagger Q)
      -\Phi_{h,r}(Q)^\dagger\Phi_{h,r}(Q)
    }_2^2
    \leq
    \frac{C_{\mathrm{mul}}k^{16}}D\norm{Q}_4^4.
  \end{equation*}
\end{lemma}

A complete proof of \cref{lem:multiplication-error} is mathematically involved. Here, we first prove the simple case $k=1$ to illustrate why random signs and hashing work.

\begin{definition}[Pauli cocycle]
\label{def:pauli-cocycle}
  For $s,t\in G$, define $\kappa(s,t)\in\{\pm1,\pm\iu\}$ satisfying
  \begin{equation*}
    P_sP_t=\kappa(s,t)P_{s+t}.
  \end{equation*}
\end{definition}

\begin{lemma}[Elementary approximate multiplicativity]\label{lem:mult}
For operators $A,B$ of Pauli degree at most one,
\begin{equation}\label{eq:mult}
  \mathbb E_{h,r}\|\Phi(AB)-\Phi(A)\Phi(B)\|_2^2
 \le\frac4D\|A\|_4^2\|B\|_4^2.
\end{equation}
\end{lemma}
\begin{proof}
Write $A=a_0I+\sum_{i,s\ne0}a_{is}P_s^{(i)}$ and
$B=b_0I+\sum_{i,s\ne0}b_{is}P_s^{(i)}$.
Then $\Phi(AB) - \Phi(A)\Phi(B)$ only includes the following collision terms:
For $i<j$ satisfying $h(i)=h(j)$, and $s,t\ne0$,
\begin{equation*}
    \Phi\br{P^{(i)}_sP^{(j)}_t} - \Phi\br{P^{(i)}_s}\Phi\br{P^{(j)}_t} =
      (1-\kappa(s,t))\chi_{r_i}(s)\chi_{r_j}(t)P^{(h(i))}_{s+t}.
\end{equation*}
Thus the two ordered products contributing to the
source label with $s$ at $i$ and $t$ at $j$ have, on a collision, combined
error coefficient
\[
 c_{ij}^{st}=a_{is}b_{jt}(1-\kappa(s,t))
             +a_{jt}b_{is}(1-\kappa(t,s)).
\]
Thus the entire error is
\[
 \sum_{i<j}\sum_{s,t\ne0}\mathbf1_{h(i)=h(j)}
 \chi_{r_i}(s)\chi_{r_j}(t)c_{ij}^{st}P_{s+t}^{(h(i))}.
\]
For fixed $h$, distinct tuples $(i,j,s,t)$ correspond to distinct source
characters. Averaging over $r$ therefore removes all cross terms, even
when their output Pauli strings coincide. Every output string has
normalized squared norm one, so
\[
 \mathbb E_r\|\Phi(AB)-\Phi(A)\Phi(B)\|_2^2
 =\sum_{i<j}\sum_{s,t\ne0}\mathbf1_{h(i)=h(j)}|c_{ij}^{st}|^2.
\]
If $s=t$, the coefficient is zero. If $s\ne t$, the two phases are
$\pm i$, and $|1-\kappa(s,t)|^2=|1-\kappa(t,s)|^2=2$.
Consequently,
\[
 |c_{ij}^{st}|^2\le4\bigl(|a_{is}|^2|b_{jt}|^2+
                              |a_{jt}|^2|b_{is}|^2\bigr).
\]
Since $\Pr[h(i)=h(j)]=1/D$, summing proves
\[
 \mathbb E_{h,r}\|\Phi(AB)-\Phi(A)\Phi(B)\|_2^2
 \le\frac4D\left(\sum_{i,s\ne0}|a_{is}|^2\right)
             \left(\sum_{j,t\ne0}|b_{jt}|^2\right).
\]
Parseval and $\|A\|_2\le\|A\|_4$ finish the proof.
\end{proof}

Lemma~\ref{lem:mult} is the most direct warm-up for the general
approximate-multiplicativity argument: a source label has at most two
ordered decompositions, and the squared error reduces to pairwise
collisions and character orthogonality. At higher degree, the number
of decompositions grows, and bounding each coefficient separately
loses the cancellations needed for polynomial dependence on degree.
Thus, for the general argument, we need a more advanced mathematical treatment  M\"obius inversion on partitions to analyze the error.

To prove \cref{lem:multiplication-error},
we now rewrite the squared multiplication error using the Pauli cocycle
introduced in \cref{def:pauli-cocycle}.
Let
\[
A=\sum_\alpha \widehat A(\alpha)P_\alpha,
\qquad
B=\sum_\beta \widehat B(\beta)P_\beta.
\]
The Pauli coefficient vector of $AB$ is
\begin{equation}
\label{eq:pauli-product-coefficients}
m_0(A,B)_\gamma
:=
\sum_{\alpha+\beta=\gamma}
\widehat A(\alpha)\widehat B(\beta)\,\kappa_n(\alpha,\beta).
\end{equation}
After hashing, we multiply Pauli matrices on $D$ qubits, so
the corresponding coefficient convolution is
\begin{equation}
\label{eq:hashed-product-coefficients}
m_h(A,B)_\gamma
:=
\sum_{\alpha+\beta=\gamma}
\widehat A(\alpha)\widehat B(\beta)\,\kappa_D(H\alpha,H\beta).
\end{equation}
In \eqref{eq:pauli-product-coefficients} and \eqref{eq:hashed-product-coefficients}, $\gamma\in G^n$; the corresponding output Pauli string
is $P_{H\gamma}$.

Fix a seed $(h,r)$ and abbreviate $\Phi=\Phi_{h,r}$.  Since $\chi_r$ is a
character and $H$ is linear,
\begin{align*}
\Phi(AB)
&=
\sum_\gamma \chi_r(\gamma)m_0(A,B)_\gamma P_{H\gamma},
\notag\\
\Phi(A)\Phi(B)
&=
\sum_{\alpha,\beta}
\widehat A(\alpha)\widehat B(\beta)
\chi_r(\alpha)\chi_r(\beta)
\kappa_D(H\alpha,H\beta)P_{H(\alpha+\beta)}
\notag\\
&=
\sum_\gamma \chi_r(\gamma)m_h(A,B)_\gamma P_{H\gamma}.
\end{align*}
Therefore the failure of $\Phi$ to be multiplicative is exactly
\begin{equation*}
\Phi(AB)-\Phi(A)\Phi(B)
=
\sum_\gamma
\chi_r(\gamma)
\bigl(m_0(A,B)_\gamma-m_h(A,B)_\gamma\bigr)P_{H\gamma}.
\end{equation*}
For fixed $h$, expand the squared norm of the preceding expression.
Character orthogonality in
the random variable $\bfr$ removes the terms with distinct $\gamma$, and hence
\begin{equation*}
\mathbb E_r
\norm{\Phi(AB)-\Phi(A)\Phi(B)}_2^2
=
\norm{m_0(A,B)-m_h(A,B)}_{\ell_2}^2.
\end{equation*}

For $\xi,\zeta\in\{0,h\}$ and
$A_1,A_2,A_3,A_4\in \Matrix_2^{\otimes n}$, define
\begin{equation}
\Lambda_{\xi,\zeta}(A_1,A_2,A_3,A_4)
:=
\sum_\gamma
m_\xi(A_1,A_2)_\gamma
\overline{m_\zeta(A_3,A_4)_\gamma}.
\label{eq:four-input-form-explicit}
\end{equation}
Here $\xi\in\{0,h\}$ specifies whether the product of $(A_1,A_2)$ uses
$\kappa_n(\alpha,\beta)$ or $\kappa_D(H\alpha,H\beta)$, respectively;
$\zeta$ specifies the same choice for $(A_3,A_4)$.  In particular,
\begin{equation*}
\Lambda_{\xi,\xi}(A,B,A,B)=\norm{m_\xi(A,B)}_{\ell_2}^2.
\end{equation*}
By \eqref{eq:four-input-form-explicit}, the squared $\ell_2$ norm above
has the exact decomposition
\begin{align*}
\mathbb E_r
\norm{\Phi(AB)-\Phi(A)\Phi(B)}_2^2
={}&
\Lambda_{0,0}(A,B,A,B)
-\Lambda_{0,h}(A,B,A,B)
\notag\\
&
-\Lambda_{h,0}(A,B,A,B)
+\Lambda_{h,h}(A,B,A,B).
\end{align*}
The following lemma bounds the expected difference between each
hash-dependent term and $\Lambda_{0,0}$.

\begin{restatable}{lemma}{lemfoldedproductestimate}
\label{lem:folded-product-estimate}
There are absolute constants $C_0,C_1>0$ with the following property.  Let
$n,k,D\geq1$, let
$A_1,A_2,A_3,A_4\in \Matrix_2^{\otimes n}$ have Pauli degree at most $k$, and
choose $h\colon[n]\to[D]$ uniformly from all functions.  Let
$m_0,m_h$ and $\Lambda_{\xi,\zeta}$ be as in
\eqref{eq:pauli-product-coefficients}, \eqref{eq:hashed-product-coefficients},
and \eqref{eq:four-input-form-explicit}.
If $D\geq C_0k^{16}$, then
\begin{equation*}
\abs{
\mathbb E_h\Lambda_{\xi,\zeta}(A_1,A_2,A_3,A_4)
-
\Lambda_{0,0}(A_1,A_2,A_3,A_4)
}
\le
\frac{C_1k^{16}}D
\prod_{j=1}^4\norm{A_j}_4
\end{equation*}
for every $\xi,\zeta\in\{0,h\}$.
\end{restatable}

\paragraph{Proof outline of~\cref{lem:folded-product-estimate}.}
The complete proof appears in \cref{app:folded-product-estimate}.
Expanding $\Lambda_{\xi,\zeta}$ involves four Pauli labels satisfying
$\alpha_1+\cdots+\alpha_4=0$.  Hashing changes the cocycle phases in this
expansion.  A tempting approach is to bound this change separately for
each label tuple and then apply the triangle inequality.  Since the
cocycle phases have modulus one, this leaves the Fourier coefficient sum
\begin{equation*}
\sum_{\alpha_1+\cdots+\alpha_4=0}
\prod_{j=1}^4\abs{\widehat A_j(\alpha_j)}.
\end{equation*}
Estimating this quantity by the size of the Fourier support introduces a
dependence on $n$, while the standard scalar hypercontractive estimate
incurs an exponential loss in $k$.  Neither bound has the dependence
required by the lemma.

The proof therefore keeps the cocycle phases and Fourier coefficients
together.  Without hashing, their sum is a trace expression:
\begin{equation*}
\abs{\Lambda_{0,0}(A_1,A_2,A_3,A_4)}
=\abs{\tau_{2^n}\!\left((A_3A_4)^\dagger A_1A_2\right)}
\le
\prod_{j=1}^4\norm{A_j}_4,
\end{equation*}
by Schatten H\"older.

First, the ratio between the hashed and source cocycle phases factors
over the hash fibers, namely the sets of coordinates sent to the same
bucket.  A first \mobius{} inversion rewrites the product of fiber phases
as a sum over partitions of the active coordinates, those on which at
least one of the four Pauli labels is nonzero.  The rank of a
partition is $r=\sum_B(|B|-1)$, where $B$ ranges over its blocks.
Requiring the hash to be constant on each block has probability $D^{-r}$.
In the full Fourier sum, the all-singleton contribution is exactly
$\Lambda_{0,0}$ and disappears upon subtraction.  We group the remaining terms by rank and
separate each partition's non-singleton block structure from its placement
among the $n$ coordinates.  At rank $r$, this structure involves at most
$2r$ coordinates.

For each fixed structure, we expand its dependence on the local Pauli
labels into characters.  Each character factors into one Pauli-label
character per input and is therefore implemented by Pauli conjugations.
These preserve the Schatten norms in the trace bound
above.  An auxiliary multivariate polynomial averages such character
multipliers and is therefore uniformly bounded by
$\prod_{j=1}^4\norm{A_j}_4$ on the unit cube.  Its mixed derivatives
encode sums over choices of coordinates.  Since its degree is controlled
by $k$, the Markov brothers' inequality bounds these derivatives without
introducing a dependence on $n$.  The derivative sums allow repeated
coordinates; a second \mobius{} inversion extracts the sums over pairwise
distinct coordinates required by the partition placements.

Finally, we combine this analytic estimate with the number and sizes of
the collision patterns at each rank.  The resulting bounds grow only
polynomially in $k$ per unit of rank.  Together with the factor $D^{-r}$,
they form a geometric series when $D\ge C_0k^{16}$, proving the claimed
estimate.

\subsubsection{Completing the Multiplication-Error Estimate}

\begin{proof}[Proof of \cref{lem:multiplication-error}]
Average the four-term decomposition above over $h$ and
add and subtract two copies of $\Lambda_{0,0}(A,B,A,B)$.  The result is
\begin{align*}
&\mathbb E_{h,r}
\norm{\Phi_{h,r}(AB)-\Phi_{h,r}(A)\Phi_{h,r}(B)}_2^2
\notag\\
&\quad=
\Bigl[
\Lambda_{0,0}(A,B,A,B)-\mathbb E_h\Lambda_{0,h}(A,B,A,B)
\Bigr]
\notag\\
&\qquad
+\Bigl[
\Lambda_{0,0}(A,B,A,B)-\mathbb E_h\Lambda_{h,0}(A,B,A,B)
\Bigr]
\notag\\
&\qquad
+\Bigl[
\mathbb E_h\Lambda_{h,h}(A,B,A,B)-\Lambda_{0,0}(A,B,A,B)
\Bigr].
\end{align*}
Apply \cref{lem:folded-product-estimate} to the three parentheses and add
the resulting bounds.  Since the four inputs are $A,B,A,B$, the triangle
inequality yields
\[
\mathbb E_{h,r}
\norm{\Phi_{h,r}(AB)-\Phi_{h,r}(A)\Phi_{h,r}(B)}_2^2
\le
\frac{3C_1k^{16}}D\norm A_4^2\norm B_4^2,
\]
which proves the first assertion with $C_{\mathrm{mul}}=3C_1$.
Finally set $A=Q^\dagger$ and
$B=Q$.  The identities
$\Phi_{h,r}(Q^\dagger)=\Phi_{h,r}(Q)^\dagger$ and
$\norm{Q^\dagger}_4=\norm Q_4$ prove the second assertion.
\end{proof}


\subsubsection{Normalization Error}
Use the input assumptions and notation of \cref{thm:folded-block-output}.
Fix a question $x\in\mathcal X$ for Alice and suppress the superscript $x$
in $Q^x$ and its entries, writing $Q=\sum_a\ket a\otimes Q_a$.
Write $\Phi=\Phi_{h,r}$; all expectations in this subsection are over the
random seed $(\mathbf{h},\bfr)$, with $x$ fixed.  Since $\Phi(\id)=\id$,
\begin{equation*}
\Phi(Q)^\dagger\Phi(Q)-\id
=\Phi(Q^\dagger Q-\id)
+\sum_a\bigl(\Phi(Q_a)^\dagger\Phi(Q_a)-\Phi(Q_a^\dagger Q_a)\bigr).
\end{equation*}
The triangle inequality and Cauchy--Schwarz give
\begin{align*}
&\norm{\sum_a\bigl(\Phi(Q_a)^\dagger\Phi(Q_a)
-\Phi(Q_a^\dagger Q_a)\bigr)}_2^2\\
&\qquad\le t\sum_a\norm{\Phi(Q_a)^\dagger\Phi(Q_a)
-\Phi(Q_a^\dagger Q_a)}_2^2,
\end{align*}
since there are at most $t$ outcomes.  Together with
\cref{prop:folded-norm-expectation,lem:multiplication-error}, this implies
\begin{align*}
\mathbb E\norm{\Phi(Q)^\dagger\Phi(Q)-\id}_2^2
&\le2\norm{Q^\dagger Q-\id}_2^2\\
&\quad+2t\sum_a\mathbb E
\norm{\Phi(Q_a)^\dagger\Phi(Q_a)-\Phi(Q_a^\dagger Q_a)}_2^2\\
&\le2\norm{Q^\dagger Q-\id}_2^2
+\frac{2C_{\mathrm{mul}}tk^{16}}D\sum_a\norm{Q_a}_4^4.
\end{align*}
Since each $Q_a^\dagger Q_a\geq0$,
\begin{equation*}
\sum_a\norm{Q_a}_4^4
=\sum_a\tau_{2^n}\!\left((Q_a^\dagger Q_a)^2\right)
\le\tau_{2^n}\!\left(\left(\sum_aQ_a^\dagger Q_a\right)^2\right)
=\norm{Q^\dagger Q}_2^2.
\end{equation*}
The assumed bound $\norm{Q^\dagger Q-\id}_2\le\delta$
implies $\norm{Q^\dagger Q}_2\le1+\delta$, so
\begin{equation*}
\grayhighlight{
\begin{aligned}
\mathbb E\norm{\Phi(Q)^\dagger\Phi(Q)-\id}_2^2
&\le 2\delta^2
+\frac{2C_{\mathrm{mul}}tk^{16}}D
(1+\delta)^2\\
&\le 2\delta^2+\frac{8C_{\mathrm{mul}}tk^{16}}D.
\end{aligned}}
\end{equation*}
The same estimate holds for every Bob block $R^y$, $y\in\mathcal Y$.

\subsubsection{Game Value Error}

\begin{lemma}
\label{lem:folded-correlation-error}
Let $\psi_{\bll}$ be a canonical state, let $n,k,D\ge1$ be integers,
and let $X,Y\in \Matrix_2^{\otimes n}$ have Pauli degree at most $2k$.
Choose $h\colon[n]\to[D]$ uniformly from all functions and
$r\in G^n$ uniformly, independently, and let $\Phi_{h,r}$ be the
randomized map in \eqref{eq:hashed-pauli-map}.  Then
\begin{equation*}
\abs{\mathbb E_{h,r}\Corr_{\psi,D}
\bigl(\Phi_{h,r}(X),\Phi_{h,r}(Y)\bigr)-\Corr_{\psi,n}(X,Y)}
\le\frac{4k^2}D\norm X_2\norm Y_2.
\end{equation*}
\end{lemma}

\begin{proof}
Write $\Phi=\Phi_{h,r}$.  Recall
\[
w_{\bll}(\alpha)=\prod_i\lambda_{\alpha_i},
\qquad
\lambda_0=1,
\qquad
|\lambda_s|\le1.
\]
For fixed $h$, orthogonality of the characters in the common random seed
$r$ gives
\begin{equation*}
\mathbb E_r\Corr_{\psi,D}(\Phi X,\Phi Y)
=
\sum_\alpha
\wh X(\alpha)\wh Y(\alpha)w_{\bll}(H\alpha).
\end{equation*}
If $h$ is injective on $\supp\alpha$, then the nonzero entries of $\alpha$
are sent to distinct coordinates of $H\alpha$.  Hence their multiset is
unchanged and
\(
w_{\bll}(H\alpha)=w_{\bll}(\alpha).
\)

For $|\alpha|\le2k$,
\[
\Pr_h[h\text{ is noninjective on }\supp\alpha]
\le
\frac{\binom{|\alpha|}{2}}D
\le
\frac{2k^2}D.
\]
On the collision event both weights have absolute value at most one, so
\[
|w_{\bll}(H\alpha)-w_{\bll}(\alpha)|\le2.
\]
Since the weight difference vanishes when $h$ is injective on
$\supp\alpha$, for every $|\alpha|\le2k$ we have
\begin{align*}
\abs{\mathbb E_h w_{\bll}(H\alpha)-w_{\bll}(\alpha)}
&\le\mathbb E_h\abs{w_{\bll}(H\alpha)-w_{\bll}(\alpha)}\\
&\le2\Pr_h[h\text{ is noninjective on }\supp\alpha]
\le\frac{4k^2}D.
\end{align*}
Subtract the Pauli expansion of $\Corr_{\psi,n}(X,Y)$ from the
expected correlation above.  Since $X$ and $Y$ have degree at most $2k$,
Cauchy--Schwarz gives
\begin{align*}
&\abs{\mathbb E_{h,r}\Corr_{\psi,D}(\Phi X,\Phi Y)
-\Corr_{\psi,n}(X,Y)}\\
&\quad=\abs{\sum_{|\alpha|\le2k}
\wh X(\alpha)\wh Y(\alpha)
\bigl(\mathbb E_h w_{\bll}(H\alpha)-w_{\bll}(\alpha)\bigr)}\\
&\quad\le\frac{4k^2}D\sum_{|\alpha|\le2k}
|\wh X(\alpha)|\,|\wh Y(\alpha)|\\
&\quad\le\frac{4k^2}D
\left(\sum_\alpha|\wh X(\alpha)|^2\right)^{1/2}
\left(\sum_\alpha|\wh Y(\alpha)|^2\right)^{1/2}\\
&\quad=\frac{4k^2}D\norm X_2\norm Y_2.
\end{align*}
The final equality is Pauli
Parseval from \cref{prop:scalar-pauli-identities}.
\end{proof}

Fix questions $x,y$ and suppress the question superscripts on $Q^x$,
$R^y$, and the predicate matrix $V^{x,y}$.  Since every entry of $V$ is either
$0$ or $1$,
\[
\norm{V}_{\mathrm{op}}
\le\left(\sum_{a,b}|V_{ab}|^2\right)^{1/2}
\le\sqrt{\tA\tB}\le t.
\]
Using $Q^\dagger Q=\sum_aQ_a^\dagger Q_a$ and cyclicity of the trace,
\begin{align*}
\sum_a\norm{Q_a^\dagger Q_a}_2^2
&=\norm{Q^\dagger Q}_2^2
-\sum_{a\ne a'}\tau_{2^n}\!\left(Q_a^\dagger Q_aQ_{a'}^\dagger Q_{a'}\right)\\
&=\norm{Q^\dagger Q}_2^2
-\sum_{a\ne a'}\norm{Q_aQ_{a'}^\dagger}_2^2\\
&\le\norm{Q^\dagger Q}_2^2\\
&\le\bigl(1+\norm{Q^\dagger Q-\id}_2\bigr)^2
\le4.
\end{align*}
The last line uses the triangle inequality and the assumed bound.  The same argument
for Bob gives
\[
\sum_b\norm{R_b^\dagger R_b}_2^2
\le\norm{R^\dagger R}_2^2
\le\bigl(1+\norm{R^\dagger R-\id}_2\bigr)^2
\le4.
\]

\begin{samepage}
For each seed $\omega=(h,r)$ and each pair of outcomes $a,b$, bilinearity
of $\Corr_{\psi,D}$ gives the decomposition
\begin{align*}
&\Corr_{\psi,D}\!\left(
\Phi_\omega(Q_a)^\dagger\Phi_\omega(Q_a),
\Phi_\omega(R_b)^\dagger\Phi_\omega(R_b)\right)-\Corr_{\psi,n}(Q_a^\dagger Q_a,R_b^\dagger R_b)\\[2pt]
&\quad=\Corr_{\psi,D}\!\left(
\Phi_\omega(Q_a^\dagger Q_a),\Phi_\omega(R_b^\dagger R_b)\right)
-\Corr_{\psi,n}(Q_a^\dagger Q_a,R_b^\dagger R_b)
&&\text{(i)}\\[6pt]
&\qquad+\Corr_{\psi,D}\!\left(
\Phi_\omega(Q_a)^\dagger\Phi_\omega(Q_a)-\Phi_\omega(Q_a^\dagger Q_a),
\Phi_\omega(R_b^\dagger R_b)\right) &&\text{(ii)}\\[6pt]
&\qquad+\Corr_{\psi,D}\!\left(
\Phi_\omega(Q_a^\dagger Q_a),
\Phi_\omega(R_b)^\dagger\Phi_\omega(R_b)-\Phi_\omega(R_b^\dagger R_b)
\right) &&\text{(iii)}\\[6pt]
&\qquad+\Corr_{\psi,D}\!\left(\begin{aligned}
&\Phi_\omega(Q_a)^\dagger\Phi_\omega(Q_a)-\Phi_\omega(Q_a^\dagger Q_a),
\Phi_\omega(R_b)^\dagger\Phi_\omega(R_b)-\Phi_\omega(R_b^\dagger R_b)
\end{aligned}\right) &&\text{(iv)}
\end{align*}
\end{samepage}

By \cref{lem:multiplication-error},
\begin{align*}
\mathbb E_\omega\sum_a
\norm{\Phi_\omega(Q_a)^\dagger\Phi_\omega(Q_a)
-\Phi_\omega(Q_a^\dagger Q_a)}_2^2
&\le\frac{4C_{\mathrm{mul}}k^{16}}D,\\
\mathbb E_\omega\sum_b
\norm{\Phi_\omega(R_b)^\dagger\Phi_\omega(R_b)
-\Phi_\omega(R_b^\dagger R_b)}_2^2
&\le\frac{4C_{\mathrm{mul}}k^{16}}D.
\end{align*}
Applying \cref{prop:folded-norm-expectation} to each $R_b^\dagger R_b$
and averaging over $h$,
\begin{equation*}
\mathbb E_\omega\sum_b\norm{\Phi_\omega(R_b^\dagger R_b)}_2^2
=\sum_b\norm{R_b^\dagger R_b}_2^2
\le4.
\end{equation*}
The same bound holds for Alice.

\smallskip\noindent
For (i), each of $Q_a^\dagger Q_a$ and $R_b^\dagger R_b$ has degree at
most $2k$.  Applying \cref{lem:folded-correlation-error} to each pair and
using the triangle inequality, we obtain
\begin{align*}
&\left|\sum_{a,b}V_{ab}\left(
\mathbb E_{h,r}\Corr_{\psi,D}\!\left(
\Phi_{h,r}(Q_a^\dagger Q_a),\Phi_{h,r}(R_b^\dagger R_b)\right)
-\Corr_{\psi,n}(Q_a^\dagger Q_a,R_b^\dagger R_b)
\right)\right|\\
&\quad\le\frac{4k^2}D\sum_{a,b}V_{ab}
\norm{Q_a^\dagger Q_a}_2\norm{R_b^\dagger R_b}_2\\
&\quad\le\frac{4k^2}D\norm{V}_{\mathrm{op}}
\left(\sum_a\norm{Q_a^\dagger Q_a}_2^2\right)^{1/2}
\left(\sum_b\norm{R_b^\dagger R_b}_2^2\right)^{1/2}\\
&\quad\le\frac{16tk^2}D.
\end{align*}

\smallskip\noindent
For (ii), apply Cauchy--Schwarz, together with
\cref{prop:corr-contraction} and $\norm{V}_{\mathrm{op}}\le t$:
\begin{align*}
&\mathbb E_\omega\left|\sum_{a,b}V_{ab}\Corr_{\psi,D}\!\left(
\Phi_\omega(Q_a)^\dagger\Phi_\omega(Q_a)-\Phi_\omega(Q_a^\dagger Q_a),
\Phi_\omega(R_b^\dagger R_b)\right)\right|\\
&\quad\le\norm{V}_{\mathrm{op}}
\left(\mathbb E_\omega\sum_a
\norm{\Phi_\omega(Q_a)^\dagger\Phi_\omega(Q_a)
-\Phi_\omega(Q_a^\dagger Q_a)}_2^2\right)^{1/2}\\
&\qquad\times
\left(\mathbb E_\omega\sum_b
\norm{\Phi_\omega(R_b^\dagger R_b)}_2^2\right)^{1/2}
\le\frac{4\sqrt{C_{\mathrm{mul}}}\,tk^8}{\sqrt D}.
\end{align*}

\smallskip\noindent
For (iii), the same argument with the two players interchanged gives
\begin{equation*}
\mathbb E_\omega\left|\sum_{a,b}V_{ab}\Corr_{\psi,D}\!\left(
\Phi_\omega(Q_a^\dagger Q_a),
\Phi_\omega(R_b)^\dagger\Phi_\omega(R_b)-\Phi_\omega(R_b^\dagger R_b)
\right)\right|
\le\frac{4\sqrt{C_{\mathrm{mul}}}\,tk^8}{\sqrt D}.
\end{equation*}

\smallskip\noindent
For (iv), Cauchy--Schwarz and \cref{prop:corr-contraction} give
\begin{align*}
&\mathbb E_\omega\left|\sum_{a,b}V_{ab}\Corr_{\psi,D}\!\left(
\Phi_\omega(Q_a)^\dagger\Phi_\omega(Q_a)-\Phi_\omega(Q_a^\dagger Q_a),
\Phi_\omega(R_b)^\dagger\Phi_\omega(R_b)-\Phi_\omega(R_b^\dagger R_b)
\right)\right|\\
&\quad\le\norm{V}_{\mathrm{op}}
\left(\mathbb E_\omega\sum_a
\norm{\Phi_\omega(Q_a)^\dagger\Phi_\omega(Q_a)
-\Phi_\omega(Q_a^\dagger Q_a)}_2^2\right)^{1/2}\\
&\qquad\times
\left(\mathbb E_\omega\sum_b
\norm{\Phi_\omega(R_b)^\dagger\Phi_\omega(R_b)
-\Phi_\omega(R_b^\dagger R_b)}_2^2\right)^{1/2}
\le\frac{4C_{\mathrm{mul}}tk^{16}}D.
\end{align*}

For each seed $\omega=(h,r)$, define
\begin{equation}
\label{eq:folded-measurement-families}
\begin{aligned}
\widetilde{\mathscr Q}_\omega
&:=\bigl(\bigl(\Phi_\omega(Q_a^x)^\dagger\Phi_\omega(Q_a^x)\bigr)_a\bigr)_x,\\
\widetilde{\mathscr R}_\omega
&:=\bigl(\bigl(\Phi_\omega(R_b^y)^\dagger\Phi_\omega(R_b^y)\bigr)_b\bigr)_y.
\end{aligned}
\end{equation}
Combining (i)--(iv) and averaging over the questions, we obtain
\begin{equation*}
\grayhighlight{
\begin{aligned}
&\abs{
\mathbb E_\omega
\Val_{\sfG}(\widetilde{\mathscr Q}_\omega,\widetilde{\mathscr R}_\omega)
-
\Val_{\sfG}(\mathscr Q,\mathscr R)
}\\
&\quad\le
4\left(
\frac{2\sqrt{C_{\mathrm{mul}}}\,tk^8}{\sqrt D}
+\frac{C_{\mathrm{mul}}tk^{16}}D
+\frac{4tk^2}D
\right).
\end{aligned}}
\end{equation*}

\subsubsection{One seed for all questions}

Let $\mathsf p_X$ and $\mathsf p_Y$ be the question marginals in
\eqref{eq:game-question-marginals}.  For a seed
$\omega=(h,r)$ put
\begin{equation}
\label{eq:folded-normalization-operators}
\wt S_x=\Phi_\omega(Q^x)^\dagger\Phi_\omega(Q^x),
\qquad
\wt T_y=\Phi_\omega(R^y)^\dagger\Phi_\omega(R^y),
\end{equation}
and, using \eqref{eq:folded-measurement-families}, define
\begin{equation}
\label{eq:seed-errors-and-value}
\begin{aligned}
N_A(\omega)
&=\mathbb E_{x\sim\mathsf p_X}
\norm{\wt S_x-\id}_2^2,\\
N_B(\omega)
&=\mathbb E_{y\sim\mathsf p_Y}
\norm{\wt T_y-\id}_2^2,\\
V(\omega)
&=\Val_{\sfG}(\widetilde{\mathscr Q}_\omega,\widetilde{\mathscr R}_\omega).
\end{aligned}
\end{equation}

Then we have the following
\begin{lemma}
\label{lem:seed-value-bound}
For the game
$\sfG=(\mathcal X,\mathcal Y,\mathcal A,\mathcal B,\mathsf p_{\sfG},V)$
and the quantities $N_A(\omega),N_B(\omega),V(\omega)$ in \eqref{eq:seed-errors-and-value},
every seed $\omega=(h,r)$ satisfies
\begin{equation*}
0\le V(\omega)
\le
\bigl(1+\sqrt{N_A(\omega)}\bigr)
\bigl(1+\sqrt{N_B(\omega)}\bigr).
\end{equation*}
\end{lemma}

\begin{proof}
The operators $\Phi_\omega(Q_a^x)^\dagger\Phi_\omega(Q_a^x)$ and
$\Phi_\omega(R_b^y)^\dagger\Phi_\omega(R_b^y)$ are positive, so
$V(\omega)\ge0$.
For fixed $(x,y)$, adding the nonnegative contributions from rejected answer
pairs bounds the conditional value by
\[
\Corr_{\psi,D}(\wt S_x,\wt T_y)
\le \norm{\wt S_x}_2\norm{\wt T_y}_2.
\]
Here the inequality follows from \cref{prop:corr-contraction}.
After averaging with respect to $\mathsf p_{\sfG}(x,y)$, Cauchy--Schwarz
separates the two question variables:
\[
V(\omega)
\le
\left(\mathbb E_x\norm{\wt S_x}_2^2\right)^{1/2}
\left(\mathbb E_y\norm{\wt T_y}_2^2\right)^{1/2}.
\]
Minkowski's inequality yields
\[
\left(\mathbb E_x\norm{\wt S_x}_2^2\right)^{1/2}
\le 1+\sqrt{N_A(\omega)},
\]
and the analogous Bob bound proves the claimed pointwise estimate.
\end{proof}

\begin{lemma}
\label{lem:common-seed}
Let $(\Omega,\mathbb P)$ be a probability space, let
$V,N_A,N_B\colon\Omega\to[0,\infty)$ be measurable functions, and let
$v_0\in\mathbb R$ and $\alpha,\Delta\geq0$.  Suppose
\begin{equation*}
V\le(1+\sqrt{N_A})(1+\sqrt{N_B}),
\qquad
\mathbb E V\ge v_0-\Delta,
\qquad
\mathbb E N_A,\mathbb E N_B\le\alpha.
\end{equation*}
If $\beta>2\alpha$, then there is $\omega_0\in\Omega$ satisfying
\(
N_A(\omega_0),N_B(\omega_0)\le\beta
\)
and
\begin{equation*}
V(\omega_0)
\ge
v_0-\Delta
-4\left(
\frac{\alpha}{\beta}
+\frac{\alpha}{\sqrt\beta}
+\alpha
\right).
\end{equation*}
\end{lemma}

\begin{proof}
Let
\(
\cG=\{N_A\le\beta,\ N_B\le\beta\}.
\)
By Markov's inequality,
\(
\displaystyle\Pr(\cG^c)\le\frac{2\alpha}{\beta}.
\)

Since $\cG^c\subseteq\{N_A>\beta\}\cup\{N_B>\beta\}$,
\begin{align*}
\mathbb E[\sqrt{N_A}\one_{\cG^c}]
&\le\mathbb E[\sqrt{N_A}\one_{\{N_A>\beta\}}]
+\mathbb E[\sqrt{N_A}\one_{\{N_B>\beta\}}]\\
&\le\frac{\mathbb E N_A}{\sqrt\beta}
+\sqrt{\mathbb E N_A}\sqrt{\Pr(N_B>\beta)}\\
&\le\frac{2\alpha}{\sqrt\beta}.
\end{align*}
The second inequality uses $\sqrt{N_A}\le N_A/\sqrt\beta$ on
$\{N_A>\beta\}$ and Cauchy--Schwarz; the last uses
$\mathbb E N_A\le\alpha$ and Markov's inequality for $N_B$.
Interchanging $A$ and $B$ gives
$\mathbb E[\sqrt{N_B}\one_{\cG^c}]\le2\alpha/\sqrt\beta$.
Also, by Cauchy--Schwarz,
\(
\mathbb E\sqrt{N_AN_B}\le\sqrt{\mathbb E N_A\mathbb E N_B}\le\alpha
\).
Substituting these estimates into the pointwise upper bound for $V$, we obtain
\begin{align*}
\mathbb E[V\one_{\cG^c}]
&\le
\Pr(\cG^c)
+\mathbb E[\sqrt{N_A}\one_{\cG^c}]
+\mathbb E[\sqrt{N_B}\one_{\cG^c}]
+\mathbb E\sqrt{N_AN_B}\\
&\le\frac{2\alpha}{\beta}
 +\frac{4\alpha}{\sqrt\beta}+\alpha\\
&\le
4\left(
\frac{\alpha}{\beta}
+\frac{\alpha}{\sqrt\beta}
+\alpha
\right).
\end{align*}
Consequently,
\[
\mathbb E[V\one_{\cG}]
\ge
v_0-\Delta
-4\left(
\frac{\alpha}{\beta}
+\frac{\alpha}{\sqrt\beta}
+\alpha
\right).
\]
Because $\beta>2\alpha$, Markov's inequality also ensures
$\Pr(\cG)>0$.  If the displayed lower bound is positive, the conditional
average of $V$ on $\cG$ is at least $\mathbb E[V\one_{\cG}]$, so some seed
in $\cG$ attains the asserted lower bound.  If the lower bound is
nonpositive, the claim follows from $V\ge0$ for any seed in the nonempty set
$\cG$.
\end{proof}


\begin{proof}[Proof of \cref{thm:folded-block-output}]
Initially take $C\ge\max\{1,C_0\}$, where $C_0$ is the constant in
\cref{lem:multiplication-error}.  The stated lower bound on $D$ ensures
$tk^{16}/D\le\delta^2$ and
$tk^8/\sqrt D\le\delta^{1/3}$.
Averaging the normalization estimate over the questions and using the
game-value estimate, we obtain
\[
\mathbb E N_A,\mathbb E N_B\le c_N\delta^2,
\qquad
\mathbb E V\ge\Val_{\sfG}(\mathscr Q,\mathscr R)-c_V\delta^{1/3},
\]
where $c_N=2+8C_{\mathrm{mul}}$ and
$c_V=8\sqrt{C_{\mathrm{mul}}}+4C_{\mathrm{mul}}+16$ are fixed independently
of the later choice of $C$.

Apply \cref{lem:common-seed} with
$v_0=\Val_{\sfG}(\mathscr Q,\mathscr R)$ and
\[
\alpha=c_N\delta^2,
\qquad
\Delta=c_V\delta^{1/3},
\qquad
\beta=4c_N\delta^{4/3}.
\]
The pointwise hypothesis is \cref{lem:seed-value-bound}, and
$\beta>2\alpha$ because $0<\delta<1$.  Hence one seed $\omega_0$
satisfies $N_A(\omega_0),N_B(\omega_0)\le4c_N\delta^{4/3}$ and
\begin{align*}
V(\omega_0)
&\ge\Val_{\sfG}(\mathscr Q,\mathscr R)
-c_V\delta^{1/3}-\delta^{2/3}
-2\sqrt{c_N}\delta^{4/3}-4c_N\delta^2\\
&\ge\Val_{\sfG}(\mathscr Q,\mathscr R)
-\bigl(c_V+1+2\sqrt{c_N}+4c_N\bigr)\delta^{1/3}.
\end{align*}
Choose $C$ also to be at least $4c_N$ and
$c_V+1+2\sqrt{c_N}+4c_N$.  By \eqref{eq:seed-errors-and-value}, these
are the two asserted bounds for the same deterministic seed $\omega_0$.
\end{proof}

\subsection{Step 4: Rounding}

In this subsection, we prove \cref{lem:polar-soundness}, restated below.

\polarsoundness*

\begin{proof}[Proof of \cref{lem:polar-soundness}]
Let $d,t\geq1$ be integers and let
$Q:\mathbb C^d\to\mathbb C^t\otimes\mathbb C^d$ be a block operator.
Put $S=Q^\dagger Q$ and take a polar decomposition
$Q=Q^\sharp\sqrt S$.  Extend the partial isometry $Q^\sharp$ on
$\ker S$ to an isometry; this is possible because the output dimension
is at least the input dimension.  The extension still satisfies
$Q=Q^\sharp\sqrt S$ and $(Q^\sharp)^\dagger Q^\sharp=\id$.
Since $Q-Q^\sharp=Q^\sharp(\sqrt S-\id)$, applying the inequality
$(\sqrt s-1)^4\leq(s-1)^2$ for $s\geq0$ to each eigenvalue of $S$ gives
\begin{equation*}
\norm{Q-Q^\sharp}_4^4
=\tau_d((\sqrt S-\id)^4)
\le\tau_d((S-\id)^2)
=\norm{S-\id}_2^2.
\end{equation*}

Apply this construction to every $Q^x$ and $R^y$, with $d=2^n$, and
denote the resulting isometries by $(Q^x)^\sharp$ and $(R^y)^\sharp$.
Define Alice's and Bob's POVM families by
\begin{align*}
A_a^{x,\sharp}
&=((Q^x)^\sharp)^\dagger
  (|a\rangle\langle a|\otimes\id)(Q^x)^\sharp,\\
B_b^{y,\sharp}
&=((R^y)^\sharp)^\dagger
  (|b\rangle\langle b|\otimes\id)(R^y)^\sharp.
\end{align*}
The isometry identities imply that these operators are positive and sum
to $\id$ for each question.  Thus they define POVM strategies
$\mathscr A^\sharp=((A_a^{x,\sharp})_{a\in\mathcal A})_{x\in\mathcal X}$
and
$\mathscr B^\sharp=((B_b^{y,\sharp})_{b\in\mathcal B})_{y\in\mathcal Y}$.

The polar-decomposition estimate and the hypotheses of the lemma imply
\[
\mathbb E_{x\sim\mathsf p_X}\norm{Q^x-(Q^x)^\sharp}_4^4\le\nu^4,
\qquad
\mathbb E_{y\sim\mathsf p_Y}\norm{R^y-(R^y)^\sharp}_4^4\le\nu^4.
\]
For each question pair $(x,y)$, write
$Q^x=(Q^x)^\sharp+\bigl(Q^x-(Q^x)^\sharp\bigr)$ and
$R^y=(R^y)^\sharp+\bigl(R^y-(R^y)^\sharp\bigr)$, and expand
$\cG_{x,y}^{\psi^{\otimes n}}(Q^x,Q^x;R^y,R^y)$.
The function in \eqref{eq:block-correlation-functional} is additive
in each of its four arguments,
so the expansion has $2^4=16$ terms.  The term containing only
isometries satisfies
\[
\mathbb E_{(x,y)\sim\mathsf p_{\sfG}}
\cG_{x,y}^{\psi^{\otimes n}}\bigl((Q^x)^\sharp,(Q^x)^\sharp;
              (R^y)^\sharp,(R^y)^\sharp\bigr)
=\Val_{\sfG}(\mathscr A^\sharp,\mathscr B^\sharp).
\]
Each of the other $15$ terms contains at least one difference.
For example, since the isometries have normalized Schatten $4$-norm
one, \cref{lem:game-expression-holder} and H\"older's inequality give
\begin{align*}
&\mathbb E_{(x,y)\sim\mathsf p_{\sfG}}
 \abs{\cG_{x,y}^{\psi^{\otimes n}}\bigl(
 Q^x-(Q^x)^\sharp,(Q^x)^\sharp;
 R^y-(R^y)^\sharp,(R^y)^\sharp\bigr)}\\
&\qquad\leq
 \mathbb E_{(x,y)\sim\mathsf p_{\sfG}}
 \norm{Q^x-(Q^x)^\sharp}_4\norm{R^y-(R^y)^\sharp}_4\\
&\qquad\leq
 \left(\mathbb E_{x\sim\mathsf p_X}
 \norm{Q^x-(Q^x)^\sharp}_4^4\right)^{1/4}
 \left(\mathbb E_{y\sim\mathsf p_Y}
 \norm{R^y-(R^y)^\sharp}_4^4\right)^{1/4}
 \leq\nu^2.
\end{align*}
The same argument applies to every other term: H\"older's inequality
with exponent $4$ for each of the four factors contributes at most $\nu$
for each difference and $1$ for each isometry.  Thus a term containing
$j$ differences has expected absolute value at most $\nu^j$.
There are $\binom4j$ such terms for each $j=1,2,3,4$, so the triangle
inequality yields
\begin{align*}
\abs{\Val_{\sfG}(\mathscr Q,\mathscr R)
      -\Val_{\sfG}(\mathscr A^\sharp,\mathscr B^\sharp)}
&\leq4\nu+6\nu^2+4\nu^3+\nu^4\\
&\leq15\nu.
\end{align*}
The last inequality uses $0\leq\nu\leq1$ and proves the claim.
\end{proof}

\subsection{Succinct Certificates}
\label{subsec:succinct-certificates}

We prove \cref{prop:succinct-verifier}, restated below.

\succinctverifier*

Fix a language in
$\MIP^{\psi}[\poly,\poly]$, a corresponding uniform
two-prover protocol, and an input $z\in\{0,1\}^N$ of bit length $N$.  On this
input, the protocol verifier induces a finite game
$\sfG=(\mathcal X,\mathcal Y,\mathcal A,\mathcal B,\mathsf p_{\sfG},V)$.
Let $R(N)$ be the number of its random bits, let $q$ be the number of question
pairs having positive probability, and put
$t:=\max\{|\mathcal A|,|\mathcal B|,2\}$.
The protocol verifier runs in polynomial time and uses polynomial-length
answers, so $R(N)\leq\poly(N)$, $q\leq2^{R(N)}$, and
$\log t\leq\poly(N)$.
Let $\Delta_0>0$ be a constant
lower bound on the protocol's completeness--soundness gap.  Use
\cref{prop:pauli-canonicalization} to fix the canonical form
$\psi_{\boldsymbol\lambda}$ of $\psi$ and its parameters
$(\lambda_s)_{s\in\mathbb F_2^2}$.  For an approximation parameter
$0<\varepsilon<1/10$, to be chosen in the final argument, let $k$ and $D$ be
the corresponding parameters in \cref{thm:finite-copy-approximation}.

We call the original randomized machine the \emph{protocol verifier} and the
$\NEXP$ machine constructed below the \emph{certificate verifier}.  Instead
of writing the POVMs as full matrices, the certificate retains the low-degree
entries $Q_a^x$ and $R_b^y$ of the block approximants.  Their associated
positive operators are
$(Q_a^x)^\dagger Q_a^x$ on Alice's side and
$(R_b^y)^\dagger R_b^y$ on Bob's side.  Thus
positivity is automatic.  All quantities needed by the certificate verifier can be
computed from the Pauli coefficients of these block entries,
without constructing a $2^D\times2^D$ matrix.

\subsubsection{Certificate Format}

For every Alice question $x$ with positive marginal probability and every
answer $a$, the witness contains a degree-$k$ operator
\begin{equation*}
  Q_a^x=\sum_{|\alpha|\leq k}q_{a,\alpha}^xP_\alpha,
\end{equation*}
and contains analogous operators $R_b^y$ for Bob.  The coefficients are
rational complex numbers.  These entries form the block operators
$Q^x=\sum_a|a\rangle\otimes Q_a^x$ and
$R^y=\sum_b|b\rangle\otimes R_b^y$, with approximate strategies
$\mathscr Q,\mathscr R$ defined in \eqref{eq:block-operator-families}.
The number of Pauli strings
of degree at most $k$
on $D$ qubits is
\begin{equation}
\label{eq:low-degree-pauli-count}
  M(D,k):=\sum_{j=0}^k3^j\binom Dj
  \leq(3eD)^k.
\end{equation}
For each question, put
\begin{equation}
\label{eq:certificate-normalization-operators}
  S_x:=(Q^x)^\dagger Q^x
  =\sum_a(Q_a^x)^\dagger Q_a^x,
  \qquad
  T_y:=(R^y)^\dagger R^y.
\end{equation}
Pauli multiplication determines all coefficients of $S_x$ and $T_y$ by
finite convolution.  Parseval's identity expresses the first normalization
test as
\begin{equation*}
  \norm{S_x-\id}_2^2
  =\sum_\gamma
    \abs{\wh S_x(\gamma)-\one_{\{\gamma=0\}}}^2.
\end{equation*}

The diagonal Pauli form of the canonical state similarly expresses the
game value as
\begin{align*}
  \Val_{\sfG}^{\psi_{\bll}^{\otimes D}}(\mathscr Q,\mathscr R)
  ={}&\sum_{x,y,a,b}\mathsf p_{\sfG}(x,y)V(a,b\mid x,y)
  \sum_\gamma w_{\bll}(\gamma)
  \notag\\
  &\quad\cdot
  \wh{\bigl((Q_a^x)^\dagger Q_a^x\bigr)}(\gamma)
  \wh{\bigl((R_b^y)^\dagger R_b^y\bigr)}(\gamma),
\end{align*}
where, by \eqref{eq:canonical-correlation-weights},
\begin{equation*}
  w_{\bll}(\gamma):=\prod_{i=1}^D\lambda_{\gamma_i}.
\end{equation*}
Every associated positive
operator occurring here has Pauli degree at most $2k$.  The normalization
tests and the value computation therefore require at most
\begin{equation*}
  \grayhighlight{\poly(q,t)D^{O(k)}}
\end{equation*}
elementary arithmetic operations.  In particular, neither test requires the
  certificate verifier to diagonalize an exponentially large matrix.
  Positivity is already guaranteed by the representation
  $(Q_a^x)^\dagger Q_a^x$.

\subsubsection{Finite Precision}

The protocol verifier uses $R(N)\leq\poly(N)$ random bits.  Consequently, every
question with nonzero marginal probability satisfies
\begin{equation*}
  \mathsf p_X(x),\mathsf p_Y(y)\geq2^{-R(N)}.
\end{equation*}
This elementary observation supplies the norm bound for each question needed for
uniform rational approximation.

\begin{lemma}
  \label{lem:questionwise-block-norm-bound}
  Let $D,t\geq1$, and let
  $\mathsf p_X$ be a probability distribution on a finite question set
  $\mathcal X$.  For each $x\in\mathcal X$, let
  $Q^x\colon(\mathbb C^2)^{\otimes D}\to
    \mathbb C^t\otimes(\mathbb C^2)^{\otimes D}$ be a block
  operator, with the Schatten norm normalized by the input trace
  $\tau_{2^D}$.  Suppose
  \begin{equation*}
    \mathbb E_{x\sim\mathsf p_X}
      \norm{(Q^x)^\dagger Q^x-\id}_2^2\leq1.
  \end{equation*}
  If $\mathsf p_X(x)>0$, then
  \begin{equation*}
    \norm{Q^x}_4
    \leq
    \left(2+\frac{2}{\mathsf p_X(x)}\right)^{1/4}.
  \end{equation*}
  The analogous statement holds for Bob.
\end{lemma}

\begin{proof}
  Set $S_x=(Q^x)^\dagger Q^x$.  Since the Schatten norms are normalized,
  $\norm{\id}_2=1$, and hence
  \begin{equation*}
    \norm{S_x}_2^2
    \leq2\norm{S_x-\id}_2^2+2.
  \end{equation*}
  On the other hand, the hypothesis of
  \cref{lem:questionwise-block-norm-bound} implies
  \begin{equation*}
    \mathsf p_X(x)\norm{S_x-\id}_2^2\leq1.
  \end{equation*}
  Therefore
  \begin{equation*}
    \norm{Q^x}_4^4=\norm{S_x}_2^2
    \leq2+\frac{2}{\mathsf p_X(x)},
  \end{equation*}
  which proves the claim.
\end{proof}

\begin{lemma}
  \label{lem:rational-block-approximation}
  Fix a game
  $\sfG=(\mathcal X,\mathcal Y,\mathcal A,\mathcal B,\mathsf p_{\sfG},V)$,
  integers $D,k\geq1$, and $D$ copies of a canonical state
  $\psi_{\boldsymbol\lambda}$.  Let
  $(Q^x)_{x\in\mathcal X}$ and
  $(R^y)_{y\in\mathcal Y}$ be collections of block approximants
  on $D$ qubits whose
  entries have Pauli degree at most $k$, and let $\mathscr Q,\mathscr R$
  be the families of measurement elements in \eqref{eq:block-operator-families}.
  Let $L_0\geq1$ be the total number of complex Pauli coefficients of
  the entries $Q_a^x,R_b^y$, and suppose every block operator has normalized
  Schatten $4$-norm at most $B\geq1$.  For every $0<\delta<1$, all coefficients can
  be replaced by rational complex numbers using
  \begin{equation*}
    O\!\left(\log B+\log L_0+\log\frac1\delta\right)
  \end{equation*}
  bits per real and imaginary part, so that both averaged squared
  normalization errors, averaged with the marginals $\mathsf p_X$ and
  $\mathsf p_Y$, and the value
  $\Val_{\sfG}^{\psi_{\boldsymbol\lambda}^{\otimes D}}(\mathscr Q,\mathscr R)$
  change by at most $\delta$.
\end{lemma}

\begin{proof}
  Round every real and imaginary coefficient to a mesh of width $\Delta$.
  For each component of a block operator,
  \begin{equation*}
    \abs{\wh Q_a^x(\alpha)}
    \leq\norm{Q_a^x}_2
    \leq\norm{Q_a^x}_4
    \leq\norm{Q^x}_4
    \leq B.
  \end{equation*}
  Here the penultimate inequality follows from positivity and
  \begin{equation*}
    \tau\!\left(\bigl((Q_a^x)^\dagger Q_a^x\bigr)^2\right)
    \leq
    \tau\!\left(\bigl((Q^x)^\dagger Q^x\bigr)^2\right).
  \end{equation*}
  Thus $O(\log B)$ bits suffice for the integer part.

  Let $E^x$ denote the perturbation of one complete block operator.  It has at
  most $tM(D,k)\leq L_0$ nonzero coefficients.  Since every Pauli string is
  unitary,
  \begin{equation*}
    \norm{E^x}_4
    \leq\norm{E^x}_\infty
    \leq\sum_{a,\alpha}\abs{\wh E_a^x(\alpha)}
    \leq C L_0\Delta.
  \end{equation*}
  Write $e=CL_0\Delta$.  For the normalization map
  $S(Q)=Q^\dagger Q$, H\"older's inequality provides the perturbation bound
  \begin{equation*}
    \norm{S(Q+E)-S(Q)}_2
    \leq(2\norm Q_4+\norm E_4)\norm E_4
    \leq(2B+e)e.
  \end{equation*}
  Moreover, $\norm{S(Q)-\id}_2\leq B^2+1$.  Combining this with the preceding
  perturbation bound shows, whenever $e\leq1$, that
  \begin{equation*}
    \abs{
      \norm{S(Q+E)-\id}_2^2-
      \norm{S(Q)-\id}_2^2
    }
    \leq C(B+1)^3e.
  \end{equation*}

  Replace the Pauli coefficients of every $Q_a^x$ and $R_b^y$ by their
  rational approximations, form the corresponding measurement elements,
  and expand the value formula \eqref{eq:block-game-value}. Applying
  \cref{lem:game-expression-holder} to each term yields
  \begin{equation*}
    \abs{\Delta\Val_{\sfG}}\leq C(B+1)^3e.
  \end{equation*}
  Because all estimates are subsequently averaged with respect to the
  question distribution, there is no multiplicative dependence on the number
  of questions.
  Taking
  \begin{equation*}
    \Delta=\frac{\delta}{C(B+1)^3L_0}
  \end{equation*}
  makes both changes at most $\delta$.  Solving this choice for the required
  precision proves the bit bound asserted in
  \cref{lem:rational-block-approximation}.
\end{proof}

For the operators supplied by \cref{thm:folded-block-output} in the
completeness proof below, the two averaged normalization errors
are smaller than one.  By \cref{lem:questionwise-block-norm-bound} and the
lower bound on nonzero question probabilities above, one may therefore take
\begin{equation*}
  B\leq2^{O(R(N))}.
\end{equation*}
Furthermore,
\begin{equation*}
  L_0\leq\poly(q,t)M(D,k).
\end{equation*}
The bounds on $R(N),q,t$ above, the count $M(D,k)$ of Pauli strings, and the
parameter bounds in \cref{thm:finite-copy-approximation} show that the precision
promised by \cref{lem:rational-block-approximation} is polynomial in $N$ for
every constant $\delta$.

The correlation weights in the certificate-value formula are determined by
the known canonical parameters.  Every relevant $\gamma$ has degree at most
$2k$, and
\begin{equation*}
  w_{\bll}(\gamma)
  =\prod_{s\in\mathbb F_2^2\setminus\{0\}}
    \lambda_s^{N_s(\gamma)},
  \qquad
  N_s(\gamma):=\abs{\{i:\gamma_i=s\}}\leq2k.
\end{equation*}

Under our computability assumption, the canonical parameters $\lambda_s$
can be approximated to error $2^{-b}$ in time $\poly(b)$, yielding error
at most $2k\,2^{-b}$ in each weight.
Since the value formula has at most $2^{\poly(N)}$ terms with coefficients
of magnitude at most $2^{\poly(N)}$, choosing $b=\poly(N)$ keeps the total
numerical error within the reserved budget and the running time
within $2^{\poly(N)}$. 

We reserve a fixed fraction of the completeness--soundness gap for numerical
error.  The local unitaries from \cref{prop:pauli-canonicalization} do not occur in the
certificate: in the completeness and soundness arguments they are absorbed
into the POVMs.

\subsubsection{Completeness, soundness, and running time}

\begin{proof}[Proof of \cref{prop:succinct-verifier}]
  Choose the approximation parameter $\varepsilon$ fixed above sufficiently
  small that the
  total dimension-reduction, rational-approximation, rounding, and numerical
  errors are at most $\Delta_0/4$.

  For completeness, start from a near-optimal strategy and follow the proof
  of \cref{thm:finite-copy-approximation}.  Immediately before the final
  rounding, retain the degree-$k$ operators supplied by
  \cref{thm:folded-block-output} with $\delta=\varepsilon^3$.
  Their value is within the reserved
  error of the original value, and their averaged squared normalization
  errors are $O(\varepsilon^4)$.  The norm bound for each question and
  \cref{lem:rational-block-approximation} replace them by rational operators
  of polynomial coefficient precision without exhausting the reserved gap.

  For soundness, choose the normalization-test threshold in the form
  $\nu_{\mathrm{cert}}^4$, where $0<\nu_{\mathrm{cert}}\leq1$ is a fixed
  constant small enough that the rounding loss is within the reserved part
  of the gap.  Suppose a rational certificate passes the two averaged
  normalization tests and the value test.  Apply
  \cref{lem:polar-soundness} in the analysis: the certified operators have polar
  parts defining POVM families for Alice and Bob.  The value of the
  resulting strategy differs from the accepted value by at most
  $15\nu_{\mathrm{cert}}$, which is within the reserved error by the choice
  of $\nu_{\mathrm{cert}}$.
  Thus a certificate
  accepted above the midpoint of the original gap cannot exist on a no
  instance.  Notice that the polar decomposition is not performed by the
  certificate verifier; it is used only to prove soundness.

  Finally, the certificate verifier enumerates the protocol verifier's
  $2^{R(N)}$ random strings to recover
  the nonzero question pairs and their probabilities.  This enumeration
  takes $2^{\poly(N)}$ time.  It then reads the Pauli coefficients and
  performs the normalization and value tests described above.
  The count of Pauli strings and arithmetic bound
  established there, together with
  \cref{lem:rational-block-approximation}, bound the number of coefficients
  and the number of arithmetic operations in these tests by
  \begin{equation*}
    \poly(q,t)(3eD)^{O(k)},
  \end{equation*}
  with $\poly(N)$ bits per real and imaginary coefficient.
  Since $\log q,\log t\leq\poly(N)$ and, for the fixed state $\psi$ and
  constant $\varepsilon$, \cref{thm:finite-copy-approximation} gives
  $k\leq\poly(\log t)$ and $D\leq\poly(t)$, we have
  \begin{equation*}
    \log\!\left(\poly(q,t)(3eD)^{O(k)}\right)\leq\poly(N).
  \end{equation*}
  Thus both the certificate length and the total running time, including
  enumeration and finite-precision arithmetic, are at most $2^{\poly(N)}$.
\end{proof}

\section{\texorpdfstring{The Case $\rho_{\max}=1$}
  {The Case rho-max = 1}}
\label{sec:maximal-correlation-one}

We prove \cref{prop:boundary-canonical,thm:boundary-re-lower-bound,lem:fixed-resource-re-upper-bound}
using the notation fixed in their statements.  The proof has two branches.
After canonicalization, the separable branch is settled by
\cref{prop:separable-resource-classical}.  In the entangled branch,
\cref{lem:boundary-hamming-partition,lem:boundary-local-extraction} supply
the local extraction, and a bilateral Pauli twirl completes
\cref{thm:boundary-epr-component}.  Its simulation consequence is then
combined with constant-answer EPR protocols and parallel repetition to
prove the lower bound.

\subsection{Canonical Form at \texorpdfstring{$\rho_{\max}=1$}{rho-max = 1}}

\begin{proof}[Proof of \cref{prop:boundary-canonical}]
  Apply \cref{prop:pauli-canonicalization}.  Since
  $\rho_{\max}(\psi)=1$, one of the three nonidentity canonical coefficients
  has modulus one.  Conjugating locally by suitable Clifford unitaries
  permutes the three nonidentity Pauli matrices and changes their signs, so
  we may move this coefficient to the $\pauliZ$ coordinate and make it positive.
  Thus it is enough
  to consider
  \begin{equation*}
    \frac14\left(
      \id\otimes\id+\lambda_1\pauliX\otimes\pauliX^{\mathsf T}
      +\lambda_2\pauliY\otimes\pauliY^{\mathsf T}
      +\pauliZ\otimes\pauliZ^{\mathsf T}
    \right).
  \end{equation*}
  The Bell basis diagonalizes this operator.  In the order
  $\Phi^+,\Phi^-,\Psi^+,\Psi^-$, its four eigenvalues are
  \begin{equation*}
    \frac{2+\lambda_1+\lambda_2}{4},\qquad
    \frac{2-\lambda_1-\lambda_2}{4},\qquad
    \frac{\lambda_1-\lambda_2}{4},\qquad
    \frac{-\lambda_1+\lambda_2}{4}.
  \end{equation*}
  Positivity of the last two eigenvalues forces
  $\lambda_1=\lambda_2=:r$.  Conjugating one subsystem by $\pauliZ$ changes the
  signs of both $\pauliX$ and $\pauliY$ coefficients while leaving the
  $\pauliZ$ coefficient
  fixed, so we may take $\lambda:=|r|\in[0,1]$.  Written in the computational
  and Pauli bases, the two surviving Bell components are precisely the two
  displayed forms of $\Psi_\lambda$ preceding
  \cref{prop:boundary-canonical}.

  Finally, the eigenvalues of the partial transpose
  $\Psi_\lambda^{\mathsf T_B}$ are
  \begin{equation*}
    \frac12,\qquad \frac12,\qquad \frac\lambda2,
    \qquad -\frac\lambda2.
  \end{equation*}
  Hence $\lambda>0$ implies entanglement, whereas
  $\Psi_0=\tfrac12(\ketbra{00}+\ketbra{11})$ is separable.
\end{proof}

\subsection{The case \texorpdfstring{$\psi$}{psi} is separable}

We prove \cref{prop:separable-resource-classical} in this subsection, restated below.

\separableresourceclassical*

\begin{proof}
  For each finite $n$ one may write
  \begin{equation*}
    \psi^{\otimes n}=\sum_z p_z\,\psi_z^A\otimes\psi_z^B.
  \end{equation*}
  For arbitrary local POVMs, the induced correlation is therefore
  \begin{equation*}
    p(a,b\mid x,y)
    =\sum_zp_z\Tr(A_a^x\psi_z^A)\Tr((B_b^y)^{\mathsf T}\psi_z^B).
  \end{equation*}
  This is a local correlation with shared random variable $z$, and hence it
  cannot exceed the optimal deterministic classical value of the game.
  Conversely, the players may ignore $\psi$ and use any classical
  strategy.  The complexity statement follows from the classical identity
  \begin{equation*}
    \MIP[\poly,O(1)]=\MIP[\poly,\poly]=\NEXP=\MIP,
  \end{equation*}
  obtained from the scaled-up PCP construction
  \cite{Mie09,FL92,BFL91}.
\end{proof}

\subsection{The case \texorpdfstring{$\psi$}{psi} is Entangled}

\subsubsection{A Perfect Hamming Partition}

Fix $0<\lambda\leq1$.  For $x,y\in\mathbb F_2^n$, let
$d_H(x,y)$ be their Hamming distance.  Expanding the $n$-th tensor power of
the computational-basis formula for $\Psi_\lambda$ gives
\begin{equation*}
  \Psi_\lambda^{\otimes n}
  =2^{-n}\sum_{x,y\in\mathbb F_2^n}
    \lambda^{d_H(x,y)}\ket{x,x}\!\bra{y,y}.
\end{equation*}

Fix $m\geq1$, put $d:=2^m$ and $n:=d-1$, and let $H$ be the
$m\times n$ binary matrix whose columns are all the nonzero vectors of
$\mathbb F_2^m$.  Define the Hamming code
$\mathcal C_m:=\ker H\subseteq\mathbb F_2^n$.  Let
$e_1,\ldots,e_n$ denote the standard basis vectors and set $e_0:=0$.

\begin{lemma}[Hamming partition]
  \label{lem:boundary-hamming-partition}
  Every $x\in\mathbb F_2^n$ has a unique representation
  \begin{equation*}
    x=c+e_i,
    \qquad c\in\mathcal C_m,
    \qquad i\in\{0,1,\ldots,n\}.
  \end{equation*}
  Moreover, $|\mathcal C_m|=2^{n-m}$ and
  $|\mathcal C_m|d=2^n$.
\end{lemma}

\begin{proof}
  If $Hx=0$, take $i=0$ and $c=x$.  Otherwise, $Hx$ equals a unique nonzero
  column $He_i$, and $c:=x+e_i$ belongs to $\mathcal C_m$.  The syndromes
  $He_0,He_1,\ldots,He_n$ are pairwise distinct, which also proves
  uniqueness.  Finally, $H$ has rank $m$, so
  $|\mathcal C_m|=2^{n-m}$.
\end{proof}

The bijection in \cref{lem:boundary-hamming-partition} defines a local
computational-basis unitary
\begin{equation}
\label{eq:hamming-extraction-unitary}
  U_m\ket{c+e_i}=\ket c_{\mathrm{code}}\ket i_{\mathrm{log}}.
\end{equation}

\begin{lemma}[Local extraction]
  \label{lem:boundary-local-extraction}
  If Alice and Bob apply $U_m$ from \eqref{eq:hamming-extraction-unitary}
  to their respective $n$ qubits and discard
  the code registers, the state of the two $d$-dimensional logical registers
  is
  \begin{equation*}
    \psi_{d,\lambda}
    =\frac1d\sum_{i,j=0}^{d-1}
      \lambda^{d_H(e_i,e_j)}\ket{i,i}\!\bra{j,j}.
  \end{equation*}
  The extraction is deterministic and uses no postselection.
\end{lemma}

\begin{proof}
  Substitute $x=c+e_i$ and $y=c'+e_j$ into
  the tensor expansion above.  Tracing out the two code registers
  removes exactly the terms with $c\neq c'$.  For the remaining terms,
  Hamming distance is invariant under the common translation by $c$, so
  \begin{equation*}
    d_H(c+e_i,c+e_j)=d_H(e_i,e_j),
  \end{equation*}
  and the common scalar is
  $|\mathcal C_m|2^{-n}=2^{-m}=1/d$.  This proves
  \cref{lem:boundary-local-extraction}.
\end{proof}

For the state in \cref{lem:boundary-local-extraction}, the only distances
that occur are
\begin{equation*}
  d_H(e_i,e_j)=
  \begin{cases}
    0,&i=j,\\
    1,&\text{exactly one of $i,j$ is zero},\\
    2,&i,j\neq0\text{ and }i\neq j.
  \end{cases}
\end{equation*}
Thus $\lambda^{d_H(e_i,e_j)}\geq\lambda^2$ for all $i,j$, independently of $d$.

\subsubsection{Bilateral Pauli Twirling}

For every $m\geq1$, set
\begin{equation}
\label{eq:epr-extraction-parameters}
  d:=2^m,
  \qquad
  N_m:=d-1+2m,
\end{equation}
then we can define
\begin{equation}
\label{eq:maximally-entangled-target}
  \ket{\Phi_d}
  :=
  \frac1{\sqrt d}\sum_{i=0}^{d-1}\ket{i,i},
  \qquad
  \Phi_d:=\ketbra{\Phi_d}.
\end{equation}
Under the identification
$\mathbb C^d\cong(\mathbb C^2)^{\otimes m}$, the state $\Phi_d$ is
precisely a tensor product of $m$ EPR pairs.

Identify each $d$-dimensional logical register with $m$ qubits, and let
$\mathcal P_m$ be the $m$-qubit Pauli operators modulo global phase.  Define
the bilateral Pauli twirl
\begin{equation}
\label{eq:bilateral-pauli-twirl}
  \mathcal W_d(\rho)
  :=\frac1{d^2}\sum_{P\in\mathcal P_m}
    (P\otimes\overline P)\rho(P\otimes\overline P)^\dagger,
\end{equation}
where the bar denotes entrywise conjugation in the computational basis.
This channel is diagonal in the generalized Bell basis and preserves the coefficient of $\Phi_d$. In
particular, $P\otimes\overline P$ fixes $\ket{\Phi_d}$ and the twirl removes
all matrix elements between distinct generalized Bell states.  Hence, for
every state $\rho$,
\begin{equation*}
  \mathcal W_d(\rho)=F\Phi_d+(1-F)\zeta,
  \qquad F=\bra{\Phi_d}\rho\ket{\Phi_d},
\end{equation*}
for a state $\zeta$ supported on the orthogonal complement of
$\ket{\Phi_d}$.

Each extra copy of $\Psi_\lambda$, when measured on both sides in the
computational basis, produces one shared unbiased bit.  Therefore
$2m$ extra copies provide the shared seed needed to sample uniformly from
the $d^2=2^{2m}$ Paulis in \eqref{eq:bilateral-pauli-twirl}.

\begin{theorem}[Dimension-independent EPR weight]
  \label{thm:boundary-epr-component}
  Fix $0<\lambda\leq1$.  For every $m\geq1$, there exist local quantum
  channels and a bipartite state $\zeta_{m,\lambda}$ such that
  \begin{equation*}
    \Psi_\lambda^{\otimes N_m}
    \longmapsto
    F_{m,\lambda}\Phi_d
    +(1-F_{m,\lambda})\zeta_{m,\lambda},
    \qquad d=2^m,N_m=2^m-1+2m
  \end{equation*}
  where
  \begin{equation}
  \label{eq:epr-component-weight}
     F_{m,\lambda}
  =
  \frac{
    d+2(d-1)\lambda+(d-1)(d-2)\lambda^2
  }{d^2}\geq\lambda^2.
  \end{equation}
\end{theorem}

Although the target dimension $d=2^m$ grows with $m$, the weight of the
exact maximally entangled component is bounded below by the fixed
quantity $\lambda^2$.  Thus every entangled boundary state
$\Psi_\lambda$ contains, in arbitrary finite dimension, an exact EPR
component whose weight does not deteriorate with the number of EPR
pairs.

\begin{proof}[Proof of \cref{thm:boundary-epr-component}]
  Apply \cref{lem:boundary-local-extraction} to the first $d-1$ copies and
  use the remaining $2m$ copies to implement the bilateral twirl.  The
  extracted state satisfies
  \begin{equation*}
    \bra{\Phi_d}\psi_{d,\lambda}\ket{\Phi_d}
    =\frac1{d^2}\sum_{i,j=0}^{d-1}
      \lambda^{d_H(e_i,e_j)}.
  \end{equation*}
  There are $d$ ordered pairs at distance zero, $2(d-1)$ at distance one,
  and $(d-1)(d-2)$ at distance two.  Substitution of these three counts into
  the preceding overlap sum gives $F_{m,\lambda}$ in
  \eqref{eq:epr-component-weight}.  Moreover,
  \begin{equation*}
    F_{m,\lambda}-\lambda^2
    =\frac{(1-\lambda)\bigl(d(1+3\lambda)-2\lambda\bigr)}{d^2}
    \geq0.
  \end{equation*}
  The resource conversion asserted in \cref{thm:boundary-epr-component} now
  follows from the twirl decomposition above.
\end{proof}


\subsubsection{Simulation and Amplification}

\begin{cor}[One-sided EPR simulation]
  \label{cor:boundary-epr-simulation}
  Suppose a strategy for a game
  $\sfG=(\mathcal X,\mathcal Y,\mathcal A,\mathcal B,\mathsf p_{\sfG},V)$
  uses $m$ EPR pairs and is accepted with
  probability $p$.  Then a strategy using $N_m$ copies of
  $\Psi_\lambda$ is accepted with probability at least
  \begin{equation*}
    F_{m,\lambda}p\geq\lambda^2p.
  \end{equation*}
\end{cor}

\begin{proof}
  First apply the local conversion in
  \cref{thm:boundary-epr-component}, and then perform the original EPR
  measurements.  The resulting acceptance operator has the form
  \begin{equation*}
    M=\sum_{x,y,a,b}\mathsf p_{\sfG}(x,y)V(a,b\mid x,y)
      A_a^x\otimes(B_b^y)^{\mathsf T}\geq0.
  \end{equation*}
  Consequently,
  \begin{align*}
    \Tr\!\left[M\left(
      F_{m,\lambda}\Phi_d+(1-F_{m,\lambda})\zeta_{m,\lambda}
    \right)\right]
    &\geq F_{m,\lambda}\Tr(M\Phi_d)\notag\\
    &=F_{m,\lambda}p.
  \end{align*}
\end{proof}

We need constant-answer EPR protocols with an arbitrarily prescribed
constant soundness.

\begin{theorem}[Constant-answer EPR protocols]
  \label{thm:constant-answer-epr-re}
  Let $L\in\RE$ and fix $\delta>0$.  There exist a constant
  $a_\delta\in\mathbb N$ and a uniform family of nonlocal games
  $H_x=(\mathcal X_x,\mathcal Y_x,\mathcal A_x,\mathcal B_x,\mathsf p_{H_x},V_x)$,
  generated by polynomial-time classical verifiers, with polynomial
  question length and answer length at most $a_\delta$ such that the
  following hold.
  \begin{enumerate}
    \item If $x\in L$, then $\omega^*(H_x)=1$.
    \item If $x\notin L$, then $\omega^*(H_x)\leq\delta$.
  \end{enumerate}
  On yes-instances, strategies using finitely many EPR pairs can achieve
  acceptance probabilities arbitrarily close to one.
\end{theorem}

The starting point is the constant-answer result of \cite{DongEtAl24}, which supplies a fixed soundness gap rather than a
prescribed soundness bound.

\begin{theorem}[{\cite[Theorems~3 and~36]{DongEtAl24}}]
  \label{thm:dfn-constant-answer-epr}
  For every $L\in\RE$, there are constants $a_0\in\mathbb N$ and
  $0<\gamma<1$ and a uniform family of entangled
  games
  $\sfG_x=(\mathcal X_x,\mathcal Y_x,\mathcal A_x,\mathcal B_x,\mathsf p_{\sfG_x},V_x)$,
  generated by polynomial-time classical verifiers, with polynomial
  question length and answer length at most $a_0$, such that the
  following hold.
  \begin{enumerate}
    \item If $x\in L$, then $\omega^*(\sfG_x)=1$.
    \item If $x\notin L$, then $\omega^*(\sfG_x)\leq1-\gamma$.
  \end{enumerate}
  On yes-instances, strategies using finitely many EPR pairs can achieve
  acceptance probabilities arbitrarily close to one.
\end{theorem}

\begin{proof}[Proof of \cref{thm:constant-answer-epr-re}]
  It suffices to prove the claim for $0<\delta<1$.
  Fix the games $(\sfG_x)_x$ and constants $a_0,\gamma$ from
  \cref{thm:dfn-constant-answer-epr}.  After increasing $a_0$ if necessary,
  use fixed-length $a_0$-bit answer encodings.  Write
  $\mathcal A_x,\mathcal B_x$ for the answer sets of $\sfG_x$.
  
  Set $S:=\max\{2a_0\log2,1\}$, so that
  $\max\{\log|\mathcal A_x\times\mathcal B_x|,1\}\leq S$ for every $x$.
  Let $C>0$ be the constant in \cref{thm:entangled-parallel-repetition}.
  Choose a fixed integer $r\geq2$ large enough that
  \begin{equation*}
    \frac{C S\log r}{\gamma^{17}r^{1/4}}\leq\delta,
  \end{equation*}
  which is possible because $\log r/r^{1/4}\to0$, and define
  $H_x:=\sfG_x^{\otimes r}$.
  For $x\notin L$, we have $\omega^*(\sfG_x)\leq1-\gamma$, so
  \cref{thm:entangled-parallel-repetition} implies
  \begin{equation*}
    \omega^*(H_x)
    \leq\frac{C S\log r}{\gamma^{17}r^{1/4}}
    \leq\delta.
  \end{equation*}
  This upper bound allows arbitrary joint entangled strategies across
  all $r$ coordinates.

  Let $x\in L$ and $0<\nu<1$.  Some strategy using a finite number $m$
  of EPR pairs succeeds in $\sfG_x$ with probability at least $1-\nu/r$.
  Running this strategy independently on
  $r$ blocks of $m$ EPR pairs succeeds in $H_x$ with probability at least
  $(1-\nu/r)^r\geq1-\nu$.  The shared resource still consists of finitely
  many EPR pairs, and $\nu$ is arbitrary, so the acceptance
  probabilities of finite EPR pairs  for $H_x$ have supremum $1$.

  The integer $r$ depends only on $a_0,\gamma,\delta$, not on the input
  length or the question-alphabet sizes.  Hard-code it in the new verifier.
  Each question is a tuple of $r$ original questions, and each answer
  contains $a_\delta:=ra_0$ bits.  Sampling and checking $r$ instances
  takes polynomial time, the question length remains polynomial, and
  the answer length remains constant.  All questions are sent together
  and all answers are returned together, so the protocol still has two
  provers and one round.
\end{proof}

Now we can prove \cref{thm:boundary-re-lower-bound}, restated below.

\boundaryrelowerbound*

\begin{proof}
  Fix $L\in\RE$.  Choose and hard-code a rational number
  $0<\underline\lambda\leq\lambda$, and set
  \begin{equation*}
    q:=\frac{\underline\lambda^2}{2},
    \qquad
    k:=\left\lceil\frac{\log3}{q}\right\rceil,
    \qquad
    \delta:=\frac1{6k}.
  \end{equation*}
  Let $(H_x)_x$ be the family from
  \cref{thm:constant-answer-epr-re} with this value of $\delta$, and define
  $K_x:=\operatorname{OR}_k(H_x)$.

  Suppose first that $x\in L$.  Since the finite-EPR value of $H_x$ has
  supremum one, some strategy using a finite number $m=m(x)$ of EPR pairs is
  accepted with probability at least $1/2$.  By
  \cref{cor:boundary-epr-simulation}, a block of copies of
  $\Psi_\lambda$ simulates this strategy with acceptance probability at
  least $\lambda^2/2\geq q$.  Independence of the resource blocks used for
  the $k$ coordinates implies
  \begin{equation*}
    \Pr[K_x\text{ accepts}]
    \geq1-(1-q)^k
    \geq1-e^{-qk}
    \geq\frac23.
  \end{equation*}

  If $x\notin L$, the upper bound in \cref{lem:or-repetition-bounds} yields
  \begin{equation*}
    \omega^*(K_x)\leq k\omega^*(H_x)
    \leq k\delta=\frac16.
  \end{equation*}
  This soundness bound holds against all finite-dimensional entangled
  strategies and therefore against the fixed-resource strategies.  Since
  $k$ and $a_\delta$ depend only on the fixed resource and fixed gap, the
  answer length remains constant.
\end{proof}

\subsection{The Upper Bound}

\fixedresourcereupperbound*

\begin{proof}
  Consider a protocol with completeness $c$, soundness $s<c$, and a rational
  threshold $s<\theta<c$.  On input $x$, dovetail over all copy numbers $n$,
  all Alice and Bob POVMs in a computable dense family, and increasing
  numerical precision.  For example, such a dense family is obtained by
  enumerating rational matrices $R_1,\ldots,R_t$, putting
  $S:=\sum_aR_a^\dagger R_a$, and, whenever $S$ is positive definite, setting
  \begin{equation*}
    A_a:=S^{-1/2}R_a^\dagger R_aS^{-1/2}.
  \end{equation*}
  Enumerate the Bob POVMs in the same way.  For each enumerated strategy,
  evaluate its acceptance probability using computable approximations
to the matrix entries of $\psi$
  and the verifier's decision predicate, enclosing the result in rational
  intervals of decreasing width.  Halt as soon as a certified lower endpoint
  exceeds $\theta$.  On a yes-instance, density and completeness guarantee
  that this eventually happens; on a no-instance, soundness guarantees that
  it never happens.  The language is therefore recursively enumerable.
\end{proof}

\section*{AI disclosure statement}
We heavily used OpenAI's GPT-5.6 and GPT-6 Astra throughout the construction, including optimizing the parameters, exploring more efficient dimension reduction, and verifying the correctness. The authors have verified the correctness and originality of all content, including references and take fully responsibility of the manuscript.

\bibliographystyle{alpha}
\bibliography{references}

\appendix
\section{\texorpdfstring{Proof of \cref{lem:complete-sdpi}}
  {Proof of the Complete SDPI Lemma}}
\label{app:complete-sdpi}

We restate \cref{lem:complete-sdpi} before proving it.
\completesdpi*

We use the following form of Gao and Rouz\'e's theorem.
\begin{theorem}[{\cite[Theorem~1.2]{GR22}}]
\label{thm:gao-rouze-complete-sdpi}
Let $\Phi\colon\Matrix_d\to\Matrix_d$ be a quantum channel with a
full-rank invariant state $\sigma=\Phi(\sigma)$, and suppose that $\Phi$
is GNS-symmetric with respect to $\sigma$.  Define
\(
  E_*(X):=\lim_{r\to\infty}\Phi^{2r}(X),
  X\in\Matrix_d.
\)
There is a constant
$0\leq c<1$ such that, for every $m\geq1$ and every density matrix
$\rho\in\Matrix_d\otimes\Matrix_m$,
\[
  \RelEnt\!\left(
    (\Phi\otimes\operatorname{id}_{\Matrix_m})(\rho)
    \,\middle\|\,
    ((\Phi\circ E_*)\otimes\operatorname{id}_{\Matrix_m})(\rho)
  \right)
  \leq c\,
  \RelEnt\!\left(
    \rho
    \,\middle\|\,
    (E_*\otimes\operatorname{id}_{\Matrix_m})(\rho)
  \right).
\]
\end{theorem}

\begin{proof}[Proof of \cref{lem:complete-sdpi}]
Let $\K:=\K_{\psi_{\boldsymbol\lambda}}$.  By
\cref{prop:correlation-channel-choi-properties}, $\K$ is completely
positive, unital, and trace preserving.  Thus $\K(\id/2)=\id/2$, so
$\id/2$ is a full-rank invariant state.  By
\cref{prop:pauli-canonicalization}, $\K$ is real diagonal in the
orthogonal Pauli basis, and hence
\[
  \tau_2\Br{X^\dagger\K(Y)}
  =\tau_2\Br{\K(X)^\dagger Y},
  \qquad X,Y\in\Matrix_2.
\]
This is the GNS-symmetry condition with respect to $\id/2$.

Moreover, $\max_{s\neq0}|\lambda_s|=\rho_{\max}(\psi)<1$, so for every
$X\in\Matrix_2$,
\[
  \K^{2r}(X)
  =\tau_2\Br{X}\id
    +\sum_{s\neq0}\lambda_s^{2r}\widehat X(s)P_s
  \xrightarrow[r\to\infty]{}\tau_2\Br{X}\id=\Pi(X).
\]
Thus $E_*=\Pi$ in \cref{thm:gao-rouze-complete-sdpi}.
Since $\K\circ\Pi=\Pi$, that theorem gives the asserted inequality
with the system and ancilla tensor factors exchanged.  The constant
$0\leq\alpha_\psi<1$ depends only on the fixed channel $\K$, hence only
on $\psi$, and is independent of the ancilla dimension~$m$.
\end{proof}

\clearpage

\section{\texorpdfstring{Proof of \cref{lem:schatten-boolean-decay}}
  {Proof of the Boolean Heat Decay Lemma}}
\label{app:schatten-boolean-decay}

\subsection{Fourier conventions on finite groups and the Boolean cube}

Fourier analysis on a finite abelian group expands functions in the
character basis, while Pauli analysis expands operators in the Pauli
basis.
For vector-valued Fourier analysis on the Boolean cube, including
normalized Bochner norms and high-degree spaces, we follow
\cite[Section~5]{MN14}.  The matrix-valued setting with normalized Schatten
norms appears in \cite[Section~1.1]{BARW08}.

\begin{definition}[Characters and dual group]
  \label{def:characters-and-dual-group}
  Let $K$ be a finite abelian group.  A \emph{character} of $K$ is a group
  homomorphism $\chi\colon K\to\U(1)$.  The set of all characters,
  equipped with pointwise multiplication, is the \emph{dual group}
  \(
    \widehat K:=\operatorname{Hom}(K,\U(1)).
  \)
\end{definition}

Throughout this subsection, $\mathbb E_{x\in K}$ denotes the normalized
average $|K|^{-1}\sum_{x\in K}$.

\begin{fact}[{\cite[Theorems~1.60 and~1.62]{Milne25}}]
  \label{fact:character-orthogonality}
  Let $K$ be a finite abelian group.
  \leavevmode
  \begin{enumerate}
    \item The constant function $\mathbf1(t)=1$ is the trivial character.
    \item If $2t=0$ for every $t\in K$, then every character takes values in
      $\{\pm1\}$.
    \item The characters form an orthonormal basis of the space of
      complex-valued functions on $K$:
      \begin{equation*}
        \mathbb E_{x\in K}\overline{\chi(x)}\eta(x)
        =\one_{\{\chi=\eta\}}
        \qquad(\chi,\eta\in\widehat K).
      \end{equation*}
  \end{enumerate}
\end{fact}


\begin{definition}[$\V$-valued Fourier expansion]
  \label{def:coefficient-valued-fourier-expansion}
  Let $K$ be a finite abelian group with the uniform probability measure, let
  $\V$ be a complex vector space, and let $f\colon K\to
  \V$.  For $\chi\in\widehat K$, define
  \begin{equation*}
    \widehat f(\chi)
    :=
    \mathbb E_{x\in K}\overline{\chi(x)}f(x)
    \in \V.
  \end{equation*}
  By \cref{fact:character-orthogonality}, Fourier inversion takes the
  coefficient-valued form
  \begin{equation*}
    f(x)=\sum_{\chi\in\widehat K}\widehat f(\chi)\chi(x).
  \end{equation*}
  Whenever characters are indexed by labels, we abbreviate
  $\widehat f(\chi_\xi)$ to $\widehat f(\xi)$.
\end{definition}

\begin{definition}[Convolution]
  \label{def:finite-abelian-convolution}
  Let $K$ be a finite abelian group, let $\V$ be a complex vector space,
  and let $f:K\to \V$ and $h:K\to\mathbb C$ be functions.  Define their
  \emph{convolution} by
  \begin{equation*}
    (h*f)(x)
    :=\mathbb E_{z\in K}h(z)f(x-z)
    =\mathbb E_{y\in K}h(x-y)f(y).
  \end{equation*}
  The function $h$ is called the \emph{kernel} of this convolution.
  For a probability measure $\nu$ on $K$, define
  \begin{equation*}
    (\nu*f)(x)
    :=\sum_{z\in K}\nu(\{z\})f(x-z)
    =\mathbb E_{z\sim\nu}f(x-z).
  \end{equation*}
  These conventions agree when $h(z)=|K|\nu(\{z\})$; the probability
  weights in the second formula already sum to one.
\end{definition}

\begin{definition}
  \label{def:fourier-support-degree-and-truncation}
  Let $f:K\to \V$ be as in
  \cref{def:coefficient-valued-fourier-expansion}.  Its Fourier support is
  \begin{equation*}
    \operatorname{supp}_{\mathrm F}(f)
    :=\{\chi\in\widehat K:\widehat f(\chi)\neq0\}.
  \end{equation*}
  To define degree and truncation, choose a function
  $w:\widehat K\to\mathbb Z_{\geq0}$ with $w(\mathbf1)=0$;
  $w(\chi)$ is the weight assigned to the character $\chi$.  Set
  \begin{equation*}
    \deg_w(f):=\max_{\chi\in\operatorname{supp}_{\mathrm F}(f)}w(\chi),
    \qquad
    f^{\leq k}:=\sum_{w(\chi)\leq k}\widehat f(\chi)\chi
    \quad(k\geq0),
  \end{equation*}
  with $\deg_w(0):=-\infty$.  These definitions depend on the choice of $w$.
  The choices for the Boolean cube and $G^n$ are specified in
  \cref{def:coefficient-valued-boolean-fourier-analysis,def:pauli-label-fourier-analysis}.
\end{definition}

\begin{definition}[Normalized Bochner norm]
  \label{def:normalized-bochner-norm}
  Let $K$ be a finite abelian group with the uniform probability measure, let
  $\V$ be a Banach space, and let $f\colon K\to \V$.  For $1\leq p<\infty$, the
  \emph{normalized Bochner $p$-norm} is
  \begin{equation*}
    \norm f_{L_p(K;\V)}
    :=
    \left(\mathbb E_{x\in K}\norm{f(x)}_\V^p\right)^{1/p}.
  \end{equation*}
  We write $L_p(K;\V)$ for the space of such functions with this norm.
  Once a weight $w$ has been chosen, $L_p^{>k}(K;\V)$ is its subspace
  consisting of functions for which $\widehat f(\chi)=0$ whenever
  $w(\chi)\leq k$.
\end{definition}

If $\V$ is a complex Hilbert space, the norm on $L_2(K;\V)$ is induced by
\begin{equation}
\label{eq:bochner-inner-product}
  \langle f,g\rangle_{L_2(K;\V)}
  :=\mathbb E_{x\in K}\langle f(x),g(x)\rangle_\V.
\end{equation}
For $\V=\mathbb C$, this is the usual normalized inner product on scalar
functions.
For $\V=\Matrix_d$ with the normalized Hilbert--Schmidt
inner product 
$\mathbb E_{x\in K}\tau_d(f(x)^\dagger g(x))$.

\begin{prop}
  \label{prop:coefficient-valued-fourier-identities}
  Let $K$ be a finite abelian group, let $\V$ be a complex vector space,
  and let $f,g:K\to \V$.  Then
  \begin{enumerate}
    \item \emph{(Parseval)}
      If $\V$ is a Hilbert space, then
      \begin{equation*}
        \norm f_{L_2(K;\V)}^2
        =\sum_{\chi\in\widehat K}\norm{\widehat f(\chi)}_\V^2.
      \end{equation*}
    \item \emph{(Inner product)}
      If $\V$ is a Hilbert space, then
      \begin{equation*}
        \langle f,g\rangle_{L_2(K;\V)}
        =\sum_{\chi\in\widehat K}
          \langle\widehat f(\chi),\widehat g(\chi)\rangle_\V.
      \end{equation*}
    \item \emph{(Constant coefficient)}
      \(\mathbb E_{x\in K}f(x)=\widehat f(\mathbf1).\)
      If $\chi_0=\mathbf1$, this coefficient is written $\widehat f(0)$.
  \end{enumerate}
\end{prop}

\begin{proof}
  The constant-coefficient identity follows
  directly from $\mathbf1(x)=1$ and the definition of the Fourier
  coefficients.  If $\V$ is a Hilbert space, expanding $f$ and $g$ in the
  character basis and using \cref{fact:character-orthogonality} proves the
  inner-product identity; setting $g=f$ then proves Parseval.
\end{proof}


\paragraph{The Boolean cube.}
\begin{definition}[Walsh characters]
  \label{def:binary-characters}
  For $N\geq1$ and $\xi,x\in\mathbb F_2^N$, define
  \begin{equation*}
    \xi\cdot x:=\sum_{j=1}^N\xi_jx_j\in\mathbb F_2,
    \qquad
    \chi_\xi(x):=(-1)^{\xi\cdot x}.
  \end{equation*}
  These are called the \emph{Walsh characters}.  We identify
  $\widehat{\mathbb F_2^N}$ with $\mathbb F_2^N$ through
  $\xi\mapsto\chi_\xi$; under this identification, $\chi_0=\mathbf1$ and
  $\chi_\xi\chi_\eta=\chi_{\xi+\eta}$.
\end{definition}

\begin{definition}
  \label{def:coefficient-valued-boolean-fourier-analysis}
  For $N\geq1$, let $\Omega_N:=\mathbb F_2^N$, with the Walsh characters
  of \cref{def:binary-characters} and the uniform probability measure.
  For $j\in[N]$, we write $\chi_{\{j\}}(x):=(-1)^{x_j}$ and use
  $\widehat f(\{j\})$ for the Fourier coefficient of $f$ at this character.
  Here we choose $w(\chi_\xi)=|\xi|$, where
  $|\xi|:=|\{j\in[N]:\xi_j=1\}|$ is the Boolean Hamming weight.
  Boolean degree, truncation, and $L_p^{>k}(\Omega_N;\V)$ use this weight.
\end{definition}

\paragraph{The Pauli label group.}
\begin{definition}
  \label{def:pauli-label-fourier-analysis}
  For $n\geq1$, identify $G^n=(\mathbb F_2^2)^n$ with
  $\mathbb F_2^{2n}$ by listing the two bits in each coordinate.
  We use the characters of \cref{def:binary-characters} under this
  identification.  Writing
  $\langle r,s\rangle:=r_1s_1+r_2s_2\in\mathbb F_2$ for $r,s\in G$,
  their coordinate form for $\rho,\alpha\in G^n$ is
  \begin{equation*}
    \chi_\rho(\alpha)
    =\prod_{i=1}^n\chi_{\rho_i}(\alpha_i)
    =(-1)^{\sum_{i=1}^n\langle\rho_i,\alpha_i\rangle}.
  \end{equation*}
  We choose $w(\chi_\alpha)=|\alpha|=|\{i\in[n]:\alpha_i\ne0\}|$
  and write $\deg_G$ for the corresponding Fourier
  degree in \cref{def:fourier-support-degree-and-truncation}.
\end{definition}


The identification $G^n\cong\Omega_{2n}$ preserves the characters and the
uniform probability measure, but not the chosen weight: if $\xi$ lists
the bits of $\alpha\in G^n$, then $|\alpha|\leq|\xi|\leq2|\alpha|$.
Thus $G$-degree and Boolean degree refer to different weights on the same
character basis.

\subsection{Proof of Lemma \ref{lem:schatten-boolean-decay}}

We restate \cref{lem:schatten-boolean-decay} before proving it.
\schattenbooleandecay*

\begin{definition}[$K$-convexity, {\cite[Section~5]{MN14}}]
Let $X$ be a Banach space.  For every integer $N\geq1$ and every
function $f:\Omega_N\to X$, define its \emph{Rademacher projection} by
\[
\operatorname{Rad}_N f
:=
\sum_{j=1}^N\widehat f(\{j\})\chi_{\{j\}}.
\]
The \emph{$K$-convexity constant} of $X$ is
\begin{equation*}
K(X)
:=
\sup_{N\ge1}
\sup_{\substack{f:\Omega_N\to X\\f\neq0}}
\frac{\norm{\operatorname{Rad}_N f}_{L_2(\Omega_N;X)}}
     {\norm{f}_{L_2(\Omega_N;X)}},
\end{equation*}
where the norm $\norm{\cdot}_{L_2(\Omega_N;X)}$ is defined in \cref{def:normalized-bochner-norm}.
We call $X$ \emph{$K$-convex} if $K(X)<\infty$.
\end{definition}

Recall that
$X_{n,d_0}=(\Matrix_{d_0}\otimes\Matrix_2^{\otimes n},\norm{\cdot}_{4/3})$.

\begin{prop}
\label{prop:schatten-k-convexity}
There is an absolute constant $K_0>1$ such that, for every $d_0,n\ge1$,
the space $X_{n,d_0}$ is $K$-convex and satisfies $K(X_{n,d_0})\le K_0$.
\end{prop}

\begin{lemma}[{\cite[Theorem~1.5]{Pisier89}}]
\label{lem:pisier-interpolation-k-convexity}
Let $(X_0,X_1)$ be an interpolation couple of complex Banach spaces,
with $X_1$ a Hilbert space.  For every $0<\theta<1$, the complex
interpolation space $[X_0,X_1]_\theta$ is $K$-convex and satisfies
\[
  K([X_0,X_1]_\theta)\le C(\theta),
\]
where $C(\theta)<\infty$ depends only on $\theta$.
\end{lemma}

\begin{proof}[Proof of \cref{prop:schatten-k-convexity}]
Set $d=d_02^n$.  For $1\le p\le2$, write
$S_p^d=(\Matrix_d,\norm{\cdot}_p)$, where
$\norm A_p=(\tau_d(|A|^p))^{1/p}$.
It is well known that (see, e.g., \cite[Eq.~(2.1)]{PX03})
\[
  [S_1^d,S_2^d]_{1/2}=S_{4/3}^d=X_{n,d_0}
\]
isometrically.
Since $S_2^d$ is Hilbert with inner product
$\langle A,B\rangle=\tau_d(A^\dagger B)$,
\cref{lem:pisier-interpolation-k-convexity} yields
\[
  K(X_{n,d_0})\le C(1/2)\le K_0:=\max\{2,C(1/2)\},
  \qquad d_0,n\ge1.
\]
\end{proof}

We now apply the following estimate of Mendel and Naor.
\begin{lemma}[{\cite[Theorem~5.1]{MN14}}]
\label{lem:mendel-naor-boolean-tail}
For every $K,p\in(1,\infty)$ there are constants $A(K,p)\in(0,1)$ and
$B(K,p),C(K,p)\in(2,\infty)$ such that, whenever $X$ is a $K$-convex
Banach space satisfying $K(X)\le K$, the following holds for every
$N,k\geq1$, every $u>0$, and every
$f\in L_p^{>k}(\Omega_N;X)$, with all Boolean-cube norms taken under
the uniform probability measure:
\begin{equation*}
\norm{\mathsf B_u f}_{L_p(\Omega_N;X)}
\le
C(K,p)\exp\!\left\{-A(K,p)k\min\{u,u^{B(K,p)}\}\right\}
\norm f_{L_p(\Omega_N;X)}.
\end{equation*}
\end{lemma}

\begin{proof}[Proof of \cref{lem:schatten-boolean-decay}]
By \cref{prop:schatten-k-convexity}, we may apply
\cref{lem:mendel-naor-boolean-tail} with $p=4/3$ and $K=K_0$.
Its constants depend only on these fixed values, so the resulting
$a_0,C_0,B$ are independent of $N,n$, and $d_0$.
\end{proof}

\section{Proof of Lemma \ref{lem:block-value-stabilitynew}}
\label{sec:appendix:block-value-stability}

\blockvaluestabilitynew*

\begin{proof}[Proof of \cref{lem:block-value-stabilitynew}]
For each $(x,y)\in\mathcal X\times\mathcal Y$, set
\begin{alignat*}{2}
S_0^x&:=A^x,\qquad
& S_1^x&:=Q^x-A^x,\\
T_0^y&:=B^y,\qquad
& T_1^y&:=R^y-B^y.
\end{alignat*}
By the normalization and approximation hypotheses,
\[
\norm{S_0^x}_4=\norm{T_0^y}_4=1,
\qquad
\norm{S_1^x}_4,\norm{T_1^y}_4\leq\eta.
\]
Recall the definition of $\cG$ in \cref{eq:block-correlation-functional}:
\begin{equation*}
  \cG_{x,y}^{\psi^{\otimes n}}(Q_1,Q_2;R_1,R_2)
  :=\sum_{a\in\mathcal A,\,b\in\mathcal B}V(a,b\mid x,y)
  \Tr\!\left[
    \psi^{\otimes n}
    \bigl(Q_{1,a}^\dagger Q_{2,a}\otimes
          (R_{1,b}^\dagger R_{2,b})^{\mathsf T}\bigr)
  \right].
\end{equation*}
Expanding the game value gives
\begin{align*}
&\Val_{\sfG}(\mathscr Q,\mathscr R)
 -\Val_{\sfG}(\mathscr A,\mathscr B)\\
&\qquad= \mathbb E_{(x,y)\sim\mathsf p_{\sfG}}\Bigl[
 \cG_{x,y}^{\psi^{\otimes n}}(Q^x,Q^x;R^y,R^y)
 -\cG_{x,y}^{\psi^{\otimes n}}(A^x,A^x;B^y,B^y)\Bigr]\\
&\qquad= \mathbb E_{(x,y)\sim\mathsf p_{\sfG}}\Bigl[
 \cG_{x,y}^{\psi^{\otimes n}}(S_0^x+S_1^x,S_0^x+S_1^x;
 T_0^y+T_1^y,T_0^y+T_1^y)\\
&\hspace{43mm}-\cG_{x,y}^{\psi^{\otimes n}}(S_0^x,S_0^x;T_0^y,T_0^y)\Bigr]\\
&\qquad=
\mathbb E_{(x,y)\sim\mathsf p_{\sfG}}
\sum_{\substack{\varepsilon\in\{0,1\}^4\\\varepsilon\neq0}}
\cG_{x,y}^{\psi^{\otimes n}}\!\left(
 S_{\varepsilon_1}^x,S_{\varepsilon_2}^x;
 T_{\varepsilon_3}^y,T_{\varepsilon_4}^y\right).
\end{align*}
For each nonzero $\varepsilon\in\{0,1\}^4$,
\cref{lem:game-expression-holder} yields
\begin{align*}
&\abs{\cG_{x,y}^{\psi^{\otimes n}}\!\left(
 S_{\varepsilon_1}^x,S_{\varepsilon_2}^x;
 T_{\varepsilon_3}^y,T_{\varepsilon_4}^y\right)}\\
&\qquad\leq
\prod_{j=1}^2\norm{S_{\varepsilon_j}^x}_4
                  \norm{T_{\varepsilon_{j+2}}^y}_4
\leq\eta^{\varepsilon_1+\varepsilon_2+\varepsilon_3+\varepsilon_4}.
\end{align*}
The triangle inequality therefore gives
\begin{align*}
\abs{\Val_{\sfG}(\mathscr Q,\mathscr R)
 -\Val_{\sfG}(\mathscr A,\mathscr B)}
&\leq\sum_{\substack{\varepsilon\in\{0,1\}^4\\\varepsilon\neq0}}
\eta^{\varepsilon_1+\varepsilon_2+\varepsilon_3+\varepsilon_4}\\
&=4\eta+6\eta^2+4\eta^3+\eta^4\\
&\leq15\eta,
\end{align*}
where the last inequality uses $0\leq\eta\leq1$.
\end{proof}

\section{\texorpdfstring{Proof of \cref{lem:folded-product-estimate}}
  {Proof of the Folded-Product Estimate}}
\label{app:folded-product-estimate}

\subsection{Pauli Cocycle}

Recall that $G=\mathbb F_2^2
    =\{(0,0),(1,0),(0,1),(1,1)\}$ in \cref{def:pauli-labels} and Pauli cocycles:
\begin{definition}[Pauli cocycle]
  For $s,t\in G$, define $\kappa(s,t)\in\{\pm1,\pm\iu\}$ satisfying
  \begin{equation*}
    P_sP_t=\kappa(s,t)P_{s+t}.
  \end{equation*}
\end{definition}

From the definition we have
\begin{prop}
  \label{prop:pauli-cocycle-identities}
  For all $r,s,t\in G$, the Pauli cocycle satisfies
  \begin{equation*}
    \kappa(0,s)=\kappa(s,0)=1,
    \qquad
    \kappa(r,s)\kappa(r+s,t)
    =\kappa(s,t)\kappa(r,s+t)
  \end{equation*}
\end{prop}

For $n\geq1$, tensorize the one-qubit phase by setting
\begin{equation}
  \kappa_n(\alpha,\beta):=\prod_{i=1}^n\kappa(\alpha_i,\beta_i),
  \qquad
  P_\alpha P_\beta=\kappa_n(\alpha,\beta)P_{\alpha+\beta}.
  \label{eq:tensor-pauli-cocycle}
\end{equation}

\begin{prop}[
  {\cite[Lemma~3.2.1]{Karpilovsky85}}]
  \label{prop:pauli-product-convolution}
  Let $A,B\in \Matrix_2^{\otimes n}$.  For every $\gamma\in G^n$, the Pauli
  coefficient of their product is
  \begin{equation*}
    \widehat{AB}(\gamma)
    =
    \sum_{\alpha+\beta=\gamma}
      \widehat A(\alpha)\widehat B(\beta)\kappa_n(\alpha,\beta).
  \end{equation*}
\end{prop}

\begin{proof}
  Expand both operators in the Pauli basis and use the tensorized
  multiplication rule above:
  \[
    AB
    =
    \sum_{\alpha,\beta\in G^n}
      \widehat A(\alpha)\widehat B(\beta)
      \kappa_n(\alpha,\beta)P_{\alpha+\beta}.
  \]
  Grouping the terms with $\alpha+\beta=\gamma$ and using uniqueness of the
  Pauli expansion proves the claimed coefficient formula.
\end{proof}

Using \cref{prop:pauli-product-convolution}, we can easily show the following:
\begin{prop}
  \label{prop:pauli-degree-subadditivity}
  For all $A,B\in \Matrix_2^{\otimes n}$,
  \(
    \deg(AB)\leq\deg(A)+\deg(B).
  \)
\end{prop}


\subsection{\mobius{} Inversion on the Partition Lattice}
In the dimension-reduction argument, we use partitions to describe which
coordinates have the same hash value.  We begin by recalling \mobius{}
inversion on the \emph{partition lattice}.  The
general incidence-algebra construction appears in
\cite[Section~3]{Rota64}.

\begin{definition}[Fibers of a hash map]
  \label{def:hash-fibers}
  Let $h\colon[n]\to[D]$.  For $j\in[D]$, the \emph{fiber} of $h$ over
  $j$ is
  \[
    B_j:=h^{-1}(j)=\{i\in[n]:h(i)=j\}.
  \]
  We write
  \begin{equation*}
    \operatorname{Fib}(h)
    :=\{B_j:j\in[D],\ B_j\ne\varnothing\}
  \end{equation*}
  for the collection of nonempty fibers.  It is the partition of $[n]$
  induced by the equivalence relation
  \[
    i\sim_h i'
    \quad\Longleftrightarrow\quad
    h(i)=h(i').
  \]
  Thus, for $C\subseteq[n]$, the restriction $h|_C$ is constant exactly
  when $C$ is contained in one block of $\operatorname{Fib}(h)$.
\end{definition}

A partition of $[n]$ is a collection of
nonempty, pairwise disjoint subsets, called \emph{blocks}, whose union is
$[n]$. We denote the set of all such partitions by $\Pi_n$. For
$\pi,\sigma\in\Pi_n$, write $\pi\leq\sigma$ if $\pi$ refines $\sigma$,
that is, if every block of $\pi$ is contained in a block of $\sigma$.
The least and greatest elements of this lattice are
\[
  \hat{0}=\bigl\{\{1\},\{2\},\ldots,\{n\}\bigr\},
  \qquad
  \hat{1}=\bigl\{[n]\bigr\},
\]
respectively.

For $\pi\leq\sigma$, the \mobius{} function of $\Pi_n$ is characterized by
\begin{equation*}
  \sum_{\pi\leq\tau\leq\sigma}\mu(\pi,\tau)
  =
  \sum_{\pi\leq\tau\leq\sigma}\mu(\tau,\sigma)
  =
  \begin{cases}
    1, & \pi=\sigma,\\
    0, & \pi<\sigma.
  \end{cases}
\end{equation*}
It has the explicit form
\begin{equation}
  \mu(\pi,\sigma)
  =\prod_{B\in\sigma}(-1)^{k_B-1}(k_B-1)!,
  \label{eq:partition-lattice-mobius}
\end{equation}
where $k_B$ is the number of blocks of $\pi$ contained in the block
$B\in\sigma$. We set $\mu(\pi,\sigma)=0$ when $\pi\nleq\sigma$. This
specialization to the partition lattice is proved in
\cite[Section~7, Corollary to Proposition~3]{Rota64}.

In this paper we will use both the upward and downward forms of \mobius{}
inversion on the partition lattice, which we now state as a lemma.

\begin{lemma}[{\cite[Section~3, Proposition~2 and Corollary~1]{Rota64}}]
  \label{lem:mobius-inversion-partition-lattice}
  Let $n\geq1$, let $V$ be a vector space, and let
  $f,g\colon\Pi_n\to V$. Then both
  of the following inversion rules hold.
  \begin{enumerate}
    \item[(i)] \emph{Downward sums:} for every $\sigma\in\Pi_n$,
    \begin{equation*}
      g(\sigma)=\sum_{\pi\leq\sigma}f(\pi)
      \quad\Longleftrightarrow\quad
      f(\sigma)=\sum_{\pi\leq\sigma}\mu(\pi,\sigma)g(\pi).
    \end{equation*}

    \item[(ii)] \emph{Upward sums:} for every $\pi\in\Pi_n$,
    \begin{equation*}
      g(\pi)=\sum_{\sigma\geq\pi}f(\sigma)
      \quad\Longleftrightarrow\quad
      f(\pi)=\sum_{\sigma\geq\pi}\mu(\pi,\sigma)g(\sigma).
    \end{equation*}
  \end{enumerate}
\end{lemma}

\begin{lemma}
  \label{lem:total-absolute-mobius-mass}
  For every integer $v\ge1$,
  \begin{equation*}
    \sum_{\sigma\in\Pi_v}\abs{\mu(\hat0,\sigma)}
    =v!.
  \end{equation*}
\end{lemma}

\begin{proof}
  By \eqref{eq:partition-lattice-mobius},
  \begin{equation}
    \abs{\mu(\hat0,\sigma)}
    =
    \prod_{B\in\sigma}(|B|-1)!.
    \label{eq:absolute-mobius-block-product}
  \end{equation}
  To prove the lemma, we give a combinatorial interpretation of the left-hand side. First partition $[v]$ into the
  supports of its disjoint cycles.  For a fixed block $B$, there are exactly
  $(|B|-1)!$ cyclic orderings of the elements of $B$, hence that many cycles
  with support $B$.  Thus, for a fixed partition $\sigma$, the product in
  \eqref{eq:absolute-mobius-block-product} counts the permutations whose
  cycle-support partition is $\sigma$.  Summing over all
  $\sigma\in\Pi_v$ counts every permutation of $[v]$ exactly once, and there
  are $v!$ such permutations.
\end{proof}

\subsection{Characters on \texorpdfstring{$\mathcal H_4$}{H4}}
\label{subsec:four-input-characters}

Recall that $G=\mathbb F_2^2$ and define
\begin{equation}
\label{eq:defofH4}
  \mathcal H_4
  =\{(t_1,t_2,t_3,t_4)\in G^4:t_1+t_2+t_3+t_4=0\}.
\end{equation}
Projection onto the first three coordinates identifies $\mathcal H_4$
with $G^3$, since $t_4=t_1+t_2+t_3$.  Under this identification, its
characters are
\[
  \chi_\nu(t_1,t_2,t_3,t_4)
  =\prod_{j=1}^3\chi_{\nu_j}(t_j),
  \qquad \nu=(\nu_1,\nu_2,\nu_3)\in G^3,
\]
where $\chi_r(s)=(-1)^{r_1s_1+r_2s_2}$ for $r,s\in G$.
Setting $\nu_4=0$ writes the same expression as
$\prod_{j=1}^4\chi_{\nu_j}(t_j)$, with one character for each input.
We use this four-input form in the Fourier sums below.


We use the Fourier convention of
\cref{def:coefficient-valued-fourier-expansion} with $\V=\mathbb C$.
For a scalar-valued function $f$ on a finite abelian group $K$, write
\begin{equation}
\label{eq:fourier-coefficient-l1-norm}
  \norm{\widehat f}_{\ell_1(\widehat K)}
  :=\sum_{\chi\in\widehat K}|\widehat f(\chi)|.
\end{equation}
This is an unnormalized sum over characters; the average defining each
Fourier coefficient remains normalized.  For $s\in K$, let
$\delta_s(t):=\mathbf1_{\{t=s\}}$.  The following scalar identities will be
used on $\mathcal H_4$ and its products.

\begin{lemma}
  \label{lem:finite-abelian-fourier-identities}
  Let $K$ be a finite abelian group and let $f,g:K\to\mathbb C$.  Then
  \begin{equation*}
    \widehat{fg}(\chi)=\sum_{\eta\in\widehat K}\widehat f(\eta)\widehat g(\eta^{-1}\chi),
    \qquad
    \norm{\widehat{fg}}_{\ell_1(\widehat K)}
    \le\norm{\widehat f}_{\ell_1(\widehat K)}
       \norm{\widehat g}_{\ell_1(\widehat K)}.
  \end{equation*}
  Moreover, for $s,t\in K$,
  \begin{equation*}
    \delta_s(t)=\frac1{|K|}\sum_{\chi\in\widehat K}\overline{\chi(s)}\chi(t).
  \end{equation*}
\end{lemma}

\begin{proof}
  Expand both factors in $f(t)g(t)$ using
  \cref{def:coefficient-valued-fourier-expansion}.  A term contributes to
  the coefficient of $\chi$ precisely when its two characters multiply to
  $\chi$.  Thus
  \[
    \widehat{fg}(\chi)
    =\sum_{\eta\in\widehat K}
      \widehat f(\eta)\widehat g(\eta^{-1}\chi).
  \]
  Taking absolute values and summing over characters, we obtain
  \begin{align*}
    \sum_{\chi\in\widehat K}|\widehat{fg}(\chi)|
    &\leq
      \sum_{\eta\in\widehat K}|\widehat f(\eta)|
      \sum_{\chi\in\widehat K}|\widehat g(\eta^{-1}\chi)|\\
    &=\norm{\widehat f}_{\ell_1(\widehat K)}
      \norm{\widehat g}_{\ell_1(\widehat K)},
  \end{align*}
  where the last equality uses that
  $\chi\mapsto\eta^{-1}\chi$ permutes $\widehat K$.
  Finally, the coefficient formula gives
  $\widehat{\delta_s}(\chi)=|K|^{-1}\overline{\chi(s)}$; its Fourier
  expansion is the claimed point-mass identity.
\end{proof}

For $K=\mathcal H_4$, every character is $\{\pm1\}$-valued and
$|\mathcal H_4|=64$.  Hence complex conjugation may be omitted from the
point-mass identity in \cref{lem:finite-abelian-fourier-identities}; if
$s\neq0$,
then
\begin{equation*}
    \frac{1}{64}\sum_{\varphi\in\widehat{\mathcal H_4}}\varphi(s) = 0.
\end{equation*}
Applying the same identity on $\mathcal H_4$ and using that its characters
are real-valued, we obtain
\begin{equation}
  \delta_s(t)=\frac1{64}\sum_{\varphi\in\widehat{\mathcal H_4}}\varphi(s)\varphi(t)
  =\frac1{64}\sum_{\varphi\in\widehat{\mathcal H_4}}\varphi(s)\bigl(\varphi(t)-1\bigr).
  \label{eq:h4-nonzero-point-mass-fourier}
\end{equation}
The last equality subtracts
$64^{-1}\sum_{\varphi\in\widehat{\mathcal H_4}}\varphi(s)$, which is zero
because $s\neq0$.

\subsection{Proof of the Folded-Product Estimate}

\lemfoldedproductestimate*

Write $\Delta_{\xi,\zeta}:=\mathbb E_h\Lambda_{\xi,\zeta}-\Lambda_{0,0}$,
with the four inputs fixed.  \Cref{fig:folded-product-proof-overview}
summarizes the proof; the other quantities are defined in the cited
equations.

\begin{figure}[htbp]
  \centering
  \begin{tikzpicture}[
      overview box/.style={draw=gray!65,fill=gray!8,rounded corners=2pt,
        align=flush center,font=\footnotesize,text width=94mm,
        inner xsep=4pt,inner ysep=4pt},
      estimate/.style={overview box,text width=44mm},
      wide/.style={text width=158mm},
      inversion/.style={draw=blue!45!gray,fill=blue!5},
      auxiliary/.style={overview box,text width=55mm,
        draw=orange!55!gray,fill=yellow!12},
      result/.style={overview box,fill=gray!15},
      arrow/.style={-{Stealth[length=1.8mm]},semithick},
      node distance=4mm
    ]
    \node[overview box,wide] (fibers) at (3.2,0)
      {\textbf{Factor the phases over the fibers of $h$},~\cref{eq:lambda-fiber-phase-expansion}\\[2pt]
       $\textstyle
       \Lambda_{\xi,\zeta}
       =\sum_{\alpha_1+\cdots+\alpha_4=0}c(\boldsymbol\alpha)
        \prod_{B\in\operatorname{Fib}(h)}
        \Theta_B^{\xi,\zeta}(\bm\alpha_B)$};

    \node[overview box,wide,inversion,below=of fibers] (first-inversion)
      {\textbf{First M\"obius inversion; take the expectation over $h$},~\cref{eq:physical-partition-expansion}\\[2pt]
       $\textstyle
       \Delta_{\xi,\zeta}
       =\sum_{\alpha_1+\cdots+\alpha_4=0}c(\boldsymbol\alpha)
       \sum_{\substack{\pi\in\operatorname{Part}(S(\boldsymbol\alpha))\\
                       \pi\ne\hat0}}
       D^{-r(\pi)}\prod_{C\in\pi}K_C(\bm\alpha_C)$\\[2pt]
       Each block must have a constant hash value; the singleton partition
       cancels.};

    \node[overview box,anchor=north] (ranks)
      at ($(first-inversion.south)+(-3.2,-0.4)$)
      {\textbf{Relabel the non-singleton blocks and group by rank},~\cref{eq:abstract-rank-decomposition,eq:rank-vertex-relations}\\[2pt]
       $\textstyle
       \Delta_{\xi,\zeta}
       =\sum_{r\ge1}D^{-r}\sum_v\frac1{v!}
        \sum_{\mathcal C\in\mathfrak P_{v,r}}\mathcal R_{\mathcal C}$\\[2pt]
       $r+1\le v\le\min\{2r,4k\}$};

    \node[overview box,below=of ranks] (characters)
      {\textbf{Expand the cumulant product using characters of $\mathcal H_4$},~\cref{eq:abstract-pattern-character-contributions}\\[2pt]
       $\textstyle \mathcal R_{\mathcal C}
       =\sum_{\boldsymbol\varphi}
        m_{\boldsymbol\varphi}\mathfrak M_{\boldsymbol\varphi}$;\\[2pt]
       the factors $\varphi_a-1$ vanish on inactive coordinates.};

    \node[estimate,anchor=north] (coefficients)
      at ($(characters.south)+(-2.5,-0.55)$)
      {\textbf{Bound the coefficients},~\cref{eq:character-coefficient-l1-bound}\\[3pt]
       $\textstyle
       \sum_{\boldsymbol\varphi}|m_{\boldsymbol\varphi}|
       \le63^v v^{2v}$};

    \node[estimate,inversion,anchor=north] (second-inversion)
      at ($(characters.south)+(2.5,-0.55)$)
      {\textbf{\mbox{Second M\"obius inversion}},\\[1pt]
       \mbox{\cref{eq:second-mobius-character-contribution,eq:block-auxiliary-polynomial-derivative}}\\[2pt]
       $\textstyle \mathfrak M_{\boldsymbol\varphi}
       =\sum_\sigma\mu(\hat0,\sigma)
        \mathfrak M_{\boldsymbol\varphi,\sigma}$\\[2pt]
       Allow repeated coordinates};

    \node[auxiliary,anchor=north west] (polynomial)
      at ($(ranks.north east)+(0.6,0)$)
      {\textbf{Auxiliary polynomial},~\cref{eq:block-auxiliary-polynomial}\\[2pt]
       $\sigma=\{B_1,\ldots,B_b\}$\\[2pt]
       $\begin{gathered}
       \textstyle F_{\sigma,\boldsymbol\chi}(w_1,\ldots,w_b)\\[1pt]
       \textstyle =\sum_{\alpha_1+\cdots+\alpha_4=0}c(\boldsymbol\alpha)\\[1pt]
       \textstyle {}\cdot\prod_{\ell=1}^b\prod_{i=1}^n
       \bigl(1-w_\ell+w_\ell\chi_{B_\ell}(\bm\alpha_i)\bigr)
       \end{gathered}$\\[3pt]
       $\deg_{w_\ell}F_{\sigma,\boldsymbol\chi}\le4k$\\[2pt]
       Pauli conjugation and Schatten H\"older give\\[2pt]
       $\textstyle \sup_{w\in[0,1]^b}
       |F_{\sigma,\boldsymbol\chi}(w)|
       \le\prod_{j=1}^4\norm{A_j}_4$};

    \node[overview box,text width=55mm,below=5mm of polynomial] (markov)
      {\textbf{Markov brothers' \mbox{inequality}},~\cref{eq:multivariate-markov-brothers}\\[2pt]
       $\begin{aligned}
       &\textstyle |\partial_{w_1}\cdots\partial_{w_b}
         F_{\sigma,\boldsymbol\chi}(0)|\\[2pt]
       &\textstyle\quad\le (2(4k)^2)^b
         \sup_{w\in[0,1]^b}|F_{\sigma,\boldsymbol\chi}(w)|\\[2pt]
       &\textstyle\quad\le (32k^2)^b
         \prod_{j=1}^4\norm{A_j}_4
       \end{aligned}$};

    \coordinate (estimate-merge) at ($(markov.south)+(-2.9,-0.4)$);
    \node[overview box,anchor=north] (distinct)
      at ($(estimate-merge)+(0,-0.4)$)
      {\textbf{Bound the sums over \mbox{distinct coordinates}},~\cref{eq:distinct-coordinate-bound}\\[2pt]
       $\textstyle\sum_\sigma|\mu(\hat0,\sigma)|=v!$,
       \quad $v\le4k$\\[2pt]
       $\textstyle
       |\mathfrak M_{\boldsymbol\varphi}|
       \le(512k^3)^v\prod_{j=1}^4\norm{A_j}_4$};

    \node[overview box,wide,anchor=north,yshift=-4mm] (pattern)
      at (fibers.center |- distinct.south)
      {\textbf{Combine the two bounds for each $\mathcal C$},~\cref{eq:fixed-abstract-pattern-bound}\\[2pt]
       $\textstyle |\mathcal R_{\mathcal C}|
       \le63^v v^{2v}(512k^3)^v\prod_{j=1}^4\norm{A_j}_4$};

    \node[result,wide,below=of pattern] (conclusion)
      {\textbf{Count the partitions and sum over $r\ge1$}\\[2pt]
       $\textstyle |\Delta_{\xi,\zeta}|
       \le\sum_{r\ge1}\left(\frac{Ck^{16}}D\right)^r
          \prod_{j=1}^4\norm{A_j}_4
       \le\frac{C_1k^{16}}D\prod_{j=1}^4\norm{A_j}_4$,
       $D\ge C_0k^{16}$};

    \draw[arrow] (fibers) -- (first-inversion);
    \draw[arrow] (ranks.north |- first-inversion.south) -- (ranks.north);
    \draw[arrow] (ranks) -- (characters);
    \draw[arrow] ($(characters.south)+(-2.5,0)$) -- (coefficients);
    \draw[arrow] ($(characters.south)+(2.5,0)$) -- (second-inversion);
    \draw[arrow] (polynomial.south) -- (markov.north);
    \draw[semithick] (second-inversion.south) |- (estimate-merge);
    \draw[semithick] (markov.south) |- (estimate-merge);
    \draw[arrow] (estimate-merge) -- (distinct.north);
    \draw[arrow] (coefficients.south) -- (coefficients.south |- pattern.north);
    \draw[arrow] (distinct.south) -- (distinct.south |- pattern.north);
    \draw[arrow] (pattern) -- (conclusion);
  \end{tikzpicture}
  \caption{Proof overview for \cref{lem:folded-product-estimate}.
  The first M\"obius inversion evaluates the average over $h$;
  the second expresses sums over distinct coordinates in terms of
  the derivatives of an auxiliary polynomial.}
  \label{fig:folded-product-proof-overview}
\end{figure}

\begin{proof}[Proof of \cref{lem:folded-product-estimate}]
Expanding the two coefficient convolutions in
\eqref{eq:four-input-form-explicit}, the defining Fourier sum is supported on
$\alpha_1+\alpha_2+\alpha_3+\alpha_4=0$, since $-\alpha=\alpha$ for
$\alpha\in G^n$.  Write
\(\boldsymbol\alpha=(\alpha_1,\alpha_2,\alpha_3,\alpha_4)\), and define the
coefficient by
\begin{equation}
c(\boldsymbol\alpha)
:=
\widehat A_1(\alpha_1)\widehat A_2(\alpha_2)
\overline{\widehat A_3(\alpha_3)\widehat A_4(\alpha_4)}
\kappa_n(\alpha_1,\alpha_2)
\overline{\kappa_n(\alpha_3,\alpha_4)}.
\label{eq:source-fourier-coefficient}
\end{equation}
Then we have
\begin{equation*}
\Lambda_{0,0}(A_1,A_2,A_3,A_4)
=
\sum_{\alpha_1+\cdots+\alpha_4=0}c(\boldsymbol\alpha).
\end{equation*}
For any element $\boldsymbol{\alpha}$ in this summation, at coordinate $i$, we have
\[
\bm{\alpha}_i
:=(\alpha_{1,i},\alpha_{2,i},\alpha_{3,i},\alpha_{4,i})
\in\mathcal H_4,
\]
where $\mathcal H_4$ is defined in
\cref{eq:defofH4}.
Membership follows because
$\alpha_1+\cdots+\alpha_4=0$ holds coordinatewise.

\paragraph{Factoring over the fibers of $h$.}
We will show that
\begin{equation}
\grayhighlight{\Lambda_{\xi,\zeta}(A_1,A_2,A_3,A_4)
=
\sum_{\alpha_1+\cdots+\alpha_4=0}
c(\boldsymbol\alpha)
\prod_{B\in\operatorname{Fib}(h)}
\Theta_B^{\xi,\zeta}(\bm{\alpha}_B),}
\label{eq:lambda-fiber-phase-expansion}
\end{equation}
where $\operatorname{Fib}(h)$ is the set of nonempty fibers of $h$, and the
factors $\Theta_B^{\xi,\zeta}(\bm{\alpha}_B)$ are defined below as ratios of
Pauli phases.

Fix $\xi,\zeta\in\{0,h\}$ and one Fourier tuple
$\boldsymbol\alpha=(\alpha_1,\ldots,\alpha_4)$ in the expansion of
$\Lambda_{\xi,\zeta}$.
The coefficient
$c(\boldsymbol\alpha)$ in \eqref{eq:source-fourier-coefficient} corresponds to the case $\xi=\zeta=0$ and uses
$\kappa_n$ for both product pairs.  We compare the 
cocycles on each hash fiber.  For $j\in[D]$, let $B_j=h^{-1}(j)$.
For any pair of labels $\alpha,\beta\in G^n$, the tensor-product
definition of the Pauli cocycle implies
\begin{align}
\kappa_D(H\alpha,H\beta)
&=
\prod_{j=1}^D
\kappa\!\left(
  \sum_{i\in B_j}\alpha_i,
  \sum_{i\in B_j}\beta_i
\right),
\notag\\
\kappa_n(\alpha,\beta)
&=
\prod_{j=1}^D\prod_{i\in B_j}
\kappa(\alpha_i,\beta_i).
\label{eq:bucketwise-cocycle-factorization}
\end{align}
For a finite coordinate set $B\subseteq[n]$ and two restricted label tuples
$\alpha_B,\beta_B\in G^B$, define
\begin{equation}
\label{eq:bucket-pair-phase-ratio}
\theta_B(\alpha_B,\beta_B)
:=
\frac{
\kappa\!\left(\sum_{i\in B}\alpha_i,
               \sum_{i\in B}\beta_i\right)}
{\prod_{i\in B}\kappa(\alpha_i,\beta_i)}.
\end{equation}
The ratio between the two terms in \cref{eq:bucketwise-cocycle-factorization} therefore factors as
\begin{equation*}
\frac{\kappa_D(H\alpha,H\beta)}{\kappa_n(\alpha,\beta)}
=
\prod_{B\in\operatorname{Fib}(h)}
\theta_B(\alpha_B,\beta_B),
\end{equation*}
where $\operatorname{Fib}(h)$ is the collection of nonempty fibers.
On each fiber, the choice $\xi=h$ contributes
$\theta_B((\alpha_1)_B,(\alpha_2)_B)$, whereas $\xi=0$ contributes $1$.
Similarly, $\zeta=h$ contributes the conjugate of
$\theta_B((\alpha_3)_B,(\alpha_4)_B)$ because the second pair enters
$\Lambda_{\xi,\zeta}$ conjugate-linearly.  For
$\bm{\alpha}_B=(\bm{\alpha}_i)_{i\in B}$, combine these contributions as
\begin{equation}
\Theta_B^{\xi,\zeta}(\bm{\alpha}_B)
:=
\theta_B((\alpha_1)_B,(\alpha_2)_B)^{\one_{\xi=h}}
\overline{
\theta_B((\alpha_3)_B,(\alpha_4)_B)
}^{\one_{\zeta=h}}.
\label{eq:exact-bucket-phase}
\end{equation}
Multiplying over the fibers expresses the full cocycle ratio as
\begin{equation}
\left(
\frac{\kappa_D(H\alpha_1,H\alpha_2)}{\kappa_n(\alpha_1,\alpha_2)}
\right)^{\one_{\xi=h}}
\overline{\left(
\frac{\kappa_D(H\alpha_3,H\alpha_4)}{\kappa_n(\alpha_3,\alpha_4)}
\right)}^{\one_{\zeta=h}}
=
\prod_{B\in\operatorname{Fib}(h)}
\Theta_B^{\xi,\zeta}(\bm{\alpha}_B),
\label{eq:fiber-phase-product}
\end{equation}
and hence we have \cref{eq:lambda-fiber-phase-expansion}.

The definition gives the following properties of
$\Theta_B^{\xi,\zeta}$; we suppress the superscripts in the statement.

\begin{lemma}
\label{prop:fiber-phase-basic-properties}
For every $B\subseteq[n]$ and every $i\in[n]$,
\begin{equation*}
|\Theta_B(\bm{\alpha}_B)|=1,
\qquad
\Theta_\varnothing=\Theta_{\{i\}}=1.
\end{equation*}
For every $x\in B$,
\begin{equation*}
\bm{\alpha}_x=0\Longrightarrow
\Theta_B(\bm{\alpha}_B)
=\Theta_{B\setminus\{x\}}(\bm{\alpha}_{B\setminus\{x\}}).
\end{equation*}
Moreover, $\Theta_B$ is invariant under simultaneous relabeling of the
coordinates in $B$ and of the corresponding local tuples.
\end{lemma}


\paragraph{The first \mobius{} inversion.} We now prepare to take the expectation over $h$ of the product in
\eqref{eq:lambda-fiber-phase-expansion}.  The fiber $B$ itself
depends on $h$.  We therefore seek an exact expansion of each
$\Theta_B(\bm{\alpha}_B)$ into products of terms attached to fixed subsets
$C\subseteq B$.  A term attached to $C$ will later require only the simple
event that $h$ be constant on $C$.

This is the first \mobius{} inversion.  It removes from $\Theta_B$ the terms
already generated by partitions of $B$ into smaller subsets.
Suppress the superscript $(\xi,\zeta)$, and let $\hat{1}_B$ be the one-block
partition of $B$.  Using the downward inversion formula from
\cref{lem:mobius-inversion-partition-lattice}, define the \emph{cumulant} associated
with $B$ by
\begin{align}
K_B(\bm{\alpha}_B)
&=
\sum_{\pi\in\operatorname{Part}(B)}
\mu(\pi,\hat{1}_B)
\prod_{C\in\pi}\Theta_C(\bm{\alpha}_C)
\notag\\
&=
\sum_{\pi\in\operatorname{Part}(B)}
(-1)^{|\pi|-1}(|\pi|-1)!
\prod_{C\in\pi}\Theta_C(\bm{\alpha}_C).
\label{eq:fiber-phase-cumulant}
\end{align}
Then partition-lattice \mobius{} inversion makes
\eqref{eq:fiber-phase-cumulant} equivalent to the following reverse formula.
\begin{lemma}
\label{lem:fiber-phase-cumulant-inversion}
For every nonempty
$B\subseteq[n]$,
\begin{equation}
\Theta_B(\bm{\alpha}_B)
=
\sum_{\pi\in\operatorname{Part}(B)}
\prod_{C\in\pi}K_C(\bm{\alpha}_C).
\label{eq:bucket-moment-cumulant}
\end{equation}
\end{lemma}

%

The following fact shows that we could constrain the cumulant $K_B$ on nonzero inputs, the number of which could be bounded by the Pauli degree of the inputs.

\begin{lemma}
\label{lem:cumulant-zero-input}
If $B\subseteq[n]$, $|B|\ge2$, and
$\bm{\alpha}_x=0$ for some
$x\in B$, then
\begin{equation}
K_B(\bm{\alpha}_B)=0.
\label{eq:cumulant-zero-input}
\end{equation}
In particular, $K_{\{i\}}=1$ for every singleton $B=\{i\}$.
\end{lemma}

It remains to average the product in \eqref{eq:lambda-fiber-phase-expansion}.  For
a fixed $\boldsymbol\alpha$, define
\begin{equation}
\label{eq:fourier-tuple-active-set}
S(\boldsymbol\alpha)
:=\{i:\bm{\alpha}_i\ne0\}.
\end{equation}
Taking the expectation over $h$ gives the following identity; its proof is deferred.
Define
\begin{equation}
\label{eq:partition-rank}
r(\pi)=\sum_{C\in\pi}(|C|-1).
\end{equation}
\begin{lemma}\label{fact:phisycal-expansion}
\begin{equation}
\grayhighlight{
\begin{aligned}
\Delta_{\xi,\zeta}
&:=
\mathbb E_h\Lambda_{\xi,\zeta}(A_1,A_2,A_3,A_4)
-\Lambda_{0,0}(A_1,A_2,A_3,A_4)
\\
&=
\sum_{\alpha_1+\cdots+\alpha_4=0}
c(\boldsymbol\alpha)
\sum_{\substack{
\pi\in\operatorname{Part}(S(\boldsymbol\alpha))\\
\pi\ne\hat0_{S(\boldsymbol\alpha)}}}
D^{-r(\pi)}
\prod_{C\in\pi}K_C(\bm{\alpha}_C).
\end{aligned}
}
\label{eq:physical-partition-expansion}
\end{equation}
\end{lemma}

The sum in \cref{fact:phisycal-expansion} runs over partitions of
$S(\boldsymbol\alpha)$.  For every contributing Fourier tuple,
$|S(\boldsymbol\alpha)|\leq\sum_{j=1}^4|\alpha_j|\leq4k$, since each
$A_j$ has Pauli degree at most $k$.  The active coordinates may occur
anywhere in $[n]$.  We next separate the partition structure from the
placement of these coordinates.

The proofs of \cref{lem:fiber-phase-cumulant-inversion}, \cref{lem:cumulant-zero-input}, \cref{fact:phisycal-expansion} are deferred to the end of this subsection.

\paragraph{Relabeling the non-singleton blocks.}
We now reorganize \eqref{eq:physical-partition-expansion} without changing
its value.  The factor $D^{-r(\pi)}$ suggests grouping the terms first by
the rank $r(\pi)$.  A direct termwise estimate is not suitable, however:
the domain $S(\boldsymbol\alpha)$ of $\pi$ varies with $\boldsymbol\alpha$, and taking absolute values before reorganizing the sum would leave an
uncontrolled Fourier $\ell_1$ norm of the inputs.  On the other hand,
singleton blocks affect neither $r(\pi)$ nor the cumulant product, since
they contribute $|C|-1=0$ and $K_{\{i\}}=1$.  Hence a term is determined by
its non-singleton block structure and by the exact coordinates on which
that structure is placed.  We separate these two pieces by labeling the
coordinates in the non-singleton blocks with a fixed set $[v]$.
This is an
exact reindexing of \eqref{eq:physical-partition-expansion}.

Fix $\boldsymbol\alpha$ and a nontrivial partition
$\pi\in\operatorname{Part}(S(\boldsymbol\alpha))$.  Let
$\pi_{\mathrm{ns}}$ be its collection of non-singleton blocks and put
\begin{equation}
\label{eq:partition-nonsingleton-support}
V(\pi):=\bigcup_{C\in\pi_{\mathrm{ns}}}C,
\qquad
v:=|V(\pi)|,
\qquad
q:=|\pi_{\mathrm{ns}}|.
\end{equation}
Since singleton blocks contribute zero in \eqref{eq:partition-rank},
the rank satisfies
\begin{equation*}
r(\pi)=v-q.
\end{equation*}
Because $\pi$ is nontrivial, $q\ge1$, and hence $v=r+q\ge r+1$.
Moreover, every block in $\pi_{\mathrm{ns}}$ contains at least two
coordinates, so
$v\ge2q=2(v-r)$, equivalently $v\le2r$.  Finally,
$V(\pi)\subseteq S(\boldsymbol\alpha)$ and each input $A_j$ has Pauli degree
at most $k$, whence
\begin{equation}
r+1\le v\le2r,
\qquad
v\le |S(\boldsymbol\alpha)|
\le\sum_{j=1}^4|\alpha_j|
\le4k.
\label{eq:rank-vertex-relations}
\end{equation}
These are the only values of $v$ that can occur in a nonzero rank-$r$
term.
The coordinates in $S(\boldsymbol\alpha)\setminus V(\pi)$ remain singleton
blocks and contribute only factors $K_{\{i\}}=1$.  Choose a labeling, namely
a bijection
\begin{equation*}
\iota:[v]\longrightarrow V(\pi)\subseteq[n],
\qquad \iota(a)=i_a,
\end{equation*}
and write the non-singleton blocks as subsets of $[v]$ using this labeling:
\begin{equation*}
\mathcal C
:=
\bigl\{\iota^{-1}(C):C\in\pi_{\mathrm{ns}}\bigr\}.
\end{equation*}
Thus $\mathcal C$ is a partition of $[v]$ with no singleton blocks.  It
records the blocks after labeling by $[v]$, while $\iota$ specifies their
coordinates in $[n]$.  Conversely, given such a pair $(\mathcal C,\iota)$ and a Fourier
tuple $\boldsymbol\alpha$, the blocks $\iota(C)$, $C\in\mathcal C$, together
with singleton blocks on
$S(\boldsymbol\alpha)\setminus\iota([v])$, recover $\pi$.
The property \cref{lem:cumulant-zero-input} allows the
condition $\iota([v])\subseteq S(\boldsymbol\alpha)$ to be omitted: a term
with $\bm{\alpha}_{\iota(a)}=0$ already vanishes.

Every partition of a $v$-element subset of $[n]$ into non-singleton blocks
has exactly $v!$ labelings by $[v]$.  Define
\begin{equation}
\label{eq:abstract-partition-family}
\mathfrak P_{v,r}
:=
\left\{
\mathcal C\in\operatorname{Part}([v]):
|B|\ge2\ \forall B\in\mathcal C,
\ |\mathcal C|=v-r
\right\}.
\end{equation}
Then \eqref{eq:physical-partition-expansion} may be rewritten exactly as
\begin{equation}
\grayhighlight{
\Delta_{\xi,\zeta}
=
\sum_{r\ge1}D^{-r}
\sum_{v\ge1}
\frac1{v!}
\sum_{\mathcal C\in\mathfrak P_{v,r}}
\sum_{\alpha_1+\cdots+\alpha_4=0}
c(\boldsymbol\alpha)
\sum_{\iota:[v]\hookrightarrow[n]}
\prod_{B\in\mathcal C}
K_B\bigl((\bm{\alpha}_{\iota(a)})_{a\in B}\bigr).}
\label{eq:abstract-rank-decomposition}
\end{equation}
Denote
\begin{equation}
\mathcal R_{\mathcal C}
:=
\sum_{\alpha_1+\cdots+\alpha_4=0}
c(\boldsymbol\alpha)
\sum_{\iota:[v]\hookrightarrow[n]}
\prod_{B\in\mathcal C}
K_B\bigl((\bm{\alpha}_{\iota(a)})_{a\in B}\bigr).
\label{eq:fixed-abstract-pattern-contribution}
\end{equation}
For fixed $\mathcal C$, the choice of coordinates in $[n]$ is specified by
the injection $\iota$, or equivalently by a tuple
$(i_1,\ldots,i_v)$ of pairwise distinct coordinates.  Crucially, $v$ and
the partition $\mathcal C$ no longer depend on the Fourier tuple
$\boldsymbol\alpha$.  We may therefore regard the cumulant product as one
fixed function on the finite space $\mathcal H_4^v$ and analyze the
remaining Fourier sum without taking coefficient-wise absolute values.

We will control each $\mathcal{R}_{\mathcal{C}}$
using the following fact,
with the proof deferred to the end of this subsection.
\begin{lemma}\label{fact:R_C-upperbound}
\begin{equation}
\abs{\mathcal R_{\mathcal C}}
\le
63^v v^{2v}(512k^3)^v
\prod_{j=1}^4\norm{A_j}_4.
\label{eq:fixed-abstract-pattern-bound}
\end{equation}
\end{lemma}

There are at most $v^v$ partitions of $[v]$, hence at most $v^v$ choices
for $\mathcal C$.  Summing
\eqref{eq:fixed-abstract-pattern-bound} over the values of $v$ allowed by
\eqref{eq:rank-vertex-relations}, using $1/v!\leq1$, and enlarging the
exponent, we obtain
\begin{equation*}
\sum_{v\ge1}\frac1{v!}
\sum_{\mathcal C\in\mathfrak P_{v,r}}
\abs{\mathcal R_{\mathcal C}}
\le
(Ck^{16})^r\prod_{j=1}^4\norm{A_j}_4.
\end{equation*}
This bound is independent of $n$.  Returning to the exact rank decomposition
\eqref{eq:abstract-rank-decomposition}, we apply the triangle inequality
and the preceding bound at each rank to obtain
\begin{align*}
\abs{\Delta_{\xi,\zeta}}
&\le
\sum_{r\ge1}D^{-r}
\sum_{v\ge1}\frac1{v!}
\sum_{\mathcal C\in\mathfrak P_{v,r}}
\abs{\mathcal R_{\mathcal C}}
\notag\\
&\le
\sum_{r\ge1}\left(\frac{Ck^{16}}D\right)^r
\prod_{j=1}^4\norm{A_j}_4.
\end{align*}
If $D\ge C_0k^{16}$ with $C_0$ sufficiently large, the series is geometric and satisfies the bound in \cref{lem:folded-product-estimate}, after enlarging $C_1$.  No Hermitianity or reality assumption was used: the last two inputs enter conjugate-linearly, exactly as encoded by the conjugated phase in \eqref{eq:exact-bucket-phase}.
\end{proof}

\begin{proof}[Proof of \cref{lem:fiber-phase-cumulant-inversion}]
For $\sigma\in\operatorname{Part}(B)$, set
\[
G(\sigma)
:=\prod_{C\in\sigma}\Theta_C(\bm{\alpha}_C),
\qquad
F(\sigma)
:=\prod_{C\in\sigma}K_C(\bm{\alpha}_C).
\]
If $\pi\le\sigma$, let $\pi|_C$ denote the restriction of $\pi$ to a
block $C\in\sigma$.  Substitute \eqref{eq:fiber-phase-cumulant} on every
block of $\sigma$ and expand the resulting product.  This produces
\begin{equation*}
F(\sigma)
=
\sum_{\pi\le\sigma}
\left(
  \prod_{C\in\sigma}
  \mu(\pi|_C,\hat1_C)
\right)
G(\pi)
=
\sum_{\pi\le\sigma}
\mu(\pi,\sigma)G(\pi).
\end{equation*}
The second equality follows from the product formula
\eqref{eq:partition-lattice-mobius}: the interval $[\pi,\sigma]$ splits
over the blocks $C\in\sigma$.  The downward inversion formula in
\cref{lem:mobius-inversion-partition-lattice} therefore reads
\[
G(\sigma)=\sum_{\pi\le\sigma}F(\pi).
\]
Taking $\sigma=\hat1_B$ yields
\eqref{eq:bucket-moment-cumulant}.
This is precisely the form of the classical moment--cumulant relation (see e.g. \cite{Speed83}).
\end{proof}

\begin{proof}[Proof of \cref{lem:cumulant-zero-input}]
Fix $x\in B$ with $\bm{\alpha}_x=0$.  For each
$\rho\in\operatorname{Part}(B\setminus\{x\})$, group together all
partitions of $B$ that reduce to $\rho$ after deleting $x$.  If $\rho$ has
$q$ blocks, there is one partition in which $x$ is a singleton; its
coefficient in \eqref{eq:fiber-phase-cumulant} is $(-1)^q q!$.
There are also $q$ partitions obtained by inserting $x$ into one block of
$\rho$, each with coefficient $(-1)^{q-1}(q-1)!$.  By \cref{prop:fiber-phase-basic-properties}, all
$q+1$ terms have the same product of $\Theta$-factors after $x$ is deleted.
Their total coefficient is
\[
(-1)^q q!+q(-1)^{q-1}(q-1)!=0.
\]
Every group therefore cancels, proving
\eqref{eq:cumulant-zero-input}.
\end{proof}

\begin{proof}[Proof of \cref{fact:phisycal-expansion}]
By
\cref{prop:fiber-phase-basic-properties}, inactive coordinates can first be
deleted from every fiber.  Apply \eqref{eq:bucket-moment-cumulant} to each
remaining fiber and multiply the expansions.  Choosing one partition inside
each fiber is equivalent to choosing a global partition $\pi$ of $S$ such
that $h$ is constant on every block $C\in\pi$.  Hence, for each fixed $h$,
\begin{equation}
\prod_{B\in\operatorname{Fib}(h)}\Theta_B(\bm{\alpha}_B)
=
\sum_{\pi\in\operatorname{Part}(S)}
\mathbf{1}_{\{h|_C\text{ is constant for every }C\in\pi\}}
\prod_{C\in\pi}K_C(\bm{\alpha}_C).
\label{eq:fixed-h-fiber-cumulant-expansion}
\end{equation}
For a block $C$ of size $s$, the probability that $h|_C$ is constant is
$D^{-(s-1)}$: the image of one coordinate is arbitrary, and the other
$s-1$ images must agree with it.  Since the blocks of $\pi$ are disjoint,
these events depend on disjoint hash variables and are independent.

Termwise expectation in \eqref{eq:fixed-h-fiber-cumulant-expansion} replaces
each partition by its factor $D^{-r(\pi)}$.  Thus
\begin{equation}
\mathbb E_h\prod_{B\in\operatorname{Fib}(h)}
\Theta_B(\bm{\alpha}_B)
=
\sum_{\pi\in\operatorname{Part}(S)}
D^{-r(\pi)}\prod_{C\in\pi}K_C(\bm{\alpha}_C).
\label{eq:fiber-cumulant-expansion}
\end{equation}
The all-singleton partition $\hat{0}_S$ contributes exactly $1$.  Take
expectation in \eqref{eq:lambda-fiber-phase-expansion} and substitute
\eqref{eq:fiber-cumulant-expansion}.  Subtracting $\Lambda_{0,0}$ removes
the all-singleton contribution and yields the expansion over partitions of
$S(\boldsymbol\alpha)$:
\begin{equation*}
\grayhighlight{
\begin{aligned}
\Delta_{\xi,\zeta}
&:=
\mathbb E_h\Lambda_{\xi,\zeta}(A_1,A_2,A_3,A_4)
-\Lambda_{0,0}(A_1,A_2,A_3,A_4)
\\
&=
\sum_{\alpha_1+\cdots+\alpha_4=0}
c(\boldsymbol\alpha)
\sum_{\substack{
\pi\in\operatorname{Part}(S(\boldsymbol\alpha))\\
\pi\ne\hat0_{S(\boldsymbol\alpha)}}}
D^{-r(\pi)}
\prod_{C\in\pi}K_C(\bm{\alpha}_C).
\end{aligned}
}
\end{equation*}
Here $\pi$ ranges over partitions of $S(\boldsymbol\alpha)\subseteq[n]$.
\end{proof}

We now prove \cref{fact:R_C-upperbound}.
\begin{proof}[Proof of \cref{fact:R_C-upperbound}]
We control each $\mathcal R_{\mathcal C}$ by a
character expansion and an auxiliary polynomial.

\paragraph{Fourier expansion of $\mathcal R_{\mathcal C}$.}
Fix $\mathcal C\in\mathfrak P_{v,r}$.  To make explicit the quantity that
must be bounded, define the scalar-valued function
$W_{\mathcal C}:\mathcal H_4^v\to\mathbb C$ by
\begin{equation}
\label{eq:cumulant-product-function}
W_{\mathcal C}(t_1,\ldots,t_v)
:=
\prod_{B\in\mathcal C}K_B((t_a)_{a\in B}).
\end{equation}
Every block of $\mathcal C$ has size at least two, so
\eqref{eq:cumulant-zero-input} implies that
$W_{\mathcal C}(t_1,\ldots,t_v)=0$ whenever some $t_a=0$.  Consequently its
exact point-mass expansion is
\begin{equation}
W_{\mathcal C}(t_1,\ldots,t_v)
=
\sum_{\boldsymbol s\in(\mathcal H_4\setminus\{0\})^v}
W_{\mathcal C}(\boldsymbol s)
\prod_{a=1}^v\delta_{s_a}(t_a).
\label{eq:cumulant-product-point-mass-expansion}
\end{equation}
For $s\in\mathcal H_4\setminus\{0\}$, Fourier inversion and character
orthogonality specialize to \eqref{eq:h4-nonzero-point-mass-fourier}.  We use this
identity below in the form
\begin{equation}
\delta_s(t)=\frac1{64}\sum_{\varphi\in\widehat{\mathcal H_4}}\varphi(s)\bigl(\varphi(t)-1\bigr).
\label{eq:nonzero-point-mass}
\end{equation}
For each $\varphi\in\widehat{\mathcal H_4}$, write
$\varphi(t)=\prod_{j=1}^4\chi_{\nu_j}(t_j)$ by taking its three
coordinate labels and setting $\nu_4=0$.  Thus multiplication of
the $j$th input's Pauli coefficient by the corresponding coordinate
character is conjugation by a one-qubit Pauli.  Substitute
\eqref{eq:cumulant-product-point-mass-expansion} into
\eqref{eq:fixed-abstract-pattern-contribution}, and then expand every point
mass by \eqref{eq:nonzero-point-mass}.  The resulting identity is
\begin{align}
\mathcal R_{\mathcal C}
={}&
\frac1{64^v}
\sum_{\boldsymbol s\in(\mathcal H_4\setminus\{0\})^v}
W_{\mathcal C}(\boldsymbol s)
\sum_{\varphi_1,\ldots,\varphi_v\in\widehat{\mathcal H_4}}
\left(\prod_{a=1}^v\varphi_a(s_a)\right)
\sum_{\alpha_1+\cdots+\alpha_4=0}
c(\boldsymbol\alpha)
\notag\\
&\quad\cdot
\sum_{\substack{i_1,\ldots,i_v\in[n]\\
i_1,\ldots,i_v\ \mathrm{pairwise\ distinct}}}\hspace{-1em}
\prod_{a=1}^v
\bigl(\varphi_a(\bm{\alpha}_{i_a})-1\bigr).
\label{eq:abstract-pattern-direct-character-expansion}
\end{align}
The final two sums are the only part that still depends on the input
operators.  Define
\begin{equation}
{\mathfrak M_{\boldsymbol\varphi}
:={}
\sum_{\alpha_1+\cdots+\alpha_4=0}
c(\boldsymbol\alpha)\hspace{-2em}
\sum_{\substack{i_1,\ldots,i_v\in[n]\\
i_1,\ldots,i_v\ \mathrm{pairwise\ distinct}}}\hspace{-1em}
\prod_{a=1}^v
\bigl(\varphi_a(\bm{\alpha}_{i_a})-1\bigr).
\label{eq:distinct-character-contribution}}
\end{equation}
For the fixed pattern $\mathcal C$, also define
\begin{equation}
m_{\boldsymbol\varphi}
:=
\frac1{64^v}
\sum_{\boldsymbol s\in(\mathcal H_4\setminus\{0\})^v}
W_{\mathcal C}(\boldsymbol s)
\prod_{a=1}^v\varphi_a(s_a).
\label{eq:abstract-pattern-character-coefficient}
\end{equation}
Reverse the order of the two finite sums in
\eqref{eq:abstract-pattern-direct-character-expansion} to obtain
\begin{equation}
\grayhighlight{
\mathcal R_{\mathcal C}
=
\sum_{\boldsymbol\varphi\in\widehat{\mathcal H_4}^{\,v}}
m_{\boldsymbol\varphi}\mathfrak M_{\boldsymbol\varphi}.}
\label{eq:abstract-pattern-character-contributions}
\end{equation}

We need an auxiliary polynomial to bound
$\mathfrak M_{\boldsymbol\varphi}$ without any dependence on $n$.

\begin{lemma}
\label{lem:auxiliary-polynomial-estimate}
Let $n,k,v\geq1$, let
$A_1,A_2,A_3,A_4\in \Matrix_2^{\otimes n}$ have Pauli degree at most $k$, and fix
characters $\varphi_1,\ldots,\varphi_v\in\widehat{\mathcal H_4}$.  For
$\boldsymbol\alpha=(\alpha_1,\ldots,\alpha_4)$ with sum zero, put
$\bm{\alpha}_i=(\alpha_{1,i},\ldots,\alpha_{4,i})\in\mathcal H_4$.  Define
the polynomial $F\colon\mathbb C^v\to\mathbb C$ by
\begin{equation}
F(z_1,\ldots,z_v)
=
\sum_{\alpha_1+\cdots+\alpha_4=0}
c(\boldsymbol\alpha)
\prod_{a=1}^v\prod_{i=1}^n
\bigl(1-z_a+z_a\varphi_a(\bm{\alpha}_i)\bigr),
\label{eq:auxiliary-polynomial}
\end{equation}
where $c$ is defined in \eqref{eq:source-fourier-coefficient}.  Then
\begin{equation}
\partial_{z_1}\cdots\partial_{z_v}F(0)
=
\sum_{\alpha_1+\cdots+\alpha_4=0}
c(\boldsymbol\alpha)
\sum_{i_1,\ldots,i_v\in[n]}
\prod_{a=1}^v
\bigl(\varphi_a(\bm{\alpha}_{i_a})-1\bigr).
\label{eq:auxiliary-derivative-all-coordinates}
\end{equation}
\begin{equation}
\sup_{z\in[0,1]^v}|F(z)|
\le
\prod_{j=1}^4\norm{A_j}_4,
\label{eq:auxiliary-polynomial-bound}
\end{equation}
and
\begin{equation}
\abs{\partial_{z_1}\cdots\partial_{z_v}F(0)}
\le
(32k^2)^v\prod_{j=1}^4\norm{A_j}_4.
\label{eq:mixed-derivative-bound}
\end{equation}
\end{lemma}
Its proof is deferred.

\paragraph{Removing repeated coordinates: the second \mobius{} inversion.}
At this point the need for a second \mobius{} inversion appears.  The contribution
\eqref{eq:distinct-character-contribution} is a sum over injective maps
$[v]\hookrightarrow[n]$, whereas the
derivative \eqref{eq:auxiliary-derivative-all-coordinates} sums over all
maps $[v]\to[n]$.  Its bound in \eqref{eq:mixed-derivative-bound} is
independent of $n$.  The repeated-coordinate terms cannot simply be
discarded, since the summands are complex and may cancel.  The following
lemma expresses the sum over injective maps as a linear combination of sums
over all maps; its proof is deferred.

\begin{lemma}
\label{lem:distinct-coordinate-mobius}
Let $n,v\geq1$.  For arbitrary functions
$f_1,\ldots,f_v:[n]\to\mathbb C$,
\begin{equation}
\sum_{\substack{i_1,\ldots,i_v\in[n]\\
i_1,\ldots,i_v\ \mathrm{pairwise\ distinct}}}
\prod_{a=1}^v f_a(i_a)=
\sum_{\sigma\in\operatorname{Part}([v])}
\mu(\hat0,\sigma)
\prod_{B\in\sigma}
\left(\sum_{i=1}^n\prod_{a\in B}f_a(i)\right),
\label{eq:distinct-coordinate-mobius}
\end{equation}
where
\begin{equation}
\mu(\hat0,\sigma)
=
\prod_{B\in\sigma}(-1)^{|B|-1}(|B|-1)!.
\label{eq:second-mobius-coefficient}
\end{equation}
\end{lemma}

Apply \cref{lem:distinct-coordinate-mobius} for each fixed
$\boldsymbol\alpha$ by taking
\[
f_a(i)=\varphi_a(\bm{\alpha}_i)-1,
\qquad a\in[v],\quad i\in[n].
\]
Its left side is exactly the injective inner sum in
\eqref{eq:distinct-character-contribution}; its right side merges the factors
whose vertex labels belong to the same block $B\in\sigma$.  Multiply this
identity by $c(\boldsymbol\alpha)$ and sum over the Fourier tuples.  We arrive
at the exact second-inversion formula
\begin{equation}
\mathfrak M_{\boldsymbol\varphi}
=
\sum_{\sigma\in\operatorname{Part}([v])}
\mu(\hat0,\sigma)
\sum_{\alpha_1+\cdots+\alpha_4=0}
c(\boldsymbol\alpha)
\prod_{B\in\sigma}
\left(
\sum_{i=1}^n
\prod_{a\in B}
(\varphi_a(\bm{\alpha}_i)-1)
\right).
\label{eq:second-mobius-character-contribution}
\end{equation}

For a block $B\in\sigma$, the function
$g_B(t)=\prod_{a\in B}(\varphi_a(t)-1)$ satisfies $g_B(0)=0$.  Each factor
$\varphi_a-1$ has Fourier $\ell_1$ norm at most $2$, and multiplication of
functions corresponds to convolution of their Fourier coefficients by
\cref{lem:finite-abelian-fourier-identities}.
Therefore
\begin{equation*}
\sum_{\chi\in\widehat{\mathcal H_4}}
\abs{\widehat{g_B}(\chi)}
\le2^{|B|}.
\end{equation*}
Since
$g_B(0)=\sum_\chi\widehat{g_B}(\chi)=0$, the coefficient of the trivial
character satisfies
\begin{equation*}
\widehat{g_B}(\mathbf1)
=
-\sum_{\chi\ne\mathbf1}\widehat{g_B}(\chi).
\end{equation*}
Thus we may write directly
\begin{equation}
g_B(t)
=
\sum_{\substack{\chi\in\widehat{\mathcal H_4}\\
                  \chi\ne\mathbf1}}
\widehat{g_B}(\chi)\bigl(\chi(t)-1\bigr),
\qquad
\sum_{\chi\ne\mathbf1}\abs{\widehat{g_B}(\chi)}
\le2^{|B|}.
\label{eq:merged-block-character-minus-one}
\end{equation}


For a fixed $\sigma\in\operatorname{Part}([v])$, denote its contribution
inside \eqref{eq:second-mobius-character-contribution} by
\begin{equation}
\label{eq:character-partition-contribution}
\mathfrak M_{\boldsymbol\varphi,\sigma}
:=
\sum_{\alpha_1+\cdots+\alpha_4=0}
c(\boldsymbol\alpha)
\prod_{B\in\sigma}
\left(
\sum_{i=1}^n g_B(\bm{\alpha}_i)
\right).
\end{equation}
Let $b:=|\sigma|$.  Expand
\eqref{eq:merged-block-character-minus-one} for every $B\in\sigma$.  Then
\begin{equation}
\mathfrak M_{\boldsymbol\varphi,\sigma}
=
\sum_{(\chi_B)_{B\in\sigma}}
\left(
\prod_{B\in\sigma}\widehat{g_B}(\chi_B)
\right)
\sum_{\alpha_1+\cdots+\alpha_4=0}
c(\boldsymbol\alpha)
\prod_{B\in\sigma}
\left(
\sum_{i=1}^n
\bigl(\chi_B(\bm{\alpha}_i)-1\bigr)
\right),
\label{eq:fixed-partition-character-expansion}
\end{equation}
where every $\chi_B$ ranges over the nontrivial characters of
$\mathcal H_4$.  To identify the inner Fourier sum with an auxiliary
polynomial derivative, enumerate the blocks as
$\sigma=\{B_1,\ldots,B_b\}$ and, for a fixed character tuple
$\boldsymbol\chi=(\chi_{B_1},\ldots,\chi_{B_b})$, define a new
$b$-variable polynomial
\begin{equation}
F_{\sigma,\boldsymbol\chi}(w_1,\ldots,w_b)
:=
\sum_{\alpha_1+\cdots+\alpha_4=0}
c(\boldsymbol\alpha)
\prod_{\ell=1}^b\prod_{i=1}^n
\bigl(1-w_\ell+w_\ell\chi_{B_\ell}(\bm{\alpha}_i)\bigr).
\label{eq:block-auxiliary-polynomial}
\end{equation}
Expanding the product of the $b$ sums gives
\begin{equation}
\prod_{\ell=1}^b
\left(
\sum_{i=1}^n\bigl(\chi_{B_\ell}(\bm{\alpha}_i)-1\bigr)
\right)
=
\sum_{i_1,\ldots,i_b\in[n]}
\prod_{\ell=1}^b
\bigl(\chi_{B_\ell}(\bm{\alpha}_{i_\ell})-1\bigr).
\label{eq:product-of-coordinate-sums-expanded}
\end{equation}
On the other hand, the mixed derivative of
\eqref{eq:block-auxiliary-polynomial}, once in every variable, is
\begin{equation}
\partial_{w_1}\cdots\partial_{w_b}
F_{\sigma,\boldsymbol\chi}(0)
=
\sum_{\alpha_1+\cdots+\alpha_4=0}
c(\boldsymbol\alpha)
\sum_{i_1,\ldots,i_b\in[n]}
\prod_{\ell=1}^b
\bigl(\chi_{B_\ell}(\bm{\alpha}_{i_\ell})-1\bigr).
\label{eq:block-auxiliary-polynomial-derivative}
\end{equation}
Equations \eqref{eq:product-of-coordinate-sums-expanded} and
\eqref{eq:block-auxiliary-polynomial-derivative} show that this derivative
is exactly the inner Fourier sum in
\eqref{eq:fixed-partition-character-expansion}.
Apply \cref{lem:auxiliary-polynomial-estimate} with $b$ characters to obtain
\begin{equation*}
\abs{
\partial_{w_1}\cdots\partial_{w_b}
F_{\sigma,\boldsymbol\chi}(0)
}
\le
(32k^2)^b\prod_{j=1}^4\norm{A_j}_4.
\end{equation*}
Together with \eqref{eq:merged-block-character-minus-one}, this implies
\begin{align*}
\abs{\mathfrak M_{\boldsymbol\varphi,\sigma}}
&\le
(32k^2)^b
\prod_{B\in\sigma}
\left(
\sum_{\chi\ne\mathbf1}\abs{\widehat{g_B}(\chi)}
\right)
\prod_{j=1}^4\norm{A_j}_4
\notag\\
&\le
(32k^2)^b\,2^{\sum_{B\in\sigma}|B|}
\prod_{j=1}^4\norm{A_j}_4
\notag\\
&=
(32k^2)^b\,2^v
\prod_{j=1}^4\norm{A_j}_4
\le
(64k^2)^v
\prod_{j=1}^4\norm{A_j}_4,
\end{align*}
where the last inequality uses $b\le v$ and $k\ge1$.

Equation \eqref{eq:second-mobius-character-contribution} now reads
$\mathfrak M_{\boldsymbol\varphi}
=\sum_\sigma\mu(\hat0,\sigma)
\mathfrak M_{\boldsymbol\varphi,\sigma}$.  Therefore
\cref{lem:total-absolute-mobius-mass} and $v\le4k$ imply
\begin{align*}
\abs{\mathfrak M_{\boldsymbol\varphi}}
&\le
\left(\sum_\sigma|\mu(\hat0,\sigma)|\right)
(64k^2)^v
\prod_{j=1}^4\norm{A_j}_4
\notag\\
&=
v!(64k^2)^v
\prod_{j=1}^4\norm{A_j}_4
\notag\\
&\le
v^v(64k^2)^v
\prod_{j=1}^4\norm{A_j}_4
\le
(256k^3)^v
\prod_{j=1}^4\norm{A_j}_4.
\end{align*}
For later use, weaken the numerical constant from $256$ to $512$ and record
the estimate in the form
\begin{equation}
\abs{\mathfrak M_{\boldsymbol\varphi}}
\le
(512k^3)^v\prod_{j=1}^4\norm{A_j}_4
\label{eq:distinct-coordinate-bound}
\end{equation}
for arbitrary
$\boldsymbol\varphi\in\widehat{\mathcal H_4}^{\,v}$.

\paragraph{Summing the ranks.}
We now return to the exact decomposition
\eqref{eq:abstract-rank-decomposition}.  Fix
$\mathcal C\in\mathfrak P_{v,r}$ with $v$ in the admissible range
\eqref{eq:rank-vertex-relations}.  At this point we use the size of the
cumulants.  If $B\in\mathcal C$ and $s=|B|$, then
$|\Theta_{C'}(\bm{\alpha}_{C'})|=1$ for every $C'\subseteq B$.  Taking
absolute values in \eqref{eq:fiber-phase-cumulant}, we obtain
\begin{equation*}
\norm{K_B}_\infty
\le
\sum_{\pi\in\operatorname{Part}(B)}(|\pi|-1)!
\le s^{2s}.
\end{equation*}
Indeed, there are at most $s^s$ partitions of $B$, and
$(|\pi|-1)!\le s!\le s^s$.  Consequently,
\begin{equation*}
\prod_{B\in\mathcal C}\norm{K_B}_\infty
\le
\prod_{B\in\mathcal C}|B|^{2|B|}
\le v^{2v}.
\end{equation*}
Indeed, $|B|\le v$ for every block and
$\sum_{B\in\mathcal C}|B|=v$.  In particular,
\eqref{eq:abstract-pattern-character-coefficient} and
$|\varphi_a(s_a)|=1$ imply the coefficient $\ell_1$ estimate
\begin{equation}
\sum_{\boldsymbol\varphi\in\widehat{\mathcal H_4}^{\,v}}
\abs{m_{\boldsymbol\varphi}}
\le
\sum_{\boldsymbol s\in(\mathcal H_4\setminus\{0\})^v}
\abs{W_{\mathcal C}(\boldsymbol s)}
\le
63^v\norm{W_{\mathcal C}}_\infty
\le
63^v v^{2v}.
\label{eq:character-coefficient-l1-bound}
\end{equation}
Here the factor $64^{-v}$ is exactly cancelled by the $64^v$ possible
character tuples, while $(\mathcal H_4\setminus\{0\})^v$ has $63^v$
elements.  Hence \eqref{eq:abstract-pattern-character-contributions},
\eqref{eq:character-coefficient-l1-bound}, and
\eqref{eq:distinct-coordinate-bound} imply
\begin{equation*}
\abs{\mathcal R_{\mathcal C}}
\le
63^v v^{2v}(512k^3)^v
\prod_{j=1}^4\norm{A_j}_4.
\end{equation*}

\end{proof}

We now prove the deferred lemmas.

\begin{proof}[Proof of \cref{lem:auxiliary-polynomial-estimate}]
Since
\[
1-z_a+z_a\varphi_a(\bm{\alpha}_i)
=
1+z_a\bigl(\varphi_a(\bm{\alpha}_i)-1\bigr),
\]
the mixed derivative in \eqref{eq:auxiliary-polynomial} is exactly
\eqref{eq:auxiliary-derivative-all-coordinates}.

We next prove the uniform bound
\eqref{eq:auxiliary-polynomial-bound}.  For each $a\in[v]$, let
$\nu_{a,1},\nu_{a,2},\nu_{a,3}\in G$ be the three coordinate labels of
$\varphi_a$ and set $\nu_{a,4}=0$.  Then
\begin{equation*}
\varphi_a(t_1,t_2,t_3,t_4)
=
\prod_{j=1}^4\chi_{\nu_{a,j}}(t_j).
\end{equation*}
Fix $z=(z_1,\ldots,z_v)\in[0,1]^v$, and let
$\varepsilon_{a,i}\in\{0,1\}$ be independent Bernoulli variables with
$\Pr(\varepsilon_{a,i}=1)=z_a$ for $i\in[n]$.  Then
\begin{align}
&\prod_{a=1}^v\prod_{i=1}^n
\bigl(1-z_a+z_a\varphi_a(\bm{\alpha}_i)\bigr)
=
\mathbb E_{\varepsilon}
\prod_{a=1}^v\prod_{i=1}^n
\varphi_a(\bm{\alpha}_i)^{\varepsilon_{a,i}}
=
\mathbb E_{\varepsilon}
\prod_{j=1}^4
\chi_{\rho_j(\varepsilon)}(\alpha_j),
\label{eq:auxiliary-product-as-random-character}
\end{align}
where $\rho_j(\varepsilon)\in G^n$ is defined coordinatewise by
\begin{equation*}
\rho_j(\varepsilon)_i
:=
\sum_{a=1}^v\varepsilon_{a,i}\nu_{a,j}.
\end{equation*}
Indeed, for every fixed $j\in[4]$ and $i\in[n]$, multiplicativity of the
characters implies
\begin{equation}
\prod_{a=1}^v
\chi_{\nu_{a,j}}(\alpha_{j,i})^{\varepsilon_{a,i}}
=
\chi_{\sum_{a=1}^v\varepsilon_{a,i}\nu_{a,j}}(\alpha_{j,i})
=
\chi_{\rho_j(\varepsilon)_i}(\alpha_{j,i}).
\label{eq:auxiliary-character-product-at-one-coordinate}
\end{equation}
The product of \eqref{eq:auxiliary-character-product-at-one-coordinate} over
$i$ equals $\chi_{\rho_j(\varepsilon)}(\alpha_j)$, which proves the final
equality in \eqref{eq:auxiliary-product-as-random-character}.

Let $U_\rho$ be the Pauli string constructed in
\cref{lem:pauli-multiplier-implementation}, and set
\begin{equation*}
A_j^{(\varepsilon)}
:=
U_{\rho_j(\varepsilon)}A_jU_{\rho_j(\varepsilon)}^\dagger.
\end{equation*}
By \cref{lem:pauli-multiplier-implementation}, the Fourier coefficient of
$A_j^{(\varepsilon)}$ at $\alpha_j$ is
$\chi_{\rho_j(\varepsilon)}(\alpha_j)\widehat A_j(\alpha_j)$.
After substituting \eqref{eq:auxiliary-product-as-random-character} into
\eqref{eq:auxiliary-polynomial}, \cref{prop:scalar-pauli-identities} rewrites the
expression as
\begin{equation}
F(z)
=
\mathbb E_{\varepsilon}
\tau_{2^n}\!\left(
\bigl(A_3^{(\varepsilon)}A_4^{(\varepsilon)}\bigr)^\dagger
A_1^{(\varepsilon)}A_2^{(\varepsilon)}
\right).
\label{eq:auxiliary-polynomial-trace-average}
\end{equation}
For every realization of $\varepsilon$, Schatten H\"older and unitary
invariance yield
\begin{equation*}
\left|
\tau_{2^n}\!\left(
\bigl(A_3^{(\varepsilon)}A_4^{(\varepsilon)}\bigr)^\dagger
A_1^{(\varepsilon)}A_2^{(\varepsilon)}
\right)
\right|
\le
\prod_{j=1}^4\norm{A_j^{(\varepsilon)}}_4
=
\prod_{j=1}^4\norm{A_j}_4.
\end{equation*}
Taking the expectation in
\eqref{eq:auxiliary-polynomial-trace-average}, and then the supremum over
$z\in[0,1]^v$, proves \eqref{eq:auxiliary-polynomial-bound}.

We record here the polynomial derivative estimate needed for the final step.
The classical one-variable Markov brothers' inequality is stated, for example,
in \cite[Equation~(0.2)]{BE95Markov}.  If $d\geq0$ and $p$ is a complex
polynomial of degree at most $d$, then
\begin{equation*}
\norm{p'}_{L_\infty([-1,1])}
\leq d^2\norm{p}_{L_\infty([-1,1])}.
\end{equation*}
For real polynomials this is the classical Markov inequality.  For a complex
polynomial, fix $x_0\in[-1,1]$ with $p'(x_0)\neq0$ and choose
$\zeta\in\mathbb C$, $|\zeta|=1$, so that
$\zeta p'(x_0)=|p'(x_0)|$.  Applying the real inequality to
$\operatorname{Re}(\zeta p)$ and taking the supremum over $x_0$ proves the
same estimate.  Under the affine change of variables $q(t)=p(2t-1)$, the
derivative acquires a factor of $2$; hence
\begin{equation*}
\norm{p'}_{L_\infty([0,1])}
\leq 2d^2\norm{p}_{L_\infty([0,1])}.
\end{equation*}
Applying this estimate successively in $z_1,\ldots,z_v$ shows that every
polynomial $P(z_1,\ldots,z_v)$ of degree at most $d$ in each variable
satisfies
\begin{equation}
\abs{\partial_{z_1}\cdots\partial_{z_v}P(0)}
\leq (2d^2)^v
\sup_{z\in[0,1]^v}\abs{P(z)}.
\label{eq:multivariate-markov-brothers}
\end{equation}

It remains to obtain the derivative estimate.  Fix a Fourier tuple
$\boldsymbol\alpha$ occurring in \eqref{eq:auxiliary-polynomial}.  If
$i\notin S(\boldsymbol\alpha)$, then $\bm{\alpha}_i=0$ and hence
$\varphi_a(\bm{\alpha}_i)=1$ for every $a$.  The corresponding factor in
\eqref{eq:auxiliary-polynomial} is therefore identically one.  Consequently,
this summand has degree in each variable $z_a$ at most
\begin{equation*}
|S(\boldsymbol\alpha)|
\le
\sum_{j=1}^4|\alpha_j|
\le4k.
\end{equation*}
Since this holds for every Fourier tuple in the sum, $F$ has degree at most
$4k$ in each variable.
Apply \eqref{eq:multivariate-markov-brothers} with $d=4k$, and then use
\eqref{eq:auxiliary-polynomial-bound}, to obtain
\begin{align*}
\abs{\partial_{z_1}\cdots\partial_{z_v}F(0)}
&\le
\bigl(2(4k)^2\bigr)^v
\sup_{z\in[0,1]^v}|F(z)|\\
&\le
(32k^2)^v\prod_{j=1}^4\norm{A_j}_4.
\end{align*}
This is \eqref{eq:mixed-derivative-bound}.
\end{proof}

\begin{proof}[Proof of \cref{lem:distinct-coordinate-mobius}]
For a map $\boldsymbol i=(i_1,\ldots,i_v):[v]\to[n]$, let
$\ker(\boldsymbol i)$ be the partition defined by
$a\sim b$ if and only if $i_a=i_b$.  For
$\rho\in\operatorname{Part}([v])$, set
\[
E_\rho
:=
\sum_{\ker(\boldsymbol i)=\rho}\prod_{a=1}^v f_a(i_a).
\]
For $\sigma\in\operatorname{Part}([v])$, define
\begin{align*}
A_\sigma
&:=
\sum_{\sigma\le\ker(\boldsymbol i)}\prod_{a=1}^v f_a(i_a)=
\prod_{B\in\sigma}
\left(\sum_{i=1}^n\prod_{a\in B}f_a(i)\right).
\end{align*}
The condition $\sigma\le\ker(\boldsymbol i)$ enforces equality inside each
$\sigma$-block but permits equality between different blocks.  Therefore
\[
A_\sigma=\sum_{\rho\ge\sigma}E_\rho.
\]
The injective maps are precisely those with
$\ker(\boldsymbol i)=\hat0_{[v]}$.  The upward form of \mobius{} inversion
from \cref{lem:mobius-inversion-partition-lattice}, evaluated at
$\hat0_{[v]}$, reads
\[
E_{\hat0}
=
\sum_{\sigma\in\operatorname{Part}([v])}
\mu(\hat0,\sigma)A_\sigma,
\]
which is \eqref{eq:distinct-coordinate-mobius}.  The coefficient formula
\eqref{eq:second-mobius-coefficient} is the specialization of
\eqref{eq:partition-lattice-mobius} to $\pi=\hat0$.
\end{proof}

\section{Quantum Markov Semigroups, Dirichlet Forms, and Entropy}

We follow the terminology of \cite[Section~2.3]{BDR20}, specialized here
to the tracial setting.

\begin{definition}[Tracial quantum Markov semigroup]
  Let $\mathcal M=\Matrix_d$ carry its normalized trace $\tau=\tau_d$.
  A \emph{tracial quantum Markov semigroup} (QMS) on $\mathcal M$ is a
  family $(\mathcal T_u)_{u\geq0}$ of completely positive, unital, and
  trace-preserving maps of the form
  \begin{equation*}
    \mathcal T_u=e^{-u\cL},\qquad u\geq0,
  \end{equation*}
  where the linear map $\cL\colon\mathcal M\to\mathcal M$ is called its
  \emph{Lindblad generator}.  In particular,
  $\mathcal T_0=\operatorname{id}$ and
  $\mathcal T_{u+v}=\mathcal T_u\circ\mathcal T_v$ for $u,v\geq0$.
  The generator is uniquely determined by
  \begin{equation*}
    \cL(X)
    =-\left.\frac{\dd}{\dd u}\mathcal T_u(X)\right|_{u=0^+}
    =\lim_{u\downarrow0}\frac{X-\mathcal T_u(X)}{u},
    \qquad X\in\mathcal M.
  \end{equation*}
  For every $u\geq0$ and $X\in\mathcal M$,
  \begin{equation*}
    \frac{\dd}{\dd u}\mathcal T_u(X)
    =-\cL\mathcal T_u(X)
    =-\mathcal T_u(\cL X).
  \end{equation*}
\end{definition}

\begin{definition}
  Let $(\mathcal T_u)_{u\geq0}$ be a tracial QMS on $\mathcal M$ with
  normalized trace $\tau$ and generator $\cL$.
  It is \emph{reversible} if every $\mathcal T_u$, equivalently $\cL$,
  is self-adjoint on $L_2(\tau)$: for all $X,Y\in\mathcal M$ and $u\geq0$,
  \begin{equation*}
    \tau\!\left(X^\dagger\mathcal T_u(Y)\right)
    =
    \tau\!\left(\mathcal T_u(X)^\dagger Y\right).
  \end{equation*}
\end{definition}

\paragraph{The depolarizing QMS.}
We now describe the semigroup structure of the depolarizing maps in
\cref{def:depolarizing-maps}.  On one qubit, set
$\cL_1:=\operatorname{id}_{\Matrix_2}-\Pi$.  For $n$ qubits, set
\begin{equation}
\label{eq:depolarizing-generator}
  \cL:=\sum_{i=1}^n
    \operatorname{id}_{\Matrix_2}^{\otimes(i-1)}
    \otimes\cL_1\otimes
    \operatorname{id}_{\Matrix_2}^{\otimes(n-i)}.
\end{equation}

\begin{prop}
  \label{prop:depolarizing-qms}
  The families $(\mathcal D_u)_{u\geq0}$ and $(\mathcal T_u)_{u\geq0}$ are
  reversible tracial QMSs, with positive generators $\cL_1$ and $\cL$,
  respectively.  In particular,
  \begin{equation*}
    \mathcal D_u=e^{-u\cL_1},
    \qquad
    \mathcal T_u=e^{-u\cL},
    \qquad
    \cL(P_\alpha)=|\alpha|P_\alpha.
  \end{equation*}
\end{prop}

\begin{proof}
  The averaging map $\Pi$ is a self-adjoint projection on $L_2(\tau_2)$.
  It is also completely positive, unital, and trace preserving, as is clear
  from its Pauli-average representation
  $\Pi(X)=\frac14\sum_{s\in G}P_sXP_s^\dagger$.
  Since $\Pi^2=\Pi$, the exponential of $\cL_1$ is
  \begin{equation*}
    e^{-u\cL_1}=\Pi+e^{-u}(\operatorname{id}_{\Matrix_2}-\Pi)=\mathcal D_u.
  \end{equation*}
  Thus $\mathcal D_0=\operatorname{id}_{\Matrix_2}$ and $\mathcal D_{u+v}=\mathcal D_u\mathcal D_v$.
  The expression $\mathcal D_u=e^{-u}\operatorname{id}_{\Matrix_2}+(1-e^{-u})\Pi$
  shows that $\mathcal D_u$ is completely positive, unital, trace preserving,
  and self-adjoint.

  The summands of $\cL$ commute because they act on distinct tensor factors.
  Hence
  \begin{equation*}
    e^{-u\cL}=(e^{-u\cL_1})^{\otimes n}
    =\mathcal D_u^{\otimes n}=\mathcal T_u.
  \end{equation*}
  Tensor products preserve the stated channel properties
  and self-adjointness, so $(\mathcal T_u)_{u\geq0}$ is a reversible
  tracial QMS.  On a Pauli string $P_\alpha$, the $i$th summand of $\cL$
  acts as the identity exactly when $\alpha_i\neq0$ and otherwise vanishes.
  Therefore
  $\cL(P_\alpha)=|\alpha|P_\alpha$, which also proves that $\cL$ is
  positive semidefinite on $L_2(\tau_{2^n})$.
\end{proof}

\begin{prop}
  \label{prop:depolarizing-product-unitaries}
  Let $n\geq1$, $u\geq0$, and let $U_1,\ldots,U_n\in\U(2)$.
  Set $U:=U_1\otimes\cdots\otimes U_n$.  For every
  $X\in \Matrix_2^{\otimes n}$,
  \begin{equation*}
    \mathcal T_u(UXU^\dagger)=U\mathcal T_u(X)U^\dagger.
  \end{equation*}
\end{prop}

\begin{proof}
  For $V\in\U(2)$ and $A\in \Matrix_2$, invariance of the trace under
  unitary conjugation implies
  \begin{equation*}
    \mathcal D_u(VAV^\dagger)
    =e^{-u}VAV^\dagger+(1-e^{-u})\tau_2(A)\id
    =V\mathcal D_u(A)V^\dagger.
  \end{equation*}
  Applying this identity in each tensor factor proves the assertion for
  $X=X_1\otimes\cdots\otimes X_n$.  Such tensors span
  $\Matrix_2^{\otimes n}$, so the general case follows by linearity.
\end{proof}

\begin{prop}
  \label{prop:depolarizing-contractivity}
  Let $n,t\geq1$, $u\geq0$, and let
  $W\colon\mathbb C^{2^n}\to\mathbb C^t\otimes\mathbb C^{2^n}$ be a
  block operator as in \eqref{eq:block-operator-form}.
  Then
  \begin{equation*}
    \norm{\mathcal T_u(W)}_\infty\leq\norm W_\infty,
  \end{equation*}
  where $\mathcal T_u$ acts entrywise on $W$.
\end{prop}

\begin{proof}
  Write $W=\sum_a|a\rangle\otimes W_a$.  For every $a$,
  \begin{equation*}
    \begin{pmatrix}W_a^\dagger W_a&W_a^\dagger\\W_a&\id\end{pmatrix}
    =\begin{pmatrix}W_a^\dagger\\\id\end{pmatrix}
      \begin{pmatrix}W_a&\id\end{pmatrix}\geq0.
  \end{equation*}
  Since $\mathcal T_u$ is unital and completely positive, applying it
  to each block yields
  \begin{equation*}
    \begin{pmatrix}
      \mathcal T_u(W_a^\dagger W_a)&\mathcal T_u(W_a)^\dagger\\
      \mathcal T_u(W_a)&\id
    \end{pmatrix}\geq0.
  \end{equation*}
  Evaluating its quadratic form on $(v,-\mathcal T_u(W_a)v)$, for
  arbitrary $v$, shows that
  $\mathcal T_u(W_a)^\dagger\mathcal T_u(W_a)
    \leq\mathcal T_u(W_a^\dagger W_a)$.
  Summing over $a$, we obtain
  \begin{align*}
    \mathcal T_u(W)^\dagger\mathcal T_u(W)
    &=\sum_a\mathcal T_u(W_a)^\dagger\mathcal T_u(W_a)\\
    &\leq\sum_a\mathcal T_u(W_a^\dagger W_a)
      =\mathcal T_u(W^\dagger W)\\
    &\leq\norm W_\infty^2\id.
  \end{align*}
  The last inequality follows from $W^\dagger W\leq\norm W_\infty^2\id$
  and positivity and unitality of $\mathcal T_u$.
  Hence $\norm{\mathcal T_u(W)}_\infty\leq\norm W_\infty$.
\end{proof}

\paragraph{Tracial entropy and Dirichlet forms.}
For $X\geq0$, define
\begin{equation}
\label{eq:tracial-entropy}
  \Ent(X)
  :=
  \tau(X\log X)-\tau(X)\log\tau(X),
\end{equation}
with $0\log0=0$.  For a reversible tracial QMS with generator $\cL$, the $L_1$ and $L_2$
Dirichlet forms are, respectively, defined as follows.
\begin{equation}
\label{eq:tracial-dirichlet-forms}
  \cE_{1,\cL}(X)
  :=
  \frac14\tau\!\left((\cL X)\log X\right),
  \qquad
  \cE_{2,\cL}(Y)
  :=
  \tau\!\left(Y^\dagger\cL(Y)\right),
\end{equation}
where $X>0$, $Y\in\mathcal M$, and $\tau=\tau_d$ for $\mathcal M=\Matrix_d$.

We need only the following tracial specialization of the quantum
Stroock--Varopoulos inequality.  The result of
\cite[Theorem~14]{BDR20} is more general: it treats arbitrary faithful
reference states using weighted noncommutative $L_p$ spaces.

\begin{theorem}[Stroock--Varopoulos inequality]
  \label{thm:quantum-stroock-varopoulos}
  Let $\mathcal T_u=e^{-u\cL}$ be a reversible tracial QMS on a
  finite-dimensional matrix algebra.  Then every positive-definite
  $X\in\mathcal M$ satisfies
  \begin{equation*}
    \cE_{1,\cL}(X)
    \geq
    \cE_{2,\cL}(\sqrt X)
    =
    \tau\!\left(\sqrt X\,\cL(\sqrt X)\right).
  \end{equation*}
\end{theorem}

The derivative of entropy along a tracial QMS has the following form;
see \cite[Section~3, Eq.~(27)]{GR22}.
\begin{prop}
  \label{prop:tracial-entropy-derivative}
  Let $\mathcal T_u=e^{-u\cL}$ be a reversible tracial QMS on a finite
  matrix algebra $\mathcal M$ with normalized trace $\tau$.  If
  $X\in\mathcal M$ is positive and $X_u:=\mathcal T_u(X)>0$ for $u\geq0$,
  then
  \begin{equation*}
    -\frac{\dd}{\dd u}\Ent(X_u)
    =
    4\cE_{1,\cL}(X_u).
  \end{equation*}
\end{prop}

\begin{proof}
  Trace preservation implies $\tau(X_u)=\tau(X)$ for every $u$.
  Differentiating this identity and using $\dot X_u=-\cL X_u$ shows that
  $\tau(\cL X_u)=0$.  Differentiation under the
  trace yields
  \[
    \frac{\dd}{\dd u}\Ent(X_u)
    =\tau(\dot X_u\log X_u)
    =-\tau\!\left((\cL X_u)\log X_u\right),
  \]
  which is the asserted identity.
\end{proof}

\begin{cor}
  \label{cor:depolarizing-stroock-varopoulos}
  Let $n\geq1$, and let $\cL(P_\alpha)=|\alpha|P_\alpha$ and
  $\mathcal T_u=e^{-u\cL}$ be the tensorized depolarizing generator and
  semigroup on $\Matrix_2^{\otimes n}$.  Every positive
  $X\in \Matrix_2^{\otimes n}$ satisfies
  \begin{equation*}
    \cE_{1,\cL}(X)
    \ge
    \cE_{2,\cL}(\sqrt X)
    =
    \tau_{2^n}\!\left(\sqrt X\,\cL(\sqrt X)\right).
  \end{equation*}
  Consequently, if $u\geq0$ and $X_u=\mathcal T_u(X)>0$, then
  \begin{equation*}
    -\frac{\dd}{\dd u}\Ent(X_u)
    \geq
    4\tau_{2^n}\!\left(\sqrt{X_u}\,\cL(\sqrt{X_u})\right).
  \end{equation*}
\end{cor}

\begin{proof}
  The depolarizing QMS is reversible by its Pauli diagonalization, so the
  first inequality is \cref{thm:quantum-stroock-varopoulos}.  The displayed
  equality follows from \eqref{eq:tracial-dirichlet-forms}.  For $X\geq0$, apply the
  inequality to $X+\varepsilon\id$ and let $\varepsilon\downarrow0$.  Applying
  \cref{prop:tracial-entropy-derivative} then proves the second inequality.
\end{proof}

\begin{prop}
  \label{prop:depolarizing-dirichlet-pauli-expansion}
  Let $n\geq1$.  For the tensorized depolarizing generator
  $\cL(P_\alpha)=|\alpha|P_\alpha$ on $\Matrix_2^{\otimes n}$, let
  $Y\in \Matrix_2^{\otimes n}$ have Pauli expansion
  $Y=\sum_{\alpha\in G^n}\widehat Y(\alpha)P_\alpha$.  Then
  \begin{equation*}
    \cE_{2,\cL}(Y)
    =
    \tau_{2^n}\!\left(Y^\dagger\cL(Y)\right)
    =
    \sum_{\alpha\in G^n}|\alpha|\,|\widehat Y(\alpha)|^2.
  \end{equation*}
\end{prop}

\begin{proof}
  The first equality follows from \eqref{eq:tracial-dirichlet-forms}.
  Since $\cL(P_\beta)=|\beta|P_\beta$, the two Pauli expansions needed below
  are
  \[
    \cL(Y)
    =
    \sum_{\beta\in G^n}|\beta|\,\widehat Y(\beta)P_\beta,
    \qquad
    Y^\dagger
    =
    \sum_{\alpha\in G^n}\overline{\widehat Y(\alpha)}P_\alpha^\dagger.
  \]
  Substituting these expansions into the Dirichlet form yields
  \begin{align*}
    \tau_{2^n}\!\left(Y^\dagger\cL(Y)\right)
    &=
    \sum_{\alpha,\beta\in G^n}
    \overline{\widehat Y(\alpha)}\,
    |\beta|\,\widehat Y(\beta)
    \tau_{2^n}(P_\alpha^\dagger P_\beta)\\
    &=
    \sum_{\alpha\in G^n}|\alpha|\,
    |\widehat Y(\alpha)|^2.
  \end{align*}
  The second equality follows from \cref{fact:pauli-orthonormal-basis}, which
  removes every term with
  $\alpha\neq\beta$.
\end{proof}

\end{document}